\documentclass[10pt,a4paper]{amsbook}
\usepackage{amssymb}
\usepackage{amsmath}
\usepackage{amsfonts}
\usepackage{graphicx}
\usepackage{a4wide}
\usepackage{microtype}
\usepackage{amsthm}
\usepackage{mathrsfs}
\usepackage{hyperref}
\usepackage{color}

\definecolor{darkred}{rgb}{0.8,0.1,0.1}
\hypersetup{
     colorlinks=true,         
     linkcolor=darkred,
     citecolor=blue,
}

\usepackage{tikz}

\usepackage{times}
\usepackage{comment}
\usepackage[margin=2.7cm]{geometry}
\usepackage[all,cmtip]{xy}

\usepackage{marvosym}

\begin{document}

\hyphenation{Schwarz-schild}
\hyphenation{phase-space}
\hyphenation{Symmetrie-reduktion}
\hyphenation{Deformations-quantisierung}

\newcommand{\bra}[1]{\langle #1 |}
\newcommand{\ket}[1]{|#1\rangle}
\newcommand{\braket}[2]{\langle #1 | #2 \rangle}
\newcommand{\me}[3]{\langle #1 | #2 | #3\rangle}
\newcommand{\com}[2]{[#1,#2]}
\newcommand{\starcom}[2]{[#1\stackrel{\star}{,}#2]}
\newcommand{\pb}[2]{\lbrace #1 , #2 \rbrace }
\newcommand{\f}[1]{\mathbf{#1}}
\newcommand{\pair}[2]{\langle #1,#2\rangle}
\newcommand{\spp}[2]{\bigl( #1 , #2 \bigr)}

\renewcommand{\arraystretch}{1.5}

\theoremstyle{plain}

\newtheorem{theo}{{\bf Theorem}}[chapter]

\newtheorem{lem}[theo]{{\bf Lemma}}

\newtheorem{propo}[theo]{{\bf Proposition}}

\newtheorem{cor}[theo]{{\bf Corollary}}

\theoremstyle{definition}

\newtheorem{defi}[theo]{{\bf Definition}}

\newtheorem{ex}[theo]{{\bf Example}}

\theoremstyle{remark}

\newtheorem{rem}[theo]{{\bf Remark}}

\renewcommand{\thepart}{\Roman{part}}

\def\nn{\nonumber}

\def\bbK{\mathbb{K}}
\def\bbR{\mathbb{R}}
\def\bbC{\mathbb{C}}
\def\bbN{\mathbb{N}}
\def\bbA{\mathbb{A}}
\def\bbB{\mathbb{B}}

\def\bfR{\mathsf{R}}
\def\bfK{\mathsf{K}}
\def\bfC{\mathsf{C}}

\def\MM{\mathcal{M}}
\def\AA{\mathcal{A}}
\def\FF{\mathcal{F}}

\def\MMM{\mathscr{M}}
\def\AAA{\mathscr{A}}

\def\Hom{\mathrm{Hom}}
\def\End{\mathrm{End}}
\def\Con{\mathrm{Con}}
\def\Sol{\mathrm{Sol}}
\def\Ker{\mathrm{Ker}}

\def\id{\mathrm{id}}
\def\supp{\mathrm{supp}}
\def\deg{\mathrm{deg}}
\def\dd{\mathrm{d}}
\def\vol{\mathrm{vol}_g}
\def\vols{\mathrm{vol}_\star}
\def\cnt{\mathrm{cnt}}
\def\top{\mathrm{top}}
\def\sc{\mathrm{sc}}
\def\dim{\mathrm{dim}}
\def\cop{\mathrm{cop}}
\def\op{\mathrm{op}}

\def\totimes{\widetilde{\otimes}}
\def\nab{\triangledown}
\def\ra{\triangleright}
\def\RA{\blacktriangleright}

\setlength{\unitlength}{1mm}
\parindent 4mm
\begin{titlepage}
\begin{center}
\vspace{20mm} ~\\
\begin{Huge}
\begin{center}
Noncommutative Gravity and\\
Quantum Field Theory on \\
Noncommutative Curved Spacetimes
\end{center}
\end{Huge}\vspace{3cm}

{\Large Dissertation zur Erlangung des } \vspace{3mm}\\ {\Large naturwissenschaftlichen Doktorgrades  }\vspace{3mm}\\ 
{\Large der Bayerischen Julius-Maximilians-Universit\"at W\"urzburg } ~\vspace{1cm}\\
\includegraphics[width=60mm]{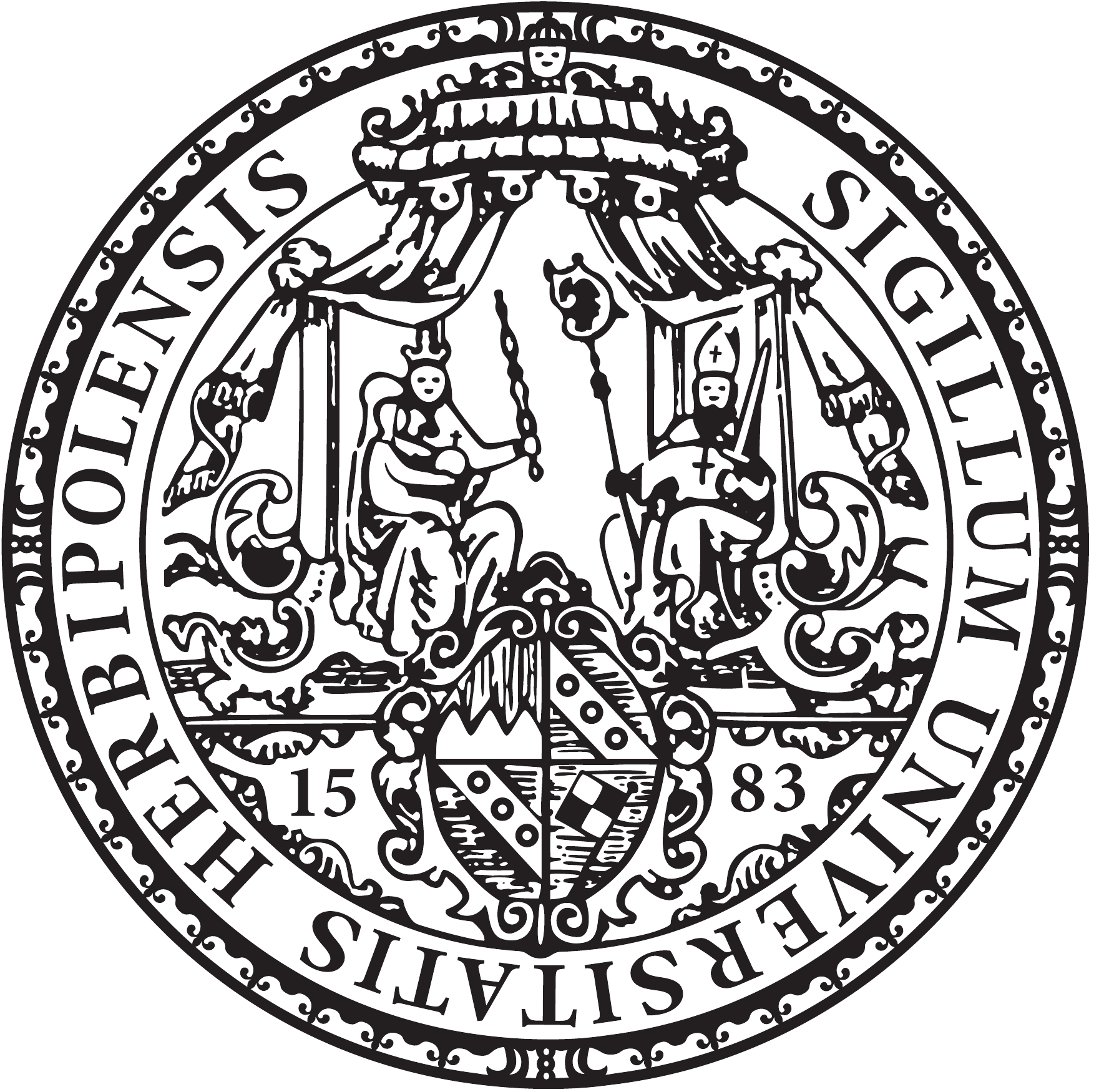} \vspace{1cm}\\
\begin{large}vorgelegt von\vspace{8mm} \\{\Large Alexander Schenkel}\vspace{8mm}\\ aus Hardheim\end{large}\\

\vspace{3cm}
{\Large W\"urzburg 2011}
\end{center}
\pagenumbering{roman}

\end{titlepage}

~~\thispagestyle{empty} \\ \newpage~~\vspace{3cm} \thispagestyle{empty} \\
\noindent {\large Eingereicht am: \quad 14.~Juni 2011}\vspace{0.5cm}\\
{\large bei der Fakult\"at f\"ur Physik und Astronomie }\vspace{1.5cm}\\

\noindent {\large 1.~Gutachter: \quad Prof.~Dr.~Thorsten Ohl }\vspace{0.5cm}\\
{\large 2.~Gutachter:  \quad Prof.~Dr.~Haye Hinrichsen }\vspace{0.5cm}\\
{\large 3.~Gutachter: \quad Prof.~Dr.~Peter Schupp~\,~(Jacobs University Bremen) }\vspace{0.5cm}\\
{\large der Dissertation }\vspace{1.5cm}\\

\noindent {\large 1.~Pr\"ufer: \quad  Prof.~Dr.~Thorsten Ohl }\vspace{0.5cm}\\
{\large 2.~Pr\"ufer: \quad Prof.~Dr.~Haye Hinrichsen }\vspace{0.5cm}\\
{\large 3.~Pr\"ufer: \quad Prof.~Dr.~Peter Schupp~\,~(Jacobs University Bremen) }\vspace{0.5cm}\\
{\large 4.~Pr\"ufer: \quad Prof.~Dr.~Thomas Trefzger }\vspace{0.5cm}\\
{\large im Promotionskolloquium}\vspace{1.5cm}\\

{\large \noindent Tag des Promotionskolloquiums: \quad 24.~Oktober 2011 }\vspace{0.5cm}\\

{\large \noindent Doktorurkunde ausgeh\"andigt am: \quad 27.~Oktober 2011 }
\newpage
~~\thispagestyle{empty}\newpage


\chapter*{Abstract}\setcounter{page}{1}
Over the past decades, noncommutative geometry has grown into an established field
in pure mathematics and theoretical physics. The discovery that noncommutative geometry emerges as a limit 
of quantum gravity and string theory has provided strong motivations to search for physics beyond the standard model 
of particle physics and also beyond Einstein's theory of general relativity within the realm of noncommutative geometries.
A very fruitful approach in the latter direction is due to Julius Wess and his group, which combines deformation 
quantization ($\star$-products) with quantum group methods.
The resulting gravity theory does not only include noncommutative effects of spacetime,
but it is also invariant under a deformed Hopf algebra of diffeomorphisms, generalizing the principle
of general covariance to the noncommutative setting.

The purpose of the first part of this thesis is to understand symmetry reduction
in noncommutative gravity, which then allows us to find exact solutions of the noncommutative Einstein equations. 
These are important investigations in order to capture the physical content of such theories
and to make contact to applications in e.g.~noncommutative cosmology and black hole physics.
We propose an extension of the usual symmetry reduction procedure, which is frequently applied to
the construction of exact solutions of Einstein's field equations, to noncommutative gravity and show that
this leads to preferred choices of noncommutative deformations of a given symmetric system. We classify
in the case of abelian Drinfel'd twists all consistent deformations of spatially
flat Friedmann-Robertson-Walker cosmologies and of the Schwarzschild black hole.
The deformed symmetry structure allows us to obtain exact solutions of 
the noncommutative Einstein equations in many of our models, for which the noncommutative metric field
coincides with the classical one.

In the second part we focus on quantum field theory on noncommutative curved spacetimes.
We develop a new formalism by combining methods from the algebraic approach to quantum field theory
with noncommutative differential geometry. The result is an algebra of observables for
scalar quantum field theories on a large class of noncommutative curved spacetimes.
A precise relation to the algebra of observables of the corresponding undeformed quantum field theory
is established.
We focus on explicit examples of deformed wave operators and find that there can be noncommutative corrections
even on the level of free field theories, which is not the case in the simplest example of the Moyal-Weyl deformed
Minkowski spacetime. The convergent deformation of simple toy-models is investigated and it is shown
that these quantum field theories have many new features compared to formal deformation quantization.
In addition to the expected nonlocality, we obtain that the relation between the deformed and the undeformed
quantum field theory is affected in a nontrivial way, leading to an improved behavior of the noncommutative quantum field theory
at short distances, i.e.~in the ultraviolet.

In the third part we develop elements of a more powerful, albeit more abstract, mathematical 
approach to noncommutative gravity.
The goal is to better understand global aspects of homomorphisms between and connections on noncommutative vector bundles, 
which are fundamental objects in the mathematical description of noncommutative gravity.
 We prove that all homomorphisms
and connections of the deformed theory can be obtained by applying a quantization isomorphism to undeformed homomorphisms
and connections. The extension of homomorphisms and connections to tensor products of modules
is clarified, and as a consequence we are able to add tensor fields of arbitrary type
to the noncommutative gravity theory of Wess et al.
As a nontrivial application of the new mathematical formalism we extend our studies of exact noncommutative
gravity solutions to more general deformations.


\chapter*{Zusammenfassung}
\"Uber die letzten Jahrzehnte hat sich 
die nichtkommutative Geometrie zu einem etablierten 
Teilgebiet der reinen Mathematik und der theoretischen Physik entwickelt.
Die Entdeckung, dass gewisse Grenzf\"alle der Quantengravitation
 und Stringtheorie zu nichtkommutativer Geometrie f\"uhren,
motivierte die Suche nach Physik jenseits des Standardmodells der Elementarteilchenphysik
und der Einstein'schen allgemeinen Relativit\"atstheorie im Rahmen von nichtkommutativen Geometrien.
Einen ergiebigen Ansatz zu letzteren Theorien, welcher Deformationsquantisierung (Sternprodukte) mit
Methoden aus der Theorie der Quantengruppen kombiniert, wurde von der Gruppe um Julius Wess entwickelt.
Die resultierende Gravitationstheorie ist nicht nur imstande nichtkommutative Effekte der Raumzeit
zu beschreiben, sondern sie erf\"ullt ebenfalls ein generalisiertes allgemeines Kovarianzprinzip,
welches durch eine deformierte Hopf Algebra von Diffeomorphismen beschrieben wird.

Gegenstand des ersten Teils dieser Dissertation ist es Symmetriereduktion im Rahmen von
nichtkommutativer Gravitation zu verstehen und damit exakte L\"osungen der nichtkommutativen Einstein'schen
Gleichungen zu konstruieren.
Diese Untersuchungen sind von gro\ss er Bedeutung um den physikalischen Inhalt dieser Theorien
herauszuarbeiten und den Kontakt zu Anwendungen, z.B.~im Rahmen nichtkommutativer Kosmologie und
Physik schwarzer L\"ocher, herzustellen.
Wir verallgemeinern die \"ubliche Methode der Symmetriereduktion, welche eine Standardtechnik
im Auffinden von L\"osungen der Einstein'schen Gleichungen ist,
auf nichtkommutative Gravitation.
Es wird gezeigt, dass unsere Methode zur nichtkommutativen Symmetriereduktion
f\"ur ein gegebenes symmetrisches System zu bevorzugten Deformationen
f\"uhrt.
F\"ur Abelsche Drinfel'd Twists klassifizieren wir alle
konsistenten Deformationen von r\"aumlich flachen Friedmann-Robertson-Walker Kosmologien
und des Schwarzschild'schen schwarzen Loches.
Aufgrund der deformierten Symmetriestruktur dieser Modelle
k\"onnen wir viele Beispiele von exakten L\"osungen der nichtkommutativen Einstein'schen Gleichungen
finden, bei welchen das nichtkommutative Metrikfeld mit dem klassischen \"ubereinstimmt.

Im Fokus des zweiten Teils sind Quantenfeldtheorien
auf nichtkommutativen gekr\"ummten Raumzeiten.
Dazu entwickeln wir einen neuen Formalismus, welcher algebraische Methoden
der Quantenfeldtheorie  mit nichtkommutativer Differentialgeometrie 
verkn\"upft.
Als Resultat unseres Ansatzes erhalten wir eine Observablenalgebra
f\"ur skalare Quantenfeldtheorien auf einer gro\ss en Klasse von nichtkommutativen
gekr\"ummten Raumzeiten.
Es wird eine pr\"azise Relation zwischen dieser Algebra und der Observablenalgebra
der undeformierten Quantenfeldtheorie hergeleitet.
Wir studieren ebenfalls explizite Beispiele von deformierten Wellenoperatoren
und finden, dass im Gegensatz zu dem einfachsten Modell des Moyal-Weyl deformierten Minkowski-Raumes,
im Allgemeinen schon die Propagation freier Felder durch die nichtkommutative Geometrie beeinflusst wird.
Die Effekte von konvergenten Deformationen werden in einfachen Spezialf\"allen untersucht,
und wir beobachten neue Aspekte in diesen Quantenfeldtheorien, welche sich in formalen Deformationen
nicht zeigten. Zus\"atzlich zu der erwarteten Nichtlokalit\"at finden wir,
dass sich die Beziehung zwischen der deformierten und der undeformierten Quantenfeldtheorie nichttrivial ver\"andert.
Wir beweisen, dass dies zu einem verbesserten Verhalten der nichtkommutativen Theorie
bei kurzen Abst\"anden, d.h.~im Ultravioletten, f\"uhrt.

Im dritten Teil dieser Arbeit entwickeln wir Elemente eines leistungsf\"ahigeren, jedoch abstrakteren,  mathematischen
Ansatzes zur Beschreibung der nichtkommutativen Gravitation.
Das Hauptaugenmerk liegt auf globalen Aspekten von Homomorphismen zwischen und Zusammenh\"angen auf nichtkommutativen Vektorb\"undeln, 
welche fundamentale Objekte in der mathematischen Beschreibung von nichtkommutativer Gravitation sind.
Wir beweisen, dass sich alle Homomorphismen und Zusammenh\"ange der deformierten Theorie
mittels eines Quantisierungsisomorphismus aus den undeformierten Homomorphismen und Zusammenh\"angen ableiten lassen.
Es wird ebenfalls untersucht wie sich Homomorphismen und Zusammenh\"ange auf Tensorprodukte von Moduln 
induzieren lassen. Das Verst\"andnis dieser Induktion erlaubt es uns die nichtkommutative Gravitationstheorie
von Wess et al.~um allgemeine Tensorfelder zu erweitern.
Als eine nichttriviale Anwendung des neuen Formalismus erweitern wir unsere Studien
zu exakten L\"osungen der nichtkommutativen Einstein'schen Gleichungen auf allgemeinere Klassen von Deformationen.


\tableofcontents

\numberwithin{equation}{chapter}

\pagenumbering{arabic}


\part*{Introduction and Outline}


\vspace{3mm}

\begin{center}
 {\bf Introduction}
\end{center}
\vspace{2mm}

Physics has gone through a number of major conceptual changes in the early twentieth century.
In particular, experiments in atomic physics revealed the quantum structure 
of nature at microscopic distances, and Einstein's theory of general relativity provided us with a deeper understanding
of space and time. Quantum mechanics was later successfully combined with special relativity, leading
to quantum field theory and eventually to the standard model of particle physics. However, the understanding of
how to combine general relativity with the concepts of quantum mechanics is not yet complete,
 and quantum gravity remains a very active field of research in mathematics, mathematical physics
and also phenomenology.
There are by now a number of serious candidates towards a theory of quantum gravity with
string theory, see e.g.~\cite{1075.81054,1075.81053}, and loop quantum gravity, see e.g.~\cite{1140.83005,1129.83004},
 being the most prominent examples. 
These theories focus on different key aspects one expects of quantum gravity:
String theory in particular on the unification of all fundamental interactions and loop quantum gravity on background
independence.
In addition to these two major frameworks, there are also other influential approaches like asymptotic safety \cite{Reuter:1996cp} 
and causal dynamical triangulations \cite{Ambjorn:1998xu}, which have already led to
interesting insights into the quantum nature of spacetime and might be able to guide future quantum gravity research.

Looking again at the two major conceptual changes in the early twentieth century one notices a puzzling feature:
Classical mechanics is described in terms of geometry of the phasespace, which is
a field in mathematics called Poisson geometry. On the other hand, quantum mechanics is described in terms
of noncommutative algebras generated by position and momentum operators satisfying
the canonical commutation relations $[\hat x^i,\hat p^j]=i\hbar\, \delta^{ij}\,\hat 1$. This noncommutativity 
and the resulting uncertainty relations are experimentally required and have been tested to a great precision.
The formulation of general relativity is based on (pseudo-) Riemannian geometry, 
which is, similarly to Poisson geometry, a subfield of differential geometry.
This means that, on the mathematical level, general relativity is comparable to classical mechanics,
because both are formulated in a geometric language, which does not include any quantum effects. 
Since the transition from classical to quantum mechanics
is a transition from geometry to noncommutative algebra, it is natural to ask the following question:
Can we come closer to a quantum theory of gravity by replacing the geometrical structures underlying general relativity
by noncommutative algebraic structures?

In mathematics, there are several examples of geometric structures which can be 
entirely described in algebraic terms.
Based on a seminal theorem by Gel'fand and Naimark \cite{0824.46060}, we can
equivalently describe topological spaces of a certain kind (locally compact Hausdorff spaces)
by {\it commutative} $C^\ast$-algebras. In addition, Serre \cite{0067.16201} and Swan \cite{0109.41601} 
have shown that vector bundles are equivalent to finitely generated and projective modules over these algebras.
The algebraic equivalent of Riemannian spin-geometry was found and investigated intensively
by Connes, see e.g.~\cite{Connes:1994yd}, and led to the definition of {\it commutative} spectral triples,
consisting of {\it commutative} algebras, Hilbert spaces and Dirac operators.
Note that all algebraic structures corresponding to classical geometries are commutative.
A natural generalization of classical geometry is thus obtained by allowing also for noncommutative algebras.
In this respect, a noncommutative topological space is a noncommutative $C^\ast$-algebra,
a noncommutative vector bundle is a finitely generated  and projective module over this algebra and
a noncommutative Riemannian spin-manifold is a noncommutative spectral triple.

Having understood the algebraic objects required to formulate noncommutative geometry
there is still the question which noncommutative algebra or which noncommutative vector bundle
we should choose in order to appropriately describe a certain physical situation.
Unfortunately, constructing a noncommutative theory from scratch
is in general very complicated, since, in contrast to classical geometry, we are often missing
physical intuition in the quantum case. An approach which turned out to be very fruitful in order to construct
noncommutative theories is quantization. In general terms, quantization is a set of rules (axioms)
how to associate to a commutative system a noncommutative one. 
A systematic approach to quantization, called deformation quantization, was 
developed in \cite{0377.53024,0377.53025}, see also Waldmann's book \cite{1139.53001} for an introduction.
 The starting point of this approach is a Poisson algebra, i.e.~a commutative algebra $A$ with a 
Poisson structure $\lbrace\cdot,\cdot\rbrace$, 
which is quantized by introducing a new {\it noncommutative} product, the $\star$-product.
This $\star$-product depends on a deformation parameter, $\hbar$ in the case of quantum mechanics,
and one demands that for $\hbar\to 0$ the $\star$-commutator reduces to the Poisson bracket
at leading order $\starcom{a}{b} = i\hbar \lbrace a,b\rbrace +\mathcal{O}(\hbar^2)$.

Motivated by the example of quantum mechanics we can start thinking about introducing a $\star$-product
on spacetime in order to quantize it. In this case a natural deformation parameter is given by the Planck length,
i.e.~the scale where we expect quantum effects of geometry to become relevant.
However, in contrast to classical mechanics, we did not yet observe a Poisson structure on spacetime.
Thus, it is not clear which Poisson tensor should be used to construct the $\star$-product,
or in other words, in which ``direction'' we should quantize.
In order to understand the basic features of a deformation let us assume some Poisson
tensor $\Theta^{\mu\nu}(x)$ on spacetime. The $\star$-commutator between coordinate functions
at leading order in the deformation parameter $\lambda$ reads $\starcom{x^\mu}{x^\nu}=
i\lambda\,\Theta^{\mu\nu}(x)+\mathcal{O}(\lambda^2)$. Thus, similar to the Planck cells
in quantum mechanics, there will be minimal areas in spacetime due to the associated coordinate uncertainty relations.
These minimal areas may depend on the position because of the $x$-dependence of the Poisson tensor.
As it has been shown in \cite{Doplicher199439,Doplicher:1994tu},
coordinate uncertainty relations are capable to limit the intrinsic resolution of spacetime,
such that black holes can not be produced in the process of sharp localization.
Moreover, the modified ultraviolet structure of spacetime immediately rises the hope to improve 
the mathematical description of physics, in particular
the ultraviolet divergences in quantum field theory and curvature singularities in
general relativity.

A possible explanation for the Poisson tensor and also the $\star$-product on spacetime is provided by string theory.
There it has been found that open strings ending on $D$-branes in a $B$-field background
$B_{\mu\nu}$ are subject to noncommutative geometry effects, see e.g.~\cite{Chu:1998qz,Schomerus:1999ug,Seiberg:1999vs}
and references therein. These effects manifest themselves on the level of string theory in terms of modified scattering amplitudes.
Taking the effective field theory limit, the $B$-field background still affects the physics on the brane,
and leads to a noncommutative Yang-Mills theory thereon.

The relation to string theory and physical motivations led phenomenologists to study in great detail possible effects
of noncommutative geometry in particle physics and cosmology.
An overview of the work done in particle physics can be obtained from
\cite{Hewett:2000zp,Hinchliffe:2002km,Jurco:2001rq,Calmet:2001na,Melic:2005fm,Melic:2005am,Ohl:2004tn,Alboteanu:2006hh,Alboteanu:2007bp}
 and for cosmology see e.g.~\cite{Lizzi:2002ib,Kim:2005tf,Akofor:2007fv,Akofor:2008gv,Koivisto:2010fk}, and references therein.
The model which was mostly used in these studies is the so-called canonical, or Moyal-Weyl, deformation,
where one assumes that $\starcom{x^\mu}{x^\nu}=i\,\lambda\,\Theta^{\mu\nu}$ is constant.
In contrast to commutative theories, the noncommutative ones displayed a violation of Lorentz invariance
in scattering amplitudes and preferred directions in the cosmic microwave background.

Assuming spacetime to be noncommutative, there is still the question
of how to describe gravitation. Besides using noncommutative metric fields, see 
e.g.~\cite{Aschieri:2005yw,Aschieri:2005zs,Kurkcuoglu:2006iw}, there are also
approaches based on hermitian metrics, see e.g.~\cite{Chamseddine:1992yx,Chamseddine:2000zu}, or vielbeins, see 
e.g.~\cite{Chamseddine:2003we,Aschieri:2009ky,Aschieri:2009mc} and references therein.
In addition to these rather conventional approaches, noncommutative geometry seems to provide
a natural mechanism for emergent gravity from noncommutative gauge theory and matrix models,
see \cite{Rivelles:2002ez,Yang:2006dk,Yang:2006hj,Steinacker:2007dq,Steinacker:2008ri,Steinacker:2010rh}.
For reviews on different approaches to noncommutative gravity see \cite{Szabo:2006wx,MullerHoissen:2007xy}.

In our work we are guided by the approach of Wess and his group to noncommutative gravity
\cite{Aschieri:2005yw,Aschieri:2005zs}. 
In this theory the symmetries of general relativity, i.e.~the diffeomorphisms,
are considered as the fundamental object and are deformed.
The generalization of the diffeomorphism symmetry is formulated in the language of Hopf algebras, a
 mathematical object which is suitable for studying quantizations of Lie groups or Lie algebras, 
see e.g.~\cite{Majid:1996kd,Kassel:1995xr} for an introduction.
A gravity theory is then constructed such that it transforms covariantly
under the deformed diffeomorphisms, which automatically results in noncommutative geometry.
Note that in noncommutative gravity we take into account quantum effects of the underlying manifold 
($\lambda$-deformation), but a phasespace quantization of the metric field ($\hbar$-deformation)
is not yet included. Thus, we expect noncommutative gravity to be a valid approximation of a full quantum gravity
theory, which should be quantized in $\hbar$ and $\lambda$, in configurations where quantum fluctuations
in the metric field are negligible.
We can also see noncommutative gravity in the following, more speculative, way.
Since the $\hbar$-quantization of general relativity is plagued by 
serious difficulties, like the perturbative nonrenormalizability,
the $\lambda$-deformation of the underlying manifold might be the missing ingredient
to improve the $\hbar$-quantization of the metric field.

In addition to noncommutative gravity, an interesting field of research is noncommutative quantum field theory.
The focus there is on quantum fields propagating on a fixed noncommutative spacetime
and the resulting effects.
Due to the minimal areas present in a noncommutative spacetime, one
expects improved mathematical properties of these quantum field theories in the ultraviolet, as well as
interesting and distinct new physical effects. Noncommutative quantum field theory comes in many different varieties,
in particular it was studied in a Euclidean and Lorentzian setting.
In the Euclidean setting, remarkable results were obtained by
Grosse and Wulkenhaar \cite{Grosse:2003aj,Grosse:2004yu,Grosse:2004by,Grosse:2009pa} 
and later also by Rivasseau and his group \cite{Rivasseau:2005bh,Gurau:2005gd,Disertori:2006nq}
after a long series of investigations.
It has been found that the $\Phi^4$-theory on the Moyal-Weyl space has interesting quantum properties,
if one includes an additional quadratic term in the action. In particular, the theory
is renormalizible to all orders in the perturbation theory and the infamous Landau pole is not present.
This is an improvement compared to the commutative $\Phi^4$-theory and
rises hope for obtaining a rigorously defined interacting $4$-dimensional Euclidean quantum field theory
by using noncommutative geometry methods. Even though there have been many attempts in this direction, 
see e.g.~the review \cite{Blaschke:2010kw}, similar results do not yet exist for noncommutative gauge theories.
In the Lorentzian case, a considerable amount of research has been done in order to understand 
perturbatively interacting quantum field theories.
The model mostly used for these studies is the Moyal-Weyl deformed Minkowski spacetime.
Different approaches to perturbation theory have been investigated, see e.g.~\cite{Bahns:2004mm,Zahn:2006mg}, 
and there was for a long time the hope that the infamous UV/IR-mixing
problem is not present in the Minkowski case. Recently, it was shown in \cite{Bahns:2010dx} and \cite{Zahn:2011bs} that the
UV/IR-mixing also occurs in the Hamiltonian and Yang-Feldman approach to noncommutative Minkowski quantum field theory,
even though the mechanism is different to the Euclidean case. These new results question the mathematical 
consistency of these approaches.

When going from the Minkowski spacetime to more general Lorentzian spacetimes, in particular
curved ones, the number of approaches to noncommutative quantum field theory reduces considerably.
An interesting formalism for deformed quantum field theory in the language
of algebraic quantum field theory was developed by Dappiaggi, Lechner and Morfa-Morales
\cite{Dappiaggi:2010hc}, which is based on the concept of warped convolutions previously studied
by Lechner and collaborators \cite{Grosse:2007vr,Grosse:2008dk,Buchholz:2010ct}.
In this approach an algebraic quantum field theory, described by a net of observable algebras,
is deformed by using methods similar to those developed by Rieffel \cite{0798.46053}.
A different approach, focusing primarily on spectral geometry, was investigated by Paschke and Verch
\cite{Paschke:2004xf}. In this thesis we will present a third approach to quantum field theory
on noncommutative curved spacetimes developed by myself and collaborators 
\cite{Ohl:2009qe,Schenkel:2010sc,Schenkel:2010jr,Schenkel:2011gw}, 
which is formulated in close contact to the noncommutative gravity theory of Wess et al.


\vspace{3mm}

\begin{center}
 {\bf Outline of this thesis}
\end{center}
\vspace{2mm}

This thesis consists of three main parts, focusing on different, but strongly connected, 
aspects of noncommutative geometry, gravity and quantum field theory. 
The purpose of this section is to provide a broad and nontechnical overview of the content of all three parts.

\subsection*{Part \ref{part:ncg}}
We are going to focus on physical aspects of the noncommutative 
gravity theory of Wess and his group. Even though this theory was already developed in
2006, no results on its application to physical situations, for example
cosmology or black hole physics, have been published until recently\footnote{
While working on this subject \cite{Ohl:2009pv} I became aware of earlier investigations by Schupp and Solodukhin
on noncommutative black holes, which were presented at conferences in 2007 \cite{Schupp:2007} 
and later published in 2009 \cite{Schupp:2009pt}.}.
A detailed investigation of explicit models 
is an essential step to capture the physical content of the noncommutative gravity theory,
and it is therefore a very important task for future developments in this field.
This provides the motivation for Part \ref{part:ncg}.
In Chapter \ref{chap:basicncg} we review the noncommutative gravity
theory under consideration \cite{Aschieri:2005yw,Aschieri:2005zs}. Since this theory makes use of mathematical
methods of Hopf algebra theory, we first give a gentle introduction to Hopf algebras and
their Drinfel'd twist deformations using explicit examples.  
Based on this we explain how spacetime, as well as its differential geometry,
can be deformed, leading to examples of noncommutative geometries.
We equip these noncommutative spacetimes with covariant derivatives, define their curvature and
eventually a noncommutative version of Einstein's equations, which are the 
underlying dynamical equations of noncommutative gravity.
After this introductory chapter we present in Chapter \ref{chap:symred} our approach to symmetry reduction
in noncommutative gravity \cite{Ohl:2008tw}. 
Remember that in classical general relativity, Friedmann-Robertson-Walker cosmologies and Schwarzschild 
black holes are characterized as configurations which are invariant under a certain symmetry group (or Lie algebra).
We generalize this definition to theories covariant under deformed Hopf algebra symmetries,
making use of the concept of infinitesimal deformed isometries, described by almost quantum Lie algebras.
As an application we classify all possible abelian twist deformations of spatially flat Friedmann-Robertson-Walker
cosmologies and Schwarzschild black holes satisfying our axioms of deformed symmetry reduction.
The physical content of these models is briefly discussed and we find a particularly interesting
cosmological model, which is invariant under all classical rotations, and a deformed black hole
model, which is invariant under all classical black hole symmetries.
In Chapter \ref{chap:ncgsol} we take the natural next step and construct exact solutions
of the noncommutative Einstein equations within our models. 
The work we present was published in \cite{Ohl:2009pv} and appeared at the same time 
as the related articles by Schupp and Solodukhin \cite{Schupp:2009pt} and Aschieri and Castellani \cite{Aschieri:2009qh},
all of them focusing on different aspects of exact noncommutative gravity solutions.
The main result of our work, which was also found in \cite{Schupp:2007,Schupp:2009pt,Aschieri:2009qh},
is that the classical metric field satisfies the noncommutative Einstein equations exactly if the
deformation is generated by sufficiently many Killing vector fields. We show that this condition is
fulfilled for most of our physically viable models of noncommutative cosmologies and black holes, thus
leading to a large class of explicit physics examples. In particular, we show that also the isotropically deformed
cosmological model solves the noncommutative Einstein equations exactly in presence of a cosmological constant.
Even though the metric field for these configurations does not receive noncommutative corrections,
the underlying manifold is quantized, leading to distinct physical effects which will be discussed.
We conclude Part \ref{part:ncg} by pointing out open problems in noncommutative gravity, which are
the motivation for the developments described in Part \ref{part:math}.

\subsection*{Part \ref{part:qft}}
After the discussion of noncommutative background spacetimes in Part \ref{part:ncg}
we focus in Part \ref{part:qft} on noncommutative quantum field theory.
This is an important step towards extracting physical observables in noncommutative
cosmology and black hole physics, for example the two-point correlation function
yielding information on the cosmological power spectrum or the Hawking radiation.
Since noncommutative quantum field theory is usually studied on the Moyal-Weyl deformed
Minkowski spacetime, the analysis of our models requires two generalizations: 
Curved spacetimes and more general deformations. In other words, we have
to develop a formalism for quantum field theory on noncommutative curved spacetimes.
In order to fix notation we first review in Chapter \ref{chap:qftbas}
the algebraic approach to quantum field theory on commutative curved spacetimes,
which has turned out to be very fruitful, see e.g.~\cite{Wald:1995yp,Bar:2007zz,Bar:2009zzb}.
We present our approach to quantum field theory on noncommutative curved spacetimes  \cite{Ohl:2009qe} 
in Chapter \ref{chap:qftdef}, which combines algebraic quantum field theory methods with noncommutative differential geometry. 
The result of this construction is a deformed algebra of observables for a scalar quantum field theory
on a large class of deformed curved spacetimes. In Chapter \ref{chap:qftcon} we explore
mathematical properties of our approach to quantum field theory on noncommutative
curved spacetimes and in particular prove that each deformed quantum field theory can be mapped bijectively to
an undeformed one \cite{Schenkel:2011gw}. Chapter \ref{chap:qftapp} is devoted to explicit examples
of field and quantum field theories on noncommutative curved spacetimes. We present
examples of deformed wave operators on noncommutative Minkowski, de Sitter, Schwarzschild and anti-de Sitter spacetimes
\cite{Schenkel:2010sc}.
We study in detail the explicit construction of a scalar quantum field theory on
the isotropically deformed de Sitter spacetime, which is a nontrivial step towards physical applications in cosmology.
As two more applications we focus on homothetic Killing deformations, yielding 
simple examples of exactly treatable models \cite{Schenkel:2010jr}, and we present a new perturbatively 
interacting quantum field theory on a nonstandard deformed Euclidean space \cite{Schenkel:2010zi}, which 
shows remarkable similarities to the recently studied Ho{\v r}ava-Lifshitz theories \cite{Horava:2009uw}
and has improved quantum properties.
We close this part with a discussion of open problems in quantum field theory on noncommutative curved spacetimes.

\subsection*{Part \ref{part:math}}
In this part we focus on mathematical aspects of noncommutative geometry, which
are based on ongoing work with Paolo Aschieri [that appeared after finishing the thesis in \cite{AlexPaolo}].
The main motivation for these studies comes from the open problems
in noncommutative gravity and quantum field theory, which have shown that in particular
metric fields and covariant derivatives are not yet completely understood in the noncommutative setting.
To explain the content of this part we remind the reader that in noncommutative geometry
spacetime is described by a noncommutative algebra $A$ and a vector bundle by a (finitely generated and projective) 
module $V$ over $A$.
We consider the situation where we have an action of a Hopf algebra $H$ on $A$ and $V$.
This is a generalization of the setting we encounter in the noncommutative gravity theory of Part \ref{part:ncg}, 
where the Hopf algebra
$H$ describes the deformed diffeomorphisms, $A$ the quantized functions on spacetime
and $V$ the quantized vector fields or differential forms.
After fixing the notation in Chapter \ref{chap:prelim} and exploring some technical
aspects of Drinfel'd twist deformations in Chapter \ref{chap:HAdef}, we focus in Chapter \ref{chap:modhom}
on module endomorphisms and homomorphisms. We show that every module endomorphism on $V$
can be quantized to yield a module endomorphism on the twist quantized module $V_\star$,
and even more that every module endomorphism on $V_\star$ can be obtained in this way. 
Thus, there is an isomorphism between the quantized and unquantized module endomorphisms.
We extend the results to homomorphisms between two modules.
As a direct consequence, we find that the quantized dual module
is isomorphic to the dual quantized module, meaning that there are no ambiguities
in considering duals. 
We conclude this chapter by studying the extension of module homomorphisms
to tensor products of modules, i.e.~tensor fields. 
In Chapter \ref{chap:con} we investigate covariant derivatives (more precisely connections) in
noncommutative geometry.
We consider connections on the module $V$ satisfying the right Leibniz rule 
and provide a quantization prescription to obtain connections on the deformed module $V_\star$.
As in case of module homomorphisms, this quantization map is an isomorphism, meaning that there
is a one-to-one correspondence between the quantized and unquantized connections.
We show that for quasi-commutative algebras and bimodules\footnote{
An algebra or bimodule is said to be quasi-commutative, if it is commutative up to the action of an $R$-matrix, see Chapter
\ref{chap:modhom}, Section \ref{sec:modhomqc} for details.
} we can extend connections canonically to tensor products of modules.
This is exactly the situation we face in noncommutative gravity when we want to extend the 
connection to tensor fields.
The curvature and torsion of connections is studied in Chapter \ref{chap:curvature}.
In Chapter \ref{chap:ncgmath} we apply our formalism to reinvestigate exact noncommutative gravity
solutions. 
In contrast to the investigations based on local coordinate patches in Chapter \ref{chap:ncgsol}, 
we can now study solutions of the noncommutative Einstein equations on a global level.
In particular, we are able to extend the known results of \cite{Schupp:2007,Schupp:2009pt,Ohl:2009pv,Aschieri:2009qh}, 
see also Chapter \ref{chap:ncgsol}, to a larger class of Drinfel'd twists.
We conclude in Chapter \ref{chap:outlookmath} by giving an outlook to 
further interesting applications that can be studied within our formalism
and point out open issues which remain to be solved for completing the construction of a noncommutative
theory of gravity.


\part{\label{part:ncg}Noncommutative Gravity}


\chapter{\label{chap:basicncg}Basics}
In this chapter we give an introduction to the noncommutative gravity theory of Wess and his group
 \cite{Aschieri:2005yw,Aschieri:2005zs}.


\section{The Hopf algebra of diffeomorphisms}
Let $\MM$ be an $N$-dimensional smooth manifold and let $\Xi$ be the space
 of complex and smooth vector fields on $\MM$. Locally, there exists a basis $\lbrace\partial_\mu\in\Xi:\mu=1,\dots,N \rbrace$,
such that every vector field $v\in\Xi$ can be written as $v=v^\mu(x)\partial_\mu$, where $v^\mu(x)\in C^\infty(\MM)$, for all
$\mu$, are the coefficient functions. The space of vector fields can be naturally equipped with a Lie bracket, i.e.~an
antisymmetric $\bbC$-bilinear map $[\cdot,\cdot]:\Xi\times\Xi\to \Xi$, which satisfies the Jacobi identity.
Locally, the Lie bracket reads
\begin{flalign}
 [v,w]= \bigl(v^\mu(x)\partial_\mu w^\nu(x) - w^\mu(x)\partial_\mu v^\nu(x)\bigr)\partial_\nu~,
\end{flalign}
for all $v,w\in\Xi$. Thus, $\bigl(\Xi,[\cdot,\cdot]\bigr)$ forms a complex Lie algebra.

The Lie algebra of vector fields $\bigl(\Xi,[\cdot,\cdot]\bigr)$ plays an important role in differential geometry,
namely it describes the infinitesimal diffeomorphisms of $\MM$. The action of $\bigl(\Xi,[\cdot,\cdot]\bigr)$
on tensor fields is via the Lie derivative $\mathcal{L}$. Note that the Lie derivative and the Lie bracket are
compatible, i.e.~$\mathcal{L}_v\circ \mathcal{L}_w - \mathcal{L}_w\circ \mathcal{L}_v = \mathcal{L}_{[v,w]}$, for 
all $v,w\in \Xi$.

Let us point out two important operations one always has in mind when dealing with Lie algebras.
These observations are essential to understand the step how to go over from Lie algebras to Hopf algebras.
We will discuss only the case of the Lie algebra $\bigl(\Xi,[\cdot,\cdot]\bigr)$, even though the same statements
hold true for every Lie algebra $\bigl(\mathfrak{g},[\cdot,\cdot]\bigr)$.
Firstly, note that for each vector field $v\in\Xi$, which we interpret as an infinitesimal diffeomorphism,
there is the inverse infinitesimal diffeomorphism $v_\text{inv}=-v\in\Xi$.
Secondly, having a product of representations, e.g.~a tensor product $\tau\otimes \tau^\prime$ of two tensor fields
$\tau,\tau^\prime$, we can apply the Leibniz rule 
\begin{flalign}
\label{eqn:leibnizsimple}
\mathcal{L}_v\bigl(\tau\otimes\tau^\prime\bigr)=\mathcal{L}_v(\tau)\otimes\tau^\prime + \tau\otimes \mathcal{L}_v(\tau^\prime)~, 
\end{flalign}
for all $v\in\Xi$. Let us introduce also a third operation, which at the moment should be interpreted
as a normalization condition. We define a map $\epsilon: \Xi\to \bbC\,,~v\mapsto \epsilon(v)=0$, which
associates to all vector fields the number zero.

From the vector space $\Xi$ we can always construct the free associative and unital algebra $\mathcal{A}_\text{free}$.
Elements of $\mathcal{A}_\text{free}$ are finite sums of finite products of vector fields and the unit element $1$.
In order to encode information on the Lie algebra structure of $\bigl(\Xi,[\cdot,\cdot]\bigr)$,
we consider the ideal $\mathcal{I}$ generated by the elements $v\,w-w\,v-[v,w]$, for all $v,w\in\Xi$.
The {\it universal enveloping algebra} of the Lie algebra $\bigl(\Xi,[\cdot,\cdot]\bigr)$ is then defined
to be the factor algebra $U\Xi := \mathcal{A}_\text{free}/\mathcal{I}$.
Provided a representation, say tensor fields, of the Lie algebra $\bigl(\Xi,[\cdot,\cdot]\bigr)$, 
we can extend it to a left representation of $U\Xi$ by defining $\mathcal{L}_{\xi\,\eta}=\mathcal{L}_{\xi}\circ \mathcal{L}_\eta$,
for all $\xi,\eta\in U\Xi$, and $\mathcal{L}_1 = \id$. The latter definition allows us to interpret
 $1$ as the trivial diffeomorphism.

We now implement the three additional operations we have for the Lie algebra into $U\Xi$, starting with
the Leibniz rule. Note that (\ref{eqn:leibnizsimple}) gives us a prescription of how to
transform products of representations. In a more abstract language, not making use of the representation,
the information of (\ref{eqn:leibnizsimple}) can be encoded into a $\bbC$-linear map $\Delta:U\Xi\to U\Xi\otimes U\Xi$,
which on the generators reads
\begin{flalign}
  \Delta(v) = v\otimes 1 +1\otimes v~,\quad \Delta(1)=1\otimes 1~,
\end{flalign}
for all $v\in\Xi$. We can extend this map to $U\Xi$ by demanding multiplicativity $\Delta(\xi\,\eta)=\Delta(\xi)\,\Delta(\eta)$,
for all $\xi,\eta\in U\Xi$. The product on $U\Xi\otimes U\Xi$ is given by $(\xi\otimes \eta)\,(\xi^\prime\otimes \eta^\prime)=
\xi\,\xi^\prime\otimes \eta\,\eta^\prime$. We can easily check that this definition is consistent with the ideal $\mathcal{I}$
\begin{flalign}
 \nn \Delta(v\,w-w\,v) &= \bigl(v\otimes 1 + 1\otimes v\bigr)\,\bigl(w\otimes 1 + 1\otimes w\bigr) - (v\leftrightarrow w)\\
 \nn &=v\,w\otimes 1 + v\otimes w + w\otimes v + 1\otimes v\,w - (v\leftrightarrow w)\\
 & = [v,w]\otimes 1 + 1\otimes [v,w] = \Delta([v,w])~,
\end{flalign}
for all $v,w\in\Xi$.

Next, we implement the inverse operation $v_\text{inv}=-v\in\Xi$. Again on a more abstract level, we are looking for
a $\bbC$-linear map $S:U\Xi\to U\Xi$, which on the generators gives $S(v)=-v$, for all $v\in\Xi$, and $S(1)=1$,
since $1$ is the trivial diffeomorphism. Since we want to interpret $S$ as a map giving the inverse
of an element $\xi\in U\Xi$, it is natural to extend it to $U\Xi$ as an antimultiplicative map, 
i.e.~$S(\xi\,\eta)=S(\eta)\,S(\xi)$, for all $\xi,\eta\in U\Xi$. We can check that $S$ defined like this
is compatible with the ideal $\mathcal{I}$
\begin{flalign}
\nn S(v\,w-w\,v) &= S(v\,w)-S(w\,v) = S(w)\,S(v) - S(v)\,S(w) \\
 &= w\,v - v\,w = -[v,w] = S([v,w])~,
\end{flalign}
for all $v,w\in\Xi$.

It remains to extend the normalization $\epsilon$ to $U\Xi$. We define the $\bbC$-linear map $\epsilon:U\Xi\to\bbC$
on the generators by $\epsilon(v)=0$, for all $v\in\Xi$, $\epsilon(1)=1$, and extend it to $U\Xi$
multiplicatively. This definition is consistent with the ideal $\mathcal{I}$
\begin{flalign}
 \epsilon(v\,w-w\,v) = \epsilon(v)\epsilon(w)- \epsilon(w)\epsilon(v) = 0 = \epsilon([v,w])~,
\end{flalign}
 for all $v,w\in\Xi$.

Let us summarize this construction: Starting from the Lie algebra $\bigl(\Xi,[\cdot,\cdot]\bigr)$
we have constructed the universal enveloping algebra $U\Xi$. The intuitive notions of Leibniz rule, inverse and normalization,
which we have on the Lie algebra, were encoded on the level of $U\Xi$ in terms of $\bbC$-linear maps
$\Delta:U\Xi\to U\Xi\otimes U\Xi$, $S: U\Xi\to U\Xi$ and $\epsilon:U\Xi\to \bbC$.
While $\Delta$ and $\epsilon$ are multiplicative maps, the map $S$ associated to inversion is antimultiplicative.
The object which is of interest in the following is the quintuple 
$H=\bigl(U\Xi,\mu,\Delta,\epsilon,S\bigr)$, where $\mu$ denotes the
multiplication map in $U\Xi$, which above was written simply as juxtaposition.

The object $H$ we have derived above from physical considerations is a structure which is well-known
in mathematics, namely a {\it Hopf algebra}. We refer to Part \ref{part:math} for a mathematical definition
of Hopf algebras and we will continue in this section with our nontechnical treatment. 
Roughly speaking, a Hopf algebra is an algebra together with three maps $\Delta$, $S$ and $\epsilon$ as above,
which satisfy certain compatibility conditions. 
The map $\Delta$ is called the {\it coproduct}, $\epsilon$ the {\it counit} and $S$ the {\it antipode}. 
We shall now show that all these conditions hold true
for our explicit example $H=\bigl(U\Xi,\mu,\Delta,\epsilon,S\bigr)$. This proves that we are indeed
dealing with a Hopf algebra, which will be of great importance later in this part.
For $H$ to be a Hopf algebra, the following three conditions have to hold true, for all $\xi\in H$,
\begin{subequations}
\label{eqn:basichopfcond}
 \begin{flalign}
  \bigl(\Delta\otimes\id\bigr)\Delta(\xi) &= \bigl(\id\otimes \Delta\bigr)\Delta(\xi) ~,\\
  \bigl(\epsilon\otimes \id\bigr)\Delta(\xi) &= \xi = \bigl(\id\otimes \epsilon\bigr)\Delta(\xi)~,\\
  \mu\Bigl( \bigl(S\otimes \id\bigr)\Delta(\xi)\Bigr) &= \epsilon(\xi)\,1=\mu\Bigl( \bigl(\id\otimes S\bigr)\Delta(\xi)\Bigr)~,
 \end{flalign}
\end{subequations}
where $\mu(\xi\otimes\eta)=\xi\,\eta$ denotes the multiplication map. In order to check the conditions (\ref{eqn:basichopfcond})
we introduce a convenient notation according to Sweedler: For any $\xi\in H$ we write for the coproduct
$\Delta(\xi)=\xi_1\otimes \xi_2$ (sum understood). In this notation the conditions (\ref{eqn:basichopfcond}) read
\begin{subequations}
\label{eqn:basichopfcondshort}
\begin{flalign}
 \xi_{1_1}\otimes \xi_{1_2}\otimes \xi_2 &= \xi_1\otimes \xi_{2_1}\otimes \xi_{2_2}~,\\
 \epsilon(\xi_1)\,\xi_2 &= \xi = \xi_1\,\epsilon(\xi_2)~,\\
 S(\xi_1)\,\xi_2 &= \epsilon(\xi)\,1 = \xi_1\, S(\xi_2)~.
\end{flalign}
\end{subequations}
We first check these conditions on the level of the generators of $H$. For $\xi=1$ the conditions trivially hold true.
For $\xi=u\in\Xi$ the left hand side of the first condition reads
\begin{flalign}
 u_{1_1}\otimes u_{1_2}\otimes u_2 = u_1\otimes u_2 \otimes 1 + 1_1\otimes 1_2\otimes u = u\otimes 1\otimes 1+1\otimes u\otimes 1
+1\otimes 1\otimes u~.
\end{flalign}
Evaluating the right hand side we obtain the same result, thus the first condition holds true.
The second condition also holds for all generators $u\in \Xi$
\begin{flalign}
 \epsilon(u_1)\,u_2 = \epsilon(u)\,1 + \epsilon(1)\,u = u = u_1\,\epsilon(u_2)~.
\end{flalign}
Analogously, the third condition holds true
\begin{flalign}
 S(u_1)\,u_2 = S(u)\,1 + S(1)\,u = -u+u =0 =\epsilon(u)\,1= u_1\,S(u_2)~.
\end{flalign}
In order prove that the conditions (\ref{eqn:basichopfcond}) are satisfied for generic elements of $H$ it is sufficient
to show that they hold true for the product $\xi\,\eta\in H$, provided they hold for the individual $\xi,\eta\in H$.
In the notation of (\ref{eqn:basichopfcondshort}) we find
\begin{subequations}
\begin{flalign}
 \nn (\xi\,\eta)_{1_1}\otimes (\xi\,\eta)_{1_2}\otimes (\xi\,\eta)_2 &= \xi_{1_1}\,\eta_{1_1}\otimes \xi_{1_2}\,\eta_{1_2}\otimes \xi_2\,\eta_2\\
\nn&= \bigl(\xi_{1_1}\otimes\xi_{1_2}\otimes\xi_2\bigr)\,\bigl(\eta_{1_1}\otimes \eta_{1_2}\otimes \eta_2\bigr)\\
\nn&= \bigl(\xi_{1}\otimes\xi_{2_1}\otimes\xi_{2_2}\bigr)\,\bigl(\eta_{1}\otimes \eta_{2_1}\otimes \eta_{2_2}\bigr)\\
&= \xi_{1}\,\eta_{1}\otimes \xi_{2_1}\,\eta_{2_1}\otimes \xi_{2_2}\,\eta_{2_2}= (\xi\,\eta)_{1}\otimes (\xi\,\eta)_{2_1}\otimes (\xi\,\eta)_{2_2}~,
\end{flalign}
\begin{flalign}
 \epsilon((\xi\,\eta)_1)\,(\xi\,\eta)_2 = \epsilon(\xi_1\,\eta_1)\,\xi_2\,\eta_2 = \epsilon(\xi_1)\epsilon(\eta_1)\, \xi_2\,\eta_2 =\epsilon(\xi_1)\,\xi_2\,\epsilon(\eta_1)\,\eta_2
=\xi\,\eta~,
\end{flalign}
and
\begin{flalign}
 S((\xi\,\eta)_1)\,(\xi\,\eta)_2 = S(\xi_1\,\eta_1)\,\xi_2\,\eta_2 = S(\eta_1)\,S(\xi_1)\,\xi_2\,\eta_2 = \epsilon(\xi)\,\epsilon(\eta)\,1 = \epsilon(\xi\,\eta)\,1~.
\end{flalign}
\end{subequations}
The right hand sides of the second and third condition (\ref{eqn:basichopfcondshort}) are shown analogously.

The calculations performed above lead us to the following conclusion: We can associate to the Lie algebra of diffeomorphisms
$\bigl(\Xi,[\cdot,\cdot]\bigr)$ a Hopf algebra $H=\bigl(U\Xi,\mu,\Delta,\epsilon,S\bigr)$. This Hopf algebra will be called
the {\it Hopf algebra of diffeomorphisms}. It includes information on the Leibniz rule (via the coproduct $\Delta$)
and on the inverse of a diffeomorphism (via the antipode $S$), and further has a normalization operation (via the counit $\epsilon$).
Vice versa, we can extract the Lie algebra $\bigl(\Xi,[\cdot,\cdot]\bigr)$ from $H$ as follows:
As a vector space, $\Xi$ is isomorphic to the space of all elements $\xi\in H$ with coproduct
\begin{flalign}
\label{eqn:lgcoprod}
\Delta(\xi) = \xi\otimes 1 + 1\otimes \xi~.
\end{flalign}
This vector space can be equipped with a Lie bracket by employing the commutator
$[\xi,\eta]=\xi\,\eta -\eta\,\xi$, for all $\xi,\eta\in H$. The resulting Lie algebra is isomorphic
to $\bigl(\Xi,[\cdot,\cdot]\bigr)$.


\section{Deformed diffeomorphisms}
The Hopf algebra of  diffeomorphisms $H$, as well as every Hopf algebra constructed from a Lie algebra
along the lines presented above, has a particular feature: Taking the coopposite coproduct $\Delta^\cop(\xi)=\xi_2\otimes \xi_1$
of a generic element $\xi\in H$ agrees with the coproduct itself, i.e.~$\Delta=\Delta^\cop$. 
This property is called {\it cocommutativity} and correspondingly the Hopf algebra $H$ is called a {\it cocommutative
Hopf algebra}. From the point of view of Hopf algebra theory this is a very special feature, which allows for 
generalizations. As we will point out in this section, deformations of $H$ will in general not be cocommutative anymore,
but the coopposite coproduct will be equal to the coproduct up to conjugation by an element $R\in H\otimes H$, called 
the universal $R$-matrix. The same $R$-matrix will also appear in the commutation relations of two quantized functions,
thus it encodes the noncommutative structure of spacetime.

In order to deform $H$ we have to introduce a deformation parameter $\lambda$, which in physical situations
shall be related to the Planck length. For our investigations we treat $\lambda$ as a formal parameter, which can be seen
as a perturbative approach. Nonperturbative, i.e.~convergent, deformations are mathematically much more involved and rare, such
that for an investigation of the leading effects of our deformations using a formal approach
seems reasonable. As a drawback of the formal approach we have to extend the Hopf algebra $H$ and the complex numbers
$\bbC$ by formal powers of $\lambda$, denoted by $H[[\lambda]]$ and $\bbC[[\lambda]]$, respectively. Elements of $H[[\lambda]]$
are given by formal power series $\xi=\sum_{n=0}^\infty\lambda^n\,\xi_{(n)}$, where $\xi_{(n)}\in H$ for all $n\geq 0$.
The zeroth order describes the classical part and higher orders describe possible corrections due to the deformation.
The extension $H[[\lambda]]$ is again a Hopf algebra: The algebra structure is given by
\begin{flalign}
 \xi+\eta = \sum\limits_{n=0}^\infty \lambda^n\,(\xi_{(n)}+\eta_{(n)})~,\quad \xi\,\eta=\sum\limits_{n=0}^\infty\lambda^n
\sum\limits_{m+k=n} \xi_{(m)}\,\eta_{(k)}~,
\end{flalign}
and the maps $\Delta,\epsilon,S$ are defined componentwise, e.g.~
\begin{flalign}
 \Delta(\xi)=\sum\limits_{n=0}^\infty\lambda^n\,\Delta(\xi_{(n)})~.
\end{flalign}
For mathematical details on formal power series, which we do not require for the present part,
 we refer to the Appendix \ref{app:basicsdefq}.
In the remaining part the formal power series extension will be implicitly understood and we drop
$[[\lambda]]$ for notational convenience. This is a typical convention in the physics literature.

We now consider deformations of the Hopf algebra $H$.
A {\it Drinfel'd twist} is an invertible element $\mathcal{F}\in H\otimes H$ satisfying the following two conditions
\begin{subequations}
\label{eqn:twistcondbas}
\begin{flalign}
 &\mathcal{F}_{12}\,(\Delta\otimes\id)\mathcal{F} = \mathcal{F}_{23}\, (\id\otimes\Delta)\mathcal{F}~,\\
 &(\epsilon\otimes \id)\mathcal{F} = 1 = (\id\otimes \epsilon)\mathcal{F}~,
\end{flalign}
\end{subequations}
where $\mathcal{F}_{12}=\mathcal{F}\otimes 1$ and $\mathcal{F}_{23}=1\otimes\mathcal{F}$.
We additionally demand that $\mathcal{F} = 1\otimes 1 +\mathcal{O}(\lambda)$ in order to leave
the zeroth order unchanged, which is reasonable since we interpret the zeroth order as the classical part.
As an aside, this object is precisely a normalized $2$-cocycle of the Hopf algebra $H$.
From the physics perspective, the two conditions (\ref{eqn:twistcondbas}) have a 
direct consequence for the noncommutative geometry we are going to construct later:
The first condition ensures that the $\star$-product is associative,
while the second leads to trivial $\star$-multiplications of functions with the unit element.
For later convenience we introduce the notation $\mathcal{F}= f^\alpha\otimes f_\alpha$ and $\mathcal{F}^{-1}=\bar f^\alpha\otimes
\bar f_\alpha$ (sum over $\alpha$ understood) for the twist and its inverse. Note that 
 $f^\alpha,f_\alpha,\bar f^\alpha,\bar f_\alpha$ are elements in $H$.

Provided a twist $\mathcal{F}$ of the Hopf algebra $H$, there is a well-known theorem telling us that we
can construct a new Hopf algebra $H^\mathcal{F} := \bigl(U\Xi,\mu,\Delta^\mathcal{F},\epsilon,S^\mathcal{F}\bigr)$
by deforming the coproduct and antipode according to
\begin{flalign}
 \Delta^\mathcal{F}(\xi):=\mathcal{F}\,\Delta(\xi)\,\mathcal{F}^{-1}~,\quad S^\mathcal{F}(\xi):=\chi\, S(\xi)\,\chi^{-1}~,
\end{flalign}
where $\chi:=f^\alpha\,S(f_\alpha)$ and $\chi^{-1} =S(\bar f^\alpha)\,\bar f_\alpha$.
See for example Majid's book \cite{Majid:1996kd} for more details.
We introduce the short notation $\Delta^\mathcal{F}(\xi)=\xi_{1_\mathcal{F}}\otimes \xi_{2_\mathcal{F}}$ 
for the deformed coproduct.

Given the deformed Hopf algebra of diffeomorphisms $H^\mathcal{F}$ the question arises if, and in which sense, 
it is different to the Hopf algebra $H$ we started with. For this remember that $H$ is a cocommutative Hopf algebra,
i.e.~$\Delta^\cop=\Delta$. If we now consider the deformed coproduct $\Delta^\mathcal{F}$ 
we obtain for an arbitrary $\xi\in H^\mathcal{F}$
\begin{flalign}
 (\Delta^{\mathcal{F}})^\cop(\xi) = 
\mathcal{F}_{21}\,\Delta^\cop(\xi)\,\mathcal{F}^{-1}_{21}=\mathcal{F}_{21}\,\mathcal{F}^{-1}\,\mathcal{F}\,\Delta(\xi)
\,\mathcal{F}^{-1}\,\mathcal{F}\,\mathcal{F}^{-1}_{21}= \mathcal{F}_{21}\,\mathcal{F}^{-1}\,\Delta^\mathcal{F}(\xi)\,\mathcal{F}
\,\mathcal{F}^{-1}_{21}~,
\end{flalign}
where $\mathcal{F}_{21}=f_\alpha\otimes f^\alpha $ and $\mathcal{F}^{-1}_{21} = \bar f_\alpha \otimes \bar f^\alpha$.
Thus, the coopposite deformed coproduct is related to $\Delta^\mathcal{F}$ by conjugation of
the element $R=\mathcal{F}_{21}\mathcal{F}^{-1}\in H^\mathcal{F}\otimes H^\mathcal{F}$, 
for all $\xi\in H^\mathcal{F}$,
\begin{flalign}
\label{eqn:Rprop1}
 (\Delta^\mathcal{F})^\cop(\xi) = R\,\Delta^\mathcal{F}(\xi)\, R^{-1}~.
\end{flalign}
This element is called a {\it universal $R$-matrix}. We introduce the convenient notation
$R = R^\alpha\otimes R_\alpha$ and $R^{-1}=\bar R^\alpha \otimes \bar R_\alpha$ (sum over $\alpha$ understood)
for the $R$-matrix and its inverse.
The $R$-matrix satisfies in addition to (\ref{eqn:Rprop1}) the following conditions
\begin{flalign}
\label{eqn:Rprop2}
 (\Delta^\mathcal{F}\otimes \id)R = R_{13}\,R_{23}~,\quad (\id\otimes \Delta^\mathcal{F})R = R_{13} R_{12}~,\quad R_{21} = R^{-1}~,
\end{flalign}
where $R_{12} = R\otimes 1$, $R_{23} = 1\otimes R$, $R_{13} =R^\alpha\otimes 1 \otimes R_\alpha $ and $R_{21} = R_\alpha\otimes R^\alpha$.
A Hopf algebra together with an $R$-matrix satisfying (\ref{eqn:Rprop1}) and (\ref{eqn:Rprop2})
is called a {\it triangular Hopf algebra}. If the last property in (\ref{eqn:Rprop2}) does not hold true,
the Hopf algebra is called {\it quasitriangular}.

Note that while both $H^\mathcal{F}$ and $H$ are triangular Hopf algebras\footnote{
The $R$-matrix of $H$ is the trivial one $1\otimes 1$.
}, $H^\mathcal{F}$ is in general not cocommutative (see the example below).
This means that $H^\mathcal{F}$ is structurally different to the Hopf algebras generated by Lie algebras via
the universal enveloping algebra construction. As a consequence of this non-cocommutative 
behavior of $H^\mathcal{F}$, we will obtain a noncommutative structure on the ``spaces'' the Hopf algebra acts on, 
e.g.~the algebra of functions on $\MM$.

Let us present an explicit example of a particular deformed Hopf algebra of diffeomorphisms.
We consider $\MM=\bbR^N$ and denote by $x^\mu$, $\mu=1,\dots,N$, global coordinate functions on $\MM$.
The derivatives $\partial_\mu$ along $x^\mu$ provide a global basis of $\Xi$, such that every vector field
$v\in \Xi$ can be written as $v=v^{\mu}(x)\partial_\mu$, with coefficient functions $v^\mu\in C^\infty(\bbR^N)$.
In particular, $\partial_\mu\in \Xi$ are globally defined vector fields for all $\mu=1,\dots,N$.
Consider the following element in $U\Xi\otimes U\Xi $
\begin{flalign}
\label{eqn:moyaltwistbas}
\mathcal{F}=\exp\left(-\frac{i\lambda}{2}\Theta^{\mu\nu}\partial_\mu\otimes\partial_\nu\right)~, 
\end{flalign}
where $\Theta^{\mu\nu}$ is a constant and antisymmetric $N\times N$-matrix.
A short calculation shows that $\mathcal{F}$ satisfies (\ref{eqn:twistcondbas}) and is invertible
via
\begin{flalign}
 \mathcal{F}^{-1} = \exp\left(\frac{i\lambda}{2}\Theta^{\mu\nu}\partial_\mu\otimes\partial_\nu\right)~.
\end{flalign}
This means that $\mathcal{F}$ is a twist of $H$, the so-called Moyal-Weyl twist. 
The $R$-matrix of the deformed Hopf algebra $H^\mathcal{F}$ reads
\begin{flalign}
 R= \mathcal{F}_{21}\,\mathcal{F}^{-1} = \exp\bigl(i\lambda\Theta^{\mu\nu}\partial_\mu\otimes\partial_\nu\bigr)~,
\end{flalign}
and thus is nontrivial. The deformed coproduct of the vector fields $\partial_\mu\in H^\mathcal{F}$ is undeformed, i.e.~
\begin{flalign}
 \Delta^\mathcal{F}(\partial_\mu) = \Delta(\partial_\mu)~,
\end{flalign}
since all $\partial_\mu$ mutually commute. However, the deformed coproduct of the vector field $v=M^\nu_\mu x^\mu\partial_\nu\in \Xi$,
where $M^\nu_\mu$ is a constant $N\times N$-matrix, reads
\begin{flalign}
 \Delta^\mathcal{F}(v) = v\otimes 1 + 1\otimes v -\frac{i\lambda}{2} \bigl(\Theta^{\rho\mu}M_\mu^\nu - \Theta^{\nu\mu}M_\mu^\rho
\bigr) \partial_\rho\otimes \partial_\nu ~.
\end{flalign}

As it will become more clear in the following section, the Hopf algebra $H^\mathcal{F}$ with $\mathcal{F}$ given by
(\ref{eqn:moyaltwistbas}) describes the diffeomorphism symmetries of the Moyal-Weyl space $\bbR^N_\Theta$, 
which is a noncommutative space.
Equipping $\bbR^N$ with the Minkowski metric $g=-dx^1\otimes dx^1+\sum_{i=2}^N dx^i\otimes dx^i$ 
and restricting the vector fields to Killing vector fields $\mathfrak{K}\subseteq \Xi$,
we obtain the deformed isometry Hopf algebra $(U\mathfrak{K},\mu,\Delta^\mathcal{F},\epsilon,S^\mathcal{F})$ 
of the noncommutative space $\bigl(\bbR_\Theta^N,g\bigr)$.
This Hopf algebra is also called the $\Theta$-twisted Poincar{\'e} Hopf algebra
\cite{Chaichian:2004za,Chaichian:2004yh}.


\section{Noncommutative differential geometry}
In classical differential geometry, the Lie algebra of vector fields $\bigl(\Xi,[\cdot,\cdot]\bigr)$
acts on tensor fields via the Lie derivative $\mathcal{L}$. As explained above, this action extends
to a left action of $H$ via $\mathcal{L}_{\xi\,\eta}=\mathcal{L}_{\xi}\circ \mathcal{L}_{\eta}$, for all
$\xi,\eta\in H$, and $\mathcal{L}_1=\id$. The aim of this section is to construct
deformations of scalar, vector and tensor fields, which transform covariantly under the deformed 
Hopf algebra of  diffeomorphisms $H^\mathcal{F}$. 
In order to simplify the notation we will suppress the symbol $\mathcal{L}$
for the Lie derivative and simply write $\xi(\cdot):=\mathcal{L}_\xi(\cdot)$, for all
$\xi\in U\Xi$.

Let us first focus on the simplest type of tensor field, namely the smooth and complex functions $C^\infty(\MM)$.
This space can be equipped with an algebra structure by employing the pointwise
multiplication $(h\,k)(x)=h(x)\,k(x)$, for all $h,k\in C^\infty(\MM)$. The algebra structure
is covariant under the Hopf algebra $H$, since the following properties hold true, for all
$1,h,k\in C^\infty(\MM)$ and $\xi\in H$,
\begin{flalign}
 \xi(h\,k) = \xi_1(h)\,\xi_2(k)~,\quad \xi(1)=\epsilon(\xi)\,1~.
\end{flalign}
In mathematical terms, this means that the algebra $\bigl(C^\infty(\MM),\cdot\bigr)$ equipped with the pointwise 
multiplication $\cdot$ is a left $H$-module algebra. 

For nontrivial deformations $\mathcal{F}$, the algebra $\bigl(C^\infty(\MM),\cdot\bigr)$ fails to be covariant
under $H^\mathcal{F}$ (see the example below). This failure is due to the nontrivial coproduct structure
on $H^\mathcal{F}$. However, the algebra $\bigl(C^\infty(\MM),\cdot\bigr)$ can be made covariant under
$H^\mathcal{F}$, if we deform the product accordingly to $\mathcal{F}$. This is a well-known and very general theorem
in mathematics, see e.g.~\cite{Majid:1996kd} and Part \ref{part:math}.
In this section we explain this deformation using examples.
For this consider the following deformed multiplication ($\star$-product)
\begin{flalign}
 h\star k := \bar f^\alpha(h)\,\bar f_\alpha(k)~,
\end{flalign}
for all $h,k\in C^\infty(\MM)$, where $\bar f^\alpha\otimes \bar f_\alpha=\mathcal{F}^{-1}$ is the inverse twist.
Due to the properties (\ref{eqn:twistcondbas}) of the twist the $\star$-product is associative, 
i.e.~$(h\star k)\star l= h\star(k\star l)$, and fulfills $1\star h = h = h\star 1$.
Thus, $\bigl(C^\infty(\MM),\star\bigr)$ is an associative algebra with unit, which however is in general noncommutative
(see the example below). The noncommutativity is governed by the inverse $R$-matrix, since
\begin{flalign}
\nn h\star k &=\bar f^\alpha(h)\,\bar f_\alpha(k) = \bar f_\alpha(k)\,\bar f^\alpha(h)\\
 &=(\bar f^\beta \, f^\gamma\,\bar f_\alpha)(k)\,(\bar f_\beta \, f_\gamma\,\bar f^\alpha)(h) 
= \bar R^\alpha(k)\star \bar R_\alpha(h)~,
\end{flalign}
for all $h,k\in C^\infty(\MM)$. In the second line we have inserted the unit $1\otimes 1 = \mathcal{F}^{-1}\,\mathcal{F}$.
 
We now focus on the action of $H^\mathcal{F}$ on the algebra $\bigl(C^\infty(\MM),\star\bigr)$.
Since $H^\mathcal{F}$ and $H$ are equal as algebras, a left action of $H^\mathcal{F}$ on the vector space
$C^\infty(\MM)$ is given by the usual Lie derivative. 
It remains to check if $\bigl(C^\infty(\MM),\star\bigr)$ is covariant under $H^\mathcal{F}$, i.e.~if
the deformed algebra is a left $H^\mathcal{F}$-module algebra. 
We obtain by an explicit calculation
\begin{flalign}
 \nn \xi(h\star k ) &= \xi\bigl(\bar f^\alpha(h)\,\bar f_\alpha(k)\bigr)=
 (\xi_1\,\bar f^\alpha)(h)\,(\xi_2\,\bar f_\alpha)(k)\\
 &=(\bar f^\beta \, f^\gamma\,\xi_1\,\bar f^\alpha)(h)\,(\bar f_\beta \,f_\gamma\,\xi_2\,\bar f_\alpha)(k) 
= \xi_{1_\mathcal{F}}(h)\star\xi_{2_\mathcal{F}}(k)~,
\end{flalign}
which means that the deformed coproduct $\Delta^\mathcal{F}$ is compatible with the $\star$-multiplication.

This observation allows us to make the following interpretation of $H^\mathcal{F}$:
Similarly as $H$ describes the diffeomorphism symmetries of the classical manifold $\MM$, 
the deformed Hopf algebra $H^\mathcal{F}$ describes the diffeomorphism symmetries of the
noncommutative space $\bigl(C^\infty(\MM),\star\bigr)$.

Before going on in the construction of a noncommutative differential geometry let us discuss
an explicit example. Let $\MM=\bbR^N$ and $\mathcal{F}$ be the Moyal-Weyl twist (\ref{eqn:moyaltwistbas}).
The corresponding $\star$-product then reads
\begin{flalign}
 h\star k = h\,e^{\frac{i\lambda}{2}\overleftarrow{\partial_\mu}\Theta^{\mu\nu}\overrightarrow{\partial_\nu}}\,k~,
\end{flalign}
which is the usual Moyal-Weyl product.
The algebra $\bigl(C^\infty(\bbR^N),\star\bigr)$ transforms covariantly under the deformed diffeomorphisms
$H^\mathcal{F}$. Due to the nontrivial $R$-matrix of $H^\mathcal{F}$, the algebra  $\bigl(C^\infty(\bbR^N),\star\bigr)$
is noncommutative. In particular, the commutation relations of the coordinate functions $x^\mu\in C^\infty(\bbR^N)$
read
\begin{flalign}
 [x^\mu\stackrel{\star}{,}x^\nu] = x^\mu\star x^\nu - x^\nu\star x^\mu = i\lambda\Theta^{\mu\nu}\,1~.
\end{flalign}

We proceed in the construction of a noncommutative differential geometry.
In classical differential geometry an object of central interest is the exterior algebra
of differential forms $\bigl(\Omega^\bullet:=\bigoplus_{n=0}^N\,\Omega^n,\wedge,\dd\bigr)$, where
$\Omega^n$ is the space of smooth and complex $n$-forms, $\wedge:\Omega^n\otimes \Omega^m\to \Omega^{m+n}$ the wedge product
and $\dd:\Omega^n\to\Omega^{n+1}$ the exterior differential. Note that $\Omega^0=C^\infty(\MM)$.
The exterior algebra is graded commutative, i.e.~$\omega\wedge\omega^\prime =(-1)^{\deg(\omega)\deg(\omega^\prime)}
 \omega^\prime\wedge\omega$, and it is covariant under $H$, since for all $\omega,\omega^\prime\in \Omega^\bullet$
and $\xi\in H$
\begin{flalign}
 \xi(\omega\wedge\omega^\prime) = \xi_1(\omega)\wedge \xi_2(\omega^\prime)~.
\end{flalign}
The differential $\dd$ is equivariant under the action of $H$, for all $\omega\in \Omega^\bullet$ and $\xi\in H$,
\begin{flalign}
\label{eqn:equvariancebas} \xi(\dd\omega) = \dd(\xi(\omega))~,
\end{flalign}
and satisfies the graded Leibniz rule, for all $\omega,\omega^\prime\in \Omega^\bullet$,
\begin{flalign}
 \dd(\omega\wedge\omega^\prime) = (\dd\omega)\wedge\omega^\prime+ (-1)^{\deg(\omega)}~\omega\wedge(\dd\omega^\prime)~.
\end{flalign}
Furthermore, the wedge product in $\Omega^\bullet$ provides us with a $\bigl(C^\infty(\MM),\cdot\bigr)$-bimodule structure
on the space of $n$-forms. This is simply the pointwise multiplication of an $n$-form by a function from
left or right, respectively. Due to the $H$-covariance of the exterior algebra,
the bimodule structure is automatically covariant under $H$, i.e.~
\begin{flalign}
 \xi(h\,\omega) = \xi_1(h)\,\xi_2(\omega)~,\quad 
\xi(\omega\,h) = \xi_1(\omega)\,\xi_2(h)~,
\end{flalign}
for all $h\in C^\infty(\MM)$, $\omega\in\Omega^n$ and $\xi\in H$.

In order to render $\Omega^\bullet$ covariant under the deformed Hopf algebra of diffeomorphisms $H^\mathcal{F}$ we
introduce the $\star$-wedge product
\begin{flalign}
 \omega\wedge_\star\omega^\prime := \bar f^\alpha(\omega)\wedge \bar f_\alpha(\omega^\prime)~,
\end{flalign}
for all $\omega,\omega^\prime\in \Omega^\bullet$. Due to (\ref{eqn:twistcondbas}) this product is associative and 
satisfies $1\wedge_\star \omega = \omega = \omega\wedge_\star 1$, for all $\omega\in \Omega^\bullet$. 
It also defines a $\bigl(C^\infty(\MM),\star\bigr)$-bimodule structure on the space on $n$-forms.
The covariance of these structures under $H^\mathcal{F}$ is easily checked and we obtain
\begin{flalign}
 \xi(\omega\wedge_\star\omega^\prime)  = \xi_{1_\mathcal{F}}(\omega)\wedge_\star
\xi_{2_\mathcal{F}}(\omega^\prime)~,
\end{flalign}
for all $\omega,\omega^\prime\in\Omega^\bullet$ and $\xi\in H^\mathcal{F}$.
Analogously to the case of the deformed algebra of functions, the inverse $R$-matrix determines the
deviation of $\bigl(\Omega^\bullet,\wedge_\star\bigr)$ from being graded commutative, more precisely we have
\begin{flalign}
\omega\wedge_\star\omega^\prime = (-1)^{\deg(\omega)\deg(\omega^\prime)} \bar R^\alpha(\omega^\prime)
\wedge_\star\bar R_\alpha(\omega)~,
\end{flalign}
for all $\omega,\omega^\prime\in \Omega^\bullet$.

The deformed exterior algebra can be equipped with a differential. Due to the equivariance property
(\ref{eqn:equvariancebas}) the undeformed differential satisfies, for all $\omega,\omega^\prime\in \Omega^\bullet$,
\begin{flalign}
 \dd(\omega\wedge_\star\omega^\prime) = (\dd\omega)\wedge_\star \omega^\prime + (-1)^{\deg(\omega)}\omega\wedge_\star
(\dd\omega^\prime)~,
\end{flalign}
and thus is a differential on $\bigl(\Omega^\bullet,\wedge_\star\bigr)$. We call 
$\bigl(\Omega^\bullet,\wedge_\star, \dd\bigr)$ the {\it deformed differential calculus}.

Another important geometric object we want to deform is the space of smooth and complex vector fields $\Xi$.
This space has no algebra structure, but it is a bimodule over the algebra $\bigl(C^\infty(\MM),\cdot\bigr)$.
The left and right action is given by the pointwise multiplication of a vector field by a function.
The bimodule structure is $H$-covariant, i.e.~for all $h\in C^\infty(\MM)$, $v\in \Xi$ and $\xi\in H$ we have
\begin{flalign}
  \xi(h\,v) = \xi_1(h)\,\xi_2(v)~,\quad 
\xi(v\,h) = \xi_1(v)\,\xi_2(h)~.
\end{flalign}
Analogously to the case of $n$-forms above we introduce the deformed left and right multiplication
\begin{flalign}
  h\star v := \bar f^\alpha(h)\,\bar f_\alpha(v)~,\quad 
 v\star h := \bar f^\alpha(v)\,\bar f_\alpha(h)~,
\end{flalign}
turning $\Xi$ into a $\bigl(C^\infty(\MM),\star\bigr)$-bimodule, which is covariant under $H^\mathcal{F}$.

Since in differential geometry vector fields and one-forms are dual to each other, we have a
contraction $\pair{\cdot}{\cdot}:\Xi\times \Omega^1\to C^\infty(\MM)$. The contraction map 
$\Omega^1\times \Xi\to C^\infty(\MM)$ will be denoted for notational simplicity by the same symbol.
These contractions are $H$-covariant in the sense that
\begin{flalign}
 \xi\bigl(\pair{v}{\omega}\bigr) = \pair{\xi_1(v)}{\xi_2(\omega)}~,
\end{flalign}
for all $v\in\Xi$, $\omega\in\Omega^1$ and $\xi\in H$. Furthermore, they satisfy the
important property
\begin{flalign}
 \pair{h\,v}{k\,\omega\,l}=h\,\pair{v\,k}{\omega}\,l~,
\end{flalign}
for all $h,k,l\in C^\infty(\MM)$, $v\in\Xi$ and $\omega\in\Omega^1$. In order to make
$\pair{\cdot}{\cdot}$ covariant under $H^\mathcal{F}$ we use again the inverse twist and define
\begin{flalign}
 \pair{v}{\omega}_\star:= \pair{\bar f^\alpha(v)}{\bar f_\alpha(\omega)}~,
\end{flalign}
for all $v\in\Xi$ and $\omega\in \Omega^1$. It follows that the $\star$-contraction
satisfies, for all $h,k,l\in C^\infty(\MM)$, $v\in\Xi$ and $\omega\in\Omega^1$,
\begin{flalign}
  \pair{h\star v}{k\star\omega\star l}_\star=h\star\pair{v\star k}{\omega}_\star\star l~.
\end{flalign}

The next step in the construction of the noncommutative differential geometry is 
to define deformed tensor fields. 
Consider the tensor algebra $\bigl(\mathcal{T},\otimes_A\bigr)$
generated by $\Omega^1$ and $\Xi$. 
The symbol $\otimes_A$ denotes the tensor product over the algebra $A=\bigl(C^\infty(\MM),\cdot\bigr)$.\footnote{
Remember that given two bimodules $V,W$ over an algebra
$A$, we can define their tensor product (over $A$) $V\otimes_A W$ as follows:
Consider the free $A$-bimodule $(V\times W)_{\mathrm{free}}$ generated by the Cartesian product $V\times W$
and the $A$-subbimodule $\mathcal{N}$ generated by the elements
\begin{subequations}
\begin{flalign}
 &(v+v^\prime,w) - (v,w) - (v^\prime,w)~,\quad (v,w+w^\prime) -(v,w) - (v,w^\prime)~,\\
 &(a\,v,w)-a\,(v,w)~,\quad(v\,a,w) - (v,a\,w)~,\quad (v,w\,a) - (v,w)\,a~,
\end{flalign}
\end{subequations}
for all $a\in A$, $v,v^\prime\in V$ and $w,w^\prime\in W$. The tensor product $V\otimes_A W$
is defined by the quotient $A$-bimodule $V\otimes_A W := (V\times W)_{\mathrm{free}}/\mathcal{N}$.
We denote by $v\otimes_A w$ the image of $(v,w)\in V\times W$ under the natural map $V\times W\to V\otimes_A W$.
}
The tensor product is by definition of the Lie derivative
$H$-covariant, i.e.~
\begin{flalign}
  \xi(\tau\otimes_A\tau^\prime) = \xi_1(\tau)\otimes_A \xi_2(\tau^\prime)~,
\end{flalign}
 for all $\tau,\tau^\prime\in\mathcal{T}$ and $\xi\in H$. 
The desired $H^\mathcal{F}$-covariance of the tensor algebra is obtained by introducing the
deformed tensor product
\begin{flalign}
  \tau\otimes_{A_\star}\tau^\prime := \bar f^\alpha(\tau)\otimes_A\bar f_\alpha(\tau^\prime)~,
\end{flalign}
for all $\tau,\tau^\prime\in \mathcal{T}$. 
Let $\tau=\tau^\alpha\otimes_{A_\star}\tau_\alpha\in\mathcal{T}$ (sum over $\alpha$ understood) 
be a tensor field with $\tau^\alpha,\tau_\alpha \in \Xi$ for all $\alpha$ (the same holds true for $\Omega^1$). 
We say that $\tau$ is symmetric/antisymmetric, if
\begin{flalign}
 \tau = \pm\,  \bar R^\beta(\tau_\alpha)\otimes_{A_\star} \bar R_\beta(\tau^\alpha)~.
\end{flalign}
We can always evaluate the $\star$-tensor product
and write the tensor field $\tau$ in terms of the usual tensor product 
$\tau = \widetilde{\tau}^\alpha\otimes_A\widetilde{\tau}_\alpha$.
Our definition of symmetry/antisymmetry in this basis reduces to the usual definition
$\tau = \pm \widetilde{\tau}_\alpha\otimes_A \widetilde{\tau}^\alpha$.
Thus, every classical symmetric/antisymmetric tensor field is also a deformed symmetric/antisymmetric tensor field.
As examples, deformed two-forms are antisymmetric tensor fields and the deformed metric field to be defined later
will be symmetric.


\section{\label{sec:QLA}Quantum Lie algebras}
In the sections above the focus was on the deformed Hopf algebra $H^\mathcal{F}$
of diffeomorphisms. Since we require this concept later, we
now show that we can associate to $H^\mathcal{F}$ a {\it quantum Lie algebra},
similarly as we could associate to $H$ the Lie algebra $\bigl(\Xi,[\cdot,\cdot]\bigr)$
of vector fields. We will be rather nontechnical in this section and refer to 
\cite{Aschieri:2005zs} for details.

In order to better understand the construction of the quantum Lie algebra, we first show that there is a Hopf
algebra $H_\star$, which is isomorphic to $H^\mathcal{F}$, but more convenient for later.
Note that the undeformed Hopf algebra $H$ acts on itself via the adjoint action
\begin{flalign}
 \text{Ad}_\xi(\eta):=  \xi_1\,\eta\,S(\xi_2)~.
\end{flalign}
More precisely, the algebra $\bigl(U\Xi,\mu\bigr)$ is a left $H$-module algebra.
Using the same twist deformation methods as above, we deform this module algebra
by introducing the $\star$-product
\begin{flalign}
 \xi\star\eta := \text{Ad}_{\bar f^\alpha}(\xi)\,\text{Ad}_{\bar f_\alpha}(\eta)~,
\end{flalign}
for all $\xi,\eta\in U\Xi$.
There is a remarkable relation between the algebras $\bigl(U\Xi,\mu_\star\bigr)$ and
$\bigl(U\Xi,\mu\bigr)$, namely they are isomorphic. 
This has been observed first in \cite{0787.17010} for this particular example
and we later found out that an analogous statement holds true for more general classes 
of algebras (see Chapter \ref{chap:HAdef}).
The algebra isomorphism $D:\bigl(U\Xi,\mu_\star\bigr)\to \bigl(U\Xi,\mu\bigr)$ is given by
\begin{flalign}
 D(\xi) := \text{Ad}_{\bar f^\alpha}(\xi)\,\bar f_\alpha~,
\end{flalign}
for all $\xi\in U\Xi$. For a proof we refer to Chapter \ref{chap:HAdef}.
Due to this isomorphism we can pull-back the Hopf algebra structure
on $H^\mathcal{F}$ to a Hopf algebra structure on $\bigl(U\Xi,\mu_\star\bigr)$.
We denote this Hopf algebra by $H_\star=\bigl(U\Xi,\mu_\star,\Delta_\star,\epsilon_\star,S_\star\bigr)$,
where the coproduct, counit and antipode are given by
\begin{subequations}
\begin{flalign}
 \Delta_\star &:= (D^{-1}\otimes D^{-1})\circ\Delta^\mathcal{F}\circ D~,\\
 \epsilon_\star &:= \epsilon\circ D~,\\
 S_\star &:= D^{-1}\circ S^\mathcal{F}\circ D~. 
\end{flalign}
\end{subequations}
This Hopf algebra is also triangular with $R$-matrix
$R_\star := (D^{-1}\otimes D^{-1})(R)$.
Since $H_\star$ and $H^\mathcal{F}$ are isomorphic Hopf algebras, any representation of $H^\mathcal{F}$
is also a representation of $H_\star$. We denote the $H_\star$-action by the $\star$-Lie derivative,
for all $\xi\in U\Xi$,
\begin{flalign}
 \label{eqn:starliederbas} 
\mathcal{L}_\xi^\star := \mathcal{L}_{D(\xi)}~.
\end{flalign}

Let us now focus on the $\star$-coproduct $\Delta_\star$ of vector fields $\Xi$. One obtains
the very compact expression
\begin{flalign}
 \Delta_\star(v)= v\otimes 1 + D^{-1}(\bar R^\alpha)\otimes \bar R_{\alpha}(v)~,
\end{flalign}
for all $v\in \Xi$. 
Acting with the $\star$-Lie derivative on, for example, a $\star$-tensor product of tensor fields
we obtain
\begin{flalign}
 \mathcal{L}^\star_{v}(\tau\otimes_{A_\star} \tau^\prime) = \mathcal{L}^\star_{v}(\tau)\otimes_{A_\star} \tau^\prime
+ \bar R^\alpha(\tau)\otimes_{A_\star} \mathcal{L}^\star_{\bar R_\alpha(v)}(\tau^\prime)~,
\end{flalign}
for all $v\in\Xi$ and $\tau,\tau^\prime\in\mathcal{T}$.
Note that this is up to the $R$-matrix the usual Leibniz rule for a vector field.

With this observation we are able to construct a quantum Lie algebra in the sense of Woronowicz \cite{Woronowicz:1989rt}.
Consider the vector fields $\Xi$ equipped with the $\star$-Lie bracket
$[\cdot,\cdot]_\star:\Xi\times\Xi\to\Xi$ defined by
\begin{flalign}
 [v,w]_\star:= [\bar f^\alpha(v),\bar f_\alpha(w)]~,
\end{flalign}
for all $v,w\in\Xi$. This bracket satisfies the deformed antisymmetry property
\begin{flalign}
 [v,w]_\star = - [\bar R^\alpha(w),\bar R_\alpha(v)]_\star~,
\end{flalign}
and the deformed Jacobi identity
\begin{flalign}
 [v,[w,z]_\star]_\star = [[v,w]_\star,z]_\star + [\bar R^\alpha(w),[\bar R_\alpha(v),z]_\star]_\star~,
\end{flalign}
for all $v,w,z\in\Xi$. 
Note that the following three properties hold true:
\begin{enumerate}
 \item $\Xi$ generates $H_\star$
 \item $[\Xi,\Xi]_\star\subseteq \Xi$
 \item $\Delta_\star(\Xi)\subseteq \Xi\otimes 1 + U\Xi\otimes \Xi$
\end{enumerate}
Thus, $\bigl(\Xi,[\cdot,\cdot]_\star\bigr)$ is a quantum Lie algebra corresponding
to $H_\star$.
This quantum Lie algebra acts on deformed tensor fields via the $\star$-Lie derivative 
(\ref{eqn:starliederbas}).

Observe that the conditions a quantum Lie algebra has to satisfy are slight deformations of the 
Lie algebra case.  In particular, the deformed Lie bracket is allowed to be antisymmetric up to an $R$-matrix
and the coproduct is allowed to deviate in a controlled way from the usual Leibniz rule.

The quantum Lie algebra $\bigl(\Xi,[\cdot,\cdot]_\star\bigr)$ should be interpreted as the infinitesimal
deformed diffeomorphisms. Covariance under $\bigl(\Xi,[\cdot,\cdot]_\star\bigr)$ implies
covariance under $H_\star$, since $\Xi$ generates $H_\star$, and vice versa, which means
that we can now work completely on the level of quantum Lie algebras, without focusing
on the Hopf algebra $H_\star$.

As a last remark, the Hopf algebra $H_\star$ is more suitable for the definition of a quantum Lie algebra,
since its quantum Lie algebra is as a vector space simply $\Xi$.
Using the isomorphism $D$, we obtain that the quantum Lie algebra of $H^\mathcal{F}$
is given as a vector space by $D(\Xi)\subseteq U\Xi$, i.e.~higher order products of vector fields.
Working with $\Xi$ is simpler than working with $D(\Xi)$, this is why we prefer $H_\star$.
However, due to the isomorphism $H_\star\simeq H^\mathcal{F}$ both approaches 
are equivalent.


\section{\label{sec:ncgbaseistein}$\star$-covariant derivatives, curvature and Einstein equations}
In this section we introduce covariant derivatives, torsion and curvature
in the framework of the noncommutative differential geometry presented in the section above.
We will again suppress the symbol $\mathcal{L}$ for the (usual) Lie derivative in order to compactify notation,
i.e.~$\xi(\cdot) = \mathcal{L}_\xi(\cdot)$.

A {\it $\star$-covariant derivative} $\nab^\star_v$ along a vector field $v\in\Xi$
is a $\bbC$-linear map $\nab^\star_v:\Xi\to\Xi$ satisfying, for all $v,w,z\in\Xi$ and $h\in C^\infty(\MM)$,
\begin{subequations}
\label{eqn:starcovderbas}
\begin{flalign}
\nab^\star_{v+w}z &= \nab^\star_v z +\nab^\star_w z~,\\
\nab^\star_{h\star v} z &= h\star \nab^\star_v z~,\\
\nab^\star_{v}(h\star z) &= \mathcal{L}_v^\star(h)\star z +\bar R^\alpha(h)\star \nab^\star_{\bar R_\alpha(v)}z~.
\end{flalign}
\end{subequations}

Given a  $\star$-covariant derivative, its {\it $\star$-torsion} and {\it $\star$-curvature}
are $\bbC$-linear maps $\text{Tor}^\star:\Xi\otimes\Xi\to\Xi$ and $\text{Riem}^\star:\Xi\otimes\Xi\otimes\Xi\to \Xi$
defined by
\begin{subequations}
\begin{flalign}
  \text{Tor}^\star(v,w) &:= \nab^\star_v w - \nab^\star_{\bar R^\alpha(w)}\bar R_\alpha(v) - [v,w]_\star~,\\
  \text{Riem}^\star(v,w,z) &:= \nab^\star_v\nab^\star_wz - \nab^\star_{\bar R^\alpha(w)}\nab^\star_{\bar R_\alpha(v)}z
-\nab^\star_{[v,w]_\star}z~,
\end{flalign}
\end{subequations}
for all $v,w,z\in\Xi$. These maps satisfy the antisymmetry properties
\begin{subequations}
\begin{flalign}
  \text{Tor}^\star(v,w) &= - \text{Tor}^\star\bigl(\bar R^\alpha(w),\bar R_\alpha(v)\bigr)~,\\
  \text{Riem}^\star(v,w,z) &= -\text{Riem}^\star\bigl(\bar R^\alpha(w),\bar R_\alpha(v),z\bigr)~,
\end{flalign}
\end{subequations}
and 
\begin{subequations}
\begin{flalign}
  \text{Tor}^\star(h\star v,k\star w) &= h\star \text{Tor}^\star(v\star k,w)~,\\
  \text{Riem}^\star(h\star v,k\star w,l\star z) &= h\star \text{Riem}^\star(v\star k,w\star l,z)~,
\end{flalign}
\end{subequations}
for all $v,w,z\in\Xi$ and $h,k,l\in C^\infty(\MM)$.

In order to construct the {\it $\star$-Ricci tensor} $\text{Ric}^\star:\Xi\otimes \Xi \to C^\infty(\MM)$
we require a (local) basis $\lbrace e_a:a=1,\dots,N \rbrace$ of $\Xi$. 
The dual basis $\lbrace \theta^a:a=1,\dots,N \rbrace$ is defined by the conditions $\pair{e_a}{\theta^b}_\star = \delta_a^b$.
The $\star$-Ricci tensor is defined by the contraction
\begin{flalign}
\text{Ric}^\star(v,w):= \sum\limits_{a=1}^N\,\pair{\theta^a}{\text{Riem}^\star(e_a,v,w)}_\star~,
\end{flalign}
for all $v,w\in\Xi$. It is independent on the choice of basis and satisfies
\begin{flalign}
 \text{Ric}^\star(h\star v,k\star w) = h\star \text{Ric}^\star(v\star k,w)~,
\end{flalign}
for all $h,k\in C^\infty(\MM)$ and $v,w\in\Xi$.

Let us consider the example of $\MM=\bbR^N$
equipped with the Moyal-Weyl twist (\ref{eqn:moyaltwistbas}).
 In this case we have the global basis $\lbrace \partial_\mu\in\Xi:\mu=1,\dots,N \rbrace$ of vector fields $\Xi$.
The $\star$-covariant derivative is completely specified by the Christoffel symbols
\begin{flalign}
 \nab^\star_{\partial_\mu}\partial_\nu = \Gamma^{\star \rho}_{\mu\nu}\star \partial_\rho=\Gamma^{\star \rho}_{\mu\nu}\,\partial_\rho~,
\end{flalign}
where in the last equality we have used that the twist acts trivially on all $\partial_\rho$.
In this basis the $\star$-torsion, $\star$-curvature and $\star$-Ricci tensor reads
\begin{subequations}
\begin{flalign}
 \text{Tor}^\star(\partial_\mu,\partial_\nu) &= \bigl(\Gamma^{\star \rho}_{\mu\nu} - \Gamma^{\star \rho}_{\nu\mu}\bigr)\,\partial_\rho~, \\
 \text{Riem}^\star(\partial_\mu,\partial_\nu,\partial_\rho) &= \bigl(\partial_\mu\Gamma_{\nu\rho}^{\star \sigma} - 
\partial_\nu\Gamma_{\mu\rho}^{\star \sigma} + \Gamma_{\nu\rho}^{\star \tau} \star \Gamma_{\mu\tau}^{\star \sigma}-
 \Gamma_{\mu\rho}^{\star \tau} \star \Gamma_{\nu\tau}^{\star\sigma}\bigr)\,\partial_\sigma ~,\\
 \text{Ric}^\star(\partial_\nu,\partial_\rho) &= \partial_\mu\Gamma_{\nu\rho}^{\star \mu} - 
\partial_\nu\Gamma_{\mu\rho}^{\star \mu} + \Gamma_{\nu\rho}^{\star \tau} \star \Gamma_{\mu\tau}^{\star \mu}-
 \Gamma_{\mu\rho}^{\star \tau} \star \Gamma_{\nu\tau}^{\star\mu} ~,
\end{flalign}
\end{subequations}
where it was extensively used that the twist acts trivially on $\partial_\rho$.
For more general deformations, or a different basis of vector fields, 
these equations will be much more complicated, in the sense that 
$R$-matrices appear.

For a noncommutative gravity theory we require one more ingredient, a metric field
$g=g^\alpha\otimes_{A_\star} g_\alpha \in\Omega^1\otimes_{A_\star} \Omega^1$.
We demand $g$ to be symmetric, real and nondegenerate.
In particular, every classical metric field $g\in\Omega^1\otimes_{A} \Omega^1$
is also a metric field in the sense above.
The $\star$-inverse metric $g^{-1}=g^{-1\alpha}\otimes_{A_\star} g^{-1}_\alpha \in\Xi\otimes_{A_\star} \Xi$ is defined via the conditions
\begin{subequations}
\begin{flalign}
  \pair{\pair{v}{g}_\star}{g^{-1}}_\star &= \pair{v}{g^\alpha}_\star \star \pair{g_\alpha}{g^{-1\beta}}_\star \star g_\beta^{-1}=v~,\\
  \pair{\pair{\omega}{g^{-1}}_\star}{g}_\star &= \pair{\omega}{g^{-1\beta}}_\star \star\pair{g^{-1}_\beta}{g^\alpha}_\star\star g_{\alpha} = \omega~,
\end{flalign}
\end{subequations}
for all $v\in\Xi$ and $\omega\in\Omega^1$.

In Einstein's theory of general relativity the metric field constitutes the dynamical
degree of freedom and the covariant derivative is fixed to be the Levi-Civita connection,
i.e.~the unique torsion-free and metric compatible covariant derivative.
In noncommutative gravity, the existence and uniqueness of a $\star$-Levi-Civita connection, 
i.e.~torsion-free $\star$-covariant derivative which is compatible with the metric tensor, is not yet completely understood.
We will come back to this issue in Chapter \ref{chap:ncgsol}, where we show that under certain conditions
there is a unique $\star$-Levi-Civita connection.

Continuing in the construction of the noncommutative Einstein equations, we
define the {\it $\star$-curvature scalar} as the contraction
\begin{flalign}
  \mathfrak{R}^\star := \text{Ric}^\star(g^{-1\alpha},g_{\alpha}^{-1})~.
\end{flalign}
The {\it $\star$-Einstein tensor} is given by the map $G^\star:\Xi\otimes\Xi\to C^\infty(\MM)$,
\begin{flalign}
  G^\star(v,w):= \text{Ric}^\star(v,w)-\frac{1}{2} g(v,w)\star\mathfrak{R}^\star~,
\end{flalign}
for all $v,w\in\Xi$, where $g(v,w):= \pair{v}{\pair{w}{g^\alpha}_\star\star g_\alpha}_\star$.
It satisfies the property
\begin{flalign}
  G^\star(h\star v,k\star w) = h\star G^\star(v\star k,w)~,
\end{flalign}
for all $h,k\in C^\infty(\MM)$  and $v,w\in\Xi$.
The noncommutative Einstein equations for vacuum are thus
\begin{flalign}
  G^\star(v,w)=0~,
\end{flalign}
for all $v,w\in\Xi$.
Provided a suitable stress-energy tensor $T^\star:\Xi\otimes \Xi\to C^\infty(\MM)$, one can also consider
the Einstein equations coupled to matter
\begin{flalign}
  G^\star(v,w)= 8\pi G_N\,T^\star(v,w)~,
\end{flalign}
for all $v,w\in\Xi$.

Let us go back to the example $\MM=\bbR^N$ deformed by the Moyal-Weyl twist (\ref{eqn:moyaltwistbas}).
Using again that the twist acts trivially on $\partial_\rho$ we can write for the metric field
$g=dx^\mu\otimes_A dx^\nu g_{\mu\nu} = dx^\mu\otimes_{A_\star} dx^\nu\star g_{\mu\nu}$, since all $\star$ drop out.
The $\star$-inverse metric field $g^{-1} = g^{\mu\nu}\,\partial_\mu\otimes_A \partial_\nu 
=g^{\mu\nu}\star \partial_\mu\otimes_{A_\star} \partial_\nu$ is determined by
\begin{flalign}
  g_{\mu\nu}\star g^{\nu\rho}= \delta_\mu^\rho~, \quad g^{\mu\nu}\star g_{\nu\rho} = \delta^\mu_\rho~.
\end{flalign}
It is simply the $\star$-inverse matrix of $g_{\mu\nu}$.
For this model there is a unique $\star$-Levi-Civita connection, see Chapter \ref{chap:ncgsol}. The corresponding 
Christoffel symbols are given by
\begin{flalign}
  \Gamma_{\mu\nu}^{\star\rho} = \frac{1}{2}\,g^{\rho\sigma}\star \left(\partial_\mu g_{\nu\sigma} +\partial_\nu g_{\mu\sigma}-
\partial_\sigma g_{\mu\nu}\right)~.
\end{flalign}
Thus, the $\star$-Einstein tensor can be expressed completely in terms of the metric and 
the noncommutative Einstein equations give rise to dynamical equations for the metric field.


\chapter{\label{chap:symred}Symmetry reduction}
In this chapter we present an approach to symmetry reduction in noncommutative gravity
and its application to noncommutative cosmology and black hole physics. 
These results have been published in the research article \cite{Ohl:2008tw}
 and the proceedings article \cite{Schenkel:2009qm}.
 

\section{Physical idea and the case of classical gravity}
Consider a physical system consisting of a metric field $g$ and
a collection of matter fields (tensor fields) $\lbrace \Phi_i\rbrace_{i\in\mathcal{I}}$
on a manifold $\MM$. Since symmetry reduction in presence of gauge fields
is more involved, we are going to neglect the effects of gauge fields in the present chapter.
An example of such a system is inflationary cosmology, where one typically
considers the metric field $g$ and an inflaton field $\Phi$
on a suitable manifold $\MM$ (e.g.~$\MM=\bbR^4$ or $\MM=\bbR\times S^3$).

We assume that the dynamics of our system is described by Einstein's equations 
$\text{Ric}-\frac{1}{2}g\,\mathfrak{R} =8\pi G_N\,T$ and geometric differential equations for $\lbrace \Phi_i\rbrace_{i\in\mathcal{I}}$.
Finding the most general solution of these highly nonlinear differential equations turns out to be extremely hard
and, even worse, practically impossible. 
In order to anyhow extract certain classes of exact solutions of the equations above, we have to make some 
physics input and assumptions. Symmetry reduction turns out to be a very powerful tool to do so.

We first explain the basic idea of symmetry reduction using the example of cosmology.
From observations we know that the universe, at large scales, is almost isotropic and homogeneous.
So it seems to be reasonable to do a zeroth order approximation and assume the universe 
to be exactly isotropic and homogeneous. The fluctuations we observe for example in the cosmic microwave background (CMB)
 are then described in a second step by considering small perturbations on top of the highly symmetric background fields.
If we demand the metric $g$ and all matter fields $\lbrace \Phi_i \rbrace_{i\in\mathcal{I}}$
to be isotropic and homogeneous, Einstein's equations reduce to  Friedmann's equations
and also the differential equations for the matter fields simplify drastically. For certain models one can then find
exact solutions, describing the evolution of a homogeneous and isotropic universe
filled with  homogeneous and isotropic matter. This model serves as a background for perturbative studies
in cosmology, in particular phenomenological investigations on the CMB.

Let us now investigate symmetry reduction from the mathematical, i.e.~differential geometric, perspective.
Let $\MM$ be a manifold, $g$ a metric field and $\lbrace\Phi_i \rbrace_{i\in\mathcal{I}}$ be a family
of tensor fields. Let $\bigl(\mathfrak{g},[\cdot,\cdot]\bigr)$ be a Lie subalgebra 
of the Lie algebra of vector fields $\bigl(\Xi,[\cdot,\cdot]\bigr)$.
We implement the symmetries described by $\mathfrak{g}$ on our system by demanding the metric
and all tensor fields to be invariant under $\mathfrak{g}$
\begin{flalign}
\label{eqn:symredcondclass}
\mathcal{L}_{\mathfrak{g}}(g) = \lbrace 0\rbrace~,\quad 
\mathcal{L}_{\mathfrak{g}}(\Phi_i) = \lbrace 0 \rbrace ~,\quad \forall i\in\mathcal{I}~.
\end{flalign}
As an aside, the symmetry condition (\ref{eqn:symredcondclass}) would be too strict for fields $\chi$ 
which also transform under some infinitesimal gauge symmetry $\mathfrak{h}$.
A reasonable definition in this case would be $\mathcal{L}_{\mathfrak{g}}(\chi) \subseteq \delta_{\mathfrak{h}}(\chi)$,
which means that gauge equivalence classes should be invariant under $\mathfrak{g}$.

Demanding the symmetry condition (\ref{eqn:symredcondclass}) on the level of the fundamental fields the question
arises if composite fields, like e.g.~the curvature or the stress-energy tensor, are 
also invariant under $\mathfrak{g}$. This propagation of symmetries is of great importance,
since it ensures that the equations of motion are consistent when restricted to the symmetry reduced fields.

In classical differential geometry the propagation of symmetries is ensured by
the coproduct of vector fields (the Leibniz rule) $\Delta(v)=v\otimes 1 + 1\otimes v$, for all
$v\in\Xi$. Since all differential geometric constructions are defined such that they transform covariantly
with respect to this coproduct, constructions out of $\mathfrak{g}$-invariant fields will be $\mathfrak{g}$-invariant.
Let us give some examples to clarify this: The tensor product of two invariant tensor fields $\tau,\tau^\prime$ is invariant, since
$\mathcal{L}_{v}\left(\tau\otimes_A\tau^\prime\right) = 
\mathcal{L}_v(\tau)\otimes_A \tau^\prime +\tau\otimes_A\mathcal{L}_v(\tau^\prime)=0$, 
for all $v\in\mathfrak{g}$. Contractions of $\tau,\tau^\prime$
are also invariant, since $\mathcal{L}_v(\pair{\tau}{\tau^\prime})=\pair{\mathcal{L}_v(\tau)}{\tau^\prime} + 
\pair{\tau}{\mathcal{L}_v(\tau^\prime)}=0$, for all $v\in\mathfrak{g}$. Furthermore, one easily shows that the 
Levi-Civita connection, Riemann tensor, Ricci tensor, curvature scalar and finally the Einstein tensor is $\mathfrak{g}$-invariant
in case the metric is.


\section{Deformed symmetry reduction: General considerations}
While in classical symmetry reduction one is interested in suitable Lie subalgebras
$\bigl(\mathfrak{g},[\cdot,\cdot]\bigr)\subseteq\bigl(\Xi,[\cdot,\cdot]\bigr)$
of the Lie algebra of vector fields on $\MM$, in the noncommutative case 
 Lie algebras are not expected to be an appropriate algebraic structure.
 This can be already seen on the level of the infinitesimal diffeomorphisms,
 which are described by the quantum Lie algebra $\bigl(\Xi,[\cdot,\cdot]_\star\bigr)$
 associated to the deformed Hopf algebra of diffeomorphisms $H_\star$.
 So the basic ingredient for a deformed symmetry reduction should be given by a suitable
 subset $\mathfrak{g}_\star\subseteq \Xi$ carrying some algebraic structures
 to be specified now. Since $\mathfrak{g}_\star$ will be later interpreted as 
 deformed infinitesimal isometries/symmetries of our models, there are natural
 conditions we would like to demand:
\begin{subequations}
\label{eqn:symredcond}
 \begin{flalign}
 \label{eqn:symredcond1}&\mathfrak{g}_\star \text{~is a vector space}\\
 \label{eqn:symredcond2}&[\mathfrak{g}_\star,\mathfrak{g}_\star]_\star\subseteq \mathfrak{g}_\star\\
 \label{eqn:symredcond3}&\Delta_\star(\mathfrak{g}_\star)\subseteq \mathfrak{g}_\star\otimes 1 + U\Xi\otimes \mathfrak{g}_\star
 \end{flalign} 
\end{subequations}
While the reason for demanding the first two conditions is obvious, we have to explain why we also demand the 
third one. As explained in detail in the last section, an important feature of classical differential geometry
is that demanding symmetries for all fundamental fields, the symmetries propagate due to the Leibniz
rule also to composite fields. Since the deformed Leibniz rule of vector fields, i.e.~the coproduct $\Delta_\star$,
does not necessarily have this property, we have to make an assumption on $\Delta_\star$
in order to ensure that symmetries propagate to composite deformed fields. 
Let us explain why the third property leads to a propagation of the symmetry using an example. Let $\tau,\tau^\prime\in\mathcal{T}$
be two tensor fields, which are $\mathfrak{g}_\star$-invariant, i.e.~
\begin{flalign}
 \mathcal{L}_{\mathfrak{g}_\star}^{\star}(\tau) = \mathcal{L}_{\mathfrak{g}_\star}^{\star}(\tau^\prime) = \lbrace 0\rbrace~. 
\end{flalign}
Due to the third condition (\ref{eqn:symredcond3}) we find for all $v\in\mathfrak{g}_\star$
\begin{flalign}
\mathcal{L}_{v}^{\star}(\tau\otimes_{A_\star}\tau^\prime) =\mathcal{L}_{v}^{\star}(\tau)\otimes_{A_\star}\tau^\prime
+ \bar R^\alpha(\tau)\otimes_{A_\star}\mathcal{L}^\star_{\bar R_\alpha(v)}(\tau^\prime) =0~.
\end{flalign}
The same holds true for all other deformed differential geometric operations, since they are covariant under
$H_\star$.

 Note that the three conditions above are much weaker than demanding 
$\mathfrak{g}_\star$ to be a quantum Lie algebra, since this would require to replace (\ref{eqn:symredcond3}) 
by the stronger condition
$\Delta_\star(\mathfrak{g}_\star)\subseteq \mathfrak{g}_\star\otimes 1 + U\mathfrak{g}_\star \otimes \mathfrak{g}_\star $.
However, since $\bigl(\mathfrak{g}_\star,[\cdot,\cdot]_\star\bigr)$ satisfies 
almost all conditions of a quantum Lie algebra we call it an {\it almost quantum Lie subalgebra 
of $\bigl(\Xi,[\cdot,\cdot]_\star\bigr)$}.

Provided an almost quantum Lie subalgebra $\bigl(\mathfrak{g}_\star,[\cdot,\cdot]_\star\bigr)\subseteq 
\bigl(\Xi,[\cdot,\cdot]_\star\bigr)$, taking the deformation parameter $\lambda\to 0$ gives us
 a Lie subalgebra $\bigl(\mathfrak{g},[\cdot,\cdot]\bigr)\subseteq \bigl(\Xi,[\cdot,\cdot]\bigr)$ of the Lie 
 algebra of vector fields on $\MM$.
 Thus, every quantum symmetry gives rise to a classical Lie algebra in the commutative limit.
 For our purpose, the other way around is more interesting.
 Let  $\bigl(\mathfrak{g},[\cdot,\cdot]\bigr)\subseteq \bigl(\Xi,[\cdot,\cdot]\bigr)$ be a Lie subalgebra,
 which we interpret as the symmetries of some classical physics model, e.g.~translations and rotations in cosmology.
 It is important to find out under which conditions we can construct an almost quantum Lie subalgebra 
 $\bigl(\mathfrak{g}_\star,[\cdot,\cdot]_\star\bigr)$ such that $\lambda\to 0$ yields $\bigl(\mathfrak{g},[\cdot,\cdot]\bigr)$.
 Note that this construction is always possible by making the choice $\mathfrak{g}_\star = \mathfrak{g} + \lambda \Xi$, 
 which means that elements of $\mathfrak{g}_\star$ are given by 
$v=\sum_{n=0}^\infty\lambda^n\,v_{(n)}$, where $v_{(0)}\in \mathfrak{g}$
 and $v_{(n)}\in\Xi$ for $n>0$. 
 Since in this case the higher orders in $\mathfrak{g}_\star$ are $\infty$-dimensional, 
 demanding invariance under $\mathfrak{g}_\star$ would provide too strong restrictions on tensor fields.
 This is why we consider this choice as unphysical and introduce the condition that
 $\mathfrak{g}_\star$ has the same dimension as $\mathfrak{g}$. 
 In the remaining part we will be even more restrictive and demand $\mathfrak{g}_\star=\mathfrak{g}$ as vector spaces.
 The motivation for this choice is threefold: Firstly, it simplifies our explicit investigations in the next sections.
 Secondly, deforming only the action of a symmetry and not its underlying vector space structure
 fits into the deformation philosophy of Chapter \ref{chap:basicncg}.
 Thirdly, this choice is mathematically distinguished from all other choices since $\mathfrak{g}_\star=\mathfrak{g}$ is
 a topologically free $\bbC[[\lambda]]$-module (see Appendix \ref{app:basicsdefq}).
 With the choice $\mathfrak{g}_\star=\mathfrak{g}$ the first condition (\ref{eqn:symredcond1}) is fulfilled. The second condition 
(\ref{eqn:symredcond2}) demands that the $\star$-Lie bracket closes on $\mathfrak{g}$. 
Taking a basis $\lbrace t_i\in\mathfrak{g}:i=1,\dots,\dim(\mathfrak{g})\rbrace$ we thus require
the existence of deformed structure constants $f_{ij}^{\star~~k}\in\bbC$, such that
\begin{flalign}
\label{eqn:starcomsymred}
[t_i,t_j]_\star = f_{ij}^{\star~~k}\,t_k~. 
\end{flalign}
By classical correspondence we know that $f_{ij}^{\star~~k}=f_{ij}^{~~k} +\mathcal{O}(\lambda)$,
where $f_{ij}^{~~k}$ are the structure constants of $\bigl(\mathfrak{g},[\cdot,\cdot]\bigr)$.
 On the level of the generators $t_i$, the third condition (\ref{eqn:symredcond3})
requires the existence of elements $Q_{i}^{\star~j}\in U\Xi$, such that
\begin{flalign}
\label{eqn:coproductsymred}
 \Delta_\star(t_i) = t_i \otimes 1 + Q_{i}^{\star~j}\otimes t_j~.
\end{flalign}

It turns out that provided the most general twist $\mathcal{F}$ and the most general Lie subalgebra $\mathfrak{g}$
the two conditions (\ref{eqn:starcomsymred}) and (\ref{eqn:coproductsymred}) are not satisfied. 
There are obstructions which give us relations between the deformation $\mathcal{F}$ and the classical symmetries 
$\mathfrak{g}$ which we want to quantize.
This result is physically very attractive, since it means that if we want to deform
$\mathfrak{g}$-symmetric systems, we can not use the most general twist, but the requirements of a consistent deformed
symmetry structure restrict us to special choices of $\mathcal{F}$. In other words,
not every classical symmetry is compatible with every Drinfel'd twist. 
This partially resolves the arbitrariness in choosing a Drinfel'd twist. However, as we will see 
later using explicit examples, there are still classes of deformations which are consistent with a given classical
symmetry in the sense above. This can be seen as a positive and also a negative feature:
It is positive, since we have different models which have different physical properties,
but it is also in some sense negative, since there is no unique way of deforming a given
classical symmetric system and we are still left with some arbitrariness in choosing $\mathcal{F}$.
Note that dropping the restriction $\mathfrak{g}_\star=\mathfrak{g}$ (as vector spaces)
we might be able to relax the conditions between a given classical symmetry Lie algebra 
$\bigl(\mathfrak{g},[\cdot,\cdot]\bigr)$ and the deformation $\mathcal{F}$. It would be an
interesting project for future research to understand quantitatively how much the space
of compatible $(\mathfrak{g},\mathcal{F})$-pairs is enlarged in this way. 
This could provide more consistent models for noncommutative cosmology and black hole physics
than those presented in Section \ref{sec:symredapplic}.


\section{Deformed symmetry reduction: Abelian Drinfel'd twists}
In order to understand the conditions (\ref{eqn:starcomsymred}) and (\ref{eqn:coproductsymred})
in more detail we are going to restrict the class of Drinfel'd twists to abelian twists.
This allows us for an explicit evaluation of (\ref{eqn:starcomsymred}) and (\ref{eqn:coproductsymred})
and provides us with a better understanding of the compatible $(\mathfrak{g},\mathcal{F})$-pairs.

Let $\lbrace X_\alpha\in\Xi\rbrace$ be a finite set of mutually commuting vector fields, 
i.e. $\com{X_\alpha}{X_\beta}=0,~\forall_{\alpha,\beta}$, on an $N$-dimensional manifold $\mathcal{M}$.
We assume the $X_\alpha$ to be linearly independent. Then the object
\begin{flalign}
\label{eqn:jstwist}
 \mathcal{F} = \exp\left(-\frac{i\lambda}{2} \Theta^{\alpha\beta} X_\alpha\otimes X_\beta \right)\in U\Xi\otimes U\Xi
\end{flalign}
is a Drinfel'd twist, if $\Theta$ is constant and antisymmetric~\cite{Reshetikhin:1990ep,Jambor:2004kc}.
 The twist (\ref{eqn:jstwist}) is called {\it abelian twist}. Note that
 $\mathcal{F}$ is not restricted to the topology $\bbR^N$ of the manifold
 $\mathcal{M}$. As a special case of an abelian twist we obtain the Moyal-Weyl twist by choosing 
$\MM=\bbR^N$ and for $X_\alpha$ the partial derivatives.

Without loss of generality we can restrict ourselves to $\Theta$ with maximal rank and an even number of vector fields $X_\alpha$.
We can therefore use the standard form
\begin{flalign}
 \Theta=\begin{pmatrix}
         0 & 1 & 0 & 0 & \cdots\\
	-1 & 0 & 0 & 0 & \cdots\\
	 0 & 0 & 0 & 1 & \cdots\\
	 0 & 0 & -1 & 0 & \cdots\\
	\vdots&\vdots&\vdots&\vdots&\ddots
        \end{pmatrix}
\end{flalign}
by applying a suitable general linear transformation on the $X_\alpha$.

The abelian twist (\ref{eqn:jstwist}) is easy to apply and in particular we obtain for the inverse and the $R$-matrix
\begin{subequations}
\begin{flalign}
\label{eqn:jsinvtwist} \mathcal{F}^{-1} &= \exp\left(\frac{i\lambda}{2} 
\Theta^{\alpha\beta} X_\alpha\otimes X_\beta \right)~,\\
\label{eqn:jsrmatrix} R&=\mathcal{F}_{21}\mathcal{F}^{-1}=
\mathcal{F}^{-2}=\exp\bigl(i \lambda\Theta^{\alpha\beta} X_\alpha\otimes X_\beta \bigr)~.
\end{flalign}
\end{subequations}

Let $\bigl(\mathfrak{g},\com{\cdot}{\cdot}\bigr)\subseteq\bigl(\Xi,\com{\cdot}{\cdot}\bigr)$ be the Lie 
algebra of the symmetry we want to deform. We choose a 
basis of this Lie algebra $\lbrace t_i : i=1,\cdots ,\dim(\mathfrak{g})\rbrace$ 
with $\com{t_i}{t_j}=f_{ij}^{~~k}t_k$. We work with the restriction of demanding
$\mathfrak{g}_\star=\mathfrak{g}$ as a vector space. Thus, the $t_i$ are also generators of $\mathfrak{g}_\star$.

We are going to check now the three conditions (\ref{eqn:symredcond}) for this particular model.
More precisely, since $\mathfrak{g}_\star=\mathfrak{g}$ is by definition a vector space
we just have to evaluate the two equivalent conditions (\ref{eqn:starcomsymred}) and (\ref{eqn:coproductsymred}) 
the generators $t_i$ have to satisfy. 
We start with the coproduct condition (\ref{eqn:coproductsymred}). Using the
explicit form of the inverse $R$-matrix (\ref{eqn:jsrmatrix}) we can calculate the $\star$-coproduct
of an arbitrary generator $t_i\in\mathfrak{g}$ and find 
\begin{flalign}
\label{eqn:symredabcoprod}
\Delta_\star(t_i) = t_i\otimes 1 + \sum\limits_{n=0}^\infty\frac{(-i\lambda)^n}{n!}\, \widetilde{X}^{\alpha_1}\cdots
\widetilde{X}^{\alpha_n}\otimes [X_{\alpha_1},\cdots [X_{\alpha_n},t_i]\cdots]~,
\end{flalign}
where we have introduced $\widetilde{X}^\alpha := \Theta^{\beta\alpha}X_\beta$. 
We find that at order $\lambda^0$ (\ref{eqn:coproductsymred}) is fulfilled,
while for the order $\lambda^1$ we obtain the compatibility condition
\begin{flalign}
\label{eqn:symredabcond}
 [X_\alpha,\mathfrak{g}]\subseteq \mathfrak{g}~,~~\forall_\alpha~.
\end{flalign}
Written in terms of generators, there have to be constants $\mathcal{N}_{\alpha i}^{j}\in\bbC$, such that
\begin{flalign}
 \com{X_\alpha}{t_i}=\mathcal{N}_{\alpha i}^{j}\,t_j~.
\end{flalign}
As a consequence we have for all $\widetilde{X}^\alpha$
\begin{flalign}
 \com{\widetilde{X}^\alpha}{t_i} = \Theta^{\beta\alpha}\com{X_\beta}{t_i}
 = \Theta^{\beta\alpha}\mathcal{N}_{\beta i}^j\,t_j =:\widetilde{N}_{i}^{\alpha j} t_j~.
\end{flalign}
Note that provided (\ref{eqn:symredabcond}) holds true, the condition (\ref{eqn:coproductsymred})
is fulfilled to all orders in $\lambda$, since
\begin{flalign}
  [X_{\alpha_1},\cdots [X_{\alpha_n},t_i]= \mathcal{N}_{\alpha_n i}^{i_n}\,\mathcal{N}_{\alpha_{n-1} i_n}^{i_{n-1}}
\cdots\mathcal{N}_{\alpha_1 i_2}^{i_1}\,t_{i_1}\in\mathfrak{g}~.
\end{flalign}

Next, we evaluate the $\star$-Lie bracket condition (\ref{eqn:symredcond2}), more precisely
its expression in terms of the generators (\ref{eqn:starcomsymred}).
Using the explicit form of the inverse twist (\ref{eqn:jsinvtwist}) and (\ref{eqn:symredabcond}) we obtain
\begin{flalign}
\nn \com{t_i}{t_j}_\star &= \sum\limits_{n=0}^\infty\frac{(i\lambda)^n}{2^n\,n!}\, 
\bigl[[\widetilde{X}^{\alpha_1},\cdots[\widetilde{X}^{\alpha_n},t_i]\cdots  ] ,[X_{\alpha_1}, \cdots [X_{\alpha_n},t_j]\cdots] \bigr]\\
\nn &= \sum\limits_{n=0}^\infty\frac{(i\lambda)^n}{2^n\,n!}\, 
\widetilde{\mathcal{N}}_{i}^{\alpha_n i_n}\cdots\widetilde{\mathcal{N}}_{i_2}^{\alpha_1 i_1}\, \mathcal{N}_{\alpha_n j}^{j_n}
\cdots\mathcal{N}_{\alpha_1 j_2}^{j_1}\, [t_{i_1},t_{j_1}] \\
 &= \sum\limits_{n=0}^\infty\frac{(i\lambda)^n}{2^n\,n!}\,
\widetilde{\mathcal{N}}_{i}^{\alpha_n i_n}\cdots\widetilde{\mathcal{N}}_{i_2}^{\alpha_1 i_1}\, \mathcal{N}_{\alpha_n j}^{j_n}
\cdots\mathcal{N}_{\alpha_1 j_2}^{j_1}\,f_{i_1 j_1}^{\quad k}\,t_k~.
\end{flalign}
Thus, the deformed structure constants are 
\begin{flalign}
f_{ij}^{\star~~k} = \sum\limits_{n=0}^\infty\frac{(i\lambda)^n}{2^n\,n!}\,
\widetilde{\mathcal{N}}_{i}^{\alpha_n i_n}\cdots\widetilde{\mathcal{N}}_{i_2}^{\alpha_1 i_1}\, \mathcal{N}_{\alpha_n j}^{j_n}
\cdots\mathcal{N}_{\alpha_1 j_2}^{j_1}\,f_{i_1 j_1}^{\quad k}~,
\end{flalign}
and  (\ref{eqn:starcomsymred}) is automatically satisfied, if (\ref{eqn:symredabcond}) or equivalently
 (\ref{eqn:coproductsymred}) is fulfilled.
We summarize the results in this
\begin{propo}
\label{propo:symred}
Let $\bigl(\mathfrak{g},[\cdot,\cdot]\bigr)\subseteq \bigl(\Xi,[\cdot,\cdot]\bigr)$ be a Lie subalgebra
and $\mathcal{F}\in U\Xi\otimes U\Xi$ be an abelian Drinfel'd twist generated by $X_\alpha$.
Then $\bigl(\mathfrak{g},[\cdot,\cdot]_\star\bigr)\subseteq \bigl(\Xi,[\cdot,\cdot]_\star\bigr)$ is
an almost quantum Lie subalgebra, if and only if
\begin{flalign}
  [X_\alpha,\mathfrak{g}]\subseteq \mathfrak{g}~,~~\forall_\alpha~.
\end{flalign}
In other words, $\bigl(\mathfrak{g},[\cdot,\cdot]_\star\bigr)\subseteq \bigl(\Xi,[\cdot,\cdot]_\star\bigr)$ is
an almost quantum Lie subalgebra, if and only if $\bigl(\mathrm{span}(X_\alpha,t_i),[\cdot,\cdot]\bigr)$ forms a Lie algebra
with ideal $\mathfrak{g}$.
\end{propo}

Let us now show explicitly that the conditions  (\ref{eqn:symredcond}) are weaker than demanding
$(\mathfrak{g},[\cdot,\cdot]_\star)$ to be a quantum Lie algebra, i.e.~demanding instead of  (\ref{eqn:symredcond3})
the condition $\Delta_\star(\mathfrak{g})\subseteq \mathfrak{g}\otimes 1 + U\mathfrak{g}\otimes \mathfrak{g}$.
From (\ref{eqn:symredabcoprod}) we find at first order in $\lambda$ 
\begin{flalign}
  \Delta_\star(t_i)\vert_{\lambda^1} =  -i\, \widetilde{X}^\alpha\otimes [X_\alpha,t_i] = -i\,\mathcal{N}_{\alpha i}^j
\widetilde{X}^\alpha\otimes t_j~.
\end{flalign}
Using that all $X_\alpha$ are linearly independent and thus also all $\widetilde{X}^\alpha$,
the condition $\mathcal{N}_{\alpha i}^j\widetilde{X}^\alpha\in U\mathfrak{g}$ is equivalent to the requirement
\begin{flalign}
  \widetilde{X}^\alpha \in\mathfrak{g}~,~~\text{if}\quad \com{X_\alpha}{\mathfrak{g}}\neq \lbrace 0\rbrace~.
\end{flalign}
This shows that demanding $(\mathfrak{g},[\cdot,\cdot]_\star)$ to be a quantum Lie algebra also restricts
the $X_\alpha$ themselves, while for an almost quantum Lie subalgebra only the action of $X_\alpha$ on $\mathfrak{g}$
is constrained via $[X_a,\mathfrak{g}]\subseteq \mathfrak{g}$.
Note that taking $X_\alpha\in\mathfrak{g}$, $\forall_\alpha$, we always obtain a quantum Lie 
algebra $(\mathfrak{g},[\cdot,\cdot]_\star)$.
This is just the quantum Lie algebra corresponding to the deformed symmetry Hopf algebra
$\bigl(U\mathfrak{g},\mu_\star,\Delta_\star,\epsilon_\star,S_\star\bigr)$.

We study the action of the $\star$-Lie derivative of $(\mathfrak{g},[\cdot,\cdot]_\star)$ on the 
deformed tensor fields. The action of the generators
 $t_i$ on $\tau\in\mathcal{T}$ can be simplified as follows
\begin{flalign}
 \nn \mathcal{L}^\star_{t_i}(\tau) &= \sum\limits_{n=0}^\infty\frac{(i\lambda)^n}{2^n\, n!} 
~\mathcal{L}_{[\widetilde{X}^{\alpha_1},\cdots[\widetilde{X}^{\alpha_n},t_i]\cdots]}\Bigl(
\mathcal{L}_{X_{\alpha_1}\cdots X_{\alpha_n}}(\tau)\Bigr) \\
\label{eqn:starliesimp} &=\sum\limits_{n=0}^\infty\frac{(i\lambda)^n}{2^n\, n!} ~\widetilde{\mathcal{N}}_{i}^{\alpha_n i_n}\cdots
\widetilde{\mathcal{N}}_{i_2}^{\alpha_1 i_1}\,\mathcal{L}_{t_{i_1}X_{\alpha_1}\cdots X_{\alpha_n}}(\tau)~.
\end{flalign}
For invariant tensor fields, the $\star$-Lie derivative has to vanish to all orders in $\lambda$.
We obtain the following useful
\begin{propo}
\label{propo:starinvariance}
 Let $\com{X_\alpha}{\mathfrak{g}}\subseteq\mathfrak{g}~,\forall_\alpha$.
Then a tensor field $\tau\in\mathcal{T}$ is $\star$-invariant under $(\mathfrak{g},\com{\cdot}{\cdot}_\star)$, 
if and only if it is invariant under the undeformed action of $(\mathfrak{g},\com{\cdot}{\cdot})$, i.e.
\begin{flalign}
  \mathcal{L}^{\star}_{\mathfrak{g}}(\tau)=\lbrace 0\rbrace\quad\Leftrightarrow\quad\mathcal{L}_{\mathfrak{g}}(\tau)=\lbrace 0\rbrace~.
\end{flalign}
\end{propo}
\begin{proof}
``$\Rightarrow$'': Let $\tau=\sum_{n=0}^\infty \lambda^n\,\tau_{(n)}$ be a $\star$-invariant tensor field. Using (\ref{eqn:starliesimp})
we obtain for all generators $t_i$
\begin{flalign}
 0=\mathcal{L}_{t_i}^\star(\tau) = 
\sum\limits_{n=0}^\infty\lambda^n \sum\limits_{m+k=n} \frac{i^m}{2^m\, m!} ~\widetilde{\mathcal{N}}_{i}^{\alpha_m i_m}\cdots
\widetilde{\mathcal{N}}_{i_2}^{\alpha_1 i_1}\,\mathcal{L}_{t_{i_1}X_{\alpha_1}\cdots X_{\alpha_m}}(\tau_{(k)})~.
\end{flalign}
At order $n=0$ we obtain $\mathcal{L}_{t_i}(\tau_{(0)}) =0$, for all $i$, and at order $n=1$ we find
\begin{flalign}
\nn  0&=\mathcal{L}_{t_i}(\tau_{(1)}) + \frac{i}{2} \widetilde{\mathcal{N}}^{\alpha j}_i\, \mathcal{L}_{t_j X_\alpha}(\tau_{(0)})
=\mathcal{L}_{t_i}(\tau_{(1)}) + \frac{i}{2} \widetilde{\mathcal{N}}^{\alpha j}_i\, \mathcal{L}_{[t_j, X_\alpha] + X_\alpha t_j}(\tau_{(0)})\\
 &=\mathcal{L}_{t_i}(\tau_{(1)}) - \frac{i}{2} \widetilde{\mathcal{N}}^{\alpha j}_i\,\mathcal{N}_{\alpha j}^k 
\mathcal{L}_{t_k }(\tau_{(0)}) =\mathcal{L}_{t_i}(\tau_{(1)})~,
\end{flalign}
where we have used that $[X_\alpha,\mathfrak{g}]\subseteq \mathfrak{g}$, for all $\alpha$, and the $n=0$ result 
$\mathcal{L}_{t_i}(\tau_{(0)}) =0$, for all $i$. By induction and using $[X_\alpha,\mathfrak{g}]\subseteq \mathfrak{g}$, 
for all $\alpha$, 
one easily shows that $\mathcal{L}_{\mathfrak{g}}(\tau_{(n)})=\lbrace 0\rbrace$ for all $n$. 
Thus, $\mathcal{L}_\mathfrak{g}(\tau)=\lbrace 0\rbrace$ and $\tau$ is a classical $\mathfrak{g}$-invariant tensor field. \vspace{1mm}\newline
``$\Leftarrow$'': Let now $\tau=\sum_{n=0}^\infty \lambda^n\,\tau_{(n)}$ be a classical $\mathfrak{g}$-invariant tensor field.
Using $[X_\alpha,\mathfrak{g}]\subseteq \mathfrak{g}$, for all $\alpha$, to permute successively
 the generators of $\mathfrak{g}$ in (\ref{eqn:starliesimp}) to the very right we obtain that $\tau$ is also $\star$-invariant.

\end{proof}
\noindent This proposition has a desirable consequence for practical applications.
If we want to make the most general ansatz for a tensor field invariant under the deformed symmetry
$\bigl(\mathfrak{g},[\cdot,\cdot]_\star\bigr)$ we can make use of existing results on
$\bigl(\mathfrak{g},[\cdot,\cdot]\bigr)$-invariant tensor fields in classical differential geometry.
In this way we will in particular be able to write down the most general ansatz for a $\star$-invariant
metric field for deformed FRW cosmology or for a deformed Schwarzschild black hole.


\section{\label{sec:symredapplic}Application to cosmology and the Schwarzschild black hole}
\subsection{Deformed FRW universes:}

In this section we investigate flat Friedmann-Robertson-Walker (FRW) universes.
The undeformed Lie algebra $\mathfrak{c}$ of this model is generated by the ``momenta'' $p_i$ and
 ``angular momenta'' $L_i$, $i\in\lbrace1,2,3\rbrace$, 
which are represented in the Lie algebra of vector fields as
\begin{flalign}
 p_i = \partial_i\quad,\quad L_i = \epsilon_{ijk}x^j \partial_k~,
\end{flalign}
where $\epsilon_{ijk}$ is the Levi-Civita symbol.
The undeformed Lie bracket relations are
\begin{flalign}
  \com{p_i}{p_j}=0~,\quad\com{p_i}{L_j}=-\epsilon_{ijk}p_k~,\quad\com{L_i}{L_j}=-\epsilon_{ijk}L_k~.
\end{flalign}

We can now explicitly evaluate the condition each twist vector field $X_\alpha$ has to satisfy given by
 $\com{X_\alpha}{\mathfrak{c}}\subseteq\mathfrak{c}$ (cf.~Proposition \ref{propo:symred}). 
Since the generators are at most linear in the spatial coordinates, $X_\alpha$ can be at most quadratic in order to 
fulfill this condition. If we make a quadratic ansatz with time dependent coefficients we obtain that each $X_\alpha$ 
has to be of the form
\begin{flalign}
\label{eqn:FRWV}
 X_\alpha = X_\alpha^0(t) \partial_t + c_\alpha^i\partial_i + d_\alpha^i L_i + f_\alpha x^i\partial_i~,
\end{flalign}
where $c_\alpha^i,~d_\alpha^i,~f_\alpha\in\bbR$ and $X_\alpha^0(t)\in C^\infty(\bbR)$ 
in order to obtain hermitian deformations.

Since our deformation requires mutually commuting vector fields, we have to implement additional conditions
on the parameters of different $X_\alpha$. A brief calculation shows that the
following requirements have to hold true:
\begin{subequations}
\label{eqn:frwcond}
\begin{align}
 \label{eqn:frwcond1}&d_\alpha^i d_\beta^j \epsilon_{ijk} = 0 ~,\\
 \label{eqn:frwcond2}&c_\alpha^i d_\beta^j \epsilon_{ijk} - c_\beta^i d_\alpha^j \epsilon_{ijk} + f_\alpha c_\beta^k - 
f_\beta c_\alpha^k =0~,\\
 \label{eqn:frwcond3}&\com{X_\alpha^0(t)\partial_t}{X_\beta^0(t)\partial_t}\equiv 0~.
\end{align}
\end{subequations}
As a first step, we  work out all consistent deformations of 
$\mathfrak{c}$ when twisted with two commuting vector fields $X_1$, $X_2$. 
To do this systematically, we divide the solutions of (\ref{eqn:frwcond}) into
 classes depending on the value of $d_\alpha^i$ and $f_\alpha$.
 We use as notation for our cosmologies $\mathfrak{C}_{AB}$, where $A\in\lbrace1,2,3\rbrace$ and $B\in\lbrace1,2\rbrace$, 
which will become clear later, when we sum up the results in Table \ref{tab:frw}.

Type $\mathfrak{C}_{11}$ is defined to be vector fields with $d_1^i=d_2^i=0$ and $f_1=f_2=0$, i.e.~
\begin{flalign}
 X_{1(\mathfrak{C}_{11})} = X_{1}^0(t)\partial_t +c_{1}^i\partial_i~,\quad 
X_{2(\mathfrak{C}_{11})} = X_{2}^0(t)\partial_t +c_{2}^i\partial_i~.
\end{flalign}
These vector fields fulfill the first two conditions (\ref{eqn:frwcond1}) and (\ref{eqn:frwcond2}). 
The solution of the third condition 
(\ref{eqn:frwcond3}) will be discussed later, since the classification we perform now does not depend on it.

Type $\mathfrak{C}_{21}$ is defined to be vector fields with $d_1^i=d_2^i=0$, $f_1\neq0$ and $f_2=0$.
 The first condition (\ref{eqn:frwcond1}) is trivially fulfilled and the second (\ref{eqn:frwcond2}) is fulfilled, 
if and only if $c_{2}^i=0,~\forall_i$, i.e.~type $\mathfrak{C}_{21}$ is given by the vector fields
\begin{flalign}
  X_{1(\mathfrak{C}_{21})} = X_{1}^0(t)\partial_t + c_{1}^i\partial_i 
+ f_{1} x^i\partial_i~,\quad X_{2(\mathfrak{C}_{21})} = X_{2}^0(t)\partial_t~.
\end{flalign}

Solutions with $d_1^i=d_2^i=0$, $f_1\neq0$ and $f_2\neq0$ lie in type 
$\mathfrak{C}_{21}$, since we can perform the twist conserving map
$X_2\to X_2 - \frac{f_2}{f_1} X_1$, which transforms $f_2$ to zero. 
Furthermore, $\mathfrak{C}_{31}$ is defined by 
$d_1^i=d_2^i=0$, $f_1=0$ and $f_2\neq0$ and is equivalent to $\mathfrak{C}_{21}$ 
by interchanging the labels of the vector fields.

We now move on to solutions with without loss of generality $\mathbf{d}_1\neq0$ 
and $\mathbf{d}_2=0$ ($\mathbf{d}$ denotes the vector). 
Note that this class contains also the class with $\mathbf{d}_1\neq0$ and $\mathbf{d}_2\neq0$.
 To see this, we use the first condition (\ref{eqn:frwcond1}) and obtain 
that $\mathbf{d}_1$ and $\mathbf{d}_2$ have to be parallel, 
i.e.~$d^i_2=\kappa d^i_1$.
With this we can transform $\mathbf{d}_2$ to zero by using the twist conserving map $X_2\to X_2 - \kappa X_1$.

Type $\mathfrak{C}_{12}$ is defined to be vector fields with $\mathbf{d}_1\neq0$, $\mathbf{d}_2=0$ and $f_1=f_2=0$.
The first condition (\ref{eqn:frwcond1}) is trivially fulfilled, while the second condition (\ref{eqn:frwcond2}) 
requires that $\mathbf{c}_2$ is parallel to $\mathbf{d}_1$, i.e.~we obtain
\begin{flalign}
 X_{1(\mathfrak{C}_{12})} = X_{1}^0(t)\partial_t + c_{1}^i\partial_i + 
d_{1}^i L_i~,\quad X_{2(\mathfrak{C}_{12})} = X_{2}^0(t)\partial_t + \kappa~d_{1}^i\partial_i~,
\end{flalign}
where $\kappa\in\bbR$ is a constant.

Type $\mathfrak{C}_{22}$ is defined to be vector fields with $\mathbf{d}_1\neq0$, $\mathbf{d}_2=0$, $f_1\neq0$ and $f_2=0$.
Solving the second condition (\ref{eqn:frwcond2}) we obtain $c_2^i=0$ and thus
\begin{flalign}
 X_{1(\mathfrak{C}_{22})} = X_1^0(t)\partial_t + c_{1}^i\partial_i + d_{1}^i L_i + f_{1} x^i\partial_i~,
\quad X_{2(\mathfrak{C}_{22})} = X_{2}^0(t)\partial_t~.
\end{flalign}
 Note that $\mathfrak{C}_{21}$ is contained in $\mathfrak{C}_{22}$ by violating the condition $\mathbf{d}_1\neq0$.

We come to the last class, type $\mathfrak{C}_{32}$, defined by
 $\mathbf{d}_1\neq0$, $\mathbf{d}_2=0$, $f_1=0$ and $f_2\neq0$. 
This class contains also the case $\mathbf{d}_1\neq0$, $\mathbf{d}_2=0$, 
$f_1\neq0$ and $f_2\neq0$ by using the twist
 conserving map $X_1\to X_1 - \frac{f_1}{f_2} X_2$. The vector fields are given by
\begin{flalign}
 X_{1(\mathfrak{C}_{32})} = X_{1}^0(t)\partial_t + \frac{d_{1}^j 
c_{2}^k\epsilon_{jki}}{f_2}  \partial_i + d_{1}^i L_i ~,\quad X_{2(\mathfrak{C}_{32})}
 = X_{2}^0(t)\partial_t + c_{2}^i\partial_i + f_{2} x^i\partial_i~.
\end{flalign}

For a better overview we additionally present the results in Table \ref{tab:frw}, containing all possible two vector field 
deformations $\mathfrak{C}_{AB}$ of the Lie algebra $\mathfrak{c}$.
From this table the notation $\mathfrak{C}_{AB}$ becomes clear.
\begin{table}
\begin{center}
\begin{tabular}{|l|l|l|}
\hline
{\large $~~\mathfrak{C}_{AB}$} & $\mathbf{d}_1=\mathbf{d}_2=0$ & $\mathbf{d}_1\neq0$~,~$\mathbf{d}_2=0$ \\ \hline
$f_1=0$,    &$X_1= X_{1}^0(t)\partial_t +c_{1}^i\partial_i$  &$X_1= X_{1}^0(t)\partial_t + c_{1}^i\partial_i + d_{1}^i L_i$   \\ 
$f_2=0$	   &$X_2= X_{2}^0(t)\partial_t +c_{2}^i\partial_i$  &$X_2= X_{2}^0(t)\partial_t + \kappa~d_{1}^i\partial_i$	     \\ \hline
$f_1\neq0$, &$X_1=X_1^0(t)\partial_t+ c_{1}^i\partial_i + f_{1} x^i\partial_i$  &$X_1= X_1^0(t)\partial_t+ c_{1}^i\partial_i + d_{1}^i L_i + f_{1} x^i\partial_i$    \\ 
$f_2=0$	   &$X_2= X_{2}^0(t)\partial_t$	                    &$X_2= X_{2}^0(t)\partial_t$	                             \\ \hline
$f_1=0$,    &$X_1= X_{1}^0(t)\partial_t$	                    &$X_1= X_{1}^0(t)\partial_t + \frac{1}{f_2}d_{1}^j c_{2}^k\epsilon_{jki}  \partial_i + d_{1}^i L_i$	\\ 
$f_2\neq0$ &$X_2= X_2^0(t)\partial_t + c_{2}^i\partial_i + f_{2} x^i\partial_i$  &$X_2= X_{2}^0(t)\partial_t + c_{2}^i\partial_i + f_{2} x^i\partial_i$	\\ \hline
\end{tabular}
\end{center}
\caption{\label{tab:frw}Two vector field deformations of the cosmological symmetry Lie algebra $\mathfrak{c}$.}
\end{table}

It remains to discuss the solutions of the third condition
 (\ref{eqn:frwcond3}) $\com{X_1^0(t)\partial_t}{X_2^0(t)\partial_t}\equiv0$. 
It is obvious that choosing either 
$X_1^0(t)\equiv0$ or $X_2^0(t)\equiv0$ and the other one arbitrary is a solution.
 Additionally, let us consider solutions with $X_1^0(t)\not \equiv0$ and $X_2^0(t)\not \equiv0$.
 For this case there has to be some point $t_0\in\bbR$, such that without
 loss of generality $X_1^0(t)$ is unequal zero in some open region 
$U\subseteq\bbR$ around $t_0$. 
In this region we can perform the diffeomorphism $t\to \tilde t(t):= \int\limits_{t_0}^t dt^\prime \frac{1}{X_1^0(t^\prime)}$ 
leading to $\tilde X_1^0(\tilde t)\equiv 1$ on $U\subseteq \bbR$. In the coordinate 
$\tilde t$ the third condition (\ref{eqn:frwcond3}) simplifies to
\begin{flalign}
 0\equiv \com{X_1^0(t)\partial_t}{X_2^0(t)\partial_t}=
\com{\tilde X_1^0(\tilde t)\partial_{\tilde t}}{\tilde X_2^0(\tilde t)\partial_{\tilde t}}
 = \Bigl(\partial_{\tilde t}\tilde X_2^0(\tilde t)\Bigr) \partial_{\tilde t}~.
\end{flalign}
This condition is solved if and only if $\tilde X_2^0(\tilde t) \equiv \mathrm{const.}$ on $U\subseteq\bbR$. 
For the subset of analytic functions $C^\omega(\bbR)\subset C^\infty(\bbR)$ we can continue this condition 
to all $\bbR$ and obtain the global relation $X_2^0(t)\equiv \kappa X_1^0(t)$, with some constant $\kappa\in\bbR$.
For non analytic, but smooth functions, we can not continue these 
relations to all $\bbR$ and therefore only obtain local conditions 
 restricting the functions in the overlap of their supports to be linearly dependent. 
In particular, non analytic functions with disjoint supports fulfill the condition (\ref{eqn:frwcond3}) trivially.

After characterizing the possible two vector field deformations of $\mathfrak{c}$ we briefly give a method how to obtain twists
 generated by a larger number of vector fields. For this purpose we use the canonical form of $\Theta$.

Assume that we want to obtain deformations with e.g.~four vector fields.
 Then of course all vector fields have to be of the form (\ref{eqn:FRWV}). 
According to the form of $\Theta$ we have two blocks of vector fields
 $(\alpha,\beta)=(1,2)$ and $(\alpha,\beta)=(3,4)$, in which the classification described
 above for two vector fields can be performed. This means that all four 
vector field twists can be obtained by using two types
 of two vector field twists.
 We label the twist by using a tuple of types, e.g.~$(\mathfrak{C}_{11},\mathfrak{C}_{22})$
 means that $X_1,X_2$ are of type $\mathfrak{C}_{11}$ and $X_3,X_4$ of type $\mathfrak{C}_{22}$.
But this does only assure that $\com{X_\alpha}{X_\beta}=0$ for $(\alpha,\beta)\in\lbrace(1,2),(3,4)\rbrace$
 and we have to demand further restrictions in
 order to fulfill $\com{X_\alpha}{X_\beta}=0$, for all $\alpha,\beta$, and that all vector fields give independent contributions 
to the twist. In particular twists constructed by linearly dependent vector fields can be reduced to a twist 
constructed by a lower number of vector fields.
This method naturally extends to a larger number of vector fields, until we cannot find anymore independent and
 mutually commuting vector fields. 

Let us provide a simple example: Consider a four vector field twist of type $(\mathfrak{C}_{11},\mathfrak{C}_{11})$.
The vector fields $X_\alpha$ mutually commute without imposing further conditions. For the special
choice $X_1=\partial_t$, $X_2=\partial_1$, $X_3=\partial_2$ and $X_4=\partial_3$ we obtain the Moyal-Weyl
twist with  $\Theta$ of rank $4$.

A first qualitative insight into the physics of a given noncommutative spacetime can be obtained by
studying the commutation relations of local coordinates.
For this, we calculate the $\star$-commutator of the linear coordinate functions $x^\mu\in C^\infty(\MM)$ 
for the various types of models up to the first order in the deformation parameter $\lambda$. It is given by
\begin{flalign}
 c^{\mu\nu}:=\starcom{x^\mu}{x^\nu} := x^\mu\star x^\nu -x^\nu\star x^\mu = i\lambda \Theta^{\alpha\beta}
 X_\alpha(x^\mu) X_\beta(x^\nu) +\mathcal{O}(\lambda^2)~.
\end{flalign}
The results are summarized in Table \ref{tab:cosmocom} and show that these commutators can be at most quadratic in 
the spatial coordinates $x^i$.

\begin{table}
\begin{center}
\begin{tabular}{|l|l|}
\hline
Type	&	$c^{\mu\nu}:=\starcom{x^\mu}{x^\nu}$ in $\mathcal{O}(\lambda^1)$\\
\hline
$\mathfrak{C}_{11}$    &	$c^{0i}=i\lambda \bigl( X_{1}^0(t) c_{2}^i - X_{2}^0(t) c_{1}^i  \bigr)$\\
~		&	$c^{ij}=i\lambda \bigl( c_{1}^i c_{2}^j - (i\leftrightarrow j) \bigr)$\\ \hline
$\mathfrak{C}_{21}$    &	$c^{0i}=-i\lambda X_{2}^0(t) \bigl( c_{1}^i + f_{1} x^i \bigr)$\\
~		&	$c^{ij}=0$  \\ \hline
$\mathfrak{C}_{12}$    &	$c^{0i}=i\lambda \bigl( X_{1}^0(t) \kappa d_{1}^i - X_{2}(t) (c_{1}^i+d_{1}^k\epsilon_{kli}x^l )   \bigr)$\\
~		&	$c^{ij}=i\lambda \kappa \bigl( (c_{1}^i +d_{1}^k \epsilon_{kli}x^l ) d_{1}^j - (i\leftrightarrow j)  \bigr)$  \\ \hline
$\mathfrak{C}_{22}$    &	$c^{0i}=-i\lambda X_{2}^0(t) \bigl( c_{1}^i + d_{1}^j\epsilon_{jki}x^k + f_{1} x^i \bigr)$\\
~		&	$c^{ij}=0$  \\ \hline
$\mathfrak{C}_{32}$    &	$c^{0i}=i\lambda \bigl( X_1^0(t) (c_2^i + f_2 x^i) - X_2^0(t) (\frac{1}{f_2}d_1^j c_2^k\epsilon_{jki} + d_1^j\epsilon_{jki}x^k)   \bigr)$\\
~		&	$c^{ij}=i\lambda \bigl( (\frac{1}{f_2} d_1^k c_2^l \epsilon_{kli} +d_1^k\epsilon_{kli} x^l )~(c_2^j + f_2 x^j) - (i\leftrightarrow j)  \bigr)$  \\ \hline
\end{tabular}
\end{center}
\caption{\label{tab:cosmocom} $\star$-commutators in the cosmological models $\mathfrak{C}_{AB}$.}
\end{table}

From the Tables \ref{tab:cosmocom} and \ref{tab:frw} we can extract some rough physics implications of our deformations.
We observe that type $\mathfrak{C}_{11}$ is a Moyal-Weyl type deformation, where the time-space commutators are
allowed to depend on time. The physics described by this model is expected to be similar to the Moyal-Weyl
case, with the advantage that we can vary the time-space deformation in time and for example switch it off
by hand for large times. However, similar to the usual Moyal-Weyl deformation 
(which is a special case of $\mathfrak{C}_{11}$) this model is not suitable for inflationary cosmology,
since the spatial commutators $\starcom{x^i}{x^j}$ are constant in {\it comoving} coordinates,
and as a consequence, the physical scale of noncommutativity grows quadratically in the scale factor.

Let us discuss the physics of $\mathfrak{C}_{22}$, which also contains type $\mathfrak{C}_{21}$. 
Choosing $c_1^i=0$ we obtain that both sides of the time-space commutator are linear in $x^i$. Thus,
the same commutation relations also hold true in physical coordinates and the physical
scale of noncommutativity is not growing in the scale factor. This makes the model attractive for inflationary cosmology.
Even better, when we also choose $d_1^i=0$, the commutation relations are symmetric under
classical rotations, meaning that this deformed universe is still isotropic around the point $x^i=0$
in the classical sense. Since the main feature of noncommutative cosmologies based on the Moyal-Weyl
deformation are preferred directions in the CMB, the model $\mathfrak{C}_{22}$ with the choices of parameters made
above will definitely be distinct.

For the model $\mathfrak{C}_{12}$ the right hand side of the spatial commutator 
is at most linear in the spatial coordinates, which means that the physical scale of noncommutativity
is growing in the scale factor and makes this model not suitable for inflationary cosmology.

Type $\mathfrak{C}_{32}$ contains suitable models for inflationary cosmology, since choosing
$c_2^i=0$, the spatial commutator is quadratic in the spatial coordinates
and the time-space commutator is linear in $x^i$. The physical scale of noncommutativity
is thus not growing in the scale factor. Even better, choosing functions $X_1^0(t)$ and $X_2^0(t)$
which decrease in time we can model a universe which starts in a strongly noncommutative phase
and descends to a phase where the time-space commutators are approximately zero.

\subsection{Deformed Schwarzschild black holes:} 
In this section we investigate abelian twist deformations of Schwarzschild black holes. 
We do this in analogy to the cosmological models and therefore do not have to explain every single step.

The undeformed Lie algebra $\mathfrak{b}$ of the Schwarzschild
black hole is generated by  the vector fields
\begin{flalign}
 p^0=\partial_t\quad,\quad L_i = \epsilon_{ijk} x^j\partial_k~.
\end{flalign}
It can be shown that each twist vector field $X_\alpha$ has to be of the form
\begin{flalign}
 X_\alpha = (c^0_\alpha(r)+N_\alpha^0 t)\partial_t + d_\alpha^i L_i + f_\alpha(r) x^i \partial_i
\end{flalign}
in order to fulfill $\com{X_\alpha}{\mathfrak{b}}\subseteq\mathfrak{b}$. Here $r=\Vert\mathbf{x}\Vert$ is the Euclidean 
norm of the spatial position vector.

For constructing an abelian twist, we additionally have to demand $\com{X_\alpha}{X_\beta} =0,~\forall_{\alpha,\beta}$, 
leading to the conditions
\begin{subequations}
\begin{flalign}
\label{eqn:blackcond1}&d_\alpha^i d_\beta^j\epsilon_{ijk}=0~, \\
\label{eqn:blackcond2}&(f_\alpha(r) x^j \partial_j - N_\alpha^0) c^0_\beta(r) - (f_\beta(r) x^j\partial_j -N_\beta^0) c^0_\alpha(r)
  \equiv 0~,\\
\label{eqn:blackcond3}& f_\alpha(r)f_\beta^\prime(r)-f_\alpha^\prime(r) f_\beta(r) \equiv 0~,
\end{flalign}
\end{subequations}
where $f_\alpha^\prime(r)$ denotes the derivative of $f_\alpha(r)$. 

The next task is to construct the two vector field deformations. 
Note that (\ref{eqn:blackcond3}) is a condition similar to (\ref{eqn:frwcond3}), and 
 has the same type of solutions. Because of this, the functions $f_1(r)$ and $f_2(r)$ have to be parallel in the overlap 
of their supports. From this we can always eliminate locally one $f_\alpha(r)$
by a twist conserving map and simplify the investigation of the condition (\ref{eqn:blackcond2}). 
At the end, the local solutions have to be glued together.
We choose without loss of generality $f_1(r)\equiv0$ for our classification of local solutions.

The solution of (\ref{eqn:blackcond1}) is that the $\mathbf{d}_\alpha$ have to be parallel and we can write
\begin{flalign}
\mathbf{d}_\alpha=\kappa_\alpha \mathbf{d}
\end{flalign}
with constants $\kappa_\alpha\in\bbR$ and some arbitrary vector $\mathbf{d}\neq0$.

We now classify the solutions of (\ref{eqn:blackcond2}) according to $N_\alpha^0$ and $f_2(r)$
 and label them by $\mathfrak{B}_{AB}$. 
We distinguish between $f_2(r)$ being the zero function or not. The result is shown in Table \ref{tab:blackhole}.
Other choices of parameters can be mapped by a twist conserving map into these classes.
Note that in particular for analytic functions $f_\alpha(r)$ the twist conserving map transforming $f_1(r)$ to zero 
can be performed globally, and with this the classification of twists given in Table \ref{tab:blackhole} is global.
\begin{table}
\begin{center}
\begin{tabular}{|l|l|l|}
\hline
 {\large$~~\mathfrak{B}_{AB}$}  &  $f_2(r)\equiv0$  &  $f_2(r)\not \equiv 0$  \\ \hline
$N_1^0=0$,	&$X_1=c_1^0(r)\partial_t +\kappa_1 d^iL_i$	&$X_1=c_1^0\partial_t +\kappa_1 d^i L_i $	\\
$N_2^0=0$	&$X_2=c_2^0(r)\partial_t+\kappa_2 d^i L_i$	&$X_2=c_2^0(r)\partial_t +\kappa_2 d^i L_i + f_2(r) x^i\partial_i$		\\ \hline
$N_1^0\neq0$,	&$X_1=(c_1^0(r)+N_1^0 t)\partial_t $  &$X_1=(c_1^0(r)+N_1^0 t)\partial_t+\kappa_1 d^i L_i$	\\
$N_2^0=0$	&$X_2=\kappa_2 d^i L_i$				&$X_2= -\frac{1}{N_1^0}f_2(r)r c_1^{0\prime}(r)\partial_t +\kappa_2 d^i L_i+f_2(r)x^i\partial_i$\\ \hline
$N_1^0=0$,	&$X_1=\kappa_1 d^i L_i$				&$X_1=c_1^0(r) \partial_t + \kappa_1 d^i L_i$,\quad\text{with (\ref{eqn:blackode})} \\
$N_2^0\neq0$	&$X_2=(c_2^0(r)+N_2^0 t)\partial_t $  &$X_2= (c_2^0(r)+N_2^0 t)\partial_t +\kappa_2 d^i L_i +f_2(r) x^i\partial_i $\\ \hline
\end{tabular}
\end{center}
\caption{\label{tab:blackhole}Two vector field deformations of the black hole symmetry Lie algebra $\mathfrak{b}$. 
Note that $c_1^0(r)\equiv c_1^0$ has to be constant in type $\mathfrak{B}_{12}$.}
\end{table}

In type $\mathfrak{B}_{32}$ we still have to solve a differential equation for $c_1^0(r)$ given by
\begin{flalign}
\label{eqn:blackode}
 c_1^0(r) = \frac{f_2(r)}{N_2^0} r c_1^{0\prime}(r)~,
\end{flalign}
for an arbitrary given $f_2(r)$.
We will not work out the solutions of this differential equation, since type $\mathfrak{B}_{32}$ is a quite unphysical model,
 in which the noncommutativity is increasing linearly in time due to $N_2^0\neq 0$.

The $\star$-commutators $c^{\mu\nu}=\starcom{x^\mu}{x^\nu}$ of the coordinate functions $x^\mu\in C^\infty(\MM)$
 in order $\lambda^1$ for these 
models are given in the Table \ref{tab:blackcom}. They can be used in order to extract physically reasonable 
 models for a noncommutative black hole.

By using the method explained in the previous section, the two vector field twists can be extended to multiple vector field
 twists. Since we do not require these twists in our work and their construction is straightforward, we do not
 present them here.

\begin{table}
\begin{center}
\begin{tabular}{|l|l|}
\hline
Type	&	$c^{\mu\nu}:=\starcom{x^\mu}{x^\nu}$ in $\mathcal{O}(\lambda^1)$\\
\hline
$\mathfrak{B}_{11}$ & $c^{0i}=i\lambda \bigl(c_1^0(r) \kappa_2 -c_2^0(r) \kappa_1\bigr) d^j\epsilon_{jki}x^k$ \\
~		& $c^{ij}=0$ \\ \hline
$\mathfrak{B}_{21}$ & $c^{0i}=i\lambda \bigl(c_1^0(r)+N_1^0 t\bigr) \kappa_2 d^j \epsilon_{jki}x^k$\\
~		 & $c^{ij}=0$\\ \hline
$\mathfrak{B}_{12}$ & $c^{0i} =i\lambda \bigl( c_1^0 (\kappa_2 d^j\epsilon_{jki} x^k+f_2(r) x^i) -\kappa_1 c_2^0(r) d^j\epsilon_{jki}x^k\bigr)$  \\
~		& $c^{ij}= i\lambda\bigl(\kappa_1 d^k\epsilon_{kli}x^l  f_2(r)x^j -(i\leftrightarrow j)\bigr)$ \\ \hline
$\mathfrak{B}_{22}$ &  $c^{0i}=i\lambda \Bigl((c_1^0(r)+N_1^0 t)(\kappa_2 d^j\epsilon_{jki} x^k +f_2(r)x^i) + \frac{1}{N_1^0} f_2(r)rc_1^{0\prime}(r)\kappa_1d^j\epsilon_{jki}x^k \Bigr)$ \\
~		&  $c^{ij}= i\lambda\bigl(\kappa_1 d^k\epsilon_{kli}x^l  f_2(r)x^j -(i\leftrightarrow j)\bigr)$\\ \hline
$\mathfrak{B}_{32}$ &  $c^{0i}=i\lambda \bigl( c_1^0(r)~(\kappa_2 d^j\epsilon_{jki} x^k+f_2(r) x^i) -(c_2^0(r) +N_2^0 t) \kappa_1 d^j\epsilon_{jki}x^k \bigr)$\\
~		&  $c^{ij}=i\lambda\bigl(\kappa_1 d^k\epsilon_{kli}x^l f_2(r)x^j -(i\leftrightarrow j)\bigr)$\\ \hline
\end{tabular}
\end{center}
\caption{\label{tab:blackcom}$\star$-commutators in the black hole models $\mathfrak{B}_{AB}$.}
\end{table}

Let us briefly discuss the physics we expect from the models $\mathfrak{B}_{AB}$
by investigating the Tables \ref{tab:blackhole} and \ref{tab:blackcom}.
Models with $N^0_\alpha\neq 0$ are not expected to be  suitable for deforming a stationary black hole,
since the noncommutativity in these cases grows linearly in time. 
The remaining models are $\mathfrak{B}_{11}$ and $\mathfrak{B}_{12}$.
The time-space commutation relations of $\mathfrak{B}_{11}$ are $r$-dependent through the
functions $c^0_1(r)$ and $c^0_2(r)$. Thus, we can model a noncommutative black hole,
which is strongly noncommutative for small $r$ and almost commutative for $r\to\infty$.
However, in case of $\kappa_\alpha\neq 0$, the model $\mathfrak{B}_{11}$ breaks the classical 
rotation symmetry of the black hole.

A very interesting model for a noncommutative black hole is given by
$\mathfrak{B}_{12}$ with the choice $\kappa_1=\kappa_2=0$. The deformation
is then invariant under all classical black hole symmetries.
From Table \ref{tab:blackcom} we obtain the commutation relations for this 
particular choice of parameters
\begin{flalign}
  \starcom{t}{x^i}= i\lambda\,c_1^0\,f_2(r)\,x^i ~,\quad \starcom{x^i}{x^j}=0~. 
\end{flalign}
Taking a suitable function $f_2(r)$, e.g.~$f_2(r)\propto \exp(-\gamma r)$ with $\gamma>0$, we can model 
a noncommutative black hole which is strongly noncommutative in its vicinity and almost commutative asymptotically
for $r\to\infty$.


\chapter{\label{chap:ncgsol}Exact solutions}
In this chapter we discuss exact solutions of the noncommutative gravity theory introduced in
Chapter \ref{chap:basicncg}. These results have been published in the research article \cite{Ohl:2009pv}
 and the proceedings article \cite{Schenkel:2009qm}.
We also comment briefly on the related work by Schupp and Solodukhin
\cite{Schupp:2007,Schupp:2009pt} and Aschieri and Castellani \cite{Aschieri:2009qh}.


\section{\label{sec:ncgbasis}Noncommutative gravity in the nice basis}
The noncommutative gravity theory presented in Chapter \ref{chap:basicncg}
is formulated in a highly abstract way in order make it suitable for mathematical studies.
However, for practical applications it is convenient to express all formulae of noncommutative
differential geometry and gravity in a suitable (local) basis of vector fields and one-forms.
This allows for straightforward calculations using ``indices'' and helps us to understand
explicit examples in more detail. Rewriting noncommutative gravity in terms
of a suitable basis of vector fields and one-forms is the subject of this section.

Using a generic local basis $\lbrace e_a\in\Xi:a=1,\dots,N \rbrace$ of vector fields
turns out to be not helpful for simplifying the formalism, since the $R$-matrices
appearing in the formulae of Chapter \ref{chap:basicncg} will complicate practical calculations.
When we have discussed the example of the Moyal-Weyl deformation in Chapter \ref{chap:basicncg},
Section  \ref{sec:ncgbaseistein}, we have observed that the resulting formulae take a simple 
form and, in particular, no $R$-matrix appeared. The reason for this is that the basis $\lbrace\partial_\mu:\mu=1,\dots,N\rbrace$
is invariant under the Moyal-Weyl twist, i.e.~the twist is generated by a subspace of vector fields which
commute with all basis vector fields $\partial_\mu$.
Assuming such an invariant basis to exist globally would be too restrictive.
However, as we will see in the next sections, for many of our models discussed in Chapter \ref{chap:symred},
we can find for almost all points $p\in \MM$ an open region $U\ni p$ and a local basis $\lbrace e_a\in\Xi:a=1,\dots,N\rbrace$, 
such that the vector fields $\mathfrak{F}\subseteq \Xi$ generating the twist commute with all $e_a$. 
Even more, we obtain that we can choose the $e_a$ to be mutually commuting, i.e.~$[e_a,e_b]=0$ for all $a,b$.
A basis with these properties is called a {\it nice basis}.
After finishing our article \cite{Ohl:2009pv}, it was shown in \cite{Aschieri:2009qh}
that for a large class of abelian twists a nice basis always exists for open regions around almost all points in $\MM$. 
The basic idea is to show that the open submanifold $\MM_\text{reg}\subseteq \MM$ of regular point, i.e.~points
$p\in\MM$ where there is an open region $U\ni p$, such that $\dim(\mathrm{span}(X_\alpha))$ is constant on $U$,
is dense in $\MM$. Under the assumption that for each $p\in \MM_\text{reg}$ there is an open region $U\subseteq \MM_\text{reg}$,
such that $\Theta^{\alpha\beta}X_\alpha\otimes X_\beta= \widetilde{\Theta}^{\tilde\alpha \tilde\beta}\widetilde{X}_{\tilde\alpha}
\otimes \widetilde{X}_{\tilde\beta}$, where the $\widetilde{X}_{\tilde\alpha}$ are {\it pointwise} linearly independent,
the abelian twist reduces locally to the Moyal-Weyl twist in an appropriate coordinate system on $U$.
The existence of a nice basis for the Moyal-Weyl deformation ensures the existence of a local
nice basis for this class of abelian twists around almost all points $p\in \MM$.
Deformed tensor fields and the $\star$-Levi-Civita connection are then treated locally on $\MM_\text{reg}$
and extended to $\MM$ by demanding smoothness.
Note that, in contrast to what was claimed in \cite{Aschieri:2009qh}, the assumption stated above
does not hold true for all abelian twists, meaning that not all abelian twists reduce locally 
to the Moyal-Weyl twist. These deformations will be called {\it exotic}.
A prime example of an exotic deformation is given by $\MM=\bbR^2$ twisted with $X_1=\partial_x$ and
$X_2=y\partial_x$, where $x,y$ denote global coordinates. Using the normalization $\Theta^{12}=1$ 
we have $\Theta^{\alpha\beta}X_\alpha\otimes X_\beta =\partial_x \otimes y\partial_x - y\partial_x\otimes \partial_x$.
The corresponding twist is nontrivial, since the tensor product is over $\bbC$, even though the
vector fields $X_1$ and $X_2$ are {\it pointwise} linearly dependent.
As we have shown in the letter \cite{Schenkel:2010zi}, see also Chapter \ref{chap:qftapp}, Section \ref{sec:z=2qft},
exotic deformations can lead to very interesting quantum field theories, which completely differ
from the usual Moyal-Weyl deformation. This is expected, since the corresponding twists are not even locally of Moyal-Weyl type.
Exotic deformations of gravity are not yet completely understood and
remain an interesting subject for future investigations. The main problem is the following:
If we can not locally reduce the deformation to the Moyal-Weyl deformation, we can not yet construct a
$\star$-Levi-Civita connection, which is one of the main ingredients of noncommutative gravity.
See, however, Chapter \ref{chap:ncgmath} for a possible way to circumvent this problem by 
quantizing directly commutative Levi-Civita connections.

For the remaining part of this section let $\MM$ be an $N$-dimensional smooth manifold and 
$\mathcal{F}\in U\mathfrak{F}\otimes U\mathfrak{F}$ be a Drinfel'd twist. 
We assume that we have around almost all points a local nice basis of vector fields $\lbrace e_a\in\Xi:a=1,\dots,N\rbrace$, 
i.e.~$[\mathfrak{F},e_a] =\lbrace 0\rbrace$ and $[e_a,e_b]=0$ for all $a,b$.
This is in particular given for all nonexotic abelian Drinfel'd twist.
All formulae which follow will be valid locally in this special basis.

As a consequence of the nice basis $\lbrace e_a\rbrace$, the $\star$-dual basis $\lbrace \theta^a\in\Omega^1:a=1,\dots,N\rbrace$
is equal to the undeformed dual basis, since
\begin{flalign}
 \delta_a^b = \pair{e_a}{\theta^b}_\star = \pair{\bar f^\alpha(e_a)}{\bar f_\alpha(\theta^b)} = \pair{e_a}{\theta^b}~.
\end{flalign}
The opposite $\star$-contraction satisfies
\begin{flalign}
  \pair{\theta^b}{e_a}_\star = \pair{\theta^b}{e_a} = \pair{e_a}{\theta^b} = \delta_a^b~, 
\end{flalign}
thus $\theta^b$ is dual to $e_a$ from both sides.
Employing the Leibniz rule we find for all $v\in\mathfrak{F}$
\begin{flalign}
  0 = v(\delta_a^b)=v(\pair{\theta^b}{e_a}) = \pair{v(\theta^b)}{e_a}~,
\end{flalign}
which implies $\mathcal{L}_\mathfrak{F}(\theta^b)=\lbrace 0\rbrace$ for all $b$, i.e.~$\lbrace \theta^b\rbrace$ is 
an invariant basis of $\Omega^1$.

Let $\tau\in\mathcal{T}$ be an arbitrary $(q,r)$-tensor field. Due to the invariance of the bases $\lbrace e_a\rbrace$ 
and $\lbrace \theta^b\rbrace$ we have locally the two equivalent expressions
\begin{subequations}
\begin{flalign}
 \tau &= \tau_{b_1\dots b_r}^{a_1\dots a_q}~\theta^{b_1}\otimes_A \cdots\otimes_A  \theta^{b_r}\otimes_A e_{a_1}\otimes_A \cdots \otimes_A e_{a_q}\,
~\\
 &=\tau_{b_1\dots b_r}^{a_1\dots a_q}\star\theta^{b_1}\otimes_{A_\star} \cdots\otimes_{A_\star}  \theta^{b_r}\otimes_{A_\star} e_{a_1}\otimes_{A_\star} \cdots \otimes_{A_\star} e_{a_q}~.
\end{flalign}
\end{subequations}
Note that the coefficient functions $\tau_{b_1\dots b_r}^{a_1\dots a_q}$ in both expressions are identical
and that the position where we put $\tau_{b_1\dots b_r}^{a_1\dots a_q}$ does not matter, e.g.~we can also write it on the right.

Let $\nab^\star$ be a $\star$-covariant derivative on $\Xi$. Using the local basis $\lbrace e_a\rbrace$ we
define the Christoffel symbols
\begin{flalign}
  \nab^\star_{e_a}e_b = \Gamma_{ab}^{\star c}~e_c = \Gamma_{ab}^{\star c}\star e_c ~.
\end{flalign}
Consider now a $\star$-covariant derivative $\nab^\star$  on one-forms $\Omega^1$.
In the local basis we can express it in terms of Christoffel symbols
\begin{flalign}
 \nab^\star_{e_a}\theta^b = \widetilde{\Gamma}_{ac}^{\star b}\,\theta^c = \widetilde{\Gamma}_{ac}^{\star b}\star\theta^c ~.
\end{flalign}
Demanding compatibility between the $\star$-covariant derivative on $\Xi$ and $\Omega^1$ we obtain
\begin{flalign}
 \nn 0&=e_a(\delta^c_b)= e_a(\pair{\theta^c}{e_b}_\star) = \pair{\nab^\star_{e_a}\theta^c}{e_b}_\star 
+ \pair{\theta^c}{\nab^\star_{e_a}e_b}_\star  \\
&=  \pair{\widetilde{\Gamma}_{ad}^{\star c}\star\theta^d}{e_b}_\star + \pair{\theta^c}{\Gamma_{ab}^{\star d}\star e_d}_\star = 
\widetilde{\Gamma}_{ab}^{\star c} + \Gamma_{ab}^{\star c}~,
\end{flalign}
i.e.~we have $\widetilde{\Gamma}_{ab}^{\star c} = -\Gamma_{ab}^{\star c}$.
The $\star$-covariant derivative is extended to tensor fields via
\begin{flalign}
  \nab^\star_{e_a}\tau = 
\Bigl(e_a(\tau_{b_1\dots b_r}^{a_1\dots a_q}) - \tau_{\widetilde b\dots b_r}^{a_1\dots a_q} \star \Gamma_{a b_1}^{\star \tilde b} 
- \dots + \tau_{b_1\dots b_r}^{a_1\dots \tilde a}\star \Gamma_{a \tilde a}^{\star a_q} \Bigr)~\theta^{b_1}\otimes_A \cdots \otimes_A e_{a_q}~.
\end{flalign}
For mathematical details on this extension see Part \ref{part:math}.
We denote the coefficient functions of $\nab^\star_{e_a}\tau$ by
\begin{flalign}
 \bigl(\nab^\star_{e_a}\tau\bigr)_{b_1\dots b_r}^{a_1\dots a_q} = e_a(\tau_{b_1\dots b_r}^{a_1\dots a_q}) - \tau_{\widetilde b\dots b_r}^{a_1\dots a_q} \star \Gamma_{a b_1}^{\star \tilde b} 
- \dots + \tau_{b_1\dots b_r}^{a_1\dots \tilde a}\star \Gamma_{a \tilde a}^{\star a_q}  ~.
\end{flalign}
Note that $\tau_{b_1\dots b_r}^{a_1\dots a_q}$ and $\Gamma_{ab}^{\star c}$ in the expressions above are coefficient {\it functions}.
This means that the action of $e_a$ is the Lie derivative on functions and all $\star$-products 
are $\star$-products between functions and do not ``act on the indices''.

The $\star$-torsion and $\star$-curvature of $\nab^\star$ take a simple form in the nice basis $\lbrace e_a\rbrace$. We obtain
\begin{flalign}
 \text{Tor}^\star(e_a,e_b) = (\Gamma_{ab}^{\star c} - \Gamma_{ba}^{\star c}\bigr)~e_c~,
\end{flalign}
and
\begin{flalign}
  \text{Riem}^\star(e_a,e_b,e_c) = \Bigl(e_a(\Gamma_{bc}^{\star d}) - e_b(\Gamma_{ac}^{\star d}) + \Gamma_{bc}^{\star e}\star\Gamma_{a e}^{\star d} -
 \Gamma_{a c}^{\star e}\star \Gamma_{b e}^{\star d}\Bigr)~e_d~.
\end{flalign}
We denote the coefficient functions of these tensors by
\begin{subequations}
\begin{flalign}
\label{eqn:bastorsion} T_{ab}^{\star c}&=\Gamma_{ab}^{\star c} - \Gamma_{ba}^{\star c}~,\\
\label{eqn:basriemann}R_{abc}^{\star ~\, d} &= e_a(\Gamma_{bc}^{\star d}) - e_b(\Gamma_{ac}^{\star d}) + \Gamma_{bc}^{\star e}\star\Gamma_{a e}^{\star d} -
 \Gamma_{a c}^{\star e}\star \Gamma_{b e}^{\star d}~.
\end{flalign}
\end{subequations}
The coefficient functions of the $\star$-Ricci tensor are given by $R^\star_{ab} = R_{cab}^{\star ~\, c}$,
where the sum over $c$ is understood.

Let $g= \theta^a\otimes_A \theta^b \,g_{ab}=\theta^a\otimes_{A_\star} \theta^b\star g_{ab}$ be a metric field.
The coefficient functions are real $g_{ab}=g^\ast_{ab}$ (provided the basis is real), symmetric $g_{ab}=g_{ba}$ and 
invertible as a matrix.
The $\star$-inverse metric field $g^{-1}=e_a\otimes_A e_b\, g^{ab} = e_a\otimes_{A_\star} e_b\star g^{ab}$
is given by the $\star$-inverse matrix
\begin{flalign}
  g_{ab}\star g^{bc} = g^{cb}\star g_{ba} = \delta_a^c~.
\end{flalign}
The $\star$-covariant derivative acts on the metric field as
\begin{flalign}
  \bigl(\nab^\star_{e_a} g\bigr)_{bc} = e_a(g_{bc}) - g_{\tilde b c} \star \Gamma_{ab}^{\star \tilde b}- 
g_{b \tilde c} \star \Gamma_{ac}^{\star \tilde c}~.
\end{flalign}
We can now prove analogously to the undeformed case that there is a unique $\star$-Levi-Civita connection,
i.e.~a $\star$-covariant derivative which is metric compatible and $\star$-torsion free.
The steps in this proof are the same as in the corresponding undeformed version of the theorem.
First, we assume that there is a $\star$-Levi-Civita connection. As a consequence,
we have
\begin{flalign}
\nn 0 &= \bigl(\nab^\star_{e_a} g\bigr)_{bc} + \bigl(\nab^\star_{e_b} g\bigr)_{ac} - \bigl(\nab^\star_{e_c} g\bigr)_{ab}\\
\nn &= e_a(g_{bc}) + e_b(g_{ac}) - e_c(g_{ab}) - g_{dc}\star \Gamma_{ab}^{\star d} - g_{bd}\star \Gamma_{ac}^{\star d} \\
\nn & \quad - 
 g_{dc}\star\Gamma_{ba}^{\star d} -  g_{ad}\star\Gamma_{bc}^{\star d} +  g_{db}\star\Gamma_{ca}^{\star d} +  g_{ad}\star\Gamma_{cb}^{\star d}\\
&\hspace{-3mm}\stackrel{T_{ab}^{\star c}=0}{=} e_a(g_{bc}) + e_b(g_{ac}) - e_c(g_{ab}) -2\,g_{cd}\star\Gamma_{ab}^{\star d}~.
\end{flalign}
This yields an expression for $\Gamma_{ab}^{\star c}$ in terms of the metric
\begin{flalign}
 \label{eqn:starlevicivita} \Gamma_{ab}^{\star c} = \frac{1}{2}\,g^{cd}\star\bigl(e_a(g_{bd}) + e_b(g_{ad}) - e_d(g_{ab})\bigr)~.
\end{flalign}
Thus, the $\star$-Levi-Civita connection is unique. It also exists, since $\nab^\star$ can be constructed from the
Christoffel symbols (\ref{eqn:starlevicivita}).

The $\star$-curvature scalar is given by
\begin{flalign}
 \mathfrak{R}^\star = g^{ab}\star R^\star_{ab}
\end{flalign}
and we obtain the $\star$-Einstein tensor
\begin{flalign}
 \label{eqn:basnceinstein} G^\star_{ab} = R^\star_{ab} - \frac{1}{2}g_{ab}\star \mathfrak{R}^\star~.
\end{flalign}
For the $\star$-Levi-Civita connection (\ref{eqn:starlevicivita}), the $\star$-Einstein tensor only depends
on the metric $g_{ab}$. The equation of motion for vacuum noncommutative Einstein gravity are thus given by
\begin{flalign}
  G^\star_{ab}=0~.
\end{flalign}
We can couple the metric to a suitable deformed stress-energy tensor 
$T^\star= \theta^a\otimes_A \theta^b\, T^\star_{ab}=\theta^a\otimes_{A_\star} \theta^b\star T^\star_{ab}$
via the equation
\begin{flalign}
G^\star_{ab} = 8\pi G_N\,T^\star_{ab}~.
\end{flalign}
The goal of the next two sections is to discuss explicit solutions of these equations.


\section{Cosmological solutions}
We study examples of exact noncommutative gravity solutions. For this we make use of the
cosmological models investigated in Chapter \ref{chap:symred}, Section \ref{sec:symredapplic}.
These models allow for a consistent deformed symmetry reduction and thus are expected to
be appropriate for noncommutative cosmology. From the discussion  
in Chapter \ref{chap:symred}, Section \ref{sec:symredapplic}, we have obtained that
for inflationary cosmology only the model $\mathfrak{C}_{22}$ with $c_1^i=0$ and $\mathfrak{C}_{32}$ with $c_2^i=0$
are suitable, since for all other models the physical scale of noncommutativity grows in the scale
factor\footnote{
Note that for $\mathfrak{C}_{22}$ we can also obtain a physically acceptable model for 
$c_1^i\neq 0$ by demanding a rapidly decreasing function $X_2^0(t)$, such that
$\vert a(t)X_2^0(t)\vert$ does not grow in time, where $a(t)$ is the scale factor.
For late times this model can be approximated by the same model with $d_1^i=f_1=0$
and is equivalent to a time-dependent Moyal-Weyl deformation.
Since we are interested in non-Moyal-Weyl deformations, this model is not of too much interest for us
and we do not discuss it further in the following.
}. Due to inflation this would lead to a highly noncommutative late universe, which we consider as unphysical.
We simplify our discussion and assume the functions $X_1^0(t)$ and $X^0_2(t)$ in model 
$\mathfrak{C}_{22}$ and $\mathfrak{C}_{32}$ to be analytic.
The vector fields $X_\alpha$ generating these twists are then given by
\begin{subequations}
\begin{flalign}
 \qquad\mathfrak{C}_{22}:&\qquad X_{1}=d_1^iL_i+f_1 x^i\partial_i~,\quad &X_2&= X(t)\partial_t~,&\qquad\\
 \qquad\mathfrak{C}_{32}:&\qquad X_{1}=\kappa_1X(t)\partial_t+d_1^iL_i ~,\quad &X_2&=\kappa_2X(t)\partial_t + f_2 x^i\partial_i~.&\qquad
\end{flalign}
\end{subequations}
Here $L_i :=\epsilon_{ijk}x^j\partial_k$ are the generators of rotations.

We can understand our models better by choosing without loss of generality $\mathbf{d}_1=(0,0,d)$ and 
transforming from Cartesian coordinates $x^i$ to spherical coordinates $(r,\zeta,\phi)$.
The vector fields then read
\begin{subequations}
\label{eqn:models}
\begin{flalign}
\label{eqn:c22}
\qquad \mathfrak{C}_{22}:&\qquad X_1=d\partial_\phi + f_1 r\partial_r~,\quad &X_2&=X(t)\partial_t~,&\qquad\\
\label{eqn:c32}
\qquad \mathfrak{C}_{32}:&\qquad X_1=\kappa_1X(t)\partial_t+d \partial_\phi~,\quad &X_2&=\kappa_2 X(t)\partial_t
 + f_2 r\partial_r~.&\qquad
\end{flalign}
\end{subequations}

The nontrivial $\star$-commutation relations between appropriate 
functions in spherical coordinates are:
\begin{subequations}
\label{eqn:commutators}
\begin{align}
\label{eqn:a22}
  \mathfrak{C}_{22}:\; &
    \left\{ \begin{aligned}
               \starcom{t}{\exp i\phi}&=2\exp i\phi~\sinh\Bigl(\frac{\lambda d}{2} X(t)\partial_t\Bigr) t\\
               \starcom{t}{r}         &=-2 i r ~\sin\Bigl(\frac{\lambda f_2}{2} X(t)\partial_t\Bigr)t
            \end{aligned} \right. \\
\label{eqn:a32}
  \mathfrak{C}_{32}:\; &
    \left\{ \begin{aligned}
               \starcom{t}{\exp i\phi}&= 2 \exp i\phi~\sinh\Bigl(\frac{\lambda d\kappa_2}{2} X(t)\partial_t\Bigr) t\\
               \starcom{t}{r}         &=2 i r ~\sin\Bigl(\frac{\lambda f_2\kappa_1}{2} X(t)\partial_t\Bigr)t\\
               \exp i\phi\star r      &= e^{-\lambda d f_2}~r\star \exp i\phi
    \end{aligned} \right.
\end{align}
\end{subequations}
In particular, our models include time-angle, time-radius and angle-radius noncommutativity.
Note that the $\star$-commutators simplify dramatically 
for the choice $X(t)\equiv 1$. This will be further explained below,
 when we discuss specific parameter choices in these models.

For our cosmological models we can make the following choice for a nice local basis
of vector fields
\begin{subequations}
\label{eqn:basiscosmo}
\begin{flalign}
 \label{eqn:basiscosmo1}\qquad &e_1=X(t)\partial_t~,&~~&e_2=r\partial_r~,&~~&e_3=\partial_\zeta~,&~~&e_4=\partial_\phi~,&~~&\text{if }X(t)\not\equiv 0~,&\qquad\\
 \label{eqn:basiscosmo2}\qquad &e_1=\partial_t~,&~~&e_2=r\partial_r~,&~~&e_3=\partial_\zeta~,&~~&e_4=\partial_\phi~,&~~&\text{if }X(t)\equiv 0~.&\qquad
\end{flalign}
\end{subequations}
The required conditions $[e_a,e_b]=0$ and $[X_\alpha,e_b]=0$ hold true for all $a,b,\alpha$. 
Note that (\ref{eqn:basiscosmo1}) is only a local basis in open regions around times $t$, where $X(t)\neq0$.
Since $X(t)$ is assumed to be analytic, the points around which we do not have a nice basis
are a null set in $\MM$. Knowing a smooth tensor field on this dense set, we can use smoothness
in order to extend it to all of $\MM$.

Let us move on to exact cosmological solutions of the noncommutative Einstein equations.
For this we can use Proposition \ref{propo:starinvariance} in order to find the right ansatz for the symmetry reduced 
metric field $g$.
This proposition tells us that a tensor field is invariant under the deformed action of the deformed symmetries, 
if and only if it is invariant under the undeformed action of the undeformed symmetries. 
We make the ansatz
 $g=dx^\mu\otimes_A dx^\nu g_{\mu\nu}$ in the commutative coordinate basis with
\begin{flalign}
 g_{\mu\nu}=\mathrm{diag}\Bigl(-1,a(t)^2,a(t)^2,a(t)^2\Bigr)_{\mu\nu}~,
\end{flalign}
and calculate the required coefficient functions $g_{ab}$ in the nice basis by solving
\begin{flalign}
 \theta^a\otimes_A\theta^b g_{ab}= dx^\mu\otimes_A dx^\nu g_{\mu\nu}~,
\end{flalign}
using the explicit expressions for the nice basis vector fields
(\ref{eqn:basiscosmo}).

It can be checked explicitly that for the choice $f_1=0$ in model $\mathfrak{C}_{22}$ (\ref{eqn:c22}) and 
the choice $\kappa_1= 0$ in model $\mathfrak{C}_{32}$ (\ref{eqn:c32}) the $\star$-Levi-Civita connection
 (\ref{eqn:starlevicivita}), the  $\star$-Riemann tensor (\ref{eqn:basriemann}) and finally the $\star$-Einstein tensor 
(\ref{eqn:basnceinstein}) receive no contributions in the deformation parameter $\lambda$, 
thus reducing to the undeformed counterparts. 
The reason is that, for the restrictions
mentioned above, one twist vector field $X_\alpha\in\mathfrak{c}$ is a Killing vector field and therefore deformed operations among
 invariant tensor fields reduce to the undeformed ones, since $X_\alpha\in \mathfrak{c}$ annihilates the tensors due to invariance.

Since the noncommutative stress-energy tensor of symmetry reduced matter fields reduces to the undeformed one due to 
the same reason,
these noncommutative models are exactly solvable, if the undeformed model is exactly solvable. Note that the reduction of the deformed
 symmetry reduced tensor fields to the undeformed ones does not mean that our models are trivial. In particular, we will obtain in general
 a deformed dynamics for fluctuations on the symmetry reduced backgrounds (see Part \ref{part:qft}),
 as well as a nontrivial coordinate algebra (\ref{eqn:commutators}).

Let us discuss physical implications of the nontrivial coordinate algebras of our models. Consider the model $\mathfrak{C}_{22}$
(\ref{eqn:c22}) with $f_1=0$ and for simplicity $X(t)\equiv1$. The coordinate algebra (\ref{eqn:a22}) 
reduces to the algebra of a quantum mechanical particle on the circle, i.\,e.
\begin{flalign}
 \com{\hat t}{\hat E}=\lambda \hat E~,
\end{flalign}
where we have introduced the abstract operators $\hat t$ and $\hat E:=\widehat{\exp i\phi}$ and set $d=1$. This 
algebra previously appeared e.g.~in
 the context of noncommutative field theory~\cite{Chaichian:2000ia,Balachandran:2004yh} and 
the noncommutative BTZ black hole~\cite{Dolan:2006hv}.
It is well-known that the operator $\hat t$ can be represented as a differential operator acting on the Hilbert
space $L^2(S^1)$ of square integrable functions on the circle and the spectrum can be shown to be given by $\sigma(\hat t)=\lambda
 (\mathbb{Z}+\delta)$, where $\delta\in[0,1)$ labels unitary inequivalent representations. 
The spectrum should be interpreted as possible time eigenvalues, i.e.~our model has a discrete time. 
This feature can be used to realize singularity avoidance in cosmology.
Consider for example an
inflationary background with $a(t)=t^p$, where $p>1$ is a parameter. This so-called power-law inflation can be 
realized by coupling a scalar field with exponential potential to the geometry even in our noncommutative model, 
since the symmetry reduced noncommutative gravity theory reduces to the undeformed one as explained above. Note that 
the scale factor goes to zero at the time $t=0$ and leads
to a singularity in the curvature scalar. But as we have discussed above, the possible time eigenvalues
 are $\lambda (\mathbb{Z}+\delta)$, which does not include the time $t=0$ for $\delta\neq0$.

For the more complicated solvable model $\mathfrak{C}_{32}$ (\ref{eqn:c32}) with $\kappa_1=0$ we obtain time-angle 
and angle-radius noncommutativity. We set without loss of generality the parameters $d=f_2=1$. 
Firstly, we choose $X(t)\equiv 0$ leading to a pure angle-radius noncommutativity. The algebra (\ref{eqn:a32})
becomes in this case
\begin{flalign}
 \hat E \hat r=e^{-\lambda} ~\hat r \hat E~.
\end{flalign}
This algebra can be represented on the Hilbert space $L^2(S^1)$ as
\begin{flalign}
 \hat E = \exp i\phi~,\quad \hat r = L \,\exp{\bigl(-i\lambda \partial_\phi\bigr)}~,
\end{flalign}
leading to the spectrum $\sigma(\hat r)=L\, \exp\bigl(\lambda (\mathbb{Z}+\delta)\bigr)$, 
where $\delta\in[0,1)$ is again a parameter labeling unitary inequivalent representations and $L$ is some length scale. 
The radius becomes discrete. 
Note that this model describes a kind of ``condensating geometry'' around the origin $r=0$, since the shells of 
constant $r$ accumulate at this point.

Secondly, we choose  $X(t)\equiv1$ in (\ref{eqn:a32}).  We obtain the abstract algebra
\begin{flalign}
\label{eqn:alg32}
 \com{\hat t}{\hat E}= \lambda \hat E~,\quad \hat E \hat r= e^{-\lambda} ~\hat r \hat E~. 
\end{flalign}
We find the representation on $L^2(S^1)$
\begin{flalign}
 \hat E = \exp i\phi~,\quad\hat t=\tau \hat 1 - i\lambda \partial_\phi~,\quad \hat r =
 L \,\exp{\bigl(-i\lambda \partial_\phi\bigr)}~.
\end{flalign}
Note that we had to introduce a real parameter $\tau\in\bbR$ and the identity operator $\hat 1$ 
in order to cover the whole spacetime.

The last cosmological model we want to briefly discuss is the isotropic model $\mathfrak{C}_{22}$ with $d=0$ (\ref{eqn:c22}).
It turns out that both $X_\alpha$ are not Killing vector fields and the
symmetry reduced noncommutative gravity theory does not automatically reduce to
the undeformed one. Thus, we expect corrections in $\lambda$ to the $\star$-Einstein equations (\ref{eqn:basnceinstein})
 and its solutions.
However, we have found a special exact solution of this model which is physically very interesting.

Consider the (undeformed) de Sitter space given by $a(t)=\exp H t$, where $H$ is the Hubble parameter.
It turns out that all $\star$-products entering the deformed geometrical quantities 
reduce to the undeformed ones, if $X(t)\equiv1$. Thus, the undeformed de Sitter space solves exactly
the $\star$-Einstein equations in presence of a cosmological constant for this particular choice of twist. 
 Note that in contrast to the solutions above, we have required the explicit form of the scale factor $a(t)$.


\section{Black hole solutions}
We now turn to examples of noncommutative black hole solutions,
making use of the models investigated in Chapter \ref{chap:symred}, Section \ref{sec:symredapplic}.
From the discussion in this section we found that the only physically viable noncommutative black hole models
are $\mathfrak{B}_{11}$ and $\mathfrak{B}_{12}$, since all other models have a noncommutativity growing
linearly in time. 
As a reminder, the twist generating vector fields $X_\alpha$ for these models are given by
\begin{subequations}
\begin{flalign}
 \qquad\mathfrak{B}_{11}:&\qquad X_{1}=c^0_1(r)\partial_t +\kappa_1 d^i L_i~,\quad &X_2&= c_2^0(r)\partial_t~,&\qquad\\
 \qquad\mathfrak{B}_{12}:&\qquad X_{1}=c_1^0 \partial_t +\kappa_1d^iL_i ~,\quad &X_2&=c_2^0(r) \partial_t + \kappa_2 d^i L_i+f_2(r) x^i\partial_i~,&\qquad
\end{flalign}
\end{subequations}
where $c_1^0$ in $\mathfrak{B}_{12}$ has to be constant and $r:=\Vert\mathbf{x}\Vert$ is the radial coordinate. 
We have set without loss of generality $\kappa_2=0$ in $\mathfrak{B}_{11}$.
Let us again choose $\mathbf{d}=(0,0,d)$ and 
transform from Cartesian coordinates $x^i$ to spherical coordinates $(r,\zeta,\phi)$.
The twist generating vector fields read in this basis
\begin{subequations}
\begin{flalign}
\label{eqn:b11} \qquad\mathfrak{B}_{11}:&\qquad X_{1}=c^0_1(r)\partial_t +\kappa_1 \partial_\phi~,\quad &X_2&= c_2^0(r)\partial_t ~,&\qquad\\
\label{eqn:b12} \qquad\mathfrak{B}_{12}:&\qquad X_{1}=c_1^0 \partial_t +\kappa_1\partial_\phi ~,\quad &X_2&=c_2^0(r) \partial_t + \kappa_2 \partial_\phi+f(r) \partial_r~.&\qquad
\end{flalign}
\end{subequations}
Note that we have defined $f(r):=f_2(r) r$ and absorbed the parameter
$d$ into $\kappa_\alpha$ in order to
simplify the expressions.
The nontrivial $\star$-commutation relations between appropriate 
functions in spherical coordinates are:
\begin{subequations}
\begin{align}
\label{eqn:a11}
  \mathfrak{B}_{11}:\; &
    \left\{ \begin{aligned}
               \starcom{t}{\exp i\phi}&= \lambda\,\exp i\phi ~ \kappa_1 c_2^0(r)
            \end{aligned} \right. \\
\label{eqn:a12}
  \mathfrak{B}_{12}:\; &
    \left\{ \begin{aligned}
               \starcom{t}{\exp i\phi} &=\exp i\phi~ \Bigl(2 \sinh\Bigl(\frac{\lambda\kappa_1}{2}\bigl(c_2^0(r)\partial_t
                                          + f(r)\partial_r\bigr)\Bigr) t -\lambda \kappa_2 c_1^0\Bigr)\\
               \starcom{t}{r}          &=i\lambda c_1^0 f(r)~\\
               \starcom{\exp i\phi}{r} &=-2 \exp i\phi ~\sinh\Bigl(\frac{\lambda\kappa_1}{2}f(r)\partial_r\Bigr)r
    \end{aligned} \right.
\end{align}
\end{subequations}
In particular, our models include time-angle, time-radius and angle-radius noncommutativity.
Note that the $\star$-commutators simplify dramatically 
for the choice $f(r)=r$ and $c_2^0(r)\equiv \mathrm{const.}$. This will be 
further explained below, when we discuss specific choices for the parameters.

We now study for which parameter choices the deformation is nonexotic.
This is required to set up the noncommutative Einstein equations. 
We assume that $c_1^0(r)$, $c_2^0(r)$ and $f(r)$ are analytic functions to simplify the investigations.
Model $\mathfrak{B}_{11}$ is nonexotic for $\kappa_1\neq 0$. A nice basis is given by
\begin{subequations}
\begin{flalign}
\label{eqn:basisbh2}
 e_1=X_2~,~~e_2=\partial_r +t \frac{c_2^{0\prime}(r)}{c_2^0(r)}\partial_t + 
\frac{\phi}{\kappa_1}\left(c_1^{0\prime}(r)-\frac{c_1^0(r)c_2^{0\prime}(r)}{c_2^0(r)}\right)\partial_t,~~e_3=\partial_\zeta~,~~
e_4=X_1~~,
\end{flalign}
where $c_1^{0\prime}(r)$ and $c_2^{0\prime}(r)$ denotes the derivative of $c_1^0(r)$ and $c_2^0(r)$, respectively. 
This basis is only defined in open regions around radii $r$ satisfying $c_2^0(r)\neq0$ and
on $\phi\in (0,2\pi)$ (due to the term linear in $\phi$).
Since $c_2^0(r)\not\equiv 0$ is assumed to be nonvanishing 
in order to have a nontrivial deformation and analytic, this is a dense set in $\MM$.

We now come to $\mathfrak{B}_{12}$ with $f(r)\not\equiv0$ (otherwise $\mathfrak{B}_{12}$ becomes part of $\mathfrak{B}_{11}$).
This model is nonexotic and we find the nice basis
\begin{flalign}
\label{eqn:basisbh1}
 e_1=\partial_t~,~~e_2=f(r)\partial_r + c_2^0(r)\partial_t~,~~e_3=\partial_\zeta~,~~e_4=\partial_\phi~,
\end{flalign}
\end{subequations}
which is defined in open regions around radii $r$ with $f(r)\neq 0$. 
Since $f(r)$ is assumed to be analytic this set is dense in $\MM$.

We use again Proposition \ref{propo:starinvariance} and make the ansatz
\begin{flalign}
\label{eqn:bhmetric}
 g_{\mu\nu}=\mathrm{diag}\Bigl(-Q(r),S(r),r^2,(r\sin\zeta)^2  \Bigr)_{\mu\nu}
\end{flalign}
for the metric field $g=dx^\mu\otimes_A dx^\nu g_{\mu\nu}$ in the commutative spherical coordinate basis. The metric in the 
nice basis can be calculated using (\ref{eqn:basisbh2}) or (\ref{eqn:basisbh1}), respectively. 
Concerning the solution of the $\star$-Einstein equations we are in a comfortable position, since 
$X_1\in\mathfrak{b}$ in $\mathfrak{B}_{12}$ is a Killing vector field, which means that the symmetry 
reduced noncommutative gravity theory reduces to the undeformed one. 
Choosing in $\mathfrak{B}_{11}$ either $c_1^0(r)\equiv \text{const.}$ or $c_2^0(r)\equiv \text{const.}$, we obtain
$X_1\in\mathfrak{b}$ or $X_2\in\mathfrak{b}$. Note that in this case $\mathfrak{B}_{11}$ becomes part
of type $\mathfrak{B}_{12}$ by setting $f(r)\equiv 0$.
This leads in the exterior of our noncommutative black hole to the metric (\ref{eqn:bhmetric}) with
\begin{flalign}
 Q(r)=S(r)^{-1} = 1-\frac{r_s}{r}~,
\end{flalign}
where $r_s$ is the Schwarzschild radius. As in the case of the cosmological models, the reduction of the symmetry 
reduced tensor fields to the undeformed counterparts does not mean that our models are trivial. 
Field fluctuations, as well as the coordinate algebras, will in general receive distinct noncommutative effects.

Taking a look at the coordinate algebra of the black hole $\mathfrak{B}_{12}$ (\ref{eqn:a12}), 
we observe that it includes in particular the algebra of a quantum mechanical particle on 
the circle for the time and angle coordinate, if we choose $c_2^0(r)\equiv0$ and $f(r)\equiv 0$. This leads to discrete times.
Another simple choice is $c_1^0=\kappa_2=0$, $c_2^0(r)\equiv0$, $\kappa_1=1$ and $f(r)=r$. The radius spectrum in this case is 
$\sigma(\hat r)=L\,\exp\bigl(\lambda(\mathbb{Z}+\delta) \bigr)$, describing a fine grained geometry around the black hole.
The phenomenological problem with this model is that the spacings between the radius eigenvalues grow exponentially in $r$.
This can be fixed by considering a modified twist like e.\,g.~$c_1^0=\kappa_2=0$, $c_2^0(r)\equiv0$, $\kappa_1=1$ and 
$f(r)=\tanh \frac{r}{L}$, where $L$ is some length scale.
The essential modification is to choose a bounded $f(r)$. 
Consider the coordinate change $r\to \eta=\log \sinh (\frac{r}{L})$.
Note that~$r$ and not~$\eta$ is the physical radial coordinate of a
Schwarzschild observer, $\eta$ is just introduced in order to simplify
the calculation.  In terms of $\eta$ the
algebra (\ref{eqn:a12}) reduces to
\begin{flalign}
 \com{\hat E}{\hat \eta}=-\frac{\lambda}{L} \hat E~,
\end{flalign}
leading to the spectrum $\sigma(\hat \eta)=\frac{\lambda}{L} \bigl(\mathbb{Z}+\delta\bigr)$.
The spectrum of the physical radius $\hat r$ is then given by $\sigma(\hat r) = L~ 
\mathrm{arcsinh}\exp\bigl(\frac{\lambda}{L}(\mathbb{Z}+\delta)\bigr) $. This spectrum approaches constant spacings
 between the eigenvalues for large $r$.

As a last model we consider $\mathfrak{B}_{12}$ with $\kappa_1=\kappa_2=0$, $c_2^0(r)\equiv 0$
and the normalization $c_1^0=1$. This model is invariant under all classical black hole symmetries.
The nontrivial commutation relation is given by
\begin{flalign}
 [\hat t,\hat r] = i\,\lambda\, f(\hat r)~.
\end{flalign}
For the special case $f(\hat r) = \hat r$ we obtain the $\kappa$-commutation relations
 $[\hat t,\hat r] = i\,\lambda\,\hat r$.
The $\kappa$-commutation relations can be represented on $L^2(\bbR)$ by the operators
\begin{flalign}
 \hat t = -i\lambda\,\frac{\partial}{\partial s}~,\quad \hat r = L\, e^{-s}~.
\end{flalign}
This representation is unique up to unitary equivalence \cite{Dabrowski:2010yk}.
The spectrum of the time and radius operator is $\sigma(\hat t)=\bbR$ and $\sigma(\hat r)=(0,\infty)$, respectively.
Thus, there is no discretization of time or radius, however, this spacetime is noncommutative
which leads to distinct effects in the propagation of fields (see Part \ref{part:qft}). 

We omit a discussion of further possible models, since our main
goal was to present the very explicit and simple examples shown
above. We conclude this section with one remark. Our class of black
hole models (\ref{eqn:b12}) has a nonvanishing overlap with the noncommutative 
black hole models found before by Schupp and Solodukhin \cite{Schupp:2007,Schupp:2009pt}. They also
observed that the symmetry reduced dynamics reduces to the undeformed one,
provided the deformation is constructed by sufficiently many Killing vector fields.
 In addition, they have constructed models based on a nonabelian and projective twist, 
that is not part of the class of abelian twists we have used for our investigations. 
Their black holes exhibit discrete radius eigenvalues as well.


\section{\label{sec:ncgsolgen}General statement on noncommutative gravity solutions}
After studying explicit examples of exact solutions of the noncommutative Einstein equations,
we gathered a sufficient understanding in order to make  more general statements on solutions
\cite{Schupp:2007,Schupp:2009pt,Ohl:2009pv,Aschieri:2009qh}. 
Let $g$ be a classical metric field and $\lbrace \Phi_i \rbrace_{i\in\mathcal{I}}$ be a collection
 of classical matter fields which are invariant under some symmetry Lie algebra $\mathfrak{g}\subseteq \Xi$
and solve classical Einstein equations and geometric differential equations for $\lbrace \Phi_i\rbrace_{i\in\mathcal{I}}$.
Consider a deformation by a nonexotic abelian Drinfel'd twist (\ref{eqn:jstwist}) with linearly independent $X_\alpha$.
If now either the vector field $X_\alpha$ or $\widetilde{X}^\alpha=\Theta^{\beta\alpha}X_\beta$
is a Killing vector field for all $\alpha$, then the classical solutions $g$ and $\Phi_i$ also solve the noncommutative
Einstein equations and the corresponding deformed geometric differential equations for $\Phi_i$.
Let us explain why:
Because everything is assumed to be smooth, it is sufficient to check if the noncommutative
equations of motion are solved in open regions around every regular point. In these regions we can choose
a nice basis $\lbrace e_a\rbrace $ satisfying $[e_a,e_b]=0$ and $[X_\alpha,e_a]=0$, for all $a,b,\alpha$,
and its dual $\lbrace \theta^a\rbrace$ satisfying $\mathcal{L}_{X_\alpha}(\theta^b)=0$, for all $b,\alpha$.
Writing a generic $\mathfrak{g}$-invariant tensor field in this basis, the coefficient functions will be annihilated
by all $X_\alpha$ and $\widetilde{X}^\alpha$ which are Killing vector fields, since the basis is 
invariant under these transformations. From the formulae of Section \ref{sec:ncgbasis} in this chapter
we find that all deformed geometric quantities are calculated in the nice basis by using
$\star$-products between coefficient functions. 
The $\star$-products between $\mathfrak{g}$-invariant tensor coefficient functions drop out, since either 
$X_\alpha$ or $\widetilde{X}^\alpha$ is a Killing vector field for all $\alpha$.
The deformed equations of motion then agree with the undeformed ones and thus are solved
exactly. Let us mention again that the equations of motion for fluctuations
around the $\mathfrak{g}$-invariant sector obtain noncommutative corrections, see Part \ref{part:qft}.


\chapter{\label{chap:ncgproblems}Open problems}
In this chapter we comment on open problems in the noncommutative gravity theory \cite{Aschieri:2005zs}, 
which we have found during our investigations.

\subsection*{Exotic deformations and global approach:} 
As discussed in detail above, the existence of a local nice basis
of vector fields is guaranteed only in case of nonexotic abelian twists, i.e.~Drinfel'd twists
which reduce to the Moyal-Weyl twist in open regions around almost all points in $\MM$.
However, for the present formulation of noncommutative gravity, this nice basis of vector fields is required,
since otherwise we can not yet prove the existence and uniqueness of a $\star$-Levi-Civita connection,
one of the main ingredients of noncommutative gravity.
As discussed in Chapter \ref{chap:qftapp}, Section \ref{sec:z=2qft}, there are very interesting exotic 
deformations of scalar quantum field theories, which might also be relevant for noncommutative gravity.
Therefore, it is important to understand noncommutative gravity beyond the nice basis.

Another motivation for a global and basis free understanding of noncommutative gravity
comes from convergent deformations. As compared to formal deformations,
the localization in open regions around points becomes impossible in the convergent setting, since the twist acts as
a nonlocal operator and in general does not preserve these patches. This means that we can not 
use a local nice basis in the convergent case,
but we have to work in a global approach. A first step towards convergent deformations
is thus to understand formal deformations on a global level, without referring to any local bases.
This is a motivation for our global investigations on noncommutative gravity in Part \ref{part:math}.
As a first successful application of this formalism we are going to study in Chapter \ref{chap:ncgmath}
solutions of the noncommutative Einstein equations in a global approach.

\subsection*{Inverse metric field:}
Since this problem also occurs for the Moyal-Weyl deformation of $\bbR^N$,
we restrict ourselves to this case in order to make our argument simple.
Consider a metric field $g = dx^\mu\otimes_{A_\star} dx^\nu\star g_{\mu\nu} = dx^\mu\otimes_A dx^\nu g_{\mu\nu}$,
where we have used that $dx^\mu$ is invariant under the twist.
By definition, the metric is symmetric $g_{\mu\nu}=g_{\nu\mu}$ and  real $g_{\mu\nu}^\ast = g_{\mu\nu}$.
However, the inverse metric field $g^{-1}=g^{\mu\nu}\partial_\mu\otimes_A \partial_\nu$
defined by
\begin{flalign}
\label{eqn:metinv}
 g_{\mu\nu}\star g^{\nu\rho} = g^{\rho\nu}\star g_{\nu\mu}=\delta_\mu^\rho
\end{flalign}
is in general neither symmetric nor real, but only hermitian $g^{\mu\nu\ast} = g^{\nu\mu}$.
This leads to an asymmetry between the metric field and its inverse, and the question arises
which one we should demand to be symmetric and real, thus paying the price that the other one is only hermitian.
Another option is to use hermitian metrics for formulating noncommutative gravity, 
see e.g.~\cite{Chamseddine:1992yx,Chamseddine:2000zu}, which however leads to new degrees
of freedom compared to Einstein's gravity.

A further issue is that the $\star$-Levi-Civita connection is in general not compatible with the inverse metric.
To see this, let us again focus on the Moyal-Weyl deformation of $\bbR^N$. The Christoffel
symbols of the $\star$-Levi-Civita connection are uniquely defined by symmetry $\Gamma_{\mu\nu}^{\star\rho} = 
\Gamma_{\nu\mu}^{\star\rho}$ and metric compatibility
\begin{flalign}
 0=\partial_\mu g_{\nu\rho} -  g_{\tau\rho}\star\Gamma_{\mu\nu}^{\star\tau} - g_{\nu\tau}\star \Gamma_{\mu\rho}^{\star \tau}~,
\end{flalign}
yielding the expression
\begin{flalign}
 \Gamma_{\mu\nu}^{\star\rho} = \frac{1}{2}\,g^{\rho\sigma}\star\bigl(\partial_\mu g_{\nu\sigma} +\partial_{\nu}g_{\mu\sigma}
 - \partial_\sigma g_{\mu\nu}\bigr)~.
\end{flalign}
As a consequence of (\ref{eqn:metinv}) the inverse metric satisfies
\begin{flalign}
 0=\partial_\mu g^{\nu\rho} + \Gamma_{\mu\sigma}^{\star \nu}\star g^{\sigma\rho} + g^{\nu\epsilon}\star g_{\tau\sigma}\star 
\Gamma_{\mu\epsilon}^{\star \tau}\star g^{\sigma\rho}~.
\end{flalign}
For noncentral Christoffel symbols this leads to a violation of the compatibility condition for the inverse metric, i.e.~
\begin{flalign}
 0\neq\partial_\mu g^{\nu\rho} + g^{\tau\rho}\star \Gamma_{\mu\tau}^{\star\nu} + g^{\nu\tau}\star\Gamma_{\mu\tau}^{\star\rho}~.
\end{flalign}
To be explicit, we consider a conformally flat metric field $g_{\mu\nu}= \Phi\,\eta_{\mu\nu}$, where
$\Phi$ is a positive definite function and $\eta_{\mu\nu}$ is the Minkowski metric.
The Christoffel symbols read
\begin{flalign}
 \Gamma_{\mu\nu}^{\star\rho} = 
\frac{1}{2} \,\Phi^{-1_\star}\star\bigl(\partial_\mu\Phi \delta_\nu^\rho + \partial_\nu \Phi \delta_\mu^\rho - \partial^\rho \Phi\eta_{\mu\nu}\bigr)~,
\end{flalign}
where the index of $\partial^\rho$ is raised with the Minkowski metric and $\Phi^{-1_\star}$ is the 
$\star$-inverse function of $\Phi$.
We obtain
\begin{flalign}
 \nn\partial_\mu g^{\nu\rho} +  g^{\tau\rho}\star \Gamma_{\mu\tau}^{\star\nu} + g^{\nu\tau}\star \Gamma_{\mu\tau}^{\star\rho}&=
\bigl(\partial_\mu \Phi^{-1_\star}  + \Phi^{-2_\star}\star\partial_\mu\Phi\bigr)\eta^{\nu\rho}\\
&=\bigl(- \Phi^{-1_\star}\star\partial_\mu\Phi\star\Phi^{-1_\star} + \Phi^{-2_\star}\star\partial_\mu\Phi\bigr)\eta^{\nu\rho}\neq 0~.
\end{flalign}

\subsection*{Covariant derivatives and their extension:}
One problem which we could solve in this thesis is the extension of covariant derivatives to tensor fields, 
see Part \ref{part:math}.
To understand why this improvement was necessary, we show explicitly where the original definition of
\cite{Aschieri:2005zs} breaks down.
As a reminder, a $\star$-covariant derivative
is defined to be a $\bbC$-linear map $\nab^\star_u:\Xi\to\Xi$ satisfying 
\begin{subequations}
\begin{flalign}
\nab^\star_{v+w}z &= \nab^\star_v z +\nab^\star_w z~,\\
\nab^\star_{h\star v} z &= h\star \nab^\star_v z~,\\
\nab^\star_{v}(h\star z) &=\mathcal{L}^\star_{v}(h)\star z + \bar R^\alpha(h)\star \nab^\star_{\bar R_\alpha(v)}z ~.
\end{flalign}
\end{subequations}
The extension of the $\star$-covariant derivative to bivector fields $\Xi\otimes_{A_\star}\Xi$ was given in \cite{Aschieri:2005zs} by
\begin{flalign}
\label{eqn:zwischenschrittopen}
 \nab^\star_v(w\otimes_{A_\star} z) = \bigl(\nab^\star_vw\bigr)\otimes_{A_\star} z 
+ \bar R^\alpha(w)\otimes_{A_\star} \nab^\star_{\bar R_\alpha(v)} z~.
\end{flalign}
Let us study this equation in more detail for the Moyal-Weyl deformation of $\bbR^N$.
By definition of the $\star$-tensor product and invariance of the basis $\partial_\mu$ we have
$(h\star \partial_\nu)\otimes_{A_\star} \partial_\rho = \partial_\nu\otimes_{A_\star}(h\star  \partial_\rho) $, for all
functions $h$.
The $\star$-covariant derivative should respect this middle linearity property in order to be well-defined.
However, we find for (\ref{eqn:zwischenschrittopen}) by an explicit calculation
\begin{subequations}
\begin{flalign}
\nn \nab^\star_{\partial_\mu}((h\star \partial_\nu)\otimes_{A_\star} \partial_\rho)
&= \nab^\star_{\partial_\mu}(h\star \partial_\nu)\otimes_{A_\star} \partial_\rho + h\star\partial_\nu\otimes_{A_\star}
\nab^\star_{\partial_\nu}\partial_\rho\\
\label{eqn:op1}&=\partial_\mu h \star \partial_\nu\otimes_{A_\star} \partial_\rho + h\star 
\bigl(\nab^\star_{\partial_\mu}\partial_\nu\otimes_{A_\star} \partial_\rho + \partial_\nu\otimes_{A_\star} \nab^\star_{\partial_\mu}\partial_\rho\bigr)
\end{flalign}
and
\begin{flalign}
 \nn \nab^\star_{\partial_\mu}(\partial_\nu\otimes_{A_\star} (h\star \partial_\rho)) 
&=\nab^\star_{\partial_\mu}\partial_\nu \otimes_{A_\star} h\star \partial_\rho + 
\partial_\nu\otimes_{A_\star} \nab^\star_{\partial_\mu}(h\star\partial_\rho)\\
\nn &=\nab^\star_{\partial_\mu}\partial_\nu \otimes_{A_\star} h\star \partial_\rho + 
\partial_\nu\otimes_{A_\star}\partial_\mu h\star \partial_\rho + \partial_\nu\otimes_{A_\star} h\star \nab^\star_{\partial_\mu}\partial_\rho\\
\label{eqn:op2}&=\partial_\mu h\star \partial_\nu\otimes_{A_\star}\partial_\rho +\nab^\star_{\partial_\mu}\partial_\nu \otimes_{A_\star} h\star \partial_\rho+
 h\star \partial_\nu\otimes_{A_\star} \nab^\star_{\partial_\mu}\partial_\rho
 ~.
\end{flalign}
\end{subequations}
Note that the second term of (\ref{eqn:op1}) is different to (\ref{eqn:op2}) 
in case the Christoffel symbols are noncentral, which is the generic case.
This problem has been resolved and we show in Chapter \ref{chap:con} that there is a consistent definition
of a $\star$-covariant derivative on tensor fields.

\subsection*{Einstein tensor:}
As already discussed above, there are issues concerning the reality of the inverse metric field.
Similarly, the $\star$-curvature and $\star$-Einstein tensor turn out to be not real.
To show this explicitly let us consider again the Moyal-Weyl deformation of $\bbR^N$ and
a conformally flat metric field $g_{\mu\nu} = \Phi\,\eta_{\mu\nu}$.
The Christoffel symbols read
\begin{flalign}
 \Gamma_{\mu\nu}^{\star\rho} = 
\frac{1}{2} \,\Phi^{-1_\star}\star\bigl(\partial_\mu\Phi \delta_\nu^\rho + \partial_\nu \Phi \delta_\mu^\rho - \partial^\rho \Phi\eta_{\mu\nu}\bigr)
=: u_\mu\delta_\nu^\rho + u_\nu\delta_\mu^\rho - u^\rho\eta_{\mu\nu}~,
\end{flalign}
where we have defined $u_\mu=\Phi^{-1_\star}\star \partial_\mu\Phi /2$.
The $\star$-Ricci tensor reads
\begin{flalign}
 R^\star_{\mu\nu} = \partial_\nu u_\mu - (N-1)\partial_\mu u_\nu - \partial_\rho u^\rho\,\eta_{\mu\nu}
+ (N-2)\,\bigl(u_\mu\star u_\nu -u_\rho\star u^\rho \,\eta_{\mu\nu}\bigr)~
\end{flalign}
and for the $\star$-curvature scalar we find
\begin{flalign}
 \mathfrak{R}^\star = \Phi^{-1_\star}\star \Bigl( 2(1-N)\,\partial_\mu u^\mu - (N-2) (N-1)\, u_\mu\star u^\mu \Bigr)~.
\end{flalign}
Complex conjugation yields
\begin{flalign}
 u_\mu^\ast =\Phi\star u_\mu\star \Phi^{-1_\star}~,\qquad
(\partial_\mu u_\nu)^\ast = \Phi\star \partial_\nu u_\mu\star \Phi^{-1_\star}~.
\end{flalign}
Using this we find for the $\star$-Ricci tensor and the $\star$-curvature scalar
\begin{flalign}
 \bigl(R^\star_{\mu\nu}\bigr)^\ast = \Phi\star R^\star_{\nu\mu}\star \Phi^{-1_\star}~,\qquad 
\bigl(\mathfrak{R}^\star\bigr)^\ast = \Phi^{2_\star}\star \mathfrak{R}^\star \star \Phi^{-2_\star}~.
\end{flalign}
The conjugation of the $\star$-Einstein tensor is
\begin{flalign}
\label{eqn:einsteinconj}
 \bigl(G^\star_{\mu\nu}\bigr)^\ast = \Phi\star G^\star_{\nu\mu}\star \Phi^{-1_\star}~.
\end{flalign}

Reality (or hermiticity) of the $\star$-Einstein tensor is an essential aspect of noncommutative gravity.
Demanding the metric field to be real and symmetric (or hermitian), it is important that the dynamics preserves
this property. From (\ref{eqn:einsteinconj}) one could expect that there is a modified hermiticity property,
namely hermiticity up to $\star$-conjugation by $\Phi$. 
Whether this is an artefact of using a conformally flat metric field is not yet clear.

\subsection*{Discussion:}
Even though the present formulation of noncommutative gravity is plagued by some unsolved issues,
the discussion of our exact solutions in Chapter \ref{chap:ncgsol} is not affected by these problems.
The reason is that we have considered symmetric spacetimes and we have chosen
the deformation such that the symmetry reduced sector of noncommutative gravity is undeformed.
This would also hold true for modified noncommutative Einstein equations, provided
they are constructed according to the principle of deformed general covariance, i.e.~by a twist.
However, for investigating gravitational fluctuations around these highly symmetric noncommutative
gravity backgrounds, a more complete understanding of noncommutative gravity
is required. Since noncommutative gravitons are not only important to study the stability of
noncommutative gravity solutions, but they are also of interest in cosmological applications,
we have resolve the issues
which are present in the current formulation.
This has been the main motivation for Part \ref{part:math} of this thesis.


\part{\label{part:qft}Quantum Field Theory on Noncommutative Curved Spacetimes}


\chapter{\label{chap:qftbas}Basics}

We review the algebraic approach to quantum field theory on commutative curved spacetimes following mostly \cite{Bar:2007zz}.
The focus is on a non-selfinteracting real scalar quantum field, which is coupled to
a given classical gravitational background.

\section{Motivation}
Quantum field theory on commutative curved spacetimes is by now a well established and far developed theory, 
see e.g.~\cite{Wald:1995yp,Bar:2007zz,Bar:2009zzb}.
Early investigations on these theories have shown that the quantum field theory on a generic curved 
spacetime differs drastically from the usual quantum field theory on the Minkowski spacetime, as presented in
every introductory textbook.
The main difference is that while on Minkowski spacetime we have a unique vacuum state for the quantum
field theory, this uniqueness fails to be present in general. This is because
the Minkowski spacetime has a very large isometry group, the Poincar{\'e} group, which poses
strong restrictions on possible vacua, while a generic spacetime has no isometries. 
We therefore do not have sufficiently many criteria at hand to choose from the huge space of all states a distinct one.
Since Stone-von Neumann's theorem does not hold true for infinitely may degrees of freedom, as in quantum field
theory, different choices of vacuum states can lead to physically distinct predictions, because the corresponding
Hilbert space representations are not unitary equivalent. In other words, the choice of vacuum is an input we have to make
which influences the physics results we extract from the quantum field theory.
Another new feature one has to deal with when going from Minkowski to a generic spacetime
is that the concept of particles becomes observer dependent. On Minkowski spacetime we have preferred observers,
living in an inertial frame, and all of them agree when measuring the particle number of a state.
However, for noninertial observers, such as the Unruh observer which moves in some direction with a
constant acceleration $a$, there is not anymore an agreement in the measured particle number,
in particular, the Unruh observer ``sees'' the Minkowski vacuum state as a thermal state with
temperature $T=a/(2\pi)$ in natural units. Since on a generic spacetime there are no preferred observers,
there is also no preferred notion of particles. Hence, on curved spacetimes
quantum field theory is indeed a theory of fields and not of particles.

Due to the nonuniqueness of the vacuum state we have to treat all possible states on the same footing and 
it is therefore not convenient to start with a specific Hilbert space representation from the beginning.
This is the reason why quantum field theories on curved spacetimes are typically treated in the algebraic framework.
The basic idea of this approach is to split the problem of constructing the quantum field theory
into two steps: Firstly, one constructs an abstract algebra of observables for the quantum field theory,
which should include the observables we can in principle measure. Some properties of these observables
can already be studied without referring to a representation. In the second step, one considers
algebraic states on this algebra, which associate to every observable operator its expectation value.
The connection to the usual framework of quantum field theory using Hilbert spaces is then made
by constructing from the algebra and state the GNS-representation.


\section{Some basics of Lorentzian geometry}
We introduce some basic notions of Lorentzian geometry which will appear frequently in this part.
We follow the presentation of \cite{Bar:2007zz}, since it is suitable for our purposes.
More details and a thorough introduction to Lorentzian geometry can e.g.~be found in \cite{0846.53001,0531.53051}.

Let $\MM$ be an $N$-dimensional smooth (Hausdorff) manifold and let $g\in \Omega^1\otimes_A \Omega^1$
be a smooth metric field of signature $(-,+,+,\cdots,+)$. We call the tuple
$(\MM,g)$ a {\it Lorentzian manifold}. At each point $p\in \MM$ the metric field $g_p\in T_p^\ast\MM\otimes T_p^\ast\MM$ 
is equivalent to the Minkowski metric.
We say that a tangent vector $v_p\in T_p\MM$ is {\it timelike}, {\it lightlike} or {\it spacelike},
if $g_p(v_p,v_p)<0$, $g_p(v_p,v_p)=0$ or $g_p(v_p,v_p)>0$, respectively.
A non-spacelike vector is also called a {\it causal} vector.
A {\it time orientation} on a Lorentzian manifold is a smooth vector field $w\in\Xi$
such that at all points $p\in\MM$ the vector $w_p\in T_p\MM$ is timelike.
A Lorentzian manifold together with a choice of such a vector field is called  {\it time-oriented}.
A causal vector $v_p\in T_p\MM$ is called {\it future/past directed} if
the inner product with the time orientation is negative/positive, i.e.~$g_p(w_p,v_p)  \lessgtr 0$.
A piecewise $C^1$-curve $\gamma$ in $\MM$ is called timelike, lightlike, spacelike, causal, future directed, or past directed,
if its tangent vectors are timelike, lightlike, spacelike, causal, future directed, or past directed, respectively.

The {\it causal future} $J_+(p)$ of a point $p\in\MM$ is the set of points that can be reached 
from $p$ by future directed causal curves and $p$ itself. The causal future of a subset
$A\subseteq \MM$ is defined to be $J_+(A):= \bigcup_{p\in A} J_+(p)$.
The {\it causal past} $J_-(A)$ of a subset $A\subseteq\MM$ is defined analogously by replacing future directed 
by past directed.
We will use the notation $J(A):=J_+(A)\cup J_-(A)$ for the domain of causal influence of a set
$A\subseteq \MM$.
Similarly, we define the {\it chronological future/past} $I_\pm(A)$ of a subset $A\subseteq \MM$
by replacing causal curves by timelike curves.

For the most general time-oriented Lorentzian manifold there can be physical pathologies, like
e.g.~the existence of closed causal curves. In order to remove these unphysical examples, we 
have to add more requirements on $(\MM,g)$. We say that a Lorentzian manifold satisfies the
{\it causality condition}, if it does not contain any closed causal curve.
It satisfies the {\it strong causality condition}, if there are no almost closed causal curves.
This means that for each $p\in \MM$ and for each open $U\ni p$, there exists an open neighborhood $U^\prime\subseteq U$
of $p$, such that each causal curve starting and ending in $U^\prime$ is entirely contained in $U$.

We can go on and define a physically very attractive class of Lorentzian manifolds.
\begin{defi}
 A time-oriented connected Lorentzian manifold is called {\it globally hyperbolic} if
it satisfies the strong causality condition and if for any $p,q\in\MM$
the intersection $J_+(p)\cap J_-(q)$ is compact.
\end{defi}
There are equivalent characterizations of globally hyperbolic Lorentzian manifolds, which we want to review
now. For this we first have to define the notion of a {\it Cauchy surface}:
A subset $\Sigma\subset \MM$ of a time-oriented connected Lorentzian manifold $(\MM,g)$
is called a Cauchy surface if each inextendible timelike curve in $\MM$ meets $\Sigma$
at exactly one point.
We can now state without a proof the following theorem \cite{Bar:2007zz}, which is in parts based on the work 
\cite{1081.53059}.
\begin{theo}[\cite{Bar:2007zz}]
 Let $(\MM,g)$ be a time-oriented connected Lorentzian manifold. Then the following
are equivalent:
\begin{itemize}
 \item[(1)] $(\MM,g)$ is globally hyperbolic.
 \item[(2)] There exists a Cauchy surface in $\MM$.
 \item[(3)] $(\MM,g)$ is isometric to $\bbR\times \Sigma$ with metric $-\beta\, dt\otimes_A dt +g_t$,
where $\beta$ is a smooth positive function, $g_t$ is a Riemannian metric on $\Sigma$ depending smoothly
on $t$ and each $\lbrace t\rbrace\times \Sigma$ is a smooth spacelike Cauchy surface in $\MM$.
\end{itemize}
\end{theo}
Let us discuss this theorem. From the physics point of view, the existence of a Cauchy surface
is strongly motivated, since it is a necessary ingredient for formulating a Cauchy problem
for differential equations on $\MM$. Because the existence of a Cauchy surface is equivalent to
global hyperbolicity, we have a motivation for this requirement.
The third equivalent formulation, which was obtained originally in \cite{1081.53059},
is very helpful for checking explicitly if a given Lorentzian manifold is globally hyperbolic.

We finish this section with a technical lemma we require later for our investigations.
\begin{lem}[\cite{Bar:2007zz}]
\label{lem:globhybsupport}
 Let $K,K^\prime\subseteq \MM$ be two compact subsets of a time-oriented, connected and 
globally hyperbolic Lorentzian manifold $(\MM,g)$, then $J_+(K)\cap J_-(K^\prime)$ is compact.
\end{lem}
\noindent For a proof of this lemma see the Appendix A.5.~of \cite{Bar:2007zz}.


\section{Normally hyperbolic operators, Green's operators and the solution space}
An interesting class of wave operators on a Lorentzian manifold $(\MM,g)$ is given by
{\it normally hyperbolic operators} $P:C^\infty(\MM)\to C^\infty(\MM)$.
An operator $P$ is normally hyperbolic, if on all local coordinate patches $U\subseteq \MM$
there are smooth functions $A^\mu,B\in C^\infty(U)$, such that it can be written as
\begin{flalign}
 P\vert_U = g^{\mu\nu}(x)\partial_\mu\partial_\nu + A^\mu(x)\partial_\mu + B(x)~,
\end{flalign}
where  $g^{\mu\nu}$ are the coefficient functions of the inverse metric field. 
An example of such a map is the Klein-Gordon operator $P=\square_g - M^2$,
where $\square_g$ is the d'Alembert operator corresponding to $g$ and $M^2\in[0,\infty)$ is a mass parameter.
An operator $P:C^\infty(\MM)\to C^\infty(\MM)$ is called {\it formally selfadjoint}
with respect to the natural scalar product on $(\MM,g)$
\begin{flalign}
 \bigl(\varphi,\psi\bigr):=\int\limits_\MM\varphi^\ast\,\psi\,\vol~,
\end{flalign}
if we have
\begin{flalign}
 \bigl(\varphi,P(\psi)\bigr)= \bigl(P(\varphi),\psi\bigr)~,
\end{flalign}
for all $\varphi,\psi\in C^\infty(\MM)$ with $\supp(\varphi)\cap \supp(\psi)$ compact.
Here $\vol$ is the canonical volume form on $(\MM,g)$.
The Klein-Gordon operator is formally selfadjoint.

Normally hyperbolic operators $P$ on time-oriented, connected and globally hyperbolic Lorentzian manifolds $(\MM,g)$
have attractive mathematical properties. In particular, the Cauchy problem can always be solved \cite{Bar:2007zz}.
For our purposes we are also interested in the existence of retarded and advanced Green's operators, since
they enter the quantum field theory construction later. Remember the following
\begin{defi}
 Let $(\MM,g)$ be a time-oriented connected Lorentzian manifold and $P:C^\infty(\MM)\to C^\infty(\MM)$
be a normally hyperbolic operator. A linear map $\Delta_\pm:C^\infty_0(\MM)\to C^\infty(\MM)$
satisfying
\begin{subequations}
\label{eqn:comgreenprop}
 \begin{flalign}
  &P\circ \Delta_\pm = \id_{C^\infty_0(\MM)}~,\\
  &\Delta_\pm \circ P\big\vert_{C^\infty_0(\MM)} = \id_{C^\infty_0(\MM)}~,\\
  &\supp\bigl(\Delta_\pm(\varphi)\bigr)\subseteq J_\pm\bigl(\supp(\varphi)\bigr)~,\quad\text{for all } 
\varphi\in C_0^\infty(\MM)~,
 \end{flalign}
\end{subequations}
is called a {\it retarded/advanced Green's operator for $P$}.
\end{defi}

For globally hyperbolic manifolds the existence and uniqueness of Green's operators is guaranteed.
\begin{theo}[\cite{Bar:2007zz}]
\label{theo:greenclas}
 Let $(\MM,g)$ be a time-oriented, connected and globally hyperbolic Lorentzian manifold and $P:C^\infty(\MM)\to C^\infty(\MM)$
be a normally hyperbolic operator.
Then there exist unique retarded and advanced Green's operators $\Delta_\pm: C^\infty_0(\MM)\to C^\infty(\MM)$ for $P$.
\end{theo}
\noindent For a proof we refer to \cite{Bar:2007zz}.

For formally selfadjoint normally hyperbolic operators $P$ we can obtain a useful relation between
the retarded and advanced Green's operators.
\begin{lem}[\cite{Bar:2007zz}]
\label{lem:antihermitianclas}
 Let $(\MM,g)$ be a time-oriented, connected and globally hyperbolic Lorentzian manifold and $P:C^\infty(\MM)\to C^\infty(\MM)$
be a formally selfadjoint normally hyperbolic operator. Let $\Delta_\pm$ be the Green's operators for $P$. Then
\begin{flalign}
 \label{eqn:antihermitian}\bigl(\Delta_\pm(\varphi),\psi\bigr) = \bigl(\varphi,\Delta_{\mp}(\psi)\bigr)~
\end{flalign}
holds true for all $\varphi,\psi\in C^\infty_0(\MM)$.
\end{lem}
\begin{proof}
 We sketch the proof following \cite{Bar:2007zz} in order to understand why in (\ref{eqn:antihermitian})
 $\pm$ on the left hand side is replaced by $\mp$ on the right hand side.
 Using the properties of the Green's operators (\ref{eqn:comgreenprop}) we obtain
\begin{flalign}
\spp{\Delta_\pm(\varphi)}{\psi} = \spp{\Delta_\pm(\varphi)}{P(\Delta_\mp(\psi))}=
\spp{P(\Delta_\pm(\varphi))}{\Delta_\mp(\psi)} = \spp{\varphi}{\Delta_\mp(\psi)}~.
\end{flalign}
For the first equality we have used that $P(\Delta_\mp(\psi)) = \psi$ by definition of the Green's operators.
The second equality holds since $P$ is formally selfadjoint and 
$\supp(\Delta_\pm(\varphi))\cap\supp(\Delta_\mp(\psi)) \subseteq J_\pm(\supp(\varphi))\cap J_\mp(\supp(\psi))$
is compact for a globally hyperbolic manifold, see Lemma \ref{lem:globhybsupport}.
For the last equality we have used that $P(\Delta_\pm(\varphi)) = \varphi$ by definition of the Green's operators.

Note that when we would use $P(\Delta_\pm(\psi)) = \psi$ in the first equality we can not do the required ``integration by parts''
in the second step, since the functions have noncompact overlap. Thus, we understand why $\pm$ on the left hand side is 
replaced by $\mp$ on the right hand side.

\end{proof}

The next step is to construct the solution space of $P$. For this let us make the following definition:
We define the space of {\it spacelike compact functions} $C^\infty_\sc(\MM)\subseteq C^\infty(\MM)$
to be the set of all $\varphi\in C^\infty(\MM)$ for which there is a compact $K\subseteq \MM$
such that $\supp(\varphi)\subseteq J(K)$.
The solution space of $P$ is then defined as
\begin{flalign}
\Sol_P := \bigl\lbrace \varphi\in C^\infty_\sc(\MM): P(\varphi)=0 \bigr\rbrace~.
\end{flalign}
It turns out that for time-oriented, connected and globally hyperbolic Lorentzian manifolds we can 
generate the solution space by using the retarded-advanced Green's operator\footnote{
In this part we use $\Delta$ for the retarded-advanced Green's operator, while we have used
$\Delta$ for the coproduct in a Hopf algebra in Part \ref{part:ncg}. From the context it should
be clear whether $\Delta$ denotes a coproduct or a Green's operator.
}
$\Delta:=\Delta_+ - \Delta_-:C^\infty_0(\MM)\to C^\infty_\sc(\MM)$.
\begin{theo}[\cite{Bar:2007zz}]
\label{theo:classcomplex}
 Let $(\MM,g)$ be a time-oriented, connected and globally hyperbolic Lorentzian manifold and $P:C^\infty(\MM)\to C^\infty(\MM)$
be a normally hyperbolic operator. Let $\Delta_\pm$ be the Green's operators for $P$. Then the sequence of linear
maps
\begin{flalign}\label{eqn:clascomplex}
 0\longrightarrow C^\infty_0(\MM) \stackrel{P}{\longrightarrow} C^\infty_0(\MM)
\stackrel{\Delta}{\longrightarrow} C^\infty_\sc(\MM) \stackrel{P}{\longrightarrow} C^\infty_\sc(\MM)
\end{flalign}
is a complex which is exact everywhere.
\end{theo}
\noindent For the proof of this theorem we refer to \cite{Bar:2007zz}. However, let us show explicitly how we can
construct the solution space from this data.

Since  (\ref{eqn:clascomplex}) is a complex, the composition of two subsequent maps yields zero.
This means in particular that acting with $\Delta$ on compactly supported functions generates solutions of $P$.
Due to the exactness of the complex, all solutions of $P$ are given in this way, i.e.~$\Delta[C^\infty_0(\MM)] = \Sol_P$.
However, $\Delta$ is no isomorphism between compactly supported functions and solutions, since it has a nontrivial kernel.
From the information that (\ref{eqn:clascomplex}) is a complex we know that $P[C^\infty_0(\MM)] \subseteq \Ker(\Delta)$,
and exactness tells us $P[C^\infty_0(\MM)] = \Ker(\Delta)$. The map $\Delta$ thus gives rise to an isomorphism
\begin{flalign}
\label{eqn:solisoclass}
 \mathcal{I}: C^\infty_0(\MM)/P[C^\infty_0(\MM)] \to \Sol_P\,,~[\varphi]\mapsto \Delta(\varphi)~.
\end{flalign}


\section{Symplectic vector space and quantization}
In this section we assume $(\MM,g)$ to be a time-oriented, connected and globally hyperbolic
Lorentzian manifold, and that our scalar field and the normally hyperbolic operator $P$ are real.
Furthermore, we assume $P$ to be formally selfadjoint.
In order to quantize the space of real solutions $\Sol^\bbR_P$ we have to introduce one more structure.
\begin{defi}
 A (weak) {\it symplectic vector space} $\bigl(V,\omega\bigr)$ is a vector space $V$ over $\bbR$ together with
an antisymmetric and $\bbR$-bilinear map $\omega:V\times V\to \bbR$, such that $\omega(v,v^\prime)=0$ for all $v^\prime\in V$
implies $v=0$. The map $\omega$ is called a (weak) symplectic structure.
\end{defi}
Similar to \cite{Bar:2007zz} we suppress the term {\it weak} in the following.
 For the system under consideration there is a canonical symplectic vector space.
Consider  $V:=C^\infty_0(\MM,\bbR)/P[C^\infty_0(\MM,\bbR)]$,
which is isomorphic to $\Sol_P^\bbR$ via (\ref{eqn:solisoclass}), and the $\bbR$-bilinear map 
\begin{flalign}\
\label{eqn:classsymplec}
 \omega: V\times V\to \bbR \,,~([\varphi], [\psi])\mapsto \omega([\varphi],[\psi])=\spp{\varphi}{\Delta(\psi)}~.
\end{flalign}
This map is well-defined and antisymmetric, since due to Lemma \ref{lem:antihermitianclas} we have
\begin{flalign}
 \spp{\varphi}{\Delta(\psi)} = -\spp{\Delta(\varphi)}{\psi}~,
\end{flalign}
for all $\varphi,\psi\in C^\infty_0(\MM,\bbR)$. It is also weakly nondegenerate, 
due to the following reason: Let $\psi\in C^\infty_0(\MM,\bbR)$ be fixed. Then
$\spp{\varphi}{\Delta(\psi)}=0$ for all $\varphi\in C^\infty_0(\MM,\bbR)$ implies $\Delta(\psi)=0$.
Due to Theorem \ref{theo:classcomplex} this means that $\psi\in P[C^\infty_0(\MM,\bbR)]$ and thus
$\psi\sim 0$ is equivalent to zero. Let us summarize the result in this
\begin{propo}
 Let $(\MM,g)$ be a time-oriented, connected and globally hyperbolic Lorentzian manifold
and $P:C^\infty(\MM,\bbR)\to C^\infty(\MM,\bbR)$ be a real and formally selfadjoint normally hyperbolic operator.
We denote the Green's operators for $P$ by $\Delta_\pm$ and $\Delta=\Delta_+ - \Delta_-$. Then $\bigl(V,\omega\bigr)$ with
$V=C^\infty_0(\MM,\bbR)/P[C^\infty_0(\MM,\bbR)]$ and $\omega$ defined in (\ref{eqn:classsymplec})
is a symplectic vector space.
\end{propo}

This symplectic vector space can be quantized in terms of CCR-algebras (also called Weyl algebras),
 yielding the observable algebras for the quantum field theory.
For this we first have to define a Weyl system.
\begin{defi}
 A {\it Weyl system} of a symplectic vector space $\bigl(V,\omega\bigr)$ consists of a unital $C^\ast$-algebra $A$
and a map $\mathcal{W}:V\to A$ such that for all $v,u\in V$ we have
\begin{subequations}
 \begin{flalign}
  \mathcal{W}(0)&=1~,\\
  \mathcal{W}(-v) &= \mathcal{W}(v)^\ast~,\\
  \mathcal{W}(v)\,\mathcal{W}(u) &= e^{-i\,\omega(v,u)/2}\,\mathcal{W}(v+u)~.
 \end{flalign}
\end{subequations}
\end{defi}
Note that the symplectic structure $\omega$ determines the noncommutativity of the Weyl system.
Associated to a Weyl system there is the notion of a CCR-algebra.
\begin{defi}
 A Weyl system $\bigl(A,\mathcal{W}\bigr)$ of a symplectic vector space $\bigl(V,\omega\bigr)$ 
is a {\it CCR-representation} of $\bigl(V,\omega\bigr)$ if $A$ is generated
as a $C^\ast$-algebra by the elements $\mathcal{W}(v)$, $v\in V$. In this case $A$ is called a {\it CCR-algebra}
of $\bigl(V,\omega\bigr)$.
\end{defi}

There are strong results on mathematical properties of Weyl systems, see e.g.~\cite{Bar:2007zz} Chapter 4.
For our purpose it is sufficient to note the uniqueness (up to $\ast$-isomorphisms)
of the CCR-representation.
\begin{theo}[\cite{Bar:2007zz}]
\label{theo:ccrunique}
 Let $\bigl(V,\omega\bigr)$ be a symplectic vector space and let $\bigl(A_1,\mathcal{W}_1\bigr)$
and $\bigl(A_2,\mathcal{W}_2\bigr)$ be two CCR-representations of $\bigl(V,\omega\bigr)$.
Then there exists a unique $\ast$-isomorphism $\pi:A_1\to A_2$, such that
$\pi(\mathcal{W}_1(v))=\mathcal{W}_2(v)$ for all $v\in V$.
\end{theo}
\noindent For a proof we refer to \cite{Bar:2007zz}. The existence of a CCR-representation has
been shown in \cite{Bar:2007zz} by an explicit construction.

A corollary of this theorem which we will require later reads
\begin{cor}[\cite{Bar:2007zz}]
\label{cor:asthomoweyl}
 Let $\bigl(V_1,\omega_1\bigr)$ and $\bigl(V_2,\omega_2\bigr)$ be two symplectic vector spaces and let
$S:V_1\to V_2$ be a symplectic linear map, i.e.~$\omega_2(S(v),S(u))=\omega_1(v,u)$ for all $v,u\in V_1$.
Then there exists a unique injective $\ast$-homomorphism $\mathfrak{S}:A_1\to A_2$ between the CCR-algebras
$A_1$ and $A_2$, such that $\mathfrak{S}(\mathcal{W}_1(v)) = \mathcal{W}_2(S(v))$, for all $v\in V_1$.
\end{cor}
Roughly speaking, this corollary shows that we can canonically associate
to each symplectic linear map a $\ast$-homomorphism between the corresponding CCR-algebras.
In practical applications, one typically has the action of symmetries, e.g.~isometries,
on the symplectic vector space of the field theory. This corollary implies that
we also have a natural action of the same transformations on the CCR-algebra.


\section{Algebraic states and representations}
Provided a $C^\ast$-algebra $A$ of observables we still require a way to associate to each element $a\in A$
its expectation value. This leads to the notion of an algebraic state.
\begin{defi}
 An {\it algebraic state} (or simply {\it state}) on a $C^\ast$-algebra $A$ is a $\bbC$-linear map $\Omega:A\to\bbC$
such that
 \begin{flalign}
  \Omega(1)=1~,\quad \Omega(a^\ast a) \geq0~,\quad\text{for all }a\in A~.
 \end{flalign}
$\Omega$ is called faithful, if $\Omega(a^\ast a)=0$ implies $a=0$.
\end{defi}
This now would be sufficient structure to make certain physics predictions, i.e.~calculate 
expectation values of suitable observables $a\in A$. However, in order to make contact
to the usual framework of quantum field theory based on Hilbert spaces we require one more
step, namely a representation of the algebra $A$ on a Hilbert space $\mathcal{H}$.
\begin{defi}
 Let $A$ be a $C^\ast$-algebra. A {\it representation} of $A$ is a tuple $\bigl(\mathcal{H},\pi\bigr)$,
where $\mathcal{H}$  is a Hilbert space and $\pi:A\to B(\mathcal{H})$ is a $\ast$-homomorphism into 
the algebra of bounded operators on $\mathcal{H}$. The representation is called {\it faithful}, if $\pi$ is injective.\vspace{1mm}

\noindent A vector $\ket{0}\in \mathcal{H}$ is called
{\it cyclic}, if $\pi[A]\ket{0}$  is dense in $\mathcal{H}$.\vspace{1mm}

\noindent A triple $\bigl(\mathcal{H},\pi,\ket{0}\bigr)$ is called a {\it cyclic representation} of $A$, if
$\bigl(\mathcal{H},\pi\bigr)$ is a representation of $A$ and $\ket{0}\in\mathcal{H}$ is cyclic.
\end{defi}
A powerful, and even constructive, theorem on representations of $C^\ast$-algebras has been found by Gel'fand, Naimark and 
Segal. 
\begin{theo}[GNS]~\,~
\begin{itemize}
\item[(i)] Let $A$ be a $C^\ast$-algebra and $\Omega$ a state on $A$.
Then there exists a cyclic representation $\bigl(\mathcal{H},\pi,\ket{0}\bigr)$ of $A$, such that
\begin{flalign}
 \Omega(a) = \bra{0} \pi(a)\ket{0}~,\quad\text{for all } a\in A~.
\end{flalign}
The triple $\bigl(\mathcal{H},\pi,\ket{0}\bigr)$ is called the GNS-representation of $\Omega$.\vspace{1mm}

\item[(ii)]$\bigl(\mathcal{H},\pi,\ket{0}\bigr)$ is unique up to unitary equivalence, i.e.~let
$\bigl(\widetilde{\mathcal{H}},\widetilde{\pi},\widetilde{\ket{0}}\bigr)$ be another cyclic representation of
$A$ satisfying $\Omega(a)=\widetilde{\bra{0}}\widetilde{\pi}(a)\widetilde{\ket{0}}$ for all $a\in A$, 
then there is a unitary $U:\mathcal{H}\to \widetilde{\mathcal{H}}$, such that $U\pi(a)U^{-1}=\widetilde{\pi}(a)$, for
all $a\in A$, and $U\ket{0} =\widetilde{\ket{0}}$.

\end{itemize}
\end{theo}
\noindent A proof of this theorem can be found in every textbook on $C^\ast$-algebras, e.g.~\cite{0905.46046}.

The GNS-Theorem shows that we can always go over from the algebraic approach using $C^\ast$-algebras
and algebraic states to a ``conventional'' Hilbert space approach to quantum field theory.
As a last remark we want to mention that the GNS-representation is not restricted to $C^\ast$-algebras and a similar theorem
also applies to unital $\ast$-algebras without a norm: Provided a unital $\ast$-algebra $A$
and an algebraic state $\Omega$ on $A$, there is a dense subspace $\mathcal{D}\subseteq \mathcal{H}$ of
a Hilbert space $\mathcal{H}$, a representation $\pi:A\to L(\mathcal{D})$ in terms of linear operators
on $\mathcal{D}$ and a vector $\ket{0}\in\mathcal{D}$, such that
\begin{subequations}
\begin{flalign}
 \Omega(a) &= \bra{0}\pi(a)\ket{0}~,\quad\text{for all }a\in A~,\\
 \mathcal{D} &= \pi[A]\ket{0} ~.
\end{flalign}
\end{subequations}

For quantum field theory there are natural conditions we can demand for the state $\Omega$.
To explain these, consider the CCR-algebra $A$ of a symplectic vector space $\bigl(V,\omega\bigr)$
with state $\Omega$ and the GNS-representation $\bigl(\mathcal{H},\pi,\ket{0}\bigr)$ of $\Omega$.
In many cases of physical interest there is a group $G$ of transformations acting on the algebra $A$, 
for example the group of isometries of symmetric manifolds $(\MM,g)$. Since we want to interpret $\Omega$ 
as a vacuum state, it is natural to demand it to be invariant under $G$. More technically, we say that
a state $\Omega$ is {\it $G$-symmetric}, if
\begin{flalign}
 \Omega(\alpha(a)) = \Omega(a)~,
\end{flalign}
for all $a\in A$ and $\alpha\in G$.
Another natural requirement is {\it regularity}, i.e.~we assume the one-parameter family of 
operators $\pi(\mathcal{W}(t\, v))$, $t\in\bbR$, to be strongly continuous for all $v\in V$.
Then there exist selfadjoint generators $\Phi_\pi(v)\in L(\mathcal{D})$, where $\mathcal{D}\subseteq \mathcal{H}$
is a dense subspace, such that formally $\pi(\mathcal{W}(v)) = e^{i\Phi_\pi(v)}$.
The operators $\Phi_\pi(v)$ are interpreted as smeared linear field operators and we define the smeared
$n$-point correlation functions by
\begin{flalign}
 \Omega_n(v_1,\dots,v_n):= \bra{0}\Phi_\pi(v_1)\cdots \Phi_\pi(v_n)\ket{0}~.
\end{flalign}
Since we are dealing with free field theories, only coupled to a fixed gravity background,
it is also natural to demand that all $n$-point functions factorize into 
products of $2$-point functions. States $\Omega$ with this property are called {\it quasi-free}.
A last condition which is typically demanded is the Hadamard condition, or equivalently the 
{\it microlocal spectrum condition} ($\mu$SC) \cite{Radzikowski:1996pa,Brunetti:1995rf}.
A regular quasi-free state $\Omega$ is said to fulfill the $\mu$SC, if the wavefront set of its 
$2$-point correlation function regarded as a distribution satisfies
\begin{flalign}
  \text{WF}(\Omega_2) \subset \bigl\lbrace
(x_1,k_1),(x_2,-k_2)\in T^\ast\MM^2\setminus\lbrace 0\rbrace :
(x_1,k_1)\sim (x_2,k_2) \text{ and } k_1^0\geq 0
\bigr\rbrace~,
\end{flalign}
where the relation $(x_1,k_1)\sim (x_2,k_2)$ means that there exists a lightlike geodesic
 connecting $x_1$ and $x_2$ with cotangent vectors $k_1$ at $x_1$ and $k_2$ at $x_2$. 
For an introduction to microlocal analysis and wavefront sets we refer to \cite{0712.35001}.
 Let us skip the mathematical details at this moment and discuss the consequences of the $\mu$SC.
The $\mu$SC poses restrictions on the singularities of the $2$-point function.
Provided it holds true, one can construct Wick-polynomials of field operators $\Phi_\pi(x)$
and eventually the stress-energy tensor. Thus, the $\mu$SC is an important condition which
is essential for extending the algebra of observables by physically important operators, such
as interaction operators or the stress-energy tensor, which all require the product of 
quantum fields at the same point. We do not discuss further developments in modern 
quantum field theory on curved spacetimes, since we do not need them for our work,
 and refer to the book \cite{Bar:2009zzb} for an overview of these topics.


\chapter{\label{chap:qftdef}Formalism}
We present the algebraic approach to quantum field theory on noncommutative
curved spacetimes developed in \cite{Ohl:2009qe}. 


\section{Notation}
In this chapter we will work in a formal deformation quantization setting.
We are going to use a more strict mathematical language,
since compared to the physical applications in Part \ref{part:ncg}
we are now going to develop a mathematical formalism, which requires this precision in notation. 
An extensive review on formal power series extensions of fields, vector spaces and linear maps, 
and their topological properties is given in the Appendix \ref{app:basicsdefq}. In this section we are
going to state the definitions and properties which are essential for this chapter.

As already discussed in Part \ref{part:ncg}, formal deformation quantization requires us to 
replace the underlying field $\bbK=\bbC$ or $\bbR$ by the commutative and unital 
ring of formal power series $\bbK[[\lambda]]$, where $\lambda$ denotes the deformation parameter.
Elements of $\bbK[[\lambda]]$ are given by $\bbK[[\lambda]]\ni\beta = \sum_{n=0}^\infty\lambda^n\,\beta_{(n)}$,
where $\beta_{(n)}\in\bbK$ for all $n$. The sum and product on $\bbK[[\lambda]]$ is induced from the field structure
on $\bbK$ and reads
\begin{flalign}
\label{eqn:sumprodform}
 \beta + \gamma :=\sum\limits_{n=0}^\infty\lambda^n\,\bigl(\beta_{(n)}+\gamma_{(n)}\bigr)~,\quad \beta\,\gamma := 
\sum\limits_{n=0}^\infty\lambda^n \sum\limits_{m+k=n} \beta_{(m)}\gamma_{(k)}~,
\end{flalign}
for all $\beta,\gamma\in\bbK[[\lambda]]$.

Let $V$ be a vector space over $\bbK$. Its formal power series extension $V[[\lambda]]$ can be equipped
with a $\bbK[[\lambda]]$-module structure by defining
\begin{flalign}
 v+v^\prime := \sum\limits_{n=0}^\infty\lambda^n\,\bigl(v_{(n)}+v^\prime_{(n)}\bigr)~,\quad \beta\,v := 
\sum\limits_{n=0}^\infty\lambda^n\sum\limits_{m+k=n} \beta_{(m)}\,v_{(k)}~,
\end{flalign}
for all $\beta\in\bbK[[\lambda]]$ and $v,v^\prime \in V[[\lambda]]$.
Note that a $\bbK[[\lambda]]$-module is not necessarily of the type $V[[\lambda]]$.
If it is, it is called a {\it topologically free $\bbK[[\lambda]]$-module}.
A {\it $\bbK[[\lambda]]$-module homomorphism} is a map between two $\bbK[[\lambda]]$-modules preserving
the $\bbK[[\lambda]]$-module structure. We also call these maps {\it $\bbK[[\lambda]]$-linear maps}.
Note that a family of $\bbK$-linear maps $P_{(n)}:V\to W$, $n\in \bbN^0$, between the vector spaces $V,W$
induces a $\bbK[[\lambda]]$-linear map $P_\star:V[[\lambda]]\to W[[\lambda]]$ by defining
for all $v\in V[[\lambda]]$
\begin{flalign}
\label{eqn:mapformal}
 P_\star(v):= \sum\limits_{n=0}^\infty\lambda^n\sum\limits_{m+k=n}P_{(m)}(v_{(k)})~.
\end{flalign}
We shall use the notation $P_\star = \sum_{n=0}^\infty \lambda^n P_{(n)}$. The other way around,
let $P_\star :V[[\lambda]] \to W[[\lambda]]$ be a $\bbK[[\lambda]]$-linear map, then it gives rise to 
a family of $\bbK$-linear maps $P_{(n)}:V\to W$, $n\in\bbN^0$, defined by
$P_{(n)}(v):= \bigl(P_\star(v)\bigr)_{(n)}$, for all $v\in V$. $P_\star$ can be expressed
in terms of the $P_{(n)}$ by (\ref{eqn:mapformal}).\footnote{
For this proof one requires the $\lambda$-adic topology on $V[[\lambda]]$, see 
Proposition \ref{eqn:topmodulemaps}.
}

An {\it algebra $A$ over $\bbK[[\lambda]]$} is a $\bbK[[\lambda]]$-module together with 
an associative $\bbK[[\lambda]]$-bilinear map $\mu: A\times A\to A\,,~(a, b)\mapsto\mu(a,b)=a\,b$
(product).  An algebra $A$ over $\bbK[[\lambda]]$ is called unital if it contains a unit element $1\in A$.
For topological algebras, which are important for treating Drinfel'd twist deformations
mathematically precise, see Appendix \ref{app:basicsdefq}.
A {\it $\ast$-algebra $A$ over $\bbC[[\lambda]]$} is an algebra over $\bbC[[\lambda]]$ equipped with a
$\bbC[[\lambda]]$-antilinear map $\ast:A\to A$ (involution) satisfying $(a\,b)^\ast=b^\ast\,a^\ast$ and $(a^\ast)^\ast=a$,
for all $a,b\in A$.
A {\it $\ast$-algebra homomorphism} is an algebra homomorphism $\kappa:A\to B$ satisfying $\bigl(\kappa(a)\bigr)^\ast
=\kappa(a^\ast)$, for all $a\in A$.


\section{Condensed review of noncommutative differential geometry}
In order to construct the deformed quantum field theory we are making use of the noncommutative 
differential geometry presented in Chapter \ref{chap:basicncg}.
Let us briefly review the essential tools in a very compact manner.
 As new elements, we will focus on reality properties and integration in this section.

Let $\MM$ be an $N$-dimensional smooth manifold and $\Xi$ be the smooth and complex vector fields
on $\MM$. The formal power series extension of the Hopf algebra of diffeomorphisms
is denoted by $H=\bigl(U\Xi[[\lambda]],\mu,\Delta,\epsilon,S\bigr)$.
Let $\mathcal{F}\in (U\Xi\otimes U\Xi)[[\lambda]]$ be a Drinfel'd twist of $H$, i.e.~$\mathcal{F}$
is invertible and satisfies
\begin{subequations}
\begin{flalign}
 \mathcal{F}_{12}\,(\Delta\otimes \id)\mathcal{F}&=\mathcal{F}_{23}\,(\id\otimes \Delta)\mathcal{F}~,\\
 (\epsilon\otimes \id)\mathcal{F} &=1 = (\id\otimes \epsilon) \mathcal{F}~,
\end{flalign}
\end{subequations}
where $\mathcal{F}_{12}=\mathcal{F}\otimes 1$ and $\mathcal{F}_{23} = 1\otimes \mathcal{F}$.
We additionally demand that $\mathcal{F}=1\otimes 1 +\mathcal{O}(\lambda)$ in order to have a good classical
limit. The inverse twist is denoted by (sum over $\alpha$ understood)
\begin{flalign}
 \mathcal{F}^{-1} = \bar f^\alpha\otimes \bar f_\alpha\in (U\Xi\otimes U\Xi)[[\lambda]]~.
\end{flalign}

In this chapter we are going to focus on reality properties of our theories.
It is convenient to choose a real twist, i.e.~to demand that $\mathcal{F}^{\ast\otimes \ast} = (S\otimes S)\mathcal{F}_{21}$.
As a consequence, the $\star$-product in the deformed algebra of functions $\bigl(C^\infty(\MM)[[\lambda]],\star\bigr)$
defined by $h\star k := \bar f^\alpha(h)\,\bar f_\alpha(k)$,
for all $h,k\in C^\infty(\MM)[[\lambda]]$, is hermitian, i.e.~$(h\star k)^\ast = k^\ast\star h^\ast$.
Remember that the Hopf algebra $H$ (and its deformation $H^\mathcal{F}$) acts on functions, vector and tensor fields via
the Lie derivative, $\xi(\cdot):=\mathcal{L}_\xi(\cdot)$, for all $\xi\in U\Xi[[\lambda]]$.

The deformed differential calculus $\bigl(\Omega^\bullet[[\lambda]],\wedge_\star,\dd\bigr)$
employs the deformed wedge-product $\omega\wedge_\star \omega^\prime := \bar f^\alpha(\omega)\wedge\bar f_\alpha(\omega^\prime)$,
for all $\omega,\omega^\prime\in\Omega^\bullet[[\lambda]]$. The differential $\dd$ is undeformed and satisfies the
graded Leibniz rule with respect to $\wedge_\star$. The involution on the deformed differential calculus
satisfies $(\omega\wedge_\star\omega^\prime)^\ast = \omega^{\prime\ast}\wedge_\star \omega^\ast$,
 in case $\mathcal{F}$ is real, and $(\dd\omega)^\ast = (-1)^{1+\deg(\omega)}\, \dd(\omega^\ast)$.

The vector fields are deformed by introducing the $\bigl(C^\infty(\MM)[[\lambda]],\star\bigr)$-bimodule structure
$h \star v = \bar f^\alpha(h)\,\bar f_\alpha(v)$ and $v\star h = \bar f^\alpha(v)\,\bar f_\alpha(h)$,
for all $h\in C^\infty(\MM)[[\lambda]]$ and $v\in \Xi[[\lambda]]$. For real twists we have
$(h\star v)^\ast = v^\ast\star h^\ast$ and $(v\star h)^\ast = h^\ast \star v^\ast$.

The tensor algebra $\bigl(\mathcal{T},\otimes_A\bigr)$ over $C^\infty(\MM)$ generated by $\Xi$ and $\Omega^1$ is 
deformed by introducing the $\star$-tensor product $\tau\otimes_{A_\star} \tau^\prime := \bar f^\alpha(\tau)\otimes_A 
\bar f_\alpha(\tau^\prime)$, for all $\tau,\tau^\prime\in\mathcal{T}[[\lambda]]$. For real twists we have
$(\tau\otimes_{A_\star}\tau^\prime)^\ast = \tau^{\prime\ast} \otimes_{A_\star}\tau^\ast$.

The contraction between vector fields and one-forms is deformed by $\pair{\cdot}{\cdot}_\star :=
\pair{\bar f^\alpha(\cdot)}{\bar f_\alpha(\cdot)} $. The involution relates the left and right contraction
via $\pair{v}{\omega}^\ast_\star = \pair{\omega^\ast}{v^\ast}_\star$, in case the twist is real.
Furthermore, we have the important relation
\begin{flalign}
 \pair{\tau\otimes_{A_\star} v\star h}{\omega\otimes_{A_\star}\tau^\prime }_\star = \tau\otimes_{A_\star}\pair{v}{h\star \omega}_\star
\star \tau^\prime~,
\end{flalign}
for all $\tau,\tau^\prime\in\mathcal{T}[[\lambda]]$, $h\in C^\infty(\MM)[[\lambda]]$, $v\in\Xi[[\lambda]]$
and $\omega\in \Omega^1[[\lambda]]$.

Let us now define an integral on top-forms $\Omega^N[[\lambda]]$. Since the classical integral
is a $\bbC$-linear map $\int_\MM:\Omega^N \to \bbC$, we can extend it to a $\bbC[[\lambda]]$-linear
map 
\begin{flalign}
 \int\limits_\MM:\Omega^N[[\lambda]]\to\bbC[[\lambda]]~,\quad \omega\mapsto \int\limits_\MM \omega = \sum\limits_{n=0}^\infty\lambda^n\int\limits_\MM \omega_{(n)}~.
\end{flalign}
Note that in general 
the integral over $\omega\wedge_\star\omega^\prime$ does not agree with the integral
over $(-1)^{\deg(\omega)\,\deg(\omega^\prime)} \omega^\prime\wedge_\star \omega$, 
for $\omega,\omega^\prime \in \Omega^\bullet[[\lambda]]$ with $\deg(\omega)+\deg(\omega^\prime)=N$.
In other words, the integral is in general not graded cyclic.\footnote{
An example for a twist without this property is the Jordanian twist $\mathcal{F}=
\exp\bigl(\frac{1}{2} H\otimes \log(1+\lambda E)\bigr)$ with $[H,E]=2E$.
}
Since graded cyclicity is a property which simplifies drastically our investigations we are going to assume it.
As pointed out in \cite{Aschieri:2009ky,Aschieri:2009mc}, graded cyclicity is fulfilled if the
twist satisfies $S(\bar f^\alpha)\bar f_\alpha =1$. This is the case for all abelian twists,
which is the class of twists we have investigated in Part \ref{part:ncg}.
Thus, we do not loose examples for explicit models when making this restriction. 
As a consequence of $S(\bar f^\alpha)\bar f_\alpha =1$ we obtain
\begin{flalign}
\label{eqn:gradedcyc}
 \int\limits_\MM\omega\wedge_\star\omega^\prime = (-1)^{\deg(\omega)\deg(\omega^\prime)}\,\int\limits_{\MM}
\omega^\prime\wedge_\star \omega = \int\limits_\MM\omega\wedge\omega^\prime~,
\end{flalign}
for all $\omega,\omega^\prime\in\Omega^\bullet[[\lambda]]$ with $\deg(\omega)+\deg(\omega^\prime)=N$
and $\supp(\omega)\cap\supp(\omega^\prime)$ compact.\footnote{
Let $\omega=\sum_{n=0}^\infty\lambda^n\,\omega_{(n)}\in\Omega^\bullet[[\lambda]]$
and $\omega^\prime=\sum_{n=0}^\infty\lambda^n\,\omega^\prime_{(n)}\in\Omega^\bullet[[\lambda]]$.
The statement $\supp(\omega)\cap\supp(\omega^\prime)$ compact
is an abbreviation for $\supp(\omega_{(n)})\cap\supp(\omega_{(m)}^\prime)$ compact
for all $n,m\in\bbN^0$.
}

An interesting topic for future research is to find out if the deformed quantum field theory
construction presented in this chapter is also possible for all real twists, not subject to
the condition $S(\bar f^\alpha)\bar f_\alpha=1$.


\section{Deformed action functional and wave operators}
Using the tools of the previous section we are in the position to construct a deformed action functional 
for a real scalar field $\Phi$. For this we require two more ingredients, 
namely a deformed metric field $g_\star\in \bigl(\Omega^1\otimes_A\Omega^1\bigr)[[\lambda]]$ and 
a deformed volume form $\vols\in\Omega^N[[\lambda]]$ on $\MM$.
Natural choices would be given by the exact noncommutative gravity solutions presented in Chapter 
\ref{chap:ncgsol}. In this case the noncommutative metric field $g_\star$ and also $\vols$ are equal to the commutative
metric field and volume form $\vol$, i.e.~they have no corrections in the deformation parameter $\lambda$.
For sake of generality, however, we do not want to exclude from the beginning models which contain corrections in
$\lambda$, since they can be valid solutions of the noncommutative Einstein equations.
In this respect, we consider deformed metric fields $g_\star\in\bigl(\Omega^1\otimes_A\Omega^1\bigr)[[\lambda]]$,
which yield at zeroth order in the deformation parameter a classical Lorentzian metric $g_\star\vert_{\lambda=0}=g$ on $\MM$.
Similarly, we consider deformed volume forms $\vols\in\Omega^N[[\lambda]]$, which at lowest order yield the classical volume
form $\vols\vert_{\lambda=0}=\vol$ and that are real $\vols^\ast = \vols$.
The $H^\mathcal{F}$-covariant construction of a deformed volume form from a deformed metric field is still an 
open problem in noncommutative gravity, which deserves for a solution. 
However, note that we can always associate to a deformed metric field $g_\star$
(regarded as a formal power series of classical tensors) a volume form $\text{vol}_{g_\star}$ via the classical construction.
The resulting volume form then would be nondegenerate and real, which are exactly the properties we require.
Fortunately, for our explicit physics examples of Part \ref{part:ncg} the semi-Killing property of the twist
makes the choice $\vols=\vol$ natural.

We call the quadruple $(\MM,g_\star,\vols,\mathcal{F})$ a {\it deformed Lorentzian manifold}.
By definition, we obtain at order $\lambda^0$ a Lorentzian manifold $(\MM,g)$. Employing the noncommutative differential
geometry we define a deformed action functional for a real scalar field $\Phi$
\begin{flalign}
\label{eqn:defaction}
 S_\star[\Phi]:= -\frac{1}{2}\int\limits_\MM \Bigl(\pair{\pair{\dd \Phi}{g_\star^{-1}}_\star}{\dd \Phi}_\star 
+ M^2\,\Phi\star\Phi \Bigr)\star\vols~.
\end{flalign}
Here  $M^2\in[0,\infty)$ is a mass parameter, $g_\star^{-1}=g^{-1\alpha}\otimes_{A_\star} g^{-1}_\alpha$ is the $\star$-inverse metric 
field corresponding  to $g_\star$ and we have used in the kinetic term the abbreviation
 $\pair{\pair{\dd \Phi}{g_\star^{-1}}_\star}{\dd \Phi}_\star =
\pair{\dd \Phi}{g^{-1\alpha}}_\star\star\pair{g^{-1}_\alpha}{\dd \Phi}_\star $.
The action $S_\star[\Phi]$ is real, for all $\Phi\in C^\infty_0(\MM,\bbR)[[\lambda]]$,
\begin{flalign}
 \nn\bigl(S_\star[\Phi]\bigr)^\ast &\stackrel{\text{RE}_\mathcal{F}}{=} -\frac{1}{2}\int\limits_\MM \vols^\ast\star
\Bigl(\pair{\dd \Phi^\ast}{g_\alpha^{-1\ast}}_\star\star \pair{g^{-1\alpha\ast}}{\dd\Phi^\ast}_\star + M^2\,\Phi^\ast\star\Phi^\ast \Bigr)\\
 \nn&\hspace{-0.5cm}\stackrel{\text{RE}_{\vols,g_\star,\Phi}}{=}-\frac{1}{2}\int\limits_\MM \vols \star
\Bigl(\pair{\dd \Phi}{g^{-1\alpha}}_\star\star \pair{g_\alpha^{-1}}{\dd\Phi}_\star + M^2\,\Phi\star\Phi \Bigr)\\
 &\,\stackrel{\text{GC}}{=}-\frac{1}{2}\int\limits_\MM 
\Bigl(\pair{\pair{\dd \Phi}{g_\star^{-1}}_\star}{\dd \Phi}_\star + M^2\,\Phi\star\Phi \Bigr)\star\vols = S_\star[\Phi]~,
\end{flalign}
where we have used the reality of the twist ($\text{RE}_\mathcal{F}$), the reality of the volume form, metric and $\Phi$
($\text{RE}_{\vols,g_\star,\Phi}$) and graded cyclicity (GC) as defined in (\ref{eqn:gradedcyc}).

Varying the action (\ref{eqn:defaction}) by functions $\delta\Phi\in C^\infty_0(\MM,\bbR)[[\lambda]]$ of compact support
we obtain
\begin{flalign}
 \delta S_\star[\Phi] = \int\limits_\MM \delta\Phi ~P^\top_\star(\Phi)~,
\end{flalign}
with the top-form valued $\bbC[[\lambda]]$-linear map $P^\top_\star: C^\infty(\MM)[[\lambda]]\to \Omega^N[[\lambda]]$
given by
\begin{flalign}
\label{eqn:defwaveoptop}
 P_\star^\top(\varphi) = \frac{1}{2}\Bigl(\square_\star(\varphi) \star \vols +\vols\star\bigl(\square_\star(\varphi^\ast)\bigr)^\ast 
- M^2 \,\varphi\star\vols - M^2\,\vols\star\varphi\Bigr)~,
\end{flalign}
for all $\varphi\in C^\infty(\MM)[[\lambda]]$. As usual, the d'Alembert operator is defined by
\begin{flalign}
 \int\limits_\MM \varphi\star \square_\star(\psi)\star \vols := 
-\int\limits_\MM \pair{\pair{\dd\varphi}{g_\star^{-1}}_\star}{\dd\psi}_\star\star\vols~,
\end{flalign}
for all $\varphi,\psi\in C^\infty(\MM)[[\lambda]]$ with $\supp(\varphi)\cap\supp(\psi)$ compact.

In order to obtain a scalar valued deformed wave operator we employ the $\bbC[[\lambda]]$-module
isomorphism $\star_\star: C^\infty(\MM)[[\lambda]]\to\Omega^N[[\lambda]]\,,~\varphi\mapsto \varphi\star\vols$
and define 
\begin{flalign}
P_\star := \star_\star^{-1}\circ P^\top_\star:C^\infty(\MM)[[\lambda]]\to C^\infty(\MM)[[\lambda]]~.
\end{flalign}
The equation of motion resulting from the action (\ref{eqn:defaction}) then reads
\begin{flalign}
 P_\star(\Phi) = 0\quad \Leftrightarrow\quad P_\star^\top(\Phi)=0~.
\end{flalign}

Let us discuss some properties of the deformed wave operator $P_\star$. This operator is formally selfadjoint
with respect to the deformed scalar product
\begin{flalign}
 \spp{\varphi}{\psi}_\star := \int\limits_\MM\varphi^\ast\star \psi\star\vols~.
\end{flalign}
More precisely, we have
\begin{flalign}
 \spp{\varphi}{P_\star(\psi)}_\star = \spp{P_\star(\varphi)}{\psi}_\star~,
\end{flalign}
for all $\varphi,\psi\in C^\infty(\MM)[[\lambda]]$ with $\supp(\varphi)\cap\supp(\psi)$ compact.
Furthermore, since it is a $\bbC[[\lambda]]$-linear endomorphism on $C^\infty(\MM)[[\lambda]]$ 
we can equivalently express it in terms of a family of $\bbC$-linear maps $P_{(n)}:C^\infty(\MM)\to C^\infty(\MM)$, $n\in \bbN^0$,
via $P_\star = \sum_{n=0}^\infty\lambda^n\,P_{(n)}$. The zeroth order is simply the undeformed Klein-Gordon operator
$P_{(0)} = \square_g - M^2 $, which is a normally hyperbolic operator.
Thus, $P_\star$ is a formal deformation of a normally hyperbolic operator on $(\MM,g)$.
The last important property we note is a modified reality property of $P_\star$.
The top-form valued equation of motion operator is real, since it is obtained by varying a real
action with respect to a real function, i.e.~$\bigl(P_\star^\top(\varphi)\bigr)^\ast= P_\star^\top(\varphi^\ast)$,
for all $\varphi\in C^\infty(\MM)[[\lambda]]$. As a consequence, $P_\star$ is not real in the conventional
sense, but satisfies $\vols\star \bigl(P_\star(\varphi)\bigr)^\ast = P_\star(\varphi^\ast)\star\vols$, for all
$\varphi\in C^\infty(\MM)[[\lambda]]$. This modified reality property will become important later in our construction
of the deformed quantum field theory.


\section{Deformed Green's operators and the solution space}
After obtaining the explicit example of a deformed Klein-Gordon operator in the previous section,
we now turn to the construction of the deformed Green's operators corresponding to $P_\star$.
Analogously to the commutative case, we say that a $\bbC[[\lambda]]$-linear
map $\Delta_{\star\pm}=\sum_{n=0}^\infty\lambda^n\,\Delta_{(n)\pm}:C^\infty_0(\MM)[[\lambda]]\to C^\infty(\MM)[[\lambda]]$
is a retarded/advanced deformed Green's operator for $P_\star$ if it satisfies
\begin{subequations}\label{eqn:green}
 \begin{flalign}
  \label{eqn:green1}&P_\star\circ \Delta_{\star\pm} = \id_{C^\infty_0(\MM)[[\lambda]]}~,\\
  \label{eqn:green2}&\Delta_{\star\pm} \circ P_\star\big\vert_{C^\infty_0(\MM)[[\lambda]]} = \id_{C^\infty_0(\MM)[[\lambda]]}~,\\
  \label{eqn:green3}&\supp\bigl(\Delta_{(n)\pm}(\varphi)\bigr)\subseteq J_\pm\bigl(\supp(\varphi)\bigr)~,\quad\text{for all } 
\varphi\in C_0^\infty(\MM)\text{ and }n\in\bbN^0~.
 \end{flalign}
\end{subequations}
In the expression above the causal future/past $J_\pm$ is with respect to the classical metric $g=g_\star\vert_{\lambda=0}$.

The existence and uniqueness of the zeroth order of the deformed Green's operators
in ensured by Theorem \ref{theo:greenclas}. We denote these operators by $\Delta_\pm := \Delta_{(0)\pm}$.
In order to prove the existence and uniqueness of $\Delta_{\star\pm}$ to all orders in $\lambda$
we have to introduce some technical support conditions on the deformation:
We assume that $(\MM,g_\star,\vols,\mathcal{F})$ is a {\it compactly deformed} time-oriented, connected and globally hyperbolic
Lorentzian manifold. This means that the underlying classical manifold $(\MM,g=g_\star\vert_{\lambda=0})$
is time-oriented, connected and globally hyperbolic, and the deformation is of compact support.
More precisely, we demand that the twist $\mathcal{F}$ is generated by compactly supported vector fields and that
the noncommutative corrections $g_{(n)}$ of the metric and $\text{vol}_{(n)}$ of the volume form are of compact support,
for all $n>0$.
As a consequence, the deformed Klein-Gordon operator $P_\star$ satisfies the support condition 
$P_{(n)}:C^\infty(\MM)\to C^\infty_0(\MM)$, for all $n>0$.

Note that similar requirements appear in the construction of perturbatively interacting quantum field theories,
where all interactions are assumed to be compactly supported in order to make the formal power series construction
well-defined. As in our case this is just a technical assumption, which is not motivated by physics.
However, from experience in perturbative quantum field theory we know that
we can even treat models not fulfilling this condition by introducing an infrared regulator.
The regularized theory is then used to make predictions for local observables
and it turns out that the results do not depend on details of the regularization, in particular
the limit of vanishing infrared regularization is well-defined, see e.g.~\cite{Duetsch:1998hf}. 
For our models we can take the same road. We regularize our deformation by deformations of compact support,
extract physical predictions and finally try to send the infrared regulator to zero.
The hope then is to capture the most important noncommutative geometry effects in this way.
The existence and properties of this limit remain to be studied in detail, see however Appendix \ref{eqn:lambda2green}
for a simple example in order $\lambda^2$.

For compactly deformed time-oriented, connected and globally hyperbolic
Lorentzian manifolds we can clarify the existence and uniqueness of deformed Green's operators.
This is not only possible for the deformed Klein-Gordon operator, but also for all
{\it compactly deformed} normally hyperbolic operators $P_\star$, 
i.e.~$P_\star= \sum_{n=0}^\infty\lambda^n\,P_{(n)}$ with $P=P_{(0)}$ normally hyperbolic
and $P_{(n)}: C^\infty(\MM)\to C^\infty_0(\MM)$ finite-order differential operators, for all $n>0$.
\begin{theo}
\label{theo:greendef}
 Let $(\mathcal{M},g_\star,\vols,\mathcal{F})$ be a compactly deformed  time-oriented, connected and globally hyperbolic 
Lorentzian manifold and let $P_\star$ be a compactly deformed normally 
hyperbolic operator acting on $C^\infty(\mathcal{M})[[\lambda]]$.
Then there exist unique deformed Green's operators $\Delta_{\star\pm}$ for $P_\star$.\vspace{2mm}

\noindent The explicit expressions for $\Delta_{(n)\pm}$, $n>0$, read
\begin{flalign}
\label{eqn:explicitgreen}
\Delta_{(n)\pm} = \sum\limits_{k=1}^{n}\sum\limits_{j_1=1}^n \dots\sum\limits_{j_k=1}^n(-1)^k \delta_{j_1+\dots+j_k, n} \, \Delta_{\pm}\circ P_{(j_1)}\circ \Delta_\pm \circ P_{(j_2)} \circ \dots \circ \Delta_\pm \circ P_{(j_k)}\circ \Delta_{\pm}~,
\end{flalign}
where $\delta_{n,m}$ is the Kronecker-delta.
\end{theo}
\begin{proof}
 We perform a proof by induction. The zeroth order of (\ref{eqn:green}) is assured
 by Theorem \ref{theo:greenclas}. Assume that we have constructed the Green's operators 
to order $\lambda^{n-1}$. In order $\lambda^n$ (\ref{eqn:green1}) reads
\begin{flalign}
\label{eqn:pr1}
 \sum_{m=0}^n P_{(m)}\circ \Delta_{(n-m)\pm}=0~.
\end{flalign}
Let $\varphi\in C_0^\infty(\mathcal{M})$ be arbitrary. We can reformulate 
(\ref{eqn:pr1}) into a Cauchy problem with respect to the
 classical operator $P=P_{(0)}$
\begin{flalign}
 P(\Delta_{(n)\pm}(\varphi)) = - \sum\limits_{m=1}^n P_{(m)}(\Delta_{(n-m)\pm}(\varphi))=:\rho\in C_0^\infty(\mathcal{M})~.
\end{flalign}
In order to satisfy (\ref{eqn:green3}), we have to impose trivial Cauchy data of vanishing field and derivative 
on a Cauchy surface past/future to $\supp(\rho)$. 
The unique solution is
\begin{flalign}
\label{eqn:sol}
 \Delta_{(n)\pm}(\varphi) =\Delta_{\pm}(\rho)= -\sum\limits_{m=1}^n \Delta_{\pm}\circ P_{(m)}\circ \Delta_{(n-m)\pm}(\varphi)~.
\end{flalign}
We obtain the support property for $n \geqq m>0$
\begin{multline}
 \supp\bigl(\Delta_{\pm}\circ P_{(m)}\circ \Delta_{(n-m)\pm}(\varphi) \bigr)\subseteq J_{\pm}\bigl(\supp(P_{(m)}\circ
\Delta_{(n-m)\pm}(\varphi))\bigr)\\
\subseteq J_{\pm}\bigl(\supp(\Delta_{(n-m)\pm}(\varphi))\bigr)\subseteq 
 J_{\pm}\bigl(J_{\pm}(\supp(\varphi))\bigr) \subseteq J_{\pm}\bigl(\supp(\varphi)\bigr)~,
\end{multline}
where we have used that $P_{(m)}$ is a finite-order differential operator, 
thus satisfying $\supp(P_{(m)}(\varphi))\subseteq \supp(\varphi)$.
 This shows that (\ref{eqn:green3}) is satisfied to order $\lambda^n$.

The equality of (\ref{eqn:sol}) and (\ref{eqn:explicitgreen}) can either be 
shown combinatorially, or by showing that (\ref{eqn:explicitgreen}) solves 
(\ref{eqn:green1}) together with the support property (\ref{eqn:green3}), and thus has to 
be equal to (\ref{eqn:sol}).

The remaining step is to prove the order $\lambda^n$ of (\ref{eqn:green2}). 
Plugging in the explicit form (\ref{eqn:explicitgreen}) one notices that every 
possible chain of operators
\begin{flalign}
 \Delta_{\pm}\circ P_{(j_1)} \circ \Delta_{\pm}\circ \dots \circ \Delta_{\pm}\circ P_{(j_k)}~,\quad \text{with~} j_1+j_2+\dots+j_k=n~,
\end{flalign}
occurs exactly twice in (\ref{eqn:green2}), but with a different sign. Thus, they cancel.

\end{proof}
\begin{rem}
 The reason for requiring the support condition $P_{(n)}:C^\infty(\MM)\to C^\infty_0(\MM)$ for all $n>0$
can be understood from (\ref{eqn:explicitgreen}). 
Since the undeformed Green's operators are in general only defined when acting on compactly supported functions,
all noncommutative corrections $P_{(n)}$ have to map to this space in order to make the compositions in 
(\ref{eqn:explicitgreen}) well-defined.
\end{rem}

The retarded and advanced deformed Green's operators fulfill an important relation
in case $P_\star$ is formally selfadjoint.
\begin{lem}
\label{lem:antihermitiandef}
 Let $(\MM,g_\star,\vols,\mathcal{F})$ be a compactly deformed time-oriented, connected and globally hyperbolic 
Lorentzian manifold and $P_\star:C^\infty(\MM)[[\lambda]]\to C^\infty(\MM)[[\lambda]]$
be a formally selfadjoint compactly deformed normally hyperbolic operator. 
Let $\Delta_{\star\pm}$ be the Green's operators for $P_\star$. Then
\begin{flalign}
 \label{eqn:antihermitiandef}\bigl(\Delta_{\star\pm}(\varphi),\psi\bigr)_\star = \bigl(\varphi,\Delta_{\star\mp}(\psi)\bigr)_\star~
\end{flalign}
holds true for all $\varphi,\psi\in C^\infty_0(\MM)[[\lambda]]$.
\end{lem}
\begin{proof}
Using the properties of the deformed Green's operators (\ref{eqn:green}) and that
$P_\star$ is formally selfadjoint we obtain for all $\varphi,\psi\in C^\infty_0(\MM)[[\lambda]]$
\begin{flalign}
 \spp{\Delta_{\star\pm}(\varphi)}{\psi}_\star = \spp{\Delta_{\star\pm}(\varphi)}{P_\star(\Delta_{\star\mp}(\psi))}_\star
= \spp{P_\star(\Delta_{\star\pm}(\varphi))}{\Delta_{\star\mp}(\psi)}_\star = \spp{\varphi}{\Delta_{\star\mp}(\psi)}_\star~.
\end{flalign}
The second equality holds, since $\supp(\Delta_{\star\pm}(\varphi))\cap \supp(\Delta_{\star\mp}(\psi))$ 
is compact.
To prove this, note that the $n$-th order of $\varphi_{\pm}:=\Delta_{\star\pm}(\varphi)$ 
satisfies
\begin{flalign}
\nn \supp(\varphi_{\pm (n)})&=\supp\left(\sum\limits_{m=0}^n \Delta_{(m)\pm}(\varphi_{(n-m)})\right)\subseteq 
\bigcup_{m=0}^n \supp\left(\Delta_{(m)\pm}(\varphi_{(n-m)})\right)\\
&\subseteq \bigcup_{m=0}^n J_\pm\bigl(\supp(\varphi_{(n-m)})\bigr)
 \subseteq J_\pm\left(\bigcup_{m=0}^n \supp(\varphi_{(n-m)})\right) = J_\pm\left(K^\varphi_{(n)}\right)~,
\end{flalign}
where $K_{(n)}^\varphi\subseteq \MM$ is compact. It follows from Lemma \ref{lem:globhybsupport} 
that $\supp(\varphi_{\pm(n)})\cap \supp(\psi_{\mp(m)})
\subseteq J_\pm(K^\varphi_{(n)})\cap J_{\mp}(K^\psi_{(m)})$ is compact, for all $n,m\in \bbN^0$.

\end{proof}

As we have seen in the commutative case, the retarded-advanced Green's operator
 plays a prominent role in the quantum field theory 
construction. For the deformed map $\Delta_\star:= \Delta_{\star +}-\Delta_{\star -}$ we obtain the following support property.
\begin{lem}
\label{lem:scsupport}
 Let $(\MM,g_\star,\vols,\mathcal{F})$ be a compactly deformed time-oriented, connected and globally hyperbolic 
Lorentzian manifold and $P_\star:C^\infty(\MM)[[\lambda]]\to C^\infty(\MM)[[\lambda]]$
be a compactly deformed normally hyperbolic operator. 
Let $\Delta_{\star\pm}$ be the Green's operators for $P_\star$.
Then for all $\varphi\in C^\infty_0(\mathcal{M})[[\lambda]]$ we have $\Delta_{\star}(\varphi)\in
 C^\infty_\mathrm{sc}(\mathcal{M})[[\lambda]]$.
\end{lem}
\begin{proof}
Let $\varphi\in C^\infty_0(\MM)[[\lambda]]$ be arbitrary and define  $\psi := \Delta_\star(\varphi)$.
We obtain for the $n$-th order of $\psi$
\begin{flalign}
\nn \supp(\psi_{(n)}) &= \supp\left(\sum\limits_{m=0}^n \Delta_{(m)}(\varphi_{(n-m)})\right)\subseteq 
\bigcup_{m=0}^n\supp\left(\Delta_{(m)}(\varphi_{(n-m)})\right)\\
&\subseteq \bigcup_{m=0}^n J\bigl(\supp(\varphi_{(n-m)})\bigr) \subseteq  J\left(\bigcup_{m=0}^n 
\supp(\varphi_{(n-m)})\right) = J\left(K_{(n)}^\varphi\right)~,
\end{flalign}
where $K_{(n)}^\varphi\subseteq \MM$ is compact. Thus, $\psi_{(n)}\in C^\infty_\sc(\MM)$ for all $n\in\bbN^0$.

\end{proof}

The space of complex solutions of the deformed wave equation is defined by
\begin{flalign}
\Sol_{P_\star}:= \bigl\lbrace\varphi\in C^\infty_\sc(\MM)[[\lambda]]: P_\star(\varphi)=0\bigr \rbrace~.
\end{flalign}
The full information on this space can be obtained from the following
\begin{theo}
\label{theo:complex}
 Let $(\MM,g_\star,\vols,\mathcal{F})$ be a compactly deformed time-oriented, connected and globally hyperbolic 
Lorentzian manifold and $P_\star:C^\infty(\MM)[[\lambda]]\to C^\infty(\MM)[[\lambda]]$
be a formally selfadjoint compactly deformed normally hyperbolic operator. 
Let $\Delta_{\star\pm}$ be the Green's operators for $P_\star$.
Then the sequence of $\bbC[[\lambda]]$-linear maps
\begin{flalign}
 0\longrightarrow C_0^\infty(\mathcal{M})[[\lambda]] \stackrel{P_\star}{\longrightarrow} C_0^\infty(\mathcal{M})[[\lambda]] \stackrel{\Delta_\star}{\longrightarrow} C_\mathrm{sc}^\infty(\mathcal{M})[[\lambda]] \stackrel{P_\star}{\longrightarrow}C_\mathrm{sc}^\infty(\mathcal{M})[[\lambda]]
\end{flalign}
is a complex which is exact everywhere.
\end{theo}
\begin{proof}
The sequence of maps forms a complex due to Theorem \ref{theo:greendef} and Lemma \ref{lem:scsupport}. 

To prove the first exactness, let $\varphi\in C^\infty_0(\mathcal{M})[[\lambda]]$ such that $P_\star(\varphi)=0$. 
Then $\varphi = \Delta_{\star\pm}\circ P_\star(\varphi) = 0$. 

To prove the second exactness, let $\varphi\in C^\infty_0(\mathcal{M})[[\lambda]]$ such that $\Delta_{\star}(\varphi)=0$. 
We define $\psi:=\Delta_{\star\pm}(\varphi)$ and obtain using Theorem \ref{theo:greendef} 
and Lemma \ref{lem:globhybsupport}
that $\psi\in C^\infty_0(\mathcal{M})[[\lambda]]$. We find $P_\star(\psi)=P_\star\circ\Delta_{\star\pm}(\varphi)=\varphi$.

To prove the third exactness, let $\varphi\in C^\infty_\mathrm{sc}(\mathcal{M})[[\lambda]]$
 such that $P_\star(\varphi)=0$.  
We can find a family of compact sets $\lbrace K_{(n)}\subseteq \MM : n\in\bbN^0\rbrace$, such that
 $\supp(\varphi_{(n)})\subseteq I_+(K_{(n)})\cup I_-(K_{(n)})$.
 We decompose analogously to \cite{Bar:2007zz} $\varphi_{(n)}=\varphi_{(n)}^+ +\varphi_{(n)}^-$, where 
$\supp(\varphi_{(n)}^\pm)\subseteq I_{\pm}(K_{(n)})\subseteq J_\pm(K_{(n)})$. We define 
$\varphi^{\pm} := \sum_{n=0}^\infty\lambda^n\,\varphi_{(n)}^\pm$ and $\psi:= \pm P_\star(\varphi^\pm)$. 
Using the support properties of $\varphi^\pm$ and Lemma \ref{lem:globhybsupport}, 
one obtains that $\psi\in C^\infty_0(\mathcal{M})[[\lambda]]$.
To show that $\Delta_{\star\pm}(\psi)=\pm\varphi^{\pm}$, let $\chi\in C^\infty_0(\mathcal{M})[[\lambda]]$ be arbitrary. We obtain
\begin{flalign}
 \nn\bigl(\chi,\Delta_{\star\pm}(\psi)\bigr)_\star &= \bigl(\Delta_{\star\mp}(\chi),\psi\bigr)_\star
 = \pm \bigl(\Delta_{\star\mp}(\chi),P_\star(\varphi^\pm)\bigr)_\star \\
&= \pm\bigl(P_\star\circ\Delta_{\star\mp}(\chi),\varphi^\pm
\bigr)_\star = \bigl(\chi,\pm \varphi^\pm\bigr)_\star~,
\end{flalign}
where we have used Lemma \ref{lem:antihermitiandef}, Theorem \ref{theo:greendef} and that $P_\star$ is formally selfadjoint. 
This shows that $\Delta_{\star}(\psi) = \varphi$.

\end{proof}
\noindent The map $\Delta_\star$ thus gives rise to a $\bbC[[\lambda]]$-module isomorphism
\begin{flalign}
\label{eqn:solisodef}
 \mathcal{I}_\star: C^\infty_0(\MM)[[\lambda]]/P_\star[C^\infty_0(\MM)[[\lambda]]] \to \Sol_{P_\star}\,
,~[\varphi]\mapsto \Delta_\star(\varphi)~.
\end{flalign}


\section{Symplectic $\bbR[[\lambda]]$-module and quantization}
In order to quantize the deformed field theory we have to equip the space of real solutions
with a symplectic structure. Since we are working in a  formal deformation quantization setting,
 we have to adjust the definition of a symplectic vector space.
\begin{defi}
 A {\it symplectic $\bbR[[\lambda]]$-module} $\bigl(V,\omega\bigr)$ is an $\bbR[[\lambda]]$-module $V$ together with
an antisymmetric and $\bbR[[\lambda]]$-bilinear map $\omega:V\times V\to \bbR[[\lambda]]$,
 such that $\omega(v,v^\prime)=0$ for all $v^\prime\in V$
implies $v=0$. The map $\omega$ is called a symplectic structure.
\end{defi}
The space of real solutions is defined by $\Sol_{P_\star}^\bbR:= \bigl\lbrace\varphi \in C^\infty_\sc(\MM,\bbR)[[\lambda]]
: P_\star(\varphi)=0\bigr\rbrace$. 
Remember that the top-form valued deformed Klein-Gordon operator $P_\star^\top$
(\ref{eqn:defwaveoptop}) is real, a natural property we also assume for our general formal deformations
of normally hyperbolic operators.
As a consequence, the scalar valued wave operator $P_\star$ and therewith the Green's operators 
$\Delta_{\star\pm}$ are not real, and the space of real solutions is {\it not} generated by restricting $\Delta_\star$
to $C^\infty_0(\MM,\bbR)[[\lambda]]$. 
The correct way to do this restriction can be obtained from
the following
\begin{propo}
\label{propo:HRprop}
Let $(\MM,g_\star,\vols,\mathcal{F})$ be a compactly deformed time-oriented, connected and globally hyperbolic 
Lorentzian manifold and $P_\star:C^\infty(\MM)[[\lambda]]\to C^\infty(\MM)[[\lambda]]$
be a formally selfadjoint compactly deformed normally hyperbolic operator, such that $P_\star^\top = \star_\star \circ P_\star$
is real.
Let $\Delta_{\star\pm}$ be the Green's operators for $P_\star$.
 Consider the $\bbR[[\lambda]]$-module
\begin{flalign}
 H_\bbR := \bigl\lbrace\varphi\in C^\infty_0(\MM)[[\lambda]] : \bigl(\Delta_{\star\pm}(\varphi)\bigr)^\ast = \Delta_{\star\pm}(\varphi)  \bigr\rbrace~.
\end{flalign}
Then the following statements hold true:
\begin{itemize}
 \item[(1)] Let $\psi\in \Sol_{P_\star}^\bbR$ be a real solution of the deformed wave equation, then
there is a $\varphi\in H_\bbR$ such that $\psi = \Delta_\star(\varphi)$.
\item[(2)] The kernel of $\Delta_\star$ restricted to $H_\bbR$ is given by
$\Ker(\Delta_\star)\big\vert_{H_\bbR}= P_\star[C^\infty_0(\MM,\bbR)[[\lambda]]]$.
\item[(3)] Let $\varphi \in H_\bbR$, then $(\varphi\star \vols)^\ast = \varphi\star \vols$.
\end{itemize}
\end{propo}
\begin{proof}
 {\it Proof of (1):}\\
Let $\psi\in \Sol_{P_\star}^\bbR$ be a real solution. 
By Theorem \ref{theo:complex} we know that there is a
 $\varphi\in C^\infty_0(\mathcal{M})[[\lambda]]$, such that $\psi = \Delta_\star(\varphi)$.
 From the reality of $\psi$ we obtain
\begin{flalign}
 \bigl(\Delta_{\star+}(\varphi)\bigr)^\ast - \Delta_{\star+}(\varphi) = 
\bigl(\Delta_{\star-}(\varphi)\bigr)^\ast - \Delta_{\star-}(\varphi)=:2 \delta\in C_0^\infty(\mathcal{M})[[\lambda]]~.
\end{flalign}
$\delta$ is of compact support due to Theorem \ref{theo:greendef} and Lemma \ref{lem:globhybsupport}.
Using $\delta^\ast = -\delta$ and $\delta = \Delta_{\star\pm}\circ P_\star(\delta)$ 
we find that 
\begin{flalign}
 \bigl(\Delta_{\star\pm}\left(\varphi + P_\star(\delta)\right)\bigr)^\ast = 
\Delta_{\star\pm}\left(\varphi + P_\star(\delta)\right)~.
\end{flalign}
Thus, $\varphi+P_\star(\delta)\in H_\bbR$ with $\Delta_\star(\varphi+P_\star(\delta)) = \Delta_{\star}(\varphi)=\psi$. 
~\vspace{2mm}\\
\noindent {\it Proof of (2):}\\
Let $\varphi=P_\star(\chi)\in P_\star[C^\infty_0(\mathcal{M},\bbR)[[\lambda]]]$. We obtain 
\begin{flalign}
 \bigl(\Delta_{\star\pm}(\varphi)\bigr)^\ast = \chi^\ast =\chi =\Delta_{\star\pm}(\varphi)~.
\end{flalign}
Thus, $\varphi=P_\star(\chi)\in H_\bbR$ and $\Delta_\star(\varphi) = \Delta_{\star}(P_\star(\chi))=0$.\\
Let now $\varphi\in H_\bbR$ such that $\Delta_\star(\varphi)=0$. By Theorem 
\ref{theo:complex} there exists a $\chi\in C^\infty_0(\mathcal{M})[[\lambda]]$, such that
 $\varphi=P_\star(\chi)$. Using the definition of $H_\bbR$ we obtain
\begin{flalign}
 0=\bigl(\Delta_{\star\pm}(\varphi)\bigr)^\ast-\Delta_{\star\pm}(\varphi) = \chi^\ast -\chi~.
\end{flalign}
Thus, $\chi\in C^\infty_0(\mathcal{M},\bbR)[[\lambda]]$.
~\vspace{2mm}\\
\noindent {\it Proof of (3):}\\
Let $\varphi\in H_\bbR$. Using reality of the top-form valued equation of motion operator
 $P_\star^\top$ we obtain
\begin{flalign}
 \left(\varphi\star\vols\right)^\ast =
 \left(P^\top_\star\left(\Delta_{\star\pm}\left(\varphi\right)\right)\right)^\ast =
 P^\top_\star\left(\left(\Delta_{\star\pm}\left(\varphi\right)\right)^\ast \right)= 
P^\top_\star\left(\Delta_{\star\pm}\left(\varphi\right) \right)=\varphi\star\vols~.
\end{flalign}

\end{proof}
From this proposition we obtain that the $\bbR[[\lambda]]$-module 
$V_\star:= H_\bbR/P_\star[C^\infty_0(\MM,\bbR)[[\lambda]]]$
is isomorphic to the space of real solutions $\Sol_{P_\star}^\bbR$.
Furthermore, we can equip $V_\star$ with a symplectic structure
\begin{flalign}
\label{eqn:defsymplec}
 \omega_\star: V_\star\times V_\star \to \bbR[[\lambda]]\,,~([\varphi], [\psi])\mapsto \omega_\star([\varphi],[\psi])= 
\spp{\varphi}{\Delta_\star(\psi)}_\star~.
\end{flalign}
This map is well-defined due to Lemma \ref{lem:antihermitiandef}.
It is weakly nondegenerate for the following reason: Let $\psi\in H_\bbR$ be fixed.
Then $\spp{\varphi}{\Delta_\star(\psi)}_\star=0$ for all $\varphi\in H_\bbR$ implies
$\Delta_\star(\psi)=0$.
Thus, by Proposition \ref{propo:HRprop} (2) we have $\psi\in P_\star[C^\infty_0(\MM,\bbR)[[\lambda]]]$, which means
that $\psi\sim 0$ is equivalent to zero.

It remains to show the antisymmetry and reality of the map $\omega_\star$.
Using Proposition \ref{propo:HRprop} (3), graded cyclicity (GC) and the reality of $\vols$ ($\text{RE}_{\vols}$),
we obtain for all $\varphi,\psi\in H_\bbR$
\begin{flalign}
\nn \spp{\varphi}{\Delta_\star(\psi)}_\star &= \int\limits_\MM \varphi^\ast\star \Delta_\star(\psi)\star \vols \stackrel{\text{GC},\,
\text{RE}_{\vols}}{=} \int\limits_\MM \bigl(\varphi\star\vols\bigr)^\ast\star \Delta_\star(\psi)\\
\nn&\hspace{-0.4mm}\stackrel{(3)}{=} \int\limits_\MM \varphi\star \vols \star\Delta_\star(\psi) \stackrel{\text{GC}}{=} \int\limits_\MM
\Delta_\star(\psi)\star\varphi\star \vols\\
&\hspace{-2.5mm}\stackrel{\psi\in H_\bbR}{=}  \spp{\Delta_\star(\psi)}{\varphi}_\star 
\stackrel{\text{Lem.~\ref{lem:antihermitiandef}}}{=} -\spp{\psi}{\Delta_\star(\varphi)}_\star~.
\end{flalign}
This shows that $\omega_\star$ is antisymmetric. Reality follows from
\begin{flalign}
 \spp{\varphi}{\Delta_\star(\psi)}^\ast_\star = 
\spp{\Delta_\star(\psi)}{\varphi}_\star \stackrel{\text{Lem.~\ref{lem:antihermitiandef}}}{=} 
-\spp{\psi}{\Delta_\star(\varphi)}_\star = \spp{\varphi}{\Delta_\star(\psi)}_\star~,
\end{flalign}
In the first equality we have used that the scalar product is hermitian and in the last one
the antisymmetry of $\omega_\star$.

The results are summarized in this
\begin{propo}
 Let $(\MM,g_\star,\vols,\mathcal{F})$ be a compactly deformed time-oriented, connected and globally hyperbolic 
Lorentzian manifold and $P_\star:C^\infty(\MM)[[\lambda]]\to C^\infty(\MM)[[\lambda]]$
be a formally selfadjoint compactly deformed normally hyperbolic operator, such that $P_\star^\top = \star_\star \circ P_\star$
is real. We denote the Green's operators for $P_\star$ by $\Delta_{\star\pm}$ and $\Delta_\star = \Delta_{\star+}-\Delta_{\star-}$.
Then $\bigl(V_\star,\omega_\star\bigr)$ with $V_\star= H_\bbR/P_\star[C^\infty_0(\MM,\bbR)[[\lambda]]]$
 and $\omega_\star$ defined in (\ref{eqn:defsymplec})
is a symplectic $\bbR[[\lambda]]$-module.
\end{propo}
\begin{rem}
 This proposition is a nontrivial result, since it shows that the field theory
on noncommutative curved spacetimes can be naturally equipped with a symplectic structure $\omega_\star$.
This is in contrast to other approaches to deformed quantum field theory based on
``symplectic or Poisson structures'' which are antisymmetric up to an $R$-matrix, 
see e.g.~\cite{Aschieri:2007sq,Aschieri:2009yq} and references therein. 
As we have shown in the Appendix \ref{app:twistqft} using the example of homothetic Killing deformations,
our approach is more flexible and allows for more general deformations of quantum field theories.
\end{rem}

Provided the symplectic $\bbR[[\lambda]]$-module we can define suitable algebras of observables
for the deformed quantum field theory.
A first possible definition is guided by the commutative case and reads
\begin{defi}
 Let $\bigl(V,\omega\bigr)$ be a symplectic $\bbR[[\lambda]]$-module.
Let $A$ be a unital $\ast$-algebra over $\bbC[[\lambda]]$ and let $\mathcal{W}:V\to A$
be a map such that for all $v,u\in V$ we have
\begin{subequations}
\begin{flalign}
 \mathcal{W}(0)&=1~,\\
 \mathcal{W}(-v)&= \mathcal{W}(v)^\ast~,\\
\label{eqn:expdefweyl} \mathcal{W}(v)\,\mathcal{W}(u) &= e^{-i\,\omega(v,u)/2}\,\mathcal{W}(v+u)~.
\end{flalign}
\end{subequations}
We call $A$ a {\it $\ast$-algebra of Weyl-type}, if it is generated by the elements $\mathcal{W}(v)$, $v\in V$.
\end{defi}
Note that, different to the commutative case, we could not demand $A$ to be a $C^\ast$-algebra, because
we are considering algebras over $\bbC[[\lambda]]$ and not over $\bbC$.
This formal extension of the algebra is required, since $\omega$ and therewith the exponential factor in (\ref{eqn:expdefweyl})
is only defined in terms of formal power series. 
Due to the missing $C^\ast$-properties, we can not use the strong mathematical
results of \cite{Bar:2007zz} for $C^\ast$-algebras, see also Chapter \ref{chap:qftbas}, and thus,
$\ast$-algebras of Weyl-type are more difficult to handle than their $C^\ast$-counterparts.
In particular, the representation theory of $\ast$-algebras over $\bbC[[\lambda]]$ is known to be very rich
\cite{1138.53316}.

A second algebra of observables, which frequently appears in commutative quantum field theory,
is the algebra of field polynomials. This algebra has no $C^\ast$-norm, and thus is not preferred
for structural mathematical studies. However, since its elements are finite sums of finite products of 
smeared linear field operators, the extraction of $n$-point correlation functions from this algebra 
is straightforward. In other words, the algebra of field polynomials is 
convenient for physical studies.
 For our purpose, this algebra is suitable since it has a straightforward extension
 to formal power series.
\begin{defi}
\label{def:fieldpoly}
 Let $(V,\omega)$ be a symplectic $\bbR[[\lambda]]$-module.
 Let $\mathcal{A}_\text{free}$ be the unital $\ast$-algebra over $\bbC[[\lambda]]$ which is
 freely generated by the elements $1$, $\Phi( v )$ and $\Phi( v )^\ast$, $v \in V$,
 and let $\mathcal{I}$ be the $\ast$-ideal generated by the elements
\begin{subequations}
\label{eqn:fieldalgebra}
\begin{flalign}
&~~\Phi(\beta\,v+\gamma\,u)-\beta\,\Phi(v)-\gamma\,\Phi(u)~,\\
&~~\Phi(v)^\ast-\Phi(v)~,\\
&~~[\Phi(v),\Phi(u)]- i\,\omega(v,u)\, 1~,
\end{flalign}
\end{subequations}
for all $v,u\in V$ and $\beta,\gamma\in \bbR[[\lambda]]$.
The {\it $\ast$-algebra of field polynomials} is defined as the quotient $\mathcal{A}_{(V,\omega)}:=
 \mathcal{A}_\text{free}/\mathcal{I}$.
\end{defi}
For notational convenience we do not write the brackets for equivalence classes in $\AA_{(V,\omega)}$
and denote the linear field operators by $\Phi(v)\in \AA_{(V,\omega)}$ and 
not by $[\Phi(v)]$.


\section{\label{sec:algstatesdef}Algebraic states and representations}
In order to extract physical observables from the $\ast$-algebras of observables of the deformed 
quantum field theory discussed in the previous section we require the notion
of an algebraic state. Defining a state on a $\ast$-algebra $A$ over $\bbC[[\lambda]]$
to be a positive $\bbC$-linear map $\Omega:A\to\bbC$ turns out to be too naive, since one is immediately
faced with convergence problems or one has to ignore the higher order corrections in $\lambda$
\cite{1138.53316}.
A more suitable definition is the following
\begin{defi}
 An {\it algebraic state} (or simply {\it state}) on a unital $\ast$-algebra $A$ over $\bbC[[\lambda]]$ is a $\bbC[[\lambda]]$-linear
map $\Omega:A\to \bbC[[\lambda]]$ such that
\begin{flalign}
 \Omega(1)=1~,\quad \Omega(a^\ast a)\geq 0~,\quad\text{for all } a\in A~.
\end{flalign}
The ordering on $\bbR[[\lambda]]$ is defined by
\begin{flalign}
 \bbR[[\lambda]]\ni \gamma=\sum\limits_{n=n_0}^\infty\lambda^n\,\gamma_{(n)} >0\quad :\Longleftrightarrow\quad \gamma_{(n_0)}>0~.
\end{flalign}
$\Omega$ is called {\it faithful}, if $\Omega(a^\ast a)=0$ implies $a=0$.
\end{defi}
A state as defined above associates to each element $a\in A$ its expectation value $\Omega(a)\in \bbC[[\lambda]]$,
which can depend on the deformation parameter $\lambda$. Thus, $\lambda$ can have effects on measurements,
which is exactly what we want. 

As argued in \cite{1138.53316}, suitable spaces on which we can represent
$\ast$-algebras over $\bbC[[\lambda]]$ are pre-Hilbert spaces over $\bbC[[\lambda]]$.
\begin{defi}~\,~
\begin{itemize}
\item[(a)] A $\bbC[[\lambda]]$-module $\mathcal{H}$ with a map $\braket{\cdot}{\cdot}: \mathcal{H}\times\mathcal{H}\to\bbC[[\lambda]]$
 is called a {\it pre-Hilbert space over $\bbC[[\lambda]]$} if
\begin{itemize}
\item[1.)] $\braket{\cdot}{\cdot}$ is $\bbC[[\lambda]]$-linear in the second argument,
\item[2.)] $\braket{\psi}{\phi}^\ast = \braket{\phi}{\psi}$ for all $\phi,\psi\in \mathcal{H}$,
\item[3.)] $\braket{\psi}{\psi}>0$ for all $\psi\neq 0$.
\end{itemize}\vspace{2mm}
\item[(b)] Let $A$ be a $\ast$-algebra over $\bbC[[\lambda]]$. A {\it representation} of $A$ is a tuple
$\bigl(\mathcal{H},\pi\bigr)$, where $\mathcal{H}$ is a pre-Hilbert space over $\bbC[[\lambda]]$ and
$\pi:A\to L_\text{ad}(\mathcal{H})$ is a $\ast$-algebra homomorphism into
the algebra of adjoinable operators on $\mathcal{H}$.\vspace{1mm}

\noindent A vector $\ket{0}\in\mathcal{H}$ is called {\it cyclic}, if $\pi[A]\ket{0}=\mathcal{H}$.\vspace{1mm}

\noindent A triple $\bigl(\mathcal{H},\pi,\ket{0}\bigr)$ is called a {\it cyclic representation} of $A$, if 
$\bigl(\mathcal{H},\pi\bigr)$ is a representation of $A$ and $\ket{0}\in\mathcal{H}$ is cyclic.
\end{itemize}
\end{defi}

Similar to the case of unital $\ast$-algebras over $\bbC$ there is a GNS-construction
for unital $\ast$-algebras over $\bbC[[\lambda]]$, see \cite{0989.53057,1138.53316} for details.
The main result of this construction is summarized in the following
\begin{theo}[\cite{0989.53057,1138.53316}]~\,~
 \begin{itemize}
  \item[(i)] Let $A$ be a unital $\ast$-algebra over $\bbC[[\lambda]]$ and $\Omega$ a state on $A$.
Then there exists a cyclic representation $\bigl(\mathcal{H},\pi,\ket{0}\bigr)$ of $A$, such that
\begin{flalign}
 \Omega(a) = \bra{0}\pi(a)\ket{0}~,\quad\text{for all } a\in A~.
\end{flalign}
The triple $\bigl(\mathcal{H},\pi,\ket{0}\bigr)$ is called the GNS-representation of $\Omega$.\vspace{1mm}
\item[(ii)]$\bigl(\mathcal{H},\pi,\ket{0}\bigr)$ is unique up to unitary equivalence, i.e.~let
$\bigl(\widetilde{\mathcal{H}},\widetilde{\pi},\widetilde{\ket{0}}\bigr)$ be another cyclic representation of
$A$ satisfying $\Omega(a)=\widetilde{\bra{0}}\widetilde{\pi}(a)\widetilde{\ket{0}}$ for all $a\in A$, 
then there is a unitary $U:\mathcal{H}\to \widetilde{\mathcal{H}}$, such that $U\pi(a)U^{-1}=\widetilde{\pi}(a)$, for
all $a\in A$, and $U\ket{0} =\widetilde{\ket{0}}$.
 \end{itemize}
\end{theo}
 This theorem shows that, even in the formal power series framework, we can go over from
 $\ast$-algebras and states to a formulation in terms of pre-Hilbert spaces and operators thereon.
However, note that there is one essential difference: All pre-Hilbert spaces 
are not equipped (and thus also not complete) with a norm $\Vert\cdot\Vert:\mathcal{H}\to\bbC$.
Since the usual spectral calculus for operators on Hilbert spaces strongly makes use of this
completeness, it is expected that a spectral calculus in the formal framework will be much more involved.

The last point we want to discuss in this chapter is the choice of an 
algebraic state for the deformed quantum field theory.
In general, this will be even more ambiguous than the choice of a state for the
undeformed quantum field theory, due to the freedom we get from the formal power series extension.
However, as we will show in Chapter \ref{chap:qftcon}, there is a $\ast$-algebra homomorphism
from the $\ast$-algebra of field polynomials of the deformed quantum field theory $\AA_{(V_\star,\omega_\star)}$ 
to the formal power series extension of the $\ast$-algebra of field polynomials of the undeformed one
$\AA_{(V,\omega)}[[\lambda]]$.
This allows us to induce states from $\AA_{(V,\omega)}[[\lambda]]$ to $\AA_{(V_\star,\omega_\star)}$.
There is even the possibility to induce faithful states from $\AA_{(V,\omega)}$ to faithful states on $\AA_{(V_\star,\omega_\star)}$.
A faithful state $\Omega:\AA_{(V,\omega)}\to\bbC$ for the undeformed (and unextended) quantum field theory
provides a faithful state $\Omega:\AA_{(V,\omega)}[[\lambda]]\to\bbC[[\lambda]]$ by 
defining, for all $a\in \AA_{(V,\omega)}[[\lambda]]$,
\begin{flalign}
\label{eqn:stateextension}
 \Omega(a):= \sum\limits_{n=0}^\infty\lambda^n\,\Omega(a_{(n)})~.
\end{flalign}
The map (\ref{eqn:stateextension}) is $\bbC[[\lambda]]$-linear, satisfies $\Omega(1)=1$ on $\AA_{(V,\omega)}[[\lambda]]$,
and
\begin{flalign}
\nn \Omega(a^\ast a) &= \sum\limits_{n=0}^\infty\lambda^n \sum\limits_{m+k=n} \Omega(a^\ast_{(m)}\,a_{(k)})\\
&= \lambda^{2 n_0} \, \Omega(a^\ast_{(n_0)}\,a_{(n_0)}) +\mathcal{O}(\lambda^{2n_0+1})>0~,~\text{for all } a\neq 0~, 
\end{flalign}
where $n_0\in \bbN^0$ denotes the first nonvanishing term in $a=\sum_{n=0}^\infty \lambda^n\,a_{(n)}$. 
This state can be pulled-back to $\AA_{(V_\star,\omega_\star)}$.


\chapter{\label{chap:qftcon}Properties}
In this chapter we first reformulate the quantum field theory on noncommutative curved spacetimes
presented in Chapter \ref{chap:qftdef} into a simpler but equivalent setting. This makes in particular reality properties
more obvious. Then we show that there exist symplectic isomorphisms between the field theory on the deformed
and the field theory on the undeformed spacetime. This eventually leads to $\ast$-algebra isomorphisms
between the corresponding $\ast$-algebras of field polynomials. The mathematical and physical consequences are studied.
The results of this chapter appeared in the proceedings article \cite{Schenkel:2011gw} and
are influenced by our investigations on special classes of deformations, the homothetic Killing deformations
\cite{Schenkel:2010jr}. 


\section{Symplectic isomorphism to a simplified formalism}
\subsection{Wave operators:}
Let us go back to the deformed wave operator $P_\star^\top:C^\infty(\MM)[[\lambda]]\to\Omega^N[[\lambda]]$
 (\ref{eqn:defwaveoptop}).
In Chapter \ref{chap:qftdef} we have used the $\bbC[[\lambda]]$-module isomorphism
\begin{flalign}
\star_\star: C^\infty(\MM)[[\lambda]]\to \Omega^N[[\lambda]]\,,~h\mapsto h\star \vols
\end{flalign}
in order to define the scalar valued wave operator 
\begin{flalign}
P_\star := \star_\star^{-1}\circ P^\top_\star : C^\infty(\MM)[[\lambda]]
\to C^\infty(\MM)[[\lambda]]~. 
\end{flalign}
The resulting operator is formally selfadjoint with respect to the deformed scalar product
\begin{flalign}
\label{eqn:defsp}
 \spp{\varphi}{\psi}_\star = \int\limits_\MM \varphi^\ast\star \psi\star \vols~.
\end{flalign}

However, instead of using $\star_\star$ we can also employ the undeformed Hodge operator $\star_g$ to
define a $\bbC[[\lambda]]$-module isomorphism
\begin{flalign}
 \star_g: C^\infty(\MM)[[\lambda]]\to \Omega^N[[\lambda]]\,,~h\mapsto h\,\vol = 
\sum\limits_{n=0}^\infty\lambda^n\,h_{(n)}\,\vol~.
\end{flalign}
The resulting scalar valued wave operator is
\begin{flalign}
 \widetilde{P}_\star := \star_g^{-1}\circ P_\star^\top: C^\infty(\MM)[[\lambda]]
\to C^\infty(\MM)[[\lambda]]~.
\end{flalign}
Since $\star_\star$ and $\star_g$ are isomorphisms, the following equations of motion are equivalent
\begin{flalign}
 P_\star(\Phi)=0\quad \Leftrightarrow\quad \widetilde{P}_\star(\Phi)=0\quad \Leftrightarrow\quad P_\star^\top(\Phi)=0~,
\end{flalign}
and therewith also the solution spaces 
$\Sol_{P_\star} = \Sol_{\widetilde{P}_\star} = \Sol_{P_\star^\top}$.

We define the $\bbC[[\lambda]]$-module automorphism
\begin{flalign}
 \iota:= \star_g^{-1}\circ \star_\star :C^\infty(\MM)[[\lambda]]\to C^\infty(\MM)[[\lambda]]~,
\end{flalign}
which relates $P_\star$ and $\widetilde{P}_\star$ as follows
\begin{flalign}
 \widetilde{P}_\star = \iota\circ P_\star~.
\end{flalign}

The deformed scalar product (\ref{eqn:defsp}) is related to the undeformed one
\begin{flalign}
\label{eqn:undefsp}
 \spp{\varphi}{\psi} = \int\limits_\MM \varphi^\ast\,\psi\,\vol~,
\end{flalign}
for all $\varphi,\psi\in C^\infty(\MM)[[\lambda]]$ with $\supp(\varphi)\cap\supp(\psi)$ compact by
\begin{subequations}
\begin{flalign}
 \spp{\varphi}{\iota(\psi)} &= \int\limits_\MM \varphi^\ast\,\iota(\psi)\,\vol = 
 \int\limits_\MM \varphi^\ast\,\bigl(\psi\star\vols\bigr) \stackrel{\text{GC}}{=} 
\int\limits_\MM \varphi^\ast\star\psi\star\vols = \spp{\varphi}{\psi}_\star~,\\
 \spp{\iota(\varphi)}{\psi} &= \int\limits_\MM \bigl(\iota(\varphi)\,\vol\bigr)^\ast\,\psi=
\int\limits_\MM \bigl(\varphi\star \vols\bigr)^\ast\,\psi \stackrel{\text{GC}}{=}
\int\limits_\MM \varphi^\ast\star\psi\star \vols = \spp{\varphi}{\psi}_\star~,
\end{flalign}
\end{subequations}
where we have used graded cyclicity (GC) (\ref{eqn:gradedcyc}) and the reality of
$\vols$.

Let $P_\star:C^\infty(\MM)[[\lambda]] \to C^\infty(\MM)[[\lambda]]$ be a wave operator
which is 1.) a formal deformation of a normally hyperbolic operator and 2.) formally selfadjoint
with respect to the deformed scalar product (\ref{eqn:defsp}).
Then we can always define the wave operator $\widetilde{P}_\star:= \iota \circ P_\star$,
which is 1.) a formal deformation of the same normally hyperbolic operator and 2.)
formally selfadjoint with respect to the undeformed scalar product (\ref{eqn:undefsp}).
The first property follows from $\iota = \id +\mathcal{O}(\lambda)$ and the second one from
a small calculation:
\begin{flalign}
\nn \spp{\widetilde{P}_\star(\varphi)}{\psi} &= \spp{\iota^{-1}\circ \widetilde{P}_\star(\varphi)}{ \psi}_\star 
= \spp{P_\star(\varphi)}{\psi}_\star \\
&= \spp{\varphi}{P_\star(\psi)}_\star = \spp{\varphi}{\iota\circ P_\star(\psi)}
=\spp{\varphi}{\widetilde{P}_\star(\psi)}~,
\end{flalign}
for all $\varphi,\psi\in C^\infty(\MM)[[\lambda]]$ with $\supp(\varphi)\cap\supp(\psi)$ compact.
Vice versa, let $\widetilde{P}_\star:C^\infty(\MM)[[\lambda]] \to C^\infty(\MM)[[\lambda]]$ be a wave operator
which is 1.) a formal deformation of a normally hyperbolic operator and 2.) formally selfadjoint
with respect to the undeformed scalar product (\ref{eqn:undefsp}).
Then we can always define the wave operator $P_\star:= \iota^{-1} \circ \widetilde{P}_\star$,
which is 1.) a formal deformation of the same normally hyperbolic operator and 2.)
formally selfadjoint with respect to the deformed scalar product (\ref{eqn:defsp}). 
Furthermore, if $(\MM,g_\star,\vols,\mathcal{F})$ is a compactly deformed Lorentzian manifold,
then demanding $P_\star$ or $\widetilde{P}_\star$ to be a compactly deformed normally hyperbolic operator
is equivalent.
Thus, studying wave operators of the type $P_\star$ or $\widetilde{P}_\star$ is equivalent,
and the relation is given by $\widetilde{P}_\star = \iota\circ P_\star$.

The advantage of studying operators of the type $\widetilde{P}_\star$, 
compared to the type $P_\star$, is that they are real, in case $\widetilde{P}_\star^\top$ is real.
To see this, note that the isomorphism $\star_g$ is real and therewith
reality of the top-form valued operator $P_\star^\top$ is equivalent to the reality of $\widetilde{P}_\star$.

\subsection{Green's operators:}
For a compactly deformed time-oriented, connected and globally hyperbolic Lorentzian manifold 
$(\MM,g_\star,\vols,\mathcal{F})$ we have shown in Theorem \ref{theo:greendef}
that there exist unique Green's operators for all compactly deformed normally hyperbolic operators.
Thus, the existence and uniqueness of Green's operators for both, $P_\star$ and $\widetilde{P}_\star$, is guaranteed.
We can also derive an explicit map relating the Green's operators $\Delta_{\star\pm}$ for $P_\star$
and the Green's operators $\widetilde{\Delta}_{\star\pm}$ for $\widetilde{P}_\star$.
For this consider the following maps
\begin{flalign}
\label{eqn:tildegreenexplicit}
\widetilde{\Delta}_{\star\pm}:= \Delta_{\star\pm}\circ \iota^{-1}: C^\infty_0(\MM)[[\lambda]]\to C^\infty(\MM)[[\lambda]]~.  
\end{flalign}
These maps are well-defined, since $\iota$ restricts to a $\bbC[[\lambda]]$-module
automorphism on compactly supported functions
 $\iota: C_0^\infty(\MM)[[\lambda]]\to C^\infty_0(\MM)[[\lambda]]$.
We explicitly check the conditions for $\widetilde{\Delta}_{\star\pm}$ to be Green's operators
for $\widetilde{P}_\star$ (\ref{eqn:green}). The first and second condition follow from small calculations
\begin{subequations}
\begin{flalign}
& \widetilde{P}_\star\circ \widetilde{\Delta}_{\star\pm} = \iota \circ P_\star\circ \Delta_{\star\pm}\circ \iota^{-1} = 
\iota\circ \iota^{-1} = \id_{C^\infty_0(\MM)[[\lambda]]}~,\\
&\widetilde{\Delta}_{\star\pm}\circ \widetilde{P}_\star\big\vert_{C_0^\infty(\MM)[[\lambda]]} = \Delta_{\star\pm}\circ\iota^{-1}
\circ\iota\circ P_\star\big\vert_{C^\infty_0(\MM)[[\lambda]]} = \id_{C^\infty_0(\MM)[[\lambda]]}~.
\end{flalign}
To prove the third condition, let $\varphi\in C^\infty_0(\MM)$ be arbitrary. We obtain
\begin{flalign}
 \nn\supp\bigl(\widetilde{\Delta}_{(n)\pm}(\varphi)\bigr) &= 
\supp\left(\sum\limits_{m=0}^n \Delta_{(m)\pm}(\iota^{-1}_{(n-m)}(\varphi))\right)\\ 
&\subseteq \bigcup_{m=0}^n J_\pm\bigl(\supp(\iota^{-1}_{(n-m)}(\varphi))\bigr)
\subseteq J_\pm(\supp(\varphi))~,
\end{flalign}
\end{subequations}
since $\supp(\iota^{-1}_{(n)}(\varphi))\subseteq \supp(\varphi)$ for all $n$. Thus,
(\ref{eqn:tildegreenexplicit}) are the unique Green's operators for $\widetilde{P}_\star$.

We define the retarded-advanced Green's operator for $\widetilde{P}_\star$ by
\begin{flalign}
 \widetilde{\Delta}_\star := \widetilde{\Delta}_{\star+} - \widetilde{\Delta}_{\star-} = \Delta_\star\circ \iota^{-1} : C^\infty_0(\MM)[[\lambda]]
\to C^\infty_\sc(\MM)[[\lambda]]~.
\end{flalign}
Then the following sequence of $\bbC[[\lambda]]$-linear maps is a complex, which is exact everywhere
\begin{flalign}
 0\longrightarrow C_0^\infty(\mathcal{M})[[\lambda]] \stackrel{\widetilde{P}_\star}{\longrightarrow} 
C_0^\infty(\mathcal{M})[[\lambda]] \stackrel{\widetilde{\Delta}_\star}{\longrightarrow} 
C_\mathrm{sc}^\infty(\mathcal{M})[[\lambda]] \stackrel{\widetilde{P}_\star}{\longrightarrow}C_\mathrm{sc}^\infty(\mathcal{M})[[\lambda]]~.
\end{flalign}
Let us make some comments on the proof of this statement. The sequence trivially is a complex,
since $\widetilde{\Delta}_{\star\pm}$ are Green's operators for $\widetilde{P}_\star$.
The exactness follows from the exactness of the corresponding complex for $P_\star$ and $\Delta_\star$,
see Theorem \ref{theo:complex}. For completeness, we perform the proof.
To show the first exactness, let $\varphi\in C^\infty_0(\MM)[[\lambda]]$ such that $\widetilde{P}_\star(\varphi) = 
\iota(P_\star(\varphi))=0$. Then due to Theorem \ref{theo:complex} and the fact that $\iota$ is an isomorphism
we have $\varphi=0$. To show the second exactness, let $\varphi\in C^\infty_0(\MM)[[\lambda]]$
such that $\widetilde{\Delta}_\star(\varphi) = \Delta_\star(\iota^{-1}(\varphi))=0$. Then
$\iota^{-1}(\varphi)\in P_\star[C_0^\infty(\MM)[[\lambda]]]$ and $\varphi\in \widetilde{P}_\star[C^\infty_0(\MM)[[\lambda]]]$.
To show the third exactness, let $\varphi\in C^\infty_\mathrm{sc}(\MM)[[\lambda]]$ such that
$\widetilde{P}_\star(\varphi)=\iota(P_\star(\varphi))=0$. Then $P_\star(\varphi)=0$ and $\varphi\in 
\Delta_\star[C^\infty_0(\MM)[[\lambda]]] = \widetilde{\Delta}_\star[C^\infty_0(\MM)[[\lambda]]]$.

The result is that we have the isomorphism
\begin{flalign}
\label{eqn:solisodeftilde}
\widetilde{\mathcal{I}}_\star : C^\infty_0(\MM)[[\lambda]]/\widetilde{P}_\star[C^\infty_0(\MM)[[\lambda]]]\to 
\Sol_{\widetilde{P}_\star}
\,,~[\varphi]\mapsto \widetilde{\Delta}_\star(\varphi)~.
\end{flalign}

\subsection{Symplectic $\bbR[[\lambda]]$-modules:}
In Proposition \ref{propo:HRprop} we have shown that the $\bbR[[\lambda]]$-submodule
\begin{flalign}
 H_\bbR:= \bigl\lbrace \varphi\in C^\infty_0(\MM)[[\lambda]]: \bigl(\Delta_{\star\pm}(\varphi)\bigr)^\ast  = \Delta_{\star\pm}(\varphi)  \bigr\rbrace
\end{flalign}
generates by the action of $\Delta_\star$ all real solutions of the wave equation given by $P_\star$.
As a consequence, the $\bbR[[\lambda]]$-module $V_\star = H_\bbR/P_\star[C^\infty_0(\MM,\bbR)[[\lambda]]]$
is isomorphic to the space of real solutions $\Sol_{P_\star}^\bbR$.
We have equipped $V_\star$ with the symplectic structure
\begin{flalign}
 \omega_\star: V_\star\times V_\star\to \bbR[[\lambda]]\,,~([\varphi],[\psi])\mapsto \spp{\varphi}{\Delta_\star(\psi)}_\star~.
\end{flalign}

For the operator $\widetilde{P}_\star$ the space of real solutions is generated by acting 
with $\widetilde{\Delta}_\star$ on the $\bbR[[\lambda]]$-module $\widetilde{V}_\star := 
C^\infty_0(\MM,\bbR)[[\lambda]]/\widetilde{P}_\star[C^\infty_0(\MM,\bbR)[[\lambda]]]$, more precisely
$\widetilde{V}_\star$ is isomorphic to $\Sol_{\widetilde{P}_\star}^\bbR = \Sol_{P_\star}^\bbR$.
The natural symplectic structure on $\widetilde{V}_\star$ reads
\begin{flalign}
 \widetilde{\omega}_\star:\widetilde{V}_\star\times\widetilde{V}_\star\to\bbR[[\lambda]]\,,~([\varphi],[\psi])\mapsto
\spp{\varphi}{\widetilde{\Delta}_\star(\psi)}~,
\end{flalign}
where now we use the undeformed scalar product (\ref{eqn:undefsp}).

There is a symplectic isomorphism between the symplectic $\bbR[[\lambda]]$-modules
$(V_\star,\omega_\star)$ and $(\widetilde{V}_\star,\widetilde{\omega}_\star)$, which we are going to construct now.
Firstly, note that the map $\iota$ provides an $\bbR[[\lambda]]$-module isomorphism
 $\iota: H_\bbR\to C^\infty_0(\MM,\bbR)[[\lambda]]$. To see this let 
$\varphi\in C^\infty_0(\MM,\bbR)[[\lambda]]$ be arbitrary, then $\iota^{-1}(\varphi)\in H_\bbR$ since
\begin{flalign}
 \bigl(\Delta_{\star\pm}(\iota^{-1}(\varphi))\bigr)^\ast = 
 \bigl(\widetilde{\Delta}_{\star\pm}(\varphi)\bigr)^\ast = \widetilde{\Delta}_{\star\pm}(\varphi) = 
\Delta_{\star\pm}(\iota^{-1}(\varphi))~,
\end{flalign}
where we have used that $\widetilde{\Delta}_{\star\pm}$ are real due to the reality of $\widetilde{P}_\star$.
Let now $\varphi\in H_\bbR$ be arbitrary, then there is a $\psi\in C^\infty_0(\MM)[[\lambda]]$, such that
$\varphi = \iota^{-1}(\psi)$. We find that $\psi$ is real
\begin{flalign}
 0=\bigl(\Delta_{\star\pm}(\varphi)\bigr)^\ast - \Delta_{\star\pm}(\varphi)  =
 \widetilde{\Delta}_{\star\pm}(\psi^\ast-\psi)\quad\Leftrightarrow\quad \psi^\ast=\psi~.
\end{flalign}
The isomorphism $\iota: H_\bbR\to C^\infty_0(\MM,\bbR)[[\lambda]]$
gives rise to an isomorphism between the factor spaces $V_\star$ and $\widetilde{V}_\star$,
since
\begin{flalign}
 \iota(\varphi +P_\star(\psi)) = \iota(\varphi) + \widetilde{P}_\star(\psi)~,
\end{flalign}
for all $\varphi\in H_\bbR$ and $\psi\in C^\infty_0(\MM,\bbR)[[\lambda]]$.
Even more, it is a symplectic isomorphism between $(V_\star,\omega_\star)$ and 
$(\widetilde{V}_\star,\widetilde{\omega}_\star)$, since
\begin{flalign}
 \widetilde{\omega}_\star\bigl([\iota(\varphi)],[\iota(\psi)]\bigr) = 
 \spp{\iota(\varphi)}{\widetilde{\Delta}_\star(\iota(\psi))} = \spp{\varphi}{\Delta_\star(\psi)}_\star
= \omega_\star\bigl([\varphi],[\psi]\bigr)~,
\end{flalign}
for all $[\varphi],[\psi]\in V_\star$.

\subsection{$\ast$-algebra of field polynomials:}
The symplectic isomorphism $\iota: V_\star \to \widetilde{V}_\star$
leads to an isomorphism between the corresponding $\ast$-algebras of field polynomials $\AA_{(V_\star,\omega_\star)}$
and $\AA_{(\widetilde{V}_\star,\widetilde{\omega}_\star)}$, defined in Definition \ref{def:fieldpoly}.
This is a consequence of the following 
\begin{propo}
\label{propo:symplectoalgebra}
Let $(V_1,\omega_1)$ and $(V_2,\omega_2)$ be two symplectic $\bbR[[\lambda]]$-modules 
and let $S:V_1\to V_2$ be a symplectic linear map. Then there exists a $\ast$-algebra homomorphism
$\mathfrak{S}:\mathcal{A}_{(V_1,\omega_1)}\to \mathcal{A}_{(V_2,\omega_2)}$, such that
\begin{subequations}
\begin{flalign}
 \mathfrak{S}(1)&=1~,\\
 \mathfrak{S}\bigl(\Phi_1(v)\bigr)&=\Phi_2(S(v))~,
\end{flalign}
\end{subequations}
for all $v\in V_1$. If $S$ is a symplectic isomorphism then $\mathfrak{S}$ is a $\ast$-algebra isomorphism.
\end{propo}
\begin{proof}
 We construct a $\ast$-algebra homomorphism $\sigma:\AA_{\text{free}_1} \to \AA_{\text{free}_2}$
between the free algebras by defining on the generators
\begin{subequations}
\begin{flalign}
 \sigma(1) &= 1~,\\
 \sigma\bigl(\Phi_1(v)\bigr) &=  \Phi_2(S(v))~,\\
 \sigma\bigl(\Phi_1(v)^\ast\bigr) &= \Phi_2(S(v))^\ast~,
\end{flalign}
\end{subequations}
for all $v\in V_1$, and extending $\sigma$ to $\AA_{\text{free}_1}$ as a $\ast$-algebra homomorphism.
The $\ast$-algebra homomorphism $\sigma$ induces a $\ast$-algebra homomorphism
$\mathfrak{S}:\AA_{(V_1,\omega_1)}\to\AA_{(V_2,\omega_2)}$, since it includes the $\ast$-ideals (\ref{eqn:fieldalgebra})
$\sigma[\mathcal{I}_1]\subseteq \mathcal{I}_2$, i.e.~
\begin{subequations}
\begin{flalign}
 &\sigma\bigl(\Phi_1(\beta\, v+\gamma\, u) -\beta\,\Phi_1(v) -\gamma\,\Phi_1(u)\bigr) = 
\Phi_2(\beta\,S(v)+\gamma\,S(u)) - \beta\,\Phi_2(S(v)) - \gamma\,\Phi_2(S(u))~,\\
 &\sigma\bigl(\Phi_1(v)^\ast - \Phi_1(v)\bigr) = \Phi_2(S(v))^\ast - \Phi_2(S(v))~,\\
 &\sigma\bigl([\Phi_1(v),\Phi_1(u)]-i\,\omega_1(v,u)\,1\bigr) =[\Phi_2(S(v)),\Phi_2(S(u))]-i\,\omega_2\bigl(S(v),S(u)\bigr)\,1 ~,
\end{flalign}
\end{subequations}
are elements of $\mathcal{I}_2$ for all $v,u\in V_1$ and $\beta,\gamma\in \bbR[[\lambda]]$.

If $S$ is a symplectic isomorphism, then the map $\sigma$ is a $\ast$-algebra isomorphism, which
induces a $\ast$-algebra isomorphism $\mathfrak{S}:\mathcal{A}_{(V_1,\omega_1)}\to \mathcal{A}_{(V_2,\omega_2)}$,
since $\sigma[\mathcal{I}_1]=\mathcal{I}_2$.

\end{proof}
This proposition has the following consequence: We can equivalently describe a deformed quantum
field theory in terms of the $\ast$-algebra of field polynomials $\AA_{(V_\star,\omega_\star)}$ or 
$\AA_{(\widetilde{V}_\star,\widetilde{\omega}_\star)}$. However,
the physical interpretation has to be adapted properly.
Due to the modified reality property of the $\bbR[[\lambda]]$-module $H_\bbR$, see Proposition \ref{propo:HRprop},
we should interpret the generators $\Phi_\star([\varphi])$ of 
$\AA_{(V_\star,\omega_\star)}$ as smeared field operators
``$\Phi_\star([\varphi]) = \int_\MM \Phi_\star\star \varphi\star\vols$''. The property
$(\varphi\star\vols)^\ast = \varphi\star\vols$, for all $\varphi\in H_\bbR$, implies that
formally $\Phi_\star([\varphi])^\ast = \Phi_\star([\varphi])$. On the other hand,
since $\widetilde{V}_\star$ is based on the space of real compactly supported functions $C_0^\infty(\MM,\bbR)[[\lambda]]$,
the corresponding generators  $\widetilde{\Phi}_\star([\varphi])$ of 
$\AA_{(\widetilde{V}_\star,\widetilde{\omega}_\star)}$ should be interpreted as smeared field operators
``$\widetilde{\Phi}_\star([\varphi]) = \int_\MM\widetilde{\Phi}_\star\,\varphi\,\vol$'' without 
$\star$-products.


\section{Symplectic isomorphisms to commutative field theory}
In this section we provide a theorem showing that the deformed symplectic $\bbR[[\lambda]]$-module 
$(\widetilde{V}_\star,\widetilde{\omega}_\star)$ is isomorphic,
via a symplectic isomorphism, to the formal power series extension of the undeformed symplectic vector space\footnote{
Note that $C^\infty_0(\MM,\bbR)[[\lambda]]/P[C^\infty_0(\MM,\bbR)[[\lambda]]]\simeq V[[\lambda]] $ via the map
$C^\infty_0(\MM,\bbR)[[\lambda]]/P[C^\infty_0(\MM,\bbR)[[\lambda]]]\to V[[\lambda]]\,,~[\varphi]\mapsto \sum_{n=0}^\infty\lambda^n\,
[\varphi_{(n)}]~.$ This natural identification will be used frequently throughout this section, without
explicitly denoting the corresponding map by some symbol.
}
$(V[[\lambda]],\omega)$,
where $V:=C^\infty_0(\mathcal{M},\bbR)/P[C^\infty_0(\mathcal{M},\bbR)]$ and
\begin{flalign}
 \omega:V[[\lambda]]\times V[[\lambda]] \to\bbR[[\lambda]]\,,~([\varphi],[\psi])\mapsto \omega([\varphi],[\psi]) 
= \bigl(\varphi,\Delta(\psi)\bigr)~.
\end{flalign} 
Here $P=\widetilde{P}_\star\vert_{\lambda=0}$ 
is the undeformed wave operator and  $\Delta=\widetilde{\Delta}_\star\vert_{\lambda=0}$ is the corresponding undeformed 
retarded-advanced Green's operator. For the undeformed field theory we have the isomorphism
\begin{flalign}
 \label{eqn:solisoclasslam}
\mathcal{I}:C^\infty_0(\MM)[[\lambda]]/P[C^\infty_0(\MM)[[\lambda]]] \to \Sol_P\,,~[\varphi]\mapsto \Delta(\varphi)~,
\end{flalign}
which leads to an isomorphism $V[[\lambda]]\simeq \Sol_P^\bbR$.

\begin{lem}
\label{lem:solutionspaceiso}
Let $(\MM,g_\star,\vols,\mathcal{F})$ be a compactly deformed time-oriented, connected and globally hyperbolic 
Lorentzian manifold and $\widetilde{P}_\star:C^\infty(\MM)[[\lambda]]\to C^\infty(\MM)[[\lambda]]$
be a real compactly deformed normally hyperbolic operator. 
Let $\widetilde{\Delta}_{\star\pm}$ be the Green's operators for $\widetilde{P}_\star$. 
 Then the maps $T_\pm:\Sol_{\widetilde{P}_\star}^\bbR\to \Sol_{P}^\bbR$ defined by 
\begin{flalign}
\label{eqn:soliso}
T_\pm:=\id_{C^\infty_\sc(\mathcal{M})[[\lambda]]}+t_{\pm}=
\id_{C^\infty_\sc(\mathcal{M})[[\lambda]]}+\Delta_{\pm}\circ\sum\limits_{n=1}^\infty\lambda^n\,\widetilde{P}_{(n)}
\end{flalign}
are $\bbR[[\lambda]]$-module isomorphisms.
\end{lem}
\begin{proof}
The compositions $\Delta_\pm\circ \widetilde{P}_{(n)}$ are well-defined due to the support condition on
$\widetilde{P}_{(n)}$, $n>0$. Additionally, the maps $\Delta_\pm\circ \widetilde{P}_{(n)}$ are real for all $n>0$.
 The composition of $P$ and $T_\pm$ is given by
\begin{flalign}
 P\circ T_\pm = P + P\circ \Delta_\pm\circ\sum\limits_{n=1}^\infty \lambda^n\, \widetilde{P}_{(n)} =
P + \sum\limits_{n=1}^\infty \lambda^n\,\widetilde{P}_{(n)} =\sum\limits_{n=0}^\infty \lambda^n\,\widetilde{P}_{(n)}=
\widetilde{P}_\star~,
\end{flalign}
where we have used that $P\circ\Delta_\pm = \text{id}_{C^\infty_0(\mathcal{M})[[\lambda]]}$.
Thus, for all $\Phi\in\Sol_{\widetilde{P}_\star}^\bbR$ we have $T_\pm(\Phi)\in\Sol_P^\bbR$.\newline
The inverse of $T_\pm$ is constructed by the geometric series
\begin{flalign}
 T_\pm^{-1}=\left(\id_{C^\infty_\sc(\mathcal{M})[[\lambda]]} +t_\pm\right)^{-1}
=\sum\limits_{m=0}^\infty \left(-t_\pm\right)^m \stackrel{\text{(\ref{eqn:explicitgreen})}}{=} 
\id_{C^\infty_\sc(\mathcal{M})[[\lambda]]} -\widetilde{\Delta}_{\star\pm}\circ\sum\limits_{n=1}^\infty\lambda^n\,\widetilde{P}_{(n)}
\end{flalign}
and satisfies
\begin{flalign}
 \widetilde{P}_\star\circ T_\pm^{-1} &= \widetilde{P}_\star - \widetilde{P}_\star\circ
\widetilde{\Delta}_{\star\pm}\circ \sum\limits_{n=1}^\infty\lambda^n\,\widetilde{P}_{(n)}
= \widetilde{P}_\star - \sum\limits_{n=1}^\infty\lambda^n\,\widetilde{P}_{(n)} = P~,
\end{flalign}
where we have used that $\widetilde{P}_\star\circ\widetilde{\Delta}_{\star\pm}=\text{id}_{C^\infty_0(\mathcal{M})[[\lambda]]}$.
Thus, for all $\Phi\in\Sol_{P}^\bbR$ we have $T^{-1}_\pm(\Phi)\in\Sol_{\widetilde{P}_\star}^\bbR$.

\end{proof}
It turns out that in general $T_+$ and $T_-$ differ. To see this let $\Phi\in\Sol_{\widetilde{P}_\star}^\bbR$ be arbitrary.
Then there is a $\varphi\in C^\infty_0(\mathcal{M},\bbR)[[\lambda]]$, 
such that $\Phi=\widetilde{\Delta}_\star(\varphi)$.
We obtain 
\begin{flalign}
T_+(\Phi) - T_-(\Phi) = \Delta\left(\sum\limits_{n=1}^\infty\lambda^n\,\widetilde{P}_{(n)}\left(\widetilde{\Delta}_\star(\varphi)
\right)\right)
=-\Delta\left(P\left(\widetilde{\Delta}_\star(\varphi)\right)\right)~.
\end{flalign}
The difference between $T_+$ and $T_-$ is thus related to the operator $\Delta\circ P\circ \widetilde{\Delta}_\star$. 
Notice that this operator is {\it not} zero in general, since the relation $\Delta\circ P =0$ just holds when acting on functions
of compact support, while $\widetilde{\Delta}_\star$ maps to functions of noncompact support. 
To be more explicit we expand the operator $\Delta\circ P\circ\widetilde{\Delta}_\star$ to first order in $\lambda$
 by using (\ref{eqn:explicitgreen}) and find
\begin{flalign}
\nonumber \Delta\circ P\circ \widetilde{\Delta}_\star &= \Delta\circ P\circ \bigl(\Delta - \lambda(\Delta_+\circ \widetilde{P}_{(1)}
\circ\Delta_+ - \Delta_-\circ \widetilde{P}_{(1)}\circ \Delta_-) \bigr)+\mathcal{O}(\lambda^2) \\
&= -\lambda \,\Delta\circ \widetilde{P}_{(1)}\circ \Delta +\mathcal{O}(\lambda^2)~.
\end{flalign} 
Since $\widetilde{P}_{(1)}$ comes from the choice of deformation, while $\Delta$ describes the commutative dynamics, 
these operators are independent and $\Delta\circ P\circ \widetilde{\Delta}_\star$ in general does not vanish.
\begin{rem}
The maps $T_\pm$ can be interpreted as retarded/advanced isomorphisms, since they 
 depend on the retarded/advanced Green's operators.
Due to the support property of $\Delta_\pm$ and the support condition
on $\widetilde{P}_{(n)}$, $n>0$, we obtain order by order in $\lambda$ that $T_\pm(\Phi)$ is equal 
to $\Phi$ for sufficiently small/large times, i.e.~for $t\to \mp\infty$.
\end{rem}

Employing the isomorphisms $T_\pm$, the isomorphism $\widetilde{\mathcal{I}}_\star$ (\ref{eqn:solisodeftilde}) 
and its commutative counterpart $\mathcal{I}$ (\ref{eqn:solisoclasslam}), we obtain the $\bbR[[\lambda]]$-module isomorphisms
$\mathbf{T}_\pm :=\mathcal{I}^{-1}\circ T_\pm\circ\widetilde{\mathcal{I}}_\star: \widetilde{V}_\star\to V[[\lambda]]$. We can
map the deformed symplectic $\bbR[[\lambda]]$-module $(\widetilde{V}_\star,\widetilde{\omega}_\star)$, 
via a symplectic isomorphism, 
to the symplectic $\bbR[[\lambda]]$-module $(V[[\lambda]],\widehat{\omega}_\star)$, where by definition
\begin{flalign}
\widehat{\omega}_\star\bigl([\varphi],[\psi]\bigr) := 
\widetilde{\omega}_\star\bigl(\mathbf{T}^{-1}_\pm\bigl([\varphi]\bigr),\mathbf{T}^{-1}_\pm\bigl([\psi]\bigr) \bigr)~,
\end{flalign}
for all $[\varphi],[\psi]\in V[[\lambda]]$. 
This expression can be simplified and we obtain 
\begin{flalign}
\nonumber \widehat{\omega}_\star\bigl([\varphi],[\psi]\bigr) 
&= \Bigl(\mathbf{T}^{-1}_\pm([\varphi]),\widetilde{\Delta}_\star\bigl(\mathbf{T}^{-1}_\pm([\psi])\bigr)\Bigr)
 = \Bigl(\mathbf{T}^{-1}_\pm([\varphi]),T^{-1}_\pm\bigl(\Delta(\psi)\bigr)\Bigr)\\
&=\Bigl(T_\pm^{-1\dagger}\bigl(\mathbf{T}^{-1}_\pm([\varphi])\bigr),\Delta(\psi)\Bigr)
=-\Bigl(\Delta\bigl(T_\pm^{-1\dagger}\bigl(\mathbf{T}^{-1}_\pm([\varphi])\bigr)\bigr),\psi\Bigr)~,
\end{flalign}
where we have used the adjoint map
\begin{flalign}
T_\pm^{-1\dagger} = \text{id}_{C^\infty_0(\mathcal{M})[[\lambda]]} - 
\sum\limits_{n=1}^\infty\lambda^n\,\widetilde{P}_{(n)}\circ \widetilde{\Delta}_{\star\mp}
\end{flalign}
and that $\Delta$ is antihermitian. Defining the map $\widehat{\Delta}_\star:C^\infty_0(\mathcal{M})[[\lambda]]\to 
\Sol_{P}$ by
\begin{flalign}
\widehat{\Delta}_\star := \Delta\circ T_\pm^{-1\dagger} \circ \mathcal{I}_\star^{-1}\circ T^{-1}_\pm\circ \Delta~,
\end{flalign}
we have for all $[\varphi],[\psi]\in V[[\lambda]]$
\begin{flalign}
\widehat{\omega}_\star\bigl([\varphi],[\psi]\bigr) = \bigl(\varphi,\widehat{\Delta}_\star(\psi)\bigr)~.
\end{flalign}
Analogously to (\ref{eqn:solisoclasslam}) and (\ref{eqn:solisodeftilde}) we define the isomorphism 
\begin{flalign}
 \widehat{\mathcal{I}}_\star: C^\infty_0(\MM)[[\lambda]]/P[C^\infty_0(\MM)[[\lambda]]] \to \Sol_{P}~,\quad 
[\varphi]\mapsto\widehat{\Delta}_\star(\varphi)~,
\end{flalign}
which leads to an isomorphism $V[[\lambda]]\simeq \Sol_P^\bbR$.
We obtain the following
\begin{theo}
 \label{theo:symplecto}
The map $\mathbf{S}: V[[\lambda]]\to V[[\lambda]]$ defined by $\mathbf{S} = \sum\limits_{n=0}^\infty\lambda^n S_{(n)}$,
where
\begin{subequations}
 \begin{flalign}
S_{(0)} &= \id_V~,\\
S_{(1)} &= \frac{1}{2} \mathcal{I}^{-1}\circ\widehat{\mathcal{I}}_{(1)}~,\\
 S_{(n)} &= \frac{1}{2}\Bigl(\mathcal{I}^{-1}\circ \widehat{\mathcal{I}}_{(n)} - \sum\limits_{m=1}^{n-1}S_{(m)}\circ S_{(n-m)}\Bigr) ~,~\forall n\geq 2~,
 \end{flalign}
\end{subequations}
provides a symplectic isomorphism between $(V[[\lambda]],\widehat{\omega}_\star)$ and $(V[[\lambda]],\omega)$,
i.e.~for all $v,u\in V[[\lambda]]$ we have $\widehat{\omega}_\star(v,u)=\omega(\mathbf{S}(v),\mathbf{S}(u))$.
\end{theo}
\begin{proof}
 The map $\mathbf{S}$ is invertible, since it is a formal deformation of the identity map.
Let $v,u\in V[[\lambda]]$ be arbitrary. We obtain
\begin{flalign}
 \nonumber \omega(\mathbf{S}(v),\mathbf{S}(u)) &= \sum\limits_{n=0}^\infty\lambda^n\sum\limits_{m+k+i+j=n}
\omega\bigl(S_{(m)}(v_{(k)}),S_{(i)}(u_{(j)})\bigr)\\
\label{eqn:zwischenschritt}&= \sum\limits_{n=0}^\infty\lambda^n\sum\limits_{k+j+l=n}~\sum\limits_{m+i=l}
\omega\bigl(S_{(m)}(v_{(k)}),S_{(i)}(u_{(j)})\bigr)~.
\end{flalign}
Note that for all $v,u\in V$ we have $\omega(v,S_{(0)}(u))=\omega(S_{(0)}(v),u)$ (trivially) and 
\begin{flalign}
 \omega(v,S_{(1)}(u))=\frac{1}{2}\widehat{\omega}_{(1)}(v,u)=-\frac{1}{2}\widehat{\omega}_{(1)}(u,v)
= \omega(S_{(1)}(v),u)~.
\end{flalign}
By induction it follows that $\omega(v,S_{(n)}(u))=\omega(S_{(n)}(v),u)$ for all $n\geq 0$.

Using this, the inner sum of (\ref{eqn:zwischenschritt}) reads
\begin{flalign}
 \sum\limits_{m+i=l}
\omega(S_{(m)}(v_{(k)}),S_{(i)}(u)_{(j)}) 
= \omega\left(v_{(k)},\sum\limits_{m=0}^l S_{(m)}\circ S_{(l-m)}(u_{(j)})\right)~.
\end{flalign}
It remains to simplify the map $\sum_{m=0}^l S_{(m)}\circ S_{(l-m)}$. For $l=0$ this is simply the identity map
and for $l=1$ it reads $2\, S_{(1)}=\mathcal{I}^{-1}\circ \widehat{\mathcal{I}}_{(1)}$. For $l\geq 2$ we find
\begin{flalign}
 \sum\limits_{m=0}^l S_{(m)}\circ S_{(l-m)} = 2\,S_{(l)} + \sum\limits_{m=1}^{l-1}S_{(m)}\circ S_{(l-m)}
 = \mathcal{I}^{-1}\circ\widehat{\mathcal{I}}_{(l)}~.
\end{flalign}
Thus, (\ref{eqn:zwischenschritt}) reads
\begin{flalign}
 \omega(\mathbf{S}(v),\mathbf{S}(u)) =  
\sum\limits_{n=0}^\infty\lambda^n\sum\limits_{k+j+l=n}\omega\bigl(v_{(k)},\mathcal{I}^{-1}\circ \widehat{\mathcal{I}}_{(l)}
(u_{(j)})\bigr)
= \widehat{\omega}_{\star}(v,u)~.
\end{flalign}

\end{proof}
As a direct consequence we obtain
\begin{cor}
\label{cor:symplecto}
 The maps $\widetilde{\mathbf{S}}_\pm:=\mathbf{S}\circ \mathbf{T}_\pm:\widetilde{V}_\star\to V[[\lambda]]$ 
are symplectic isomorphisms between $(\widetilde{V}_\star,\widetilde{\omega}_\star)$ and $(V[[\lambda]],\omega)$, 
i.e.~$\omega\bigl(\widetilde{\mathbf{S}}_\pm(v),\widetilde{\mathbf{S}}_\pm(u)\bigr)
=\widetilde{\omega}_\star\bigl(v,u\bigr)$ for all $v,u\in \widetilde{V}_\star$.

\noindent Moreover, the maps $\mathbf{S}_\pm:=\widetilde{\mathbf{S}}_\pm\circ\iota: V_\star\to V[[\lambda]]$ are
 symplectic isomorphisms between $(V_\star,\omega_\star)$ and $(V[[\lambda]],\omega)$.
\end{cor}


\section{Consequences of the symplectic isomorphisms}
We study the consequences of the symplectic isomorphisms $\mathbf{S}_\pm: V_\star\to V[[\lambda]]$.
Analogous results hold true for the symplectic isomorphisms
$\widetilde{\mathbf{S}}_\pm: \widetilde{V}_\star\to V[[\lambda]]$.
\subsection{$\ast$-algebra of field polynomials:}
Consider the $\ast$-algebras of field polynomials $\AA_{(V_\star,\omega_\star)}$ and
$\AA_{(V[[\lambda]],\omega)}$ of the deformed and undeformed quantum field theory, respectively.
From Proposition \ref{propo:symplectoalgebra} we immediately obtain
\begin{cor}
\label{cor:symplalg}
 There exist $\ast$-algebra isomorphisms $\mathfrak{S}_\pm:\AA_{(V_\star,\omega_\star)}\to\AA_{(V[[\lambda]],\omega)}$.
\end{cor}
This means that we can {\it mathematically} describe the quantum field theory on noncommutative
curved spacetimes by a formal power series extension of the corresponding 
quantum field theory on commutative curved spacetimes. However, the {\it physical interpretation}
has to be adapted properly: If we want to probe the deformed quantum field theory with
a set of smearing functions $\lbrace [\varphi_i] \rbrace$ in order to extract physical observables
(e.g.~correlation functions) we have to probe the corresponding commutative quantum field theory
with a different set of smearing functions $\lbrace \mathbf{S}_\pm([\varphi_i])\rbrace $
in order to obtain the same result. We will come back to this point when we discuss $n$-point functions.

There is a mathematical subtlety at this point: Instead of the algebra $\AA_{(V[[\lambda]],\omega)}$
for the undeformed quantum field theory we could also consider the algebra
$\AA_{(V,\omega)}[[\lambda]]$, where we first construct the $\ast$-algebra of field polynomials
over $\bbC$ and afterwards extend by formal power series. 
There is a natural identification of elements in $\AA_{(V[[\lambda]],\omega)}$ with elements
in $\AA_{(V,\omega)}[[\lambda]]$ given by regarding the generators $\Phi(v)$, $v\in V[[\lambda]]$, 
as formal power series $\Phi(v)=\sum_{n=0}^\infty\lambda^n\,\Phi(v_{(n)})$, where on the right hand
side $\Phi(v_{(n)})$, $v_{(n)}\in V$, denote the generators of $\AA_{(V,\omega)}$.
However, this identification just provides a $\ast$-algebra homomorphism
$\AA_{(V[[\lambda]],\omega)} \to \AA_{(V,\omega)}[[\lambda]]$, which is no isomorphism in general.
Employing methods from the Appendix \ref{app:basicsdefq}, we obtain that $\AA_{(V,\omega)}[[\lambda]]$
is the $\lambda$-adic completion of $\AA_{(V[[\lambda]],\omega)}$.
For the deformed quantum field theory this means that there is 
in addition to the $\ast$-algebra isomorphism $\AA_{(V_\star,\omega_\star)} \simeq \AA_{(V[[\lambda]],\omega)}$ also
a $\ast$-algebra homomorphism $\AA_{(V_\star,\omega_\star)}\to \AA_{(V,\omega)}[[\lambda]]$, 
which is not necessarily invertible.

\subsection{Symplectic automorphisms:}
An important class of symmetries of quantum field theories are those which are induced by
symplectic automorphisms of the corresponding symplectic $\bbR[[\lambda]]$-module.
\begin{defi}
 Let $(V,\omega)$ be a symplectic $\bbR[[\lambda]]$-module. A map $\alpha\in \End_{\bbR[[\lambda]]}(V)$
is called a {\it symplectic automorphism}, if it is invertible and if $\omega(\alpha(v),\alpha(u))=\omega(v,u)$,
for all $v,u\in V$. The set $\mathcal{G}_{(V,\omega)}\subseteq \End_{\bbR[[\lambda]]}(V)$ of all symplectic automorphisms
forms a group under the usual composition $\circ$ of endomorphisms. This group is called the
 {\it group of symplectic automorphisms}.
\end{defi}
Due to the symplectic isomorphisms $\mathbf{S}_\pm:V_\star\to V[[\lambda]]$
we obtain the following
\begin{cor}
\label{cor:symplauto}
 There exist group isomorphisms $\mathbf{S}_{\pm\mathcal{G}}: \mathcal{G}_{(V[[\lambda]],\omega)}\to
\mathcal{G}_{(V_\star,\omega_\star)}\,,~\alpha\mapsto \mathbf{S}_\pm^{-1}\circ \alpha\circ \mathbf{S}_\pm$.
\end{cor}
The group of symplectic automorphisms of a symplectic $\bbR[[\lambda]]$-module $(V,\omega)$
can be represented on its $\ast$-algebra of field polynomials
by defining on the generators, for all $\alpha\in\mathcal{G}_{(V,\omega)}$,
\begin{flalign}
 \alpha(1)&=1~,\quad \alpha\bigl(\Phi(v)\bigr)=\Phi(\alpha(v))~,~\text{for all }v\in V~,
\end{flalign}
and extending to $\AA_{(V,\omega)}$ as $\ast$-algebra homomorphisms.
\begin{ex}
Let $(\MM,g_\star,\vols,\mathcal{F})$ be a compactly deformed time-oriented, connected and globally hyperbolic 
Lorentzian manifold and $\widetilde{P}_\star:C^\infty(\MM)[[\lambda]]\to C^\infty(\MM)[[\lambda]]$
be the deformed Klein-Gordon operator. 
Assume that the underlying classical Lorentzian manifold $(\MM,g)$ has a group $G_{\text{iso}}$ of isometries.
Due to the $G_\text{iso}$-invariance of the classical Klein-Gordon operator and the classical Green's operators
we have $G_\text{iso}\subseteq \mathcal{G}_{(V[[\lambda]],\omega)}$, where the action of $G_\text{iso}$ 
on $V[[\lambda]]$ is given by the geometric action.
Corollary \ref{cor:symplauto} states that the deformed quantum field theory
$\AA_{(V_\star,\omega_\star)}$ enjoys the same transformations as symplectic automorphisms, i.e.~$G_\text{iso}
\subseteq \mathcal{G}_{(V_\star,\omega_\star)}$, not dependent on details of the deformation.
The action is given by $\alpha_\star:=\mathbf{S}^{-1}_\pm\circ\alpha\circ \mathbf{S}_\pm$
and it is in general different to the geometric action.
Thus, the deformation does not break isometries of the undeformed quantum field theory,
but replaces geometric symmetries by nongeometric ones.
\end{ex}

\subsection{Algebraic states:}
Due to the symplectic isomorphisms $\mathbf{S}_\pm:V_\star\to V[[\lambda]]$
the space of algebraic states on $\AA_{(V_\star,\omega_\star)}$ and $\AA_{(V[[\lambda]],\omega)}$
can be related in a precise way. To explain this, we require the following
\begin{lem}
 Let $A_1$ and $A_2$ be two unital $\ast$-algebras over $\bbC[[\lambda]]$ and 
$\kappa:A_1\to A_2$ be a $\ast$-algebra homomorphism. Then each state $\Omega_2$ on $A_2$ 
induces a state $\Omega_1$ on $A_1$ by defining
\begin{flalign}
 \Omega_1(a):=\Omega_2(\kappa(a))~,
\end{flalign}
for all $a\in A_1$.
\end{lem}
\begin{proof}
 $\Omega_1:A_1\to \bbC[[\lambda]]$ is a $\bbC[[\lambda]]$-linear map. We have $\Omega_1(1) =\Omega_2(\kappa(1))=\Omega_2(1)=1$
and 
\begin{flalign}
 \Omega_1(a^\ast a) = \Omega_2\bigl(\kappa(a^\ast a)\bigr) =\Omega_2\bigl((\kappa(a))^\ast\kappa(a)\bigr)\geq 0~,
\end{flalign}
for all $a\in A_1$.

\end{proof}
The state $\Omega_1$ defined above is called the {\it pull-back} of $\Omega_2$ under the $\ast$-algebra homomorphism
$\kappa$. Note that for $\ast$-algebra isomorphisms $\kappa:A_1\to A_2$ there is bijection between the states on
$A_1$ and $A_2$.
We immediately obtain from Corollary \ref{cor:symplalg} and the Lemma above
\begin{cor}
 The $\ast$-algebra isomorphisms $\mathfrak{S}_\pm:\AA_{(V_\star,\omega_\star)}\to\AA_{(V[[\lambda]],\omega)}$
 provide bijections between the states on $\AA_{(V_\star,\omega_\star)}$ and $\AA_{(V[[\lambda]],\omega)}$.
$\mathcal{G}_{(V[[\lambda]],\omega)}$-symmetric states are pulled-back to $\mathcal{G}_{(V_\star,\omega_\star)}$-symmetric
states, and vice versa.
\end{cor}

As discussed before in Chapter \ref{chap:qftdef}, Section \ref{sec:algstatesdef},
we can induce faithful states on $\AA_{(V,\omega)}[[\lambda]]$ from faithful states
on the undeformed and unextended quantum field theory $\AA_{(V,\omega)}$.
The $\ast$-algebra homomorphism $\AA_{(V_\star,\omega_\star)}\to \AA_{(V,\omega)}[[\lambda]]$
allows us to pull-back also these states on the deformed quantum field theory.

\subsection{$n$-point functions:}
We now focus on $n$-point functions in the deformed and undeformed quantum field theory.
Let $\Omega$ be an algebraic state on $\AA_{(V[[\lambda]],\omega)}$. We define the
$n$-point correlation functions by
\begin{flalign}
\Omega_n: C^\infty_0(\MM,\bbR)[[\lambda]]^{\times n}\to \bbC[[\lambda]]\,,\quad(\varphi_1,\dots,\varphi_n)\mapsto 
 \Omega\bigl(\Phi([\varphi_1])\cdots \Phi([\varphi_n])\bigr)~.
\end{flalign}
Using the $\ast$-algebra isomorphism $\mathfrak{S}_\pm$ we can pull-back the state $\Omega$ to a state 
$\Omega_\star$ on the deformed quantum field theory $\AA_{(V_\star,\omega_\star)}$.
The deformed $n$-point correlation functions can be expressed in terms of the undeformed ones via
\begin{flalign}
 \Omega_{\star n}: H_\bbR^{\times n} \to \bbC[[\lambda]]\,,\quad (\varphi_1,\dots,\varphi_n)\mapsto 
\Omega\bigl(\Phi(\mathbf{S}_\pm([\varphi_1]))\cdots \Phi(\mathbf{S}_\pm([\varphi_n])) \bigr)~.
\end{flalign}
Note that $\Omega_{\star n}$ is defined on $H_\bbR$, such that we can not directly compare it to the undeformed
correlation functions $\Omega_{n}$, acting on compactly supported real functions.
This is due to the different smearing interpretation we have for the generators 
$\Phi([\varphi])$ of $\AA_{(V[[\lambda]],\omega)}$ and $\Phi_\star([\varphi])$ of $\AA_{(V_\star,\omega_\star)}$ 
(the latter ones are smeared with $\star$-products).
Composing $\Omega_{\star n}$ with the isomorphism $\bigl(\iota^{-1}\bigr)^{\times n}:C^\infty_0(\MM,\bbR)[[\lambda]]^{\times n}
 \to H_\bbR^{\times n}$
we obtain the correlation functions for the quantum field theory $\AA_{(\widetilde{V}_\star,\widetilde{\omega}_\star)}$
\begin{flalign}
 \nn \widetilde{\Omega}_{\star n} := \Omega_{\star n}\circ \bigl(\iota^{-1}\bigr)^{\times n}: ~& 
C^\infty_0(\MM,\bbR)[[\lambda]]^{\times n} \to\bbC[[\lambda]]~,\\
& (\varphi_1,\dots,\varphi_n)\mapsto \Omega\bigl(\Phi(\widetilde{\mathbf{S}}_\pm([\varphi_1]))\cdots \Phi(\widetilde{\mathbf{S}}_\pm
([\varphi_n])\bigr)~.
\end{flalign}
Since both, $\Omega_n$ and $\widetilde{\Omega}_{\star n}$, are maps on real functions of compact support, 
we can compare them. Due to the maps $\widetilde{\mathbf{S}}_\pm$, the undeformed and deformed correlation 
functions {\it evaluated on the same functions} do not agree
\begin{flalign}
 \Omega_n(\varphi_1,\dots,\varphi_n)\neq \widetilde{\Omega}_{\star n}(\varphi_1,\dots,\varphi_n)~.
\end{flalign}
Because the functions $\varphi_i$ which localize the correlation functions are fixed by details
of the experimental setup, using the same apparatus we could in principle distinguish between
the deformed and undeformed quantum field theory in the state $\widetilde{\Omega}_\star$ and $\Omega$, respectively.
For example, in cosmology the usual prescription to extract the primordial power spectrum
is to evaluate the $2$-point correlation function, regarded as a distribution kernel, on equal-time hypersurfaces.
If we apply the same prescription to the deformed quantum field theory we obtain
\begin{flalign}
 \Omega_2(x_1,x_2) \neq \widetilde{\Omega}_{\star 2}(x_1,x_2)~.
\end{flalign}

This discussion shows that the deformed quantum field theory in the induced state $\widetilde{\Omega}_\star$ 
is physically inequivalent to the undeformed quantum field theory in the state $\Omega$.
It remains to exclude the possibility that there is a different state $\Omega^\prime$ of the undeformed quantum field
theory such that all $n$-point functions $\Omega_n^\prime$ in this state coincide with the deformed
$n$-point functions $\widetilde{\Omega}_{\star n}$. If this would be the case, then we would
in particular have, for all $\varphi,\psi\in C^\infty_0(\MM,\bbR)[[\lambda]]$,
\begin{flalign}
\widetilde{\Omega}_{\star 2}(\varphi,\psi) = \Omega_2^\prime(\varphi,\psi)~.
\end{flalign}
This is, however, inconsistent with the canonical commutation relations, since for a nontrivial $\widetilde{\mathbf{S}}_\pm$
\begin{flalign}
 \widetilde{\Omega}_{\star 2}(\varphi,\psi) - \widetilde{\Omega}_{\star 2}(\psi,\varphi) =\widetilde{\omega}_\star(\varphi,\psi)
\neq \omega(\varphi,\psi) = \Omega_2^\prime(\varphi,\psi) - \Omega_2^\prime(\psi,\varphi)~.
\end{flalign}
This proves that the deformation can lead to new physical effects, even though
the underlying algebras of observables are isomorphic, i.e.~mathematically equivalent.


\chapter{\label{chap:qftapp}Applications}
We now study explicit examples of deformed classical and quantum field theories.
This chapter is mainly based on the articles \cite{Schenkel:2010sc,Schenkel:2010zi,Schenkel:2010jr}.


\section{\label{sec:exampleswaveop}Examples of deformed wave operators}
For discussing explicit examples of deformed wave operators it is convenient to
rewrite the deformed action functional (\ref{eqn:defaction}) for the real scalar field $\Phi$
in terms of local bases. 

Firstly, let us use a generic local coordinate basis $\lbrace \partial_\mu\in \Xi[[\lambda]]:\mu=1,\dots,N\rbrace$
for the vector fields $\Xi[[\lambda]]$. We define the $\star$-dual basis of one-forms 
$\lbrace dx_\star^\mu\in\Omega^1[[\lambda]]:\mu=1,\dots,N\rbrace$ by
$\pair{\partial_\mu}{dx_\star^\nu}_\star = \delta_\mu^\nu$. As a consequence of $\mathcal{F}=1\otimes 1+\mathcal{O}(\lambda)$
we obtain $dx_\star^\mu = dx^\mu +\mathcal{O}(\lambda)$, for all $\mu$, where $\lbrace dx^\mu\rbrace $ is the
undeformed dual basis of $\lbrace \partial_\mu\rbrace$ defined by $\pair{\partial_\mu}{dx^\nu}=\delta_\mu^\nu$.
The differential $\dd$ acting on functions corresponds to the undeformed one.
Thus, we have in the undeformed basis $\dd \varphi = dx^\mu\,\partial_\mu\varphi$, for all $\varphi\in C^\infty(\MM)[[\lambda]]$.
Since $\lbrace dx_\star^\mu\rbrace$ is also a basis of $\Omega^1[[\lambda]]$ we can find for all $\varphi$ 
coefficient functions $\partial_{\star\mu}\varphi\in C^\infty(\MM)[[\lambda]]$, such that 
$\dd \varphi = dx^\mu\,\partial_\mu \varphi= dx_\star^\mu\star\partial_{\star\mu}\varphi$.
Evaluating the $\star$-products and $dx_\star^\mu$ on the right hand side of this expression we
can construct the deformed partial derivatives $\partial_{\star\mu}$ order by order in $\lambda$.
The inverse metric field can be expressed in the local basis as
$g^{-1}_\star = \partial_\mu\otimes_{A_\star} g_\star^{\mu\nu}\star\partial_\nu$.
The involution $\ast$ on $\Xi[[\lambda]]$ satisfies $\partial_\mu^\ast =-\partial_\mu$ \cite{Aschieri:2005zs}.
From this, the reality of the twist and the hermiticity of $g_\star^{-1}$
we obtain for the coefficient functions
\begin{flalign}
 \bigl(g_\star^{\mu\nu}\bigr)^\ast =  g_\star^{\nu\mu}~.
\end{flalign}
Furthermore, for real functions $\varphi$ we have 
\begin{flalign}
\dd\varphi = -(\dd\varphi)^\ast = 
-(dx_\star^\mu\star\partial_{\star\mu}\varphi)^\ast 
= -\bigl(\partial_{\star\mu}\varphi\bigr)^\ast\star dx_\star^{\mu\ast}~.
\end{flalign}
With this we can express the kinetic term in the deformed action (\ref{eqn:defaction})
in terms of the local coordinate basis
\begin{flalign}
 \nn\pair{\pair{\dd\Phi}{g_\star^{-1}}_\star}{\dd\Phi}_\star &= 
-\bigl(\partial_{\star\mu}\Phi\bigr)^\ast\star \pair{dx_\star^{\mu\ast}}{\partial_\nu}_\star \star g_\star^{\nu\rho}
\star\pair{\partial_\rho}{dx_\star^\sigma}_\star\star\partial_{\star\sigma}\Phi\\
\nn &=-\bigl(\partial_{\star\mu}\Phi\bigr)^\ast\star\pair{\partial_\nu^\ast}{dx_\star^\mu}_\star^\ast\star g_\star^{\nu\rho}\star
\partial_{\star\rho}\Phi\\
&\hspace{-4.5mm}\stackrel{\partial_\nu^\ast=-\partial_\nu}{=} \bigl(\partial_{\star\mu}\Phi\bigr)^\ast \star g_\star^{\mu\rho}\star
\partial_{\star\rho}\Phi~.
\end{flalign}
The deformed action (\ref{eqn:defaction}) then reads  in the basis
\begin{flalign}
\label{eqn:defactioncobas}
 S_\star[\Phi] = -\frac{1}{2}\int\limits_\MM \Bigl(
\bigl(\partial_{\star\mu}\Phi\bigr)^\ast \star g_\star^{\mu\nu}\star
\partial_{\star\nu}\Phi + M^2\,\Phi\star\Phi
\Bigr)\star\vols~,
\end{flalign}
where we could also express the volume form in terms of the basis,  
$\vols = \gamma_\star\star dx_\star^1\wedge_\star\cdots \wedge_\star dx_\star^N$.
By classical correspondence we know that $\gamma_\star = \sqrt{-\text{det}(g_{\mu\nu})} +\mathcal{O}(\lambda)$.
The equation of motion in the basis can be derived from (\ref{eqn:defactioncobas}) order by order in $\lambda$ via
multiple integrations by parts.
Even though the deformed action (\ref{eqn:defactioncobas}) looks pretty familiar and
simple, there are two issues which in general complicate the application of this formula:
One requires the inverse metric field $g_\star^{\mu\nu}$ in the $\star$-basis and one has to
determine the deformed partial derivatives $\partial_{\star\mu}$, which are higher differential operators,
 to all orders in $\lambda$ recursively.

As we have observed in Chapter \ref{chap:ncgsol}, Section \ref{sec:ncgbasis}, 
the formulae of noncommutative differential geometry simplify when we use a nice basis.
As a reminder, a nice basis $\lbrace e_a\in \Xi[[\lambda]]:a=1,\dots,N \rbrace$
satisfies the conditions $[e_a,e_b]=0$ and $[\mathfrak{F},e_a]=\lbrace 0\rbrace$, for all $a,b$,
where $\mathfrak{F}\subseteq \Xi$ is the Lie subalgebra generating the twist, i.e.~$\mathcal{F}\in 
(U\mathfrak{F}\otimes U\mathfrak{F})[[\lambda]]$. For a large class of abelian twists, the nonexotic ones, such a basis
exists for opens $U\subseteq\MM$ around almost all points $p\in\MM$ \cite{Aschieri:2009qh}.
Let us now derive an expression for the deformed action (\ref{eqn:defaction})
in the nice basis $\lbrace e_a \rbrace$ and its dual $\lbrace \theta^a\rbrace$, which is also 
$\mathfrak{F}$-invariant $\mathcal{L}_\mathfrak{F}(\theta^a)=\lbrace 0\rbrace$.
For the inverse metric field we can write $g_\star^{-1} = e_a\otimes_{A_\star} g_\star^{ab}\star e_b = e_a\otimes_A g_\star^{ab}\,e_b $.
The differential of a function $\varphi$ reads $\dd\varphi=\theta^a\star e_a(\varphi) = \theta^a\,e_a(\varphi)$,
where $e_a(\cdot)$ denotes the vector field action (Lie derivative) of $e_a$ on functions.
The basis expression for the kinetic term is
\begin{flalign}
 \pair{\pair{\dd\Phi}{g_\star^{-1}}_\star}{\dd\Phi}_\star = 
e_a(\Phi)\star \pair{\theta^a}{e_b}_\star\star g_\star^{bc}\star \pair{e_c}{\theta^d}_\star \star e_d(\Phi) = 
e_a(\Phi)\star g_\star^{ab}\star e_b(\Phi)~.
\end{flalign}
Plugging this into the action yields
\begin{flalign}
 \label{eqn:defactionnicebas}
 S_\star[\Phi] = -\frac{1}{2}\int\limits_\MM \Bigl(
\bigl(e_a(\Phi) \star g_\star^{ab}\star e_b(\Phi) + M^2\,\Phi\star\Phi
\Bigr)\star\gamma_\star\star \cnt~,
\end{flalign}
where $\cnt = \theta^1\wedge_\star \cdots\wedge_\star\theta^N = \theta^1\wedge \cdots\wedge\theta^N $
is the basis top-form satisfying $\mathcal{L}_\mathfrak{F}(\cnt)=\lbrace 0\rbrace$ and $\gamma_\star$ is the volume factor
determined by $\vols = \gamma_\star \star \cnt = \gamma_\star\,\cnt$.
The obvious advantage of (\ref{eqn:defactionnicebas}) compared to (\ref{eqn:defactioncobas})
is that no higher differential operators $\partial_{\star\mu}$ appear.
Using graded cyclicity (\ref{eqn:gradedcyc}), integration by parts, $\mathfrak{F}$-invariance
of $\cnt$ and $\mathcal{L}_{e_a}(\cnt)=0$, for all $a$,
we can easily calculate the equation of motion operator by varying (\ref{eqn:defactionnicebas}) and obtain
\begin{flalign}
 P_\star^\top(\Phi) = \frac{1}{2}\Bigl( e_a\bigl(g_\star^{ab} \star e_b(\Phi)\star\gamma_\star\bigr) 
+ e_a\bigl(\gamma_\star\star e_b(\Phi)\star g_\star^{ba}\bigr) - M^2 \bigl(\Phi\star\gamma_\star + \gamma_\star\star\Phi\bigr) 
 \Bigr)\star\cnt~.
\end{flalign}
The scalar valued equation of motion operator is given by
\begin{flalign}
\label{eqn:eombas}
 P_\star(\Phi) = \frac{1}{2}\Bigl( e_a\bigl(g_\star^{ab} \star e_b(\Phi)\star\gamma_\star\bigr) 
+ e_a\bigl(\gamma_\star\star e_b(\Phi)\star g_\star^{ba}\bigr) - M^2 \bigl(\Phi\star\gamma_\star + 
\gamma_\star\star\Phi\bigr)  \Bigr)\star\gamma_\star^{-1_\star}~,
\end{flalign}
where $\gamma_\star^{-1_\star}$ is the $\star$-inverse of $\gamma_\star$, i.e.~$\gamma_\star\star\gamma_\star^{-1_\star}=
\gamma_\star^{-1_\star}\star\gamma_\star =1$.
For completeness, the real equation of motion operator $\widetilde{P}_\star$ reads
\begin{flalign}
\label{eqn:eombastilde}
 \widetilde{P}_\star(\Phi) = \frac{1}{2}\Bigl( e_a\bigl(g_\star^{ab} \star e_b(\Phi)\star\gamma_\star\bigr) 
+ e_a\bigl(\gamma_\star\star e_b(\Phi)\star g_\star^{ba}\bigr) - M^2 \bigl(\Phi\star\gamma_\star + 
\gamma_\star\star\Phi\bigr)  \Bigr)\, \gamma^{-1}~,
\end{flalign}
where $\gamma^{-1}$ is the inverse volume factor of the undeformed volume form $\vol = \gamma \,\cnt$.

We are now going to derive explicit examples of these operators, where the twist is given by an abelian twist
\begin{flalign}
\label{eqn:abeliantwistqft}
 \mathcal{F} = \exp\left(-\frac{i\lambda}{2} \Theta^{\alpha\beta} X_\alpha\otimes X_\beta\right)~,
\end{flalign}
with $[X_\alpha,X_\beta]=0$ for all $\alpha,\beta$ and $\Theta$ canonical.

\subsection{Noncommutative Minkowski spacetimes:}
The simplest example is to consider the $4$-dimensional Minkowski spacetime, i.e.~$\MM=\bbR^4$ with metric
$g_\star^{-1}=g^{-1}=\partial_\mu\otimes_A g^{\mu\nu}\partial_\nu$, where in Cartesian coordinates
$g^{\mu\nu}=\text{diag}(-1,1,1,1)^{\mu\nu}$.
The volume form is taken to be the classical one $\vols = \vol = dt\wedge dx^1\wedge dx^2\wedge dx^3$.

If we deform this spacetime with the Moyal-Weyl twist, the nice basis and coordinate basis chosen above
coincide and we find for the equation of motion operators (\ref{eqn:eombas}) and (\ref{eqn:eombastilde})
\begin{flalign}
 P_\star(\Phi) = \widetilde{P}_\star(\Phi) = P(\Phi) = g^{\mu\nu}\partial_\mu\partial_\nu\Phi -M^2\,\Phi~.
\end{flalign}
Thus, the wave operator for the Moyal-Weyl deformed Minkowski spacetime is undeformed, which is
a well-known result.

Consider now the abelian twist (\ref{eqn:abeliantwistqft}) with $X_1=r\partial_r$ and $X_2=\partial_t$,
where $r$ is the spatial radius coordinate. 
The resulting commutation relations are
$[t\stackrel{\star}{,}r] = -i\,\lambda\,r$, i.e.~we are dealing with a $\kappa$-deformed Minkowski spacetime.
Note that since $X_2$ is a Killing vector field, the undeformed metric field solves 
the noncommutative vacuum Einstein equations exactly.
A nice basis for this deformation is given by
\begin{flalign}
\label{eqn:nicebasqftmink}
 e_1=\partial_t~,\quad e_2=r\partial_r~,\quad e_3=\partial_\zeta~,\quad e_4=\partial_\phi~,
\end{flalign}
where $\zeta$ is the polar and $\phi$ is the azimuthal angle. The dual basis is
\begin{flalign}
 \theta^1=dt~,\quad \theta^2 = \frac{dr}{r} ~,\quad \theta^3 = d\zeta~,\quad\theta^4 = d\phi~.
\end{flalign}
The inverse metric field in this basis reads $g_\star^{ab} = \text{diag}\bigl(-1,r^{-2},r^{-2},(r\,\sin\zeta)^{-2}\bigr)^{ab}$
and the volume factor is $\gamma_\star= \gamma = r^3\,\sin\zeta$. Note that $\gamma_\star^{-1_\star} = \gamma^{-1} = 
r^{-3}\,(\sin\zeta)^{-1}$.
Evaluating all $\star$-products in the equation of motion operator (\ref{eqn:eombas})
we obtain
\begin{flalign}
 P_\star(\Phi) = -\frac{1}{2} \bigl(1+e^{i3\lambda \partial_t}\bigr)\bigl(\partial_t^2 + M^2\bigr)\Phi +\frac{1}{2}\bigl(
e^{-i\lambda\partial_t} + e^{i4\lambda\partial_t}\bigr)\bigtriangleup\Phi~,
\end{flalign}
where $\bigtriangleup$ is the spatial Laplacian. The tilded operator reads
\begin{flalign}
 \widetilde{P}_\star(\Phi) = -\cosh\left(\frac{3\lambda}{2}\,i\partial_t\right)\bigl(\partial_t^2 + M^2\bigr) \Phi
+\cosh\left(\frac{5\lambda}{2}\,i\partial_t\right)\bigtriangleup\Phi~.
\end{flalign}
For this derivation the following identities are useful
\begin{flalign}
\label{eqn:relationsqftapp}
 (r^n)\star h = r^n\, e^{\frac{in\lambda}{2}\partial_t}h~,\quad h\star(r^n) = r^n\,e^{-\frac{in\lambda}{2}\partial_t}h~,
\end{flalign}
for all $n\in \mathbb{Z}$ and $h\in C^\infty(\MM)[[\lambda]]$.

Thus, we obtain a nontrivial wave operator for the scalar field on the deformed Minkowski spacetime.
The second order corrections to the Green's operators for this model are calculated in the Appendix 
\ref{eqn:lambda2green}.

\subsection{Noncommutative de Sitter spacetimes:}
We consider the flat slicing of the $4$-dimensional de Sitter spacetime, i.e.~$\MM=\bbR^4$ 
with metric $g_\star^{-1} = g^{-1} =\partial_\mu\otimes_A g^{\mu\nu}\,\partial_\nu$,
where in Cartesian coordinates the metric coefficient functions read 
$g^{\mu\nu} = \text{diag}\bigl(-1,e^{-2Ht},e^{-2Ht},e^{-2Ht}\bigr)^{\mu\nu}$.
Here $H>0$ denotes the Hubble parameter, i.e.~the expansion rate of the universe.
The volume form is taken to be the classical one $\vols=\vol = e^{3Ht}\,dt\wedge dx^1\wedge dx^2\wedge dx^3$.

We deform this spacetime by an abelian twist with $X_1 = r\partial_r$ and $X_2=\partial_t$,
where again $r$ denotes the spatial radius coordinate. The commutation relations are
$[t\stackrel{\star}{,}r]=-i\,\lambda\, r$. As we have discussed in Chapter \ref{chap:ncgsol}
the noncommutative Einstein equations in presence of a cosmological constant are solved exactly for this model.
A nice basis for this deformation is given by (\ref{eqn:nicebasqftmink}). In this basis
the inverse metric field reads $g_\star^{ab}=\text{diag}\bigl(-1,e^{-2Ht}r^{-2},e^{-2Ht}r^{-2},
e^{-2Ht}(r\,\sin\zeta)^{-2}\bigr)^{ab}$ and the volume factor is
$\gamma_\star=\gamma=  e^{3Ht}\,r^3\,\sin\zeta$. Note that $\gamma_\star^{-1_\star} = \gamma^{-1} = 
e^{-3Ht}\, r^{-3}\,(\sin\zeta)^{-1}$.
Evaluating all $\star$-products in the equation of motion operator (\ref{eqn:eombas})
we obtain
\begin{flalign}
 P_\star(\Phi) = -\frac{1}{2}\bigl(1+e^{i3\lambda \mathcal{D}} \bigr)\bigl( \partial_t^2 + 3H\partial_t +M^2 \bigr)\Phi
 +\frac{1}{2}\bigl(e^{-i\lambda \mathcal{D}} + e^{i4\lambda\mathcal{D}}\bigr)e^{-2Ht}\bigtriangleup\Phi~,
\end{flalign}
where $\mathcal{D}:=\partial_t -H r\partial_r$. The tilded operator reads
\begin{flalign}
\label{eqn:isotropicdesitter}
 \widetilde{P}_\star(\Phi) = -\cosh\left(\frac{3\lambda}{2}\,i\mathcal{D}\right)\bigl( \partial_t^2 + 3H\partial_t +M^2 \bigr)\Phi
 +\cosh\left(\frac{5\lambda}{2}\,i\mathcal{D}\right)e^{-2Ht}\bigtriangleup\Phi~.
\end{flalign}
In this derivation we have used the identities
\begin{flalign}
 \bigl(e^{nHt}r^n\bigr)\star h = e^{nHt}r^n\,e^{\frac{in\lambda}{2}\mathcal{D}}h~,\quad 
h\star \bigl(e^{nHt}r^n\bigr) = e^{nHt}r^n\,e^{-\frac{in\lambda}{2}\mathcal{D}}h~,
\end{flalign}
for all $n\in\mathbb{Z}$ and $h\in C^\infty(\MM)[[\lambda]]$.

As a next deformation we consider an abelian twist with $X_1=\partial_\phi$ and $X_2=\partial_t$,
leading to a nontrivial angle-time commutator. This model has been discussed in Chapter \ref{chap:ncgsol}.
Since $X_1$ is a Killing vector field, the undeformed de Sitter metric solves the noncommutative Einstein
equations in presence of a cosmological constant exactly. As a nice basis for this deformation we can simply use
the spherical coordinate basis
\begin{flalign}
 e_1 =\partial_t~,\quad e_2=\partial_r~,\quad e_3 = \partial_\zeta~,\quad e_4 = \partial_\phi~,
\end{flalign}
and its dual
\begin{flalign}
 \theta^1=dt~,\quad \theta^2 = dr~,\quad \theta^3 = d\zeta~,\quad \theta^4=d\phi~.
\end{flalign}
The inverse metric in this basis is given by $g_\star^{ab} = \text{diag}\bigl(-1,e^{-2Ht},e^{-2Ht}r^{-2},e^{-2Ht}(r\,\sin\zeta)^{-2}
\bigr)^{ab}$ and the volume factor reads $\gamma_\star =\gamma= e^{3Ht}\,r^2\,\sin\zeta$.
We find $\gamma_\star^{-1_\star} = \gamma^{-1} = e^{-3Ht}\, r^{-2}\,(\sin\zeta)^{-1}$.
Evaluating all $\star$-products in the equation of motion operator (\ref{eqn:eombas})
we obtain
\begin{flalign}
 P_\star(\Phi) =  -\frac{1}{2}\bigl(1+e^{-i3\lambda H\partial_\phi} \bigr)\bigl( \partial_t^2 + 3H\partial_t +M^2 \bigr)\Phi
 +\frac{1}{2}\bigl(e^{i\lambda H\partial_\phi} + e^{-i4\lambda H\partial_\phi}\bigr)e^{-2Ht}\bigtriangleup\Phi~.
\end{flalign}
The tilded operator reads
\begin{flalign}
\label{eqn:frwextimeangle}
 \widetilde{P}_\star(\Phi) = -\cosh\left(\frac{3\lambda H}{2}\,i\partial_\phi\right)\bigl( \partial_t^2 + 3H\partial_t +M^2 \bigr)\Phi
 +\cosh\left(\frac{5\lambda H}{2}\,i\partial_\phi\right)e^{-2Ht}\bigtriangleup\Phi~.
\end{flalign}
Note that the deformation is governed by the product $\lambda H$. Thus,
in a slowly expanding universe noncommutativity will be small and
in a rapidly expanding one large. This is a physically very interesting feature which
might be able to explain why today we have not yet observed effects of noncommutative geometry.

As a last model we consider an abelian twist with $X_1=\partial_\phi$ and $X_2=r\partial_r$,
leading to a nontrivial angle-radius commutator. This model has been discussed in Chapter \ref{chap:ncgsol}.
Since $X_1$ is a Killing vector field, the undeformed de Sitter metric solves the noncommutative Einstein
equations in presence of a cosmological constant exactly. A nice basis for this model is given by 
(\ref{eqn:nicebasqftmink}). The inverse metric in this basis reads 
$g_\star^{ab}=\text{diag}\bigl(-1,e^{-2Ht}r^{-2},e^{-2Ht}r^{-2},
e^{-2Ht}(r\,\sin\zeta)^{-2}\bigr)^{ab}$ and the volume factor is
$\gamma_\star= \gamma=  e^{3Ht}\,r^3\,\sin\zeta$. Note that $\gamma_\star^{-1_\star} = \gamma^{-1} = 
e^{-3Ht}\, r^{-3}\,(\sin\zeta)^{-1}$.
Evaluating all $\star$-products in the equation of motion operator (\ref{eqn:eombas})
we obtain
\begin{flalign}
 P_\star(\Phi) =  -\frac{1}{2}\bigl(1+e^{-i3\lambda\partial_\phi} \bigr)\bigl( \partial_t^2 + 3H\partial_t +M^2 \bigr)\Phi
 +\frac{1}{2}\bigl(e^{i\lambda\partial_\phi} + e^{-i4\lambda\partial_\phi}\bigr)e^{-2Ht}\bigtriangleup\Phi~.
\end{flalign}
The tilded operator reads
\begin{flalign}
 \widetilde{P}_\star(\Phi) = -\cosh\left(\frac{3\lambda}{2}\,i\partial_\phi\right)\bigl( \partial_t^2 + 3H\partial_t +M^2 \bigr)\Phi
 +\cosh\left(\frac{5\lambda}{2}\,i\partial_\phi\right)e^{-2Ht}\bigtriangleup\Phi~.
\end{flalign}
This operator has a similar structure as (\ref{eqn:frwextimeangle}), with the difference that
the deformation is governed by a dimensionless $\lambda$, while the latter one is governed by $\lambda H$.

\subsection{Noncommutative Schwarzschild spacetime:}
We consider the exterior of the Schwarzschild black hole, i.e.~$\MM\subset\bbR^4$ 
with metric $g_\star^{-1} = g^{-1} =\partial_\mu\otimes_A g^{\mu\nu}\,\partial_\nu$,
where in spherical coordinates the metric coefficient functions read 
$g^{\mu\nu} = \text{diag}\bigl(-Q(r)^{-1},Q(r),r^{-2},(r\,\sin\zeta)^{-2}\bigr)^{\mu\nu}$
with $Q(r) = 1-\frac{r_s}{r}$.
The volume form is taken to be the classical one $\vols=\vol = r^2\,\sin\zeta\,dt\wedge dr\wedge d\zeta\wedge d\phi$.

We deform the black hole with an abelian twist constructed by $X_1=\partial_t$ and $X_2=r\partial_r$,
which leads to a time-radius noncommutativity.
This model was discussed in Chapter \ref{chap:ncgsol}, where it was found that
it is invariant under all classical black hole symmetries and that it solves the noncommutative Einstein equations
exactly. A nice basis for this model is given by (\ref{eqn:nicebasqftmink}). The inverse metric field in this
basis reads $g_\star^{ab} = \text{diag}\bigl(-Q(r)^{-1},Q(r)\, r^{-2},r^{-2},(r\sin\zeta)^{-2}\bigr)^{ab}$
and the volume factor is $\gamma_\star =\gamma= r^3\,\sin\zeta$. Note that $\gamma_\star^{-1_\star} = \gamma^{-1} = 
r^{-3}\,(\sin\zeta)^{-1}$.
Evaluating the equation of motion operator (\ref{eqn:eombas}) using (\ref{eqn:relationsqftapp})
we find
\begin{flalign}
 \nn P_\star(\Phi) &= -\frac{1}{2} \bigl(Q^{-1}\star \partial_t^2\Phi + 
e^{-i3\lambda\partial_t}\partial_t^2\Phi \star Q^{-1}\bigr) - \frac{M^2}{2}\bigl( 1+e^{-i3\lambda\partial_t}  \bigr)\Phi\\
&+\frac{1}{2r^2}\,\partial_r\Bigl[r^2\bigl(Q\star e^{i\lambda\partial_t}\partial_r\Phi + e^{-i4\lambda\partial_t}\partial_r\Phi\star Q\bigr)\Bigr]
+\frac{1}{2r^2} \bigl(e^{i\lambda\partial_t} + e^{-i4\lambda\partial_t}\bigr)\bigtriangleup_{S^2}\Phi~,
\end{flalign}
where $\bigtriangleup_{S^2} = (\sin\zeta)^{-1}\partial_\zeta\sin\zeta\partial_\zeta + (\sin\zeta)^{-2}\partial_\phi^2$
is the Laplacian on the unit two-sphere.
The tilded operator reads
\begin{flalign}
 \nn \widetilde{P}_\star(\Phi) &= -\frac{1}{2} \bigl(Q^{-1}\star e^{\frac{i3\lambda}{2}\partial_t}\partial_t^2\Phi + 
e^{-\frac{i3\lambda}{2}\partial_t}\partial_t^2\Phi \star Q^{-1}\bigr) - M^2\,\cosh\left(\frac{3\lambda}{2}\,i\partial_t\right)\Phi\\
&+\frac{1}{2r^2}\,\partial_r\Bigl[r^2\bigl(Q\star e^{\frac{i5\lambda}{2}\partial_t}\partial_r\Phi + e^{-\frac{i5\lambda}{2}\partial_t}\partial_r\Phi\star Q\bigr)\Bigr]
+\frac{1}{r^2} \cosh\left(\frac{5\lambda}{2}\,i\partial_t\right)\bigtriangleup_{S^2}\Phi~.
\end{flalign}
In addition to the exponentials of time derivatives, which we also found in the previous examples,
there are $\star$-products involving either $Q(r)$ or $Q^{-1}(r)$. While the former ones are
easily evaluated, since $Q(r)$ is a sum of eigenfunctions of the dilation operator $r\partial_r$, 
this is not the case for the latter. Nevertheless, one can evaluate these products up to the desired order
in the deformation parameter $\lambda$ by using the explicit expression of the $\star$-product and calculating
the scale derivatives $r\partial_r$ of $Q^{-1}$.

\subsection{Noncommutative anti-de Sitter spacetime:}
We consider the Poincar{\'e} patch of the $5$-dimensional anti-de Sitter spacetime, i.e.~$\MM=\bbR^5$
with metric $g_\star^{-1} = g^{-1} =\partial_M\otimes_A g^{MN}\,\partial_N = \partial_\mu\otimes_A e^{2ky}g^{\mu\nu}\,\partial_\nu
+\partial_y\otimes_A \partial_y$,
where $g^{\mu\nu} = \text{diag}\bigl(-1,1,1,1\bigr)^{\mu\nu}$ is the $4$-dimensional
Minkowski metric and $k$ is related to the curvature.
The volume form is taken to be the classical one $\vols=\vol = e^{-4ky}\,dt\wedge dx^1\wedge dx^2\wedge dx^3\wedge dy$.

We deform this spacetime by an abelian twist generated by $2n$ vector fields $X_\alpha$ defined as follows
\begin{flalign}
 X_{2a-1} = T_{2a-1}^\mu\partial_\mu~,\quad X_{2a} = \vartheta(y)\,T_{2a}^\mu\partial_\mu~,\qquad \text{for } a=1,\dots,n~,
\end{flalign}
where $T_\alpha^\mu$ are constant and real matrices and $\vartheta(y)$ is a smooth and real function.
Note that $X_{2a-1}$ are Killing vector fields for all $a$, and therewith 
the undeformed anti-de Sitter metric solves the noncommutative Einstein
equations in presence of a cosmological constant exactly. The commutation relations for this model read
\begin{flalign}
 \starcom{x^\mu}{x^\nu} = i\lambda\,\vartheta(y)\,\Theta^{\alpha\beta}\,T_\alpha^\mu\, T_\beta^\nu~,\quad \starcom{x^\mu}{y}=0~.
\end{flalign}
Thus, we have a canonical noncommutativity on the $\bbR^4$-slices of constant $y$, which scales along $y$ with $\vartheta(y)$.
Since in particle physics the coordinate $y$ is interpreted as an extradimension, while $x^\mu$
are coordinates of our $4$-dimensional world, this means that the strength
of noncommutativity depends on our position in the higher dimensional space.
For first phenomenological investigations in models of this kind see \cite{Ohl:2010bh}.

Instead of constructing a nice basis for this model, it is more convenient to evaluate the deformed action functional
in the coordinate basis (\ref{eqn:defactioncobas}). The metric in the $\star$-product basis
$g_\star^{-1}=\partial_M\otimes_{A_\star} g_\star^{MN}\star\partial_N$ is found to agree with $g^{MN}$
and we have the useful relation $\mathcal{L}_{X_\alpha}(\vol) =0$ for all $\alpha$.
Using this and graded cyclicity of the integral (\ref{eqn:gradedcyc}) we obtain for the action
\begin{flalign}
 S_\star[\Phi] = -\frac{1}{2}\int\limits_\MM \Bigl(\bigl(\partial_{\star\mu}\Phi\bigr)^\ast\,e^{2ky}\,g^{\mu\nu}\,\partial_{\star\nu}\Phi
+ \bigl(\partial_{\star y}\Phi\bigr)^\ast\,\partial_{\star y}\Phi\Bigr)\,e^{-4 ky} dt\wedge dx^1\wedge dx^2\wedge dx^3\wedge dy~.
\end{flalign}
We have set the ``bulk mass'' $M^2=0$ and obtain effective mass terms later via the Kaluza-Klein reduction.
It remains to calculate the deformed partial derivatives $\partial_{\star M}$ by comparing
both sides of $dx^M\,\partial_M\Phi = dx^M\star\partial_{\star M}\Phi$. We find the exact expression
\begin{flalign}
 \partial_{\star\mu} = \partial_\mu~,\quad \partial_{\star y} = \partial_y + \frac{i\lambda}{2}\vartheta^\prime(y)
\sum\limits_{a=1}^n T_{2a-1}^\mu T_{2a}^\nu\partial_\mu\partial_\nu =: \partial_y +\frac{i\lambda}{2}\vartheta^\prime(y)\,
\mathbb{T}~,
\end{flalign}
where $\vartheta^\prime$ denotes the derivative of $\vartheta$. Note that the deformed
partial derivative $\partial_{\star y}$ is as expected a higher differential operator.
Inserting this into the action yields
\begin{flalign}
 S_\star[\Phi] = -\frac{1}{2}\int\limits_\MM \Bigl(\partial_\mu\Phi\,e^{2ky}\,g^{\mu\nu}\,\partial_{\nu}\Phi
+ \partial_y\Phi\,\partial_{y}\Phi + \frac{\lambda^2}{4} \vartheta^\prime(y)^2 ~\mathbb{T}\Phi\,\mathbb{T}\Phi\Bigr)\,\vol~.
\end{flalign}
Note again that this is an exact expression valid to all orders in $\lambda$.

To obtain an effective $4$-dimensional theory we introduce two $4$-dimensional branes 
at the positions $y=0$ and $y=R\pi$, with $R$ denoting the ``radius'' of the extradimension,
and restrict the spacetime to the interval $y\in[0,R\pi]$. This is the so-called Randall-Sundrum model \cite{Randall:1999ee}.
The noncommutative Einstein equations are also solved exactly in presence of these branes as discussed in
\cite{Ohl:2010bh}.
We make the Kaluza-Klein ansatz for the scalar field 
\begin{flalign}
 \Phi(x^\mu,y) = \sum\limits_{i=0}^\infty \Phi_i(x^\mu)\,t_{i}(y)~,
\end{flalign}
where $\Phi_i$ are the effective $4$-dimensional fields and $\lbrace t_i\rbrace$ is a complete set of 
eigenfunctions of the mass operator $\hat{O}:=-e^{2ky}\partial_ye^{-4ky}\partial_y$ 
satisfying Neumann or Dirichlet boundary conditions. The eigenfunctions $\lbrace t_i\rbrace$ are
orthonormal with respect to the standard scalar product $\int_0^{R\pi}dy \,e^{-2ky}\, t_i t_j =\delta_{ij}$. 
We obtain the Kaluza-Klein reduced action
\begin{flalign}
\label{eqn:kkaction}
 S_\star[\lbrace\Phi_i\rbrace] = -\frac{1}{2}\sum\limits_{i=0}^\infty~ \int\limits_{\bbR^4} \Bigl(
\partial_\mu\Phi_i g^{\mu\nu}\partial_\nu\Phi_i + M_i^2\,\Phi_i^2 + \lambda^2 \sum\limits_{j=0}^\infty C_{ij}~ \mathbb{T}\Phi_i\,
\mathbb{T}\Phi_j \Bigr)\,\mathrm{vol}_4~,
\end{flalign}
where $\mathrm{vol}_4 = dt\wedge dx^1\wedge dx^2\wedge dx^3$ and the masses $M_i^2$ and couplings $C_{ij}$
are given by
\begin{flalign}
 \hat{O}t_i =M_i^2\,t_i~,\quad C_{ij} = \int\limits_{0}^{R\pi} dy \,\frac{\vartheta^\prime(y)^2}{4} \,e^{-4ky}\, t_i(y)\, t_j(y)~.
\end{flalign}

We finish this section by an interesting observation, which was first made in \cite{Schenkel:2010zi}.
We can specialize our deformation to obtain a $z{=}2$-Ho{\v r}ava-Lifshitz propagator  \cite{Horava:2009uw} 
for the $4$-dimensional scalar fields $\Phi_i$, while not affecting local potentials. 
To this end, we set $n=3$ and choose $T_{2a-1}^\mu = T_{2a}^\mu=\delta_a^\mu$, resulting in the spatial Laplacian
$\mathbb{T}=\bigtriangleup$. Choosing $\vartheta(y)\sim e^{ky}$ such that $C_{ij}=C_i\delta_{ij}$
is diagonal, we obtain from the deformed action (\ref{eqn:kkaction}) propagator denominators of the form
\begin{flalign}
 E^2 - \mathbf{k}^2 -\lambda^2 C_i \mathbf{k}^4 -M_i^2~,
\end{flalign}
for all individual Kaluza-Klein modes. Here $E$ denotes the energy and $\mathbf{k}$ the spatial momentum.
It is known that propagators of this form improve the high energy behavior of quantum field theories,
without introducing ghost states, see e.g.~\cite{Horava:2009uw} and references therein.
See also Section \ref{sec:z=2qft} of this chapter.


\section{\label{sec:homothetic}Homothetic Killing deformations of FRW universes}
We consider deformations by homothetic Killing vector fields, leading to exactly treatable convergent deformations
of scalar quantum field theories. For this, we first provide a short review on homothetic Killing vector fields.

Let $(\MM, g)$ be an $N$-dimensional smooth Lorentzian manifold and let $\Xi$
be the smooth and complex vector fields on $\MM$.
\begin{defi}
A vector field $v\in\Xi$ is called a {\it Killing vector field}, if $\mathcal{L}_v (g)=0$.
It is called a {\it homothetic Killing vector field}, if $\mathcal{L}_v(g)= c_v\, g$
with $c_v\in\bbC$. $v$ is called {\it proper}, if $c_v\neq 0$.
\end{defi}
Obviously, each Killing vector field is also a homothetic Killing vector field with $c_v=0$. 
A homothetic Killing vector field is a special case of a conformal Killing vector field $v\in\Xi$,
 satisfying $\mathcal{L}_v(g)=h\,g$ with $h\in C^\infty(\MM)$. We do not discuss general conformal Killing vector fields
 in the following.

We remind the reader of the following standard result.
\begin{propo}
The homothetic Killing vector fields form a Lie subalgebra $(\mathfrak{H},\com{~}{~})\subseteq (\Xi,\com{~}{~})$
of the Lie algebra of vector fields on $\mathcal{M}$.
The Killing vector fields form a Lie subalgebra $(\mathfrak{K},\com{~}{~})\subseteq (\mathfrak{H},\com{~}{~})$ and
the following inclusion holds true
\begin{flalign}
\com{\mathfrak{H}}{\mathfrak{H}}\subseteq \mathfrak{K}~.
\end{flalign}
\end{propo}
\begin{proof}
Let $v,w\in\mathfrak{H}$. Then $\mathcal{L}_{\beta\,v+\gamma\,w}(g)=\beta\mathcal{L}_v(g)+\gamma\mathcal{L}_w(g)
=(\beta\,c_v +\gamma\,c_w)g$, for all $\beta,\gamma\in\bbC$. Thus, $\mathfrak{H}$ is a vector space over $\bbC$
and $\mathfrak{K}\subseteq\mathfrak{H}$ is a vector subspace.
Furthermore,
\begin{flalign}
\mathcal{L}_{\com{v}{w}}(g)=(\mathcal{L}_{v}\circ\mathcal{L}_w-\mathcal{L}_w\circ\mathcal{L}_v)(g)= (c_v\,c_w-c_w\,c_v)g=0~,
\end{flalign}
such that $[v,w]\in\mathfrak{K}\subseteq \mathfrak{H}$. For $v,w\in\mathfrak{K}$ we trivially find $\com{v}{w}\in\mathfrak{K}$.

\end{proof}
\noindent 
It can be shown that $\dim(\mathfrak{K})\leq\dim(\mathfrak{H})\leq\dim(\mathfrak{K})+1$.
 To prove this, assume
that there are two proper homothetic Killing vector fields $v,w\in\mathfrak{H}$, satisfying
$\mathcal{L}_v(g)=c_v\,g$ and $\mathcal{L}_w(g)=c_w\,g$ with $c_v,c_w\neq0$.
Then $u:=c_w\, v-c_v\,w$ is a Killing vector field, since
 $\mathcal{L}_u(g)=\mathcal{L}_{c_w\, v-c_v\,w}(g)=(c_w c_v-c_v c_w)g=0$, and
 $w=(c_w\,v-u)/c_v$ is a linear combination of a proper homothetic Killing vector field and a Killing vector field.

Let us provide some examples of manifolds allowing for proper homothetic Killing vector fields.
\begin{ex}
\label{ex:minkhomothetic}
Let $\mathcal{M}=\bbR^N$ and let $g=g_{\mu\nu}dx^\mu\otimes_A dx^\nu$,
where $x^\mu$ are global coordinates on $\bbR^N$ and $g_{\mu\nu}=\text{diag}(-1,1,\dots,1)_{\mu\nu}$.
It is well-known that $\mathfrak{K}$ is the Lie algebra of the Poincar{\'e} group $SO(N-1,1)\ltimes \bbR^N$.
A proper homothetic Killing vector field is given by the dilation $v=x^\mu\partial_\mu$,
satisfying $\mathcal{L}_v(g) = 2\,g$.
\end{ex}
\begin{ex}[\cite{Eardley:1973fm}]
\label{ex:frwhomothetic}
Let $\mathcal{M}=(0,\infty) \times \bbR^{N-1}$ and let $g=-dt\otimes_A dt +a(t)^2\,\delta_{ij}dx^i\otimes_A dx^j$,
where $t\in(0,\infty)$ is the cosmological time, $x^i,~i=1,\dots,\,N-1$, are comoving coordinates and $a(t)$ is the scale
factor of the universe. If we assume that $a(t)\propto t^p$, where $p\in\bbR$, we have a
proper homothetic Killing vector field
\begin{flalign}
v=t\partial_t +\left(1-t\,\frac{\dot a(t)}{a(t)}\right)x^i\partial_i=t\partial_t+(1-p)\,x^i\partial_i~,
\end{flalign}
satisfying $\mathcal{L}_v(g)= 2\,g$. These spacetimes are relevant in cosmology, since a perfect fluid with
equation of state $P=\omega\rho$, where $P$ is the pressure, $\rho$ is the energy density and $\omega\in\bbR$
is a parameter, leads to a scale factor $a(t)\propto t^{\frac{2}{3(\omega+1)}}$, i.e.~$p=\frac{2}{3(\omega+1)}$. 
\end{ex}
For more examples, including the Kasner spacetime and the  plane-wave spacetime, as well as a 
construction principle for spacetimes allowing for a proper homothetic Killing vector field see \cite{Eardley:1973fm}.

\subsection{Formal aspects of homothetic Killing deformations:}
Let $\mathcal{F}$ be an abelian twist generated by two real vector fields $X_1,X_2\in \Xi$.
We call the deformation a {\it homothetic Killing deformation by two vector fields}, if
without loss of generality $X_1\in\mathfrak{K}$ and $X_2\in\mathfrak{H}$.
We shall use the normalization $\mathcal{L}_{X_2}(g) = c\,g$, with $c\in\bbR$.
From Chapter \ref{chap:ncgsol} we know that the noncommutative Einstein equations
reduce to the undeformed ones, since $X_1$ is a Killing vector field. Thus, exact solutions
where the deformed and undeformed metric field coincide exist for these deformations.

Consider the deformed action for a real scalar field
\begin{flalign}
\label{eqn:defactionhom}
 S_\star[\Phi] :=-\frac{1}{2}\int\limits_\mathcal{M}\Bigl(\pair{\pair{\dd\Phi}{g_\star^{-1}}_\star}{\dd\Phi}_\star +
 \xi \,\Phi\star \mathfrak{R}\star \Phi\Bigr) \star \vol~,
\end{flalign}
where we have set the mass to zero $M^2=0$, but introduced a coupling to curvature $\xi\in\bbR$.
Since $X_1$ is Killing, the deformed and undeformed curvature of $g$ coincide.

We obtain for the deformed wave operators
\begin{propo}
 Consider a homothetic Killing deformation by two vector fields $X_1\in\mathfrak{K}$ and $X_2\in\mathfrak{H}$
of an $N$-dimensional smooth Lorentzian manifold $(\mathcal{M},g)$. 
Then $g_\star^{-1} = g^{-1}$, where $g^{-1}\in \Xi\otimes_A\Xi$ is the undeformed inverse metric field.
The equation of motion corresponding to the scalar field action (\ref{eqn:defactionhom}) is given by
(suppressing the symbol $\mathcal{L}$ for Lie derivatives)
\begin{flalign}
\label{eqn:eomoperatorhom}
 \widetilde{P}_\star(\Phi)= \cosh\left(\frac{\lambda c}{4}\left(N+2\right) iX_1\right)\,
\bigl(\square_g -\xi\, \mathfrak{R}\bigr)\Phi=0~,
\end{flalign}
where $\square_g$ is the undeformed d'Alembert operator.
\end{propo}
\begin{proof}
The $\star$-inverse metric $g_\star^{-1}\in\bigl(\Xi\otimes_A\Xi\bigr)[[\lambda]]$ of $g$ exists and is unique. 
We show that the undeformed inverse metric $g^{-1}\in\Xi\otimes_A\Xi$ defined by $\pair{g^{-1}}{\pair{g}{v}}=v$ 
and $\pair{g}{\pair{g^{-1}}{\omega}}=\omega$, for all $v\in\Xi$ and $\omega\in\Omega^1$,
 is equal to $g_\star^{-1}$.
For $g^{-1}$ one easily finds $\mathcal{L}_{X_1}(g^{-1})=0$ and $\mathcal{L}_{X_2}(g^{-1})=-c\,g^{-1}$.
Using the homothetic Killing property we obtain
\begin{flalign}
\label{eqn:hkprop1}
\pair{g}{v}_\star = \pair{g}{e^{-\frac{i\lambda c}{2}X_1}v}~,\quad \pair{g^{-1}}{\omega}_\star=\pair{g^{-1}}{e^{\frac{i\lambda c}{2}X_1}\omega}~,
\end{flalign}
for all $v\in\Xi[[\lambda]]$ and $\omega\in\Omega^1[[\lambda]]$. Thus,
\begin{flalign}
 \pair{g^{-1}}{\pair{g}{v}_\star}_\star = \pair{g^{-1}}{\pair{g}{v}}=v~,\quad
\pair{g}{\pair{g^{-1}}{\omega}_\star}_\star = \pair{g}{\pair{g^{-1}}{\omega}}=\omega~,
\end{flalign}
by using the invariance of $g$ and $g^{-1}$ under $X_1$.

\noindent For the metric volume form one finds that $\mathcal{L}_{X_1}(\vol)=0$ and $\mathcal{L}_{X_2}(\vol)=\frac{c N}{2}\vol$.
For the curvature we have $\mathcal{L}_{X_1}(\mathfrak{R})=0$ and $\mathcal{L}_{X_2}(\mathfrak{R})=-c\,\mathfrak{R}$ 
\cite{Aschieri:2009qh}.
Using this, (\ref{eqn:hkprop1}) and graded cyclicity in order to remove one $\star$ under the integral, we obtain
for the variation of the action (\ref{eqn:defactionhom}) by functions $\delta\Phi$ of compact support
\begin{flalign}
 \delta S_\star[\Phi] = 
\int\limits_\mathcal{M}\delta\Phi\,\vol\,\cosh\left(\frac{\lambda c}{4}(N+2)i X_1\right)\,\bigl(\square_g-\xi\,\mathfrak{R}\bigr) \Phi~.
\end{flalign}

\end{proof}
\begin{rem}
Note that in case we deform by two Killing vector fields $X_1,X_2\in\mathfrak{K}$ we have $c=0$ and the equation
of motion operator $\widetilde{P}_\star$ (\ref{eqn:eomoperatorhom}) is undeformed. This is a generalization of the well-known
result that the dynamics of a free scalar field theory on the Moyal-Weyl deformed Minkowski spacetime is undeformed.
\end{rem}

Let now $(\MM,g)$ be a connected, time-oriented and globally hyperbolic Lorentzian manifold.
The construction of the Green's operators corresponding to the deformed equation of motion operator (\ref{eqn:eomoperatorhom})
is straightforward. 
We define
\begin{flalign}
\label{eqn:greenoperatorhom}
 \widetilde{\Delta}_{\star\pm}:= \Delta_{\pm}\circ \cosh\left(\frac{\lambda c}{4}\left(N+2\right)iX_1\right)^{-1}~,
\end{flalign}
where the inverse of $\cosh\left(\frac{\lambda c}{4}\left(N+2\right)iX_1\right)$ is understood
in terms of formal power series and $\Delta_\pm$ are the unique retarded and advanced
 Green's operators corresponding to the undeformed equation of motion operator $P=\square_g-\xi\,\mathfrak{R}$.
We find
\begin{subequations}
\begin{flalign}
 \widetilde{P}_\star\circ \widetilde{\Delta}_{\star\pm} & = 
\id_{C^\infty_0(\mathcal{M})[[\lambda]]}~,\\
\widetilde{\Delta}_{\star\pm}\circ \widetilde{P}_\star\vert_{C^\infty_0(\mathcal{M})[[\lambda]]} 
&= \id_{C^\infty_0(\mathcal{M})[[\lambda]]}~,
\end{flalign}
and the support property
\begin{flalign}
 \supp\bigl(\widetilde{\Delta}_{\star\pm}(\varphi)\bigr)\subseteq J_\pm\bigl(\supp(\varphi)\bigr)~,
\end{flalign}
\end{subequations}
for all $\varphi\in C_0^\infty(\mathcal{M})$, since the noncommutative corrections to $\Delta_{\pm}$ 
are finite order differential operators at every order in $\lambda$. 

The retarded-advanced Green's operator for this model is
\begin{flalign}
 \widetilde{\Delta}_\star := \widetilde{\Delta}_{\star +} -\widetilde{\Delta}_{\star -} = \Delta \circ 
\cosh\left(\frac{\lambda c}{4}\left(N+2\right)iX_1\right)^{-1}~,
\end{flalign}
resulting in the following symplectic structure on $\widetilde{V}_\star = 
C^\infty_0(\MM,\bbR)[[\lambda]]/\widetilde{P}_\star[C^\infty_0(\MM,\bbR)[[\lambda]]]$
\begin{flalign}
 \widetilde{\omega}_\star([\varphi],[\psi]) = \spp{\varphi}{\widetilde{\Delta}_\star(\psi)}~.
\end{flalign}
The symplectic isomorphism to the undeformed symplectic $\bbR[[\lambda]]$-module $(V[[\lambda]],\omega)$
is obviously given by
\begin{flalign}
 \mathbf{S}: \widetilde{V}_\star \to V[[\lambda]]\,,~[\varphi] \mapsto [S\varphi]~,
\end{flalign}
where
\begin{flalign}
 S= \sqrt{\cosh\left(\frac{\lambda c}{4}\left(N+2\right)iX_1\right)^{-1}}~.
\end{flalign}
We have $\widetilde{\omega}_\star([\varphi],[\psi]) = \omega(\mathbf{S}[\varphi],\mathbf{S}[\psi])$, 
for all $[\varphi],[\psi]\in \widetilde{V}_\star$.
\begin{rem}
 Note that the retarded and advanced symplectic isomorphism $\mathbf{S}_\pm$ of Chapter \ref{chap:qftcon}
coincide for homothetic Killing deformations by two vector fields.
\end{rem}
\noindent The quantum field theory can be constructed along the lines of Chapter \ref{chap:qftdef}.

This shows that the formal theory of homothetic Killing deformations by two vector fields
is very simple. This allows us to discuss convergent examples, which is the main goal of the
remaining part of this section.

\subsection{\label{ex:frwqft}A spatially flat FRW toy-model:}
We apply the formalism presented in the previous subsection to a toy-model.
We use a special choice of the FRW spacetime of Example \ref{ex:frwhomothetic}. 
Let $\mathcal{M}=(0,\infty)\times\bbR^3$, i.e.~$N=4$,
and let $t$ and $x^i,~i\in\lbrace1,2,3\rbrace,$ be global coordinates.
The metric field we consider is given by $g=-dt\otimes_A dt + t^2\,\delta_{ij}dx^i\otimes_A dx^j$.
Note that in our conventions the spatial coordinates $x^i$ are dimensionless. The reason for choosing
the scale factor $a(t)\propto t$ is that in this case a proper homothetic Killing vector field
is given by $X_2=t\partial_t$ and {\it all} Killing vector fields commute with $X_2$.
The most general real Killing vector field
is  $k_{(\xi,\eta)} :=\xi^i\partial_i + \eta^{k}\epsilon_{kij}x^i\partial_j$, where $\xi,\eta\in\bbR^3$.

\subsubsection*{{\bf The undeformed theory:}}
We start by collecting useful formulae of the undeformed free, real, massless and curvature coupled scalar quantum field theory
 on our particular FRW spacetime. They will be used later to study the noncommutative deformation. We frequently
use the Fourier transformation on the spatial hypersurfaces $\bbR^3$ defined by $t=\text{const}$.
We indicate this transformation by a hat and use the conventions
\begin{flalign}
 \widehat{\varphi}(t,k)= \int\limits_{\bbR^3}d^3x\,e^{ikx}\,\varphi(t,x)\quad ,\qquad \varphi(t,x)=
\int\limits_{\bbR^3}\frac{d^3k}{(2\pi)^3}\,e^{-ikx}\,\widehat{\varphi}(t,k)~.
\end{flalign}
The wave operator $P=\square_g-\xi\,\mathfrak{R}$ in Fourier space is given by
\begin{flalign}
 \widehat{P}\bigl(\widehat{\varphi}\bigr)(t,k) = -\left(\partial_t^2+\frac{3}{t}\partial_t + 
\frac{k^2+6\xi}{t^2}\right)\widehat{\varphi}(t,k)~.
\end{flalign}
The corresponding retarded and advanced Green's operators read
\begin{flalign}
\label{eqn:greenfrw}
 \widehat{\Delta}_{\pm}\bigl(\widehat{\varphi}\bigr)(t,k)=-\int\limits_{t_\pm}^t d\tau\tau^3\,\widehat{\Delta}(t,\tau,k)\,
\widehat{\varphi}(\tau,k)~,
\end{flalign}
where $t_+=0$, $t_-=\infty$ and
\begin{flalign}
\widehat{\Delta}(t,\tau,k)= \frac{t^{\sqrt{1-k^2-6\xi}}\tau^{-\sqrt{1-k^2-6\xi}}-
t^{-\sqrt{1-k^2-6\xi}}\tau^{\sqrt{1-k^2-6\xi}}}{2t\tau\sqrt{1-k^2-6\xi}}~.
\end{flalign}
We obtain for the retarded-advanced Green's operator $\Delta=\Delta_+-\Delta_-$
\begin{flalign}
  \widehat{\Delta}\bigl(\widehat{\varphi}\bigr)(t,k)=-\int\limits_{0}^{\infty} d\tau\tau^3\,\widehat{\Delta}(t,\tau,k)\,
\widehat{\varphi}(\tau,k)~,
\end{flalign}
resulting in the pre-symplectic structure
\begin{flalign}
 \omega(\varphi,\psi) = -\int\limits_{0}^\infty dt t^3 \int\limits_{0}^\infty d\tau \tau^3\int\limits_{\bbR^3} 
\frac{d^3k}{(2\pi)^3}\,\widehat{\varphi}(t,-k)\,\widehat{\Delta}(t,\tau,k)\,\widehat{\psi}(\tau,k)~.
\end{flalign}

We define the geometric action of $(R,a)\in SO(3)\ltimes\bbR^3$ on $C^\infty(\mathcal{M})$ by
\begin{flalign}
 \bigl(\alpha_{(R,a)} \varphi\bigr)(t,x) :=\varphi\bigl(t,R^{-1}(x-a)\bigr)~.
\end{flalign}
In Fourier space, these transformations are given by
\begin{flalign}
\label{eqn:geometricactionfrw}
 \bigl(\widehat{\alpha}_{(R,a)}\widehat{\varphi}\bigr)(t,k)=e^{ika}\,\widehat{\varphi}(t,R^{-1}k)~.
\end{flalign}
We easily obtain that $SO(3)\ltimes\bbR^3\subseteq \mathcal{G}_{(V,\omega)}$ are symplectic automorphisms
of the symplectic vector space $(V,\omega)=
\bigl(C^\infty_0(\mathcal{M},\bbR)/P[C^\infty_0(\mathcal{M},\bbR)],\omega\bigr)$.

In this section we are working in a convergent framework and all symplectic modules are vector spaces 
over $\bbR$. Thus, we can apply the powerful theory of CCR-algebras in order to quantize the
symplectic vector space $(V,\omega)$, i.e.~in order to define the quantum field theory, see Chapter \ref{chap:qftbas}.

\subsubsection*{{\bf The deformed theory with {\boldmath $X_1=\partial_1$}:}}
We study a convergent deformation of our FRW model. We choose $X_1=\partial_1$, 
i.e.~a translation along the $x^1$-direction. 
The flow generated by $X_1=\partial_1$ is noncompact.
The condition $[X_1,X_2]=0$, which is required for our deformations, is satisfied. 
Our strategy is to make a convergent definition of the maps 
$S=\sqrt{\cosh(3\lambda i \partial_1)^{-1}}$ and $S^{-1}=\sqrt{\cosh(3\lambda i\partial_1)}$,
which enter the construction of the deformed quantum field theory. Using these
maps we construct the deformed quantum field theory and investigate its properties. 

We first define a convenient space of functions $C^\infty_{0,\mathscr{S}}(\mathcal{M})\subset C^\infty(\mathcal{M})$,
where $\mathcal{M}=(0,\infty)\times \bbR^3$. 
A function $\varphi\in C^\infty(\mathcal{M})$ is in $C^\infty_{0,\mathscr{S}}(\mathcal{M})$, if
$(1.)$ for all fixed $x\in\bbR^3$ $\varphi(t,x)\in C^\infty_0((0,\infty))$
and $(2.)$ for all fixed $t\in (0,\infty)$ $\varphi(t,x)\in \mathscr{S}(\bbR^3)$ is a Schwartz function.
We have $C^\infty_0(\mathcal{M})\subset C^\infty_{0,\mathscr{S}}(\mathcal{M})\subset C^\infty(\mathcal{M})$.
The spatial Fourier transformation
is an automorphism of $ C^\infty_{0,\mathscr{S}}(\mathcal{M})$, i.e.~let $\varphi(t,x)\in C^\infty_{0,\mathscr{S}}(\mathcal{M})$
 then $\widehat{\varphi}(t,k)\in C^\infty_{0,\mathscr{S}}(\mathcal{M})$ and vice versa.

Using the spatial Fourier transformation we define the map $S:C^\infty_{0,\mathscr{S}}(\mathcal{M})\to 
 C^\infty_{0,\mathscr{S}}(\mathcal{M})$ by
\begin{flalign}
\label{eqn:convsymp}
 \bigl(\widehat{S}\widehat{\varphi}\bigr)(t,k) := \sqrt{\cosh(3\lambda k_1)^{-1}}\,\widehat{\varphi}(t,k)~.
\end{flalign} 
This map is injective. The inverse map $S^{-1}:S\bigl[C^\infty_{0,\mathscr{S}}(\mathcal{M})\bigr]\to 
C^\infty_{0,\mathscr{S}}(\mathcal{M})$ is given by
\begin{flalign}
 \bigl(\widehat{S}^{-1}\widehat{\varphi}\bigr)(t,k) = \sqrt{\cosh(3\lambda k_1)}\,\widehat{\varphi}(t,k)~.
\end{flalign}
Note that $S\bigl[C^\infty_{0,\mathscr{S}}(\mathcal{M})\bigr]\subset C^\infty_{0,\mathscr{S}}(\mathcal{M})$, since
$S\bigl[C^\infty_{0,\mathscr{S}}(\mathcal{M})\bigr]$ includes only functions with Fourier spectra
decreasing faster than $e^{-3\lambda \vert k_1\vert /2}$ for large $\vert k_1\vert$.

We can now construct the deformed quantum field theory. 
For simplifying the notation we drop the tilde on the deformed maps.
 We define the deformed Green's operators by
\begin{flalign}
 \widehat{\Delta}_{\star\pm}\bigl(\widehat{\varphi}\bigr)(t,k):=\widehat{\Delta}_{\pm}\bigl(\widehat{S}^2\widehat{\varphi}\bigr)(t,k)
=-\int\limits_{t_\pm}^t d\tau\tau^3\,\frac{\widehat{\Delta}(t,\tau,k)}{\cosh(3\lambda k_1)}\,\widehat{\varphi}(\tau,k)~.
\end{flalign}
This results in the deformed retarded-advanced Green's operator
\begin{flalign}
\label{eqn:frwdefsympl}
 \widehat{\Delta}_{\star}\bigl(\widehat{\varphi}\bigr)(t,k)=-\int\limits_{0}^\infty d\tau\tau^3\,
\frac{\widehat{\Delta}(t,\tau,k)}{\cosh(3\lambda k_1)}\,\widehat{\varphi}(\tau,k)~,
\end{flalign}
and the deformed pre-symplectic structure on $C^\infty_0(\mathcal{M},\bbR)$
\begin{flalign}
 \omega_\star(\varphi,\psi) = -\int\limits_{0}^\infty dt t^3 \int\limits_{0}^\infty d\tau \tau^3\int\limits_{\bbR^3} 
\frac{d^3k}{(2\pi)^3}\,\widehat{\varphi}(t,-k)\,\frac{\widehat{\Delta}(t,\tau,k)}{\cosh(3\lambda k_1)}\,
\widehat{\psi}(\tau,k)~.
\end{flalign}
The deformed retarded-advanced Green's operator satisfies $\text{Ker}(\Delta_\star) = \text{Ker}(\Delta)$.
To prove this, let $\varphi\in C^\infty_0(\mathcal{M},\bbR)$. Using (\ref{eqn:frwdefsympl}) we find
\begin{flalign}
\widehat{\Delta}_\star\bigl(\widehat{\varphi}\bigr)(t,k)=\cosh(3\lambda k_1)^{-1}\,\widehat{\Delta}\bigl(\widehat{\varphi}\bigr)(t,k)~,
\end{flalign}
and the proof follows from the positivity of $\cosh(3\lambda k_1)^{-1}$.
From the general considerations in Chapter \ref{chap:qftbas} we know that 
$\text{Ker}(\Delta)=P[C^\infty_0(\mathcal{M},\bbR)]$.
Thus, we can define the deformed symplectic vector space as $(V_\star,\omega_\star):= 
\bigl(C^\infty_0(\mathcal{M},\bbR)/P[C^\infty_0(\mathcal{M},\bbR)], \omega_\star\bigr)$.
Different to formal deformation quantization this is now a vector space over $\bbR$ and not a module over 
the ring $\bbR[[\lambda]]$.

The construction of the deformed quantum field theory in terms of a CCR-representation of $(V_\star,\omega_\star)$
can be performed analogously to the undeformed case (see Chapter \ref{chap:qftbas}),
 since $(V_\star,\omega_\star)$ is symplectic vector space over $\bbR$.
This results in a unique (up to $\ast$-isomorphisms) deformed CCR-algebra $A_\star$
describing the deformed quantum field theory.

\subsubsection*{{\bf Physical features of the deformed theory:}}
Let us investigate some physical features of the deformed field theory. In the following we assume that $\lambda>0$.
We study in more detail the map $S^2$ acting on $C^\infty_0(\mathcal{M})$. Note that in 
position space this map is given by the following convolution
\begin{flalign}
\label{eqn:convolution}
 \bigl(S^2\varphi\bigr)(t,x) = \int\limits_\bbR dy^1\,\frac{1}{6\lambda\,\cosh\bigl(\pi (x^1-y^1)/6\lambda\bigr)}\,
\varphi(t,y^1,x^2,x^3)~,
\end{flalign}
for all $\varphi\in C^\infty_0(\mathcal{M})$. It is easy to check that the image $S^2[C^\infty_0(\mathcal{M})]\not\subseteq 
C^\infty_0(\mathcal{M})$. For this assume that $\varphi\in C^\infty_0(\mathcal{M})$ is a positive semidefinite function localized
in some compact region $K\subset\mathcal{M}$, e.g.~a bump function. Since the convolution kernel is of noncompact support
and strictly positive, the resulting function $S^2\varphi$ is of noncompact support in the $x^1$-direction.
However, $S^2\varphi$ will be of rapid decrease, since the convolution kernel is a Schwartz function.

Physically, this means that causality is lost. We immediately obtain that the  relation
$\Delta_{\star\pm}(\varphi)\subseteq J_\pm\bigl(\supp(\varphi)\bigr)$, for all $\varphi\in C^\infty_0(\mathcal{M})$, 
is violated. 
Thus, external sources couple in a nonlocal way to our deformed field theory, which is a feature not
present in the commutative counterpart. Since $\Delta_{\star\pm}$ depends on the value of $\lambda$ through (\ref{eqn:convolution}),
we can determine $\lambda$ (in principle) by measuring the response of the field to external excitations.

Consider now the deformed pre-symplectic structure in position space
\begin{flalign}
\label{eqn:nonlocalsymp}
 \omega_\star(\varphi,\psi)=\int\limits_\mathcal{M} \varphi\,\Delta(S^2\psi)\,\vol~.
\end{flalign}
Due to the appearance of the nonlocal map $S^2$ (\ref{eqn:convolution}) there
are functions $\varphi,\psi$ with causally disconnected support satisfying $\omega_\star(\varphi,\psi)\neq 0$. 
For our choice of deformation ($X_1=\partial_1$ and $X_2=t\partial_t$)  this nonlocality 
only affects the $x^1$-direction, but its range is infinite. We still obtain that $\omega_\star(\varphi,\psi)=0$ if
arbitrary translations of  $\supp(\psi)$ along the $x^1$-direction and $\supp(\varphi)$ are
causally disconnected.
In the quantum field theory described by $A_\star$ this nonlocal behavior leads to a noncommutativity between observables located in
 spacelike separated regions in spacetime.

\subsubsection*{{\boldmath $\ast$}{\bf-isomorphism between the deformed and a nonstandard undeformed quantum field theory:}}
We have shown in the previous paragraph that 
$S^2[C^\infty_0(\mathcal{M},\bbR)]\not\subseteq C^\infty_0(\mathcal{M},\bbR)$.
The map $S$ defined in (\ref{eqn:convsymp}) can also be written in terms of a convolution in position
space, similar to $S^2$, but the convolution kernel is more complicated. Its explicit form is not of importance to us.

From $S^2[C^\infty_0(\mathcal{M},\bbR)]\not\subseteq C^\infty_0(\mathcal{M},\bbR)$ we can infer
that $S[C^\infty_0(\mathcal{M},\bbR)]\not\subseteq C^\infty_0(\mathcal{M},\bbR)$. To prove this, assume
that $S[C^\infty_0(\mathcal{M},\bbR)]\subseteq C^\infty_0(\mathcal{M},\bbR)$ then we would find
$S^2[C^\infty_0(\mathcal{M},\bbR)]\subseteq C^\infty_0(\mathcal{M},\bbR)$ by applying $S$ twice, 
which contradicts the observation above.

The bijective map $S:C^\infty_0(\mathcal{M},\bbR)\to C^\infty_\text{img}(\mathcal{M},\bbR)
\subset C^\infty_{0,\mathscr{S}}(\mathcal{M},\bbR)$ induces a symplectic isomorphism between
$(V_\star,\omega_\star)$ and $(V_{\text{img}},\omega):=\bigl(C^\infty_\text{img}(\mathcal{M},\bbR)
/S\bigl[P[C_0^\infty(\mathcal{M},\bbR)]\bigr], \omega\bigr)$, i.e.~the deformed field theory
can be related to a nonstandard undeformed one. Using Corollary \ref{cor:asthomoweyl},
this map induces a unique $\ast$-isomorphism between the CCR-algebras $A_\star$ and 
$A_\text{img}$.

Let us consider the symmetries of the deformed quantum field theory. As we have seen above, the undeformed symplectic
vector space $(V,\omega)$ contains $SO(3)\ltimes \bbR^3$ in the group of symplectic automorphisms.
The representation is given by the geometric action (\ref{eqn:geometricactionfrw}).
However, the space $C^\infty_{\text{img}}(\mathcal{M},\bbR)\subset C^\infty_{0,\mathscr{S}}(\mathcal{M},\bbR)$,
 which serves as a pre-symplectic vector space for $(V_\text{img},\omega)$, is not invariant under the action of 
$SO(3)\ltimes \bbR^3$. The preserved subgroup is $SO(2)\ltimes \bbR^3$, where 
the $SO(2)$-rotation preserves the $x^1$-axis. Thus, the deformed quantum field theory $A_\star$ is $\ast$-isomorphic
 to a nonsymmetric undeformed quantum field theory $A_\text{img}$. This, in particular, shows that
the results of Chapter \ref{chap:qftcon} are restricted to formal deformation quantization.

\subsubsection*{{\bf Physical interpretation:}}
As shown above, we do not have a $\ast$-isomorphism between the deformed quantum field theory 
$A_\star$ and the {\it standard} undeformed quantum field theory $A$ for convergent deformations. 
This fact has a very natural physical interpretation, which we will explain now.

Consider the extended symplectic vector space
$(V_\text{ext},\omega):= \bigl(C^\infty_{0,\mathscr{S}}(\mathcal{M},\bbR)/\text{Ker}(\Delta),\omega\bigr)$.
This vector space carries a representation of $SO(3)\ltimes\bbR^3$ via the geometric
action (\ref{eqn:geometricactionfrw}). 
A CCR-representation of $(V_\text{ext},\omega)$ yields the extended CCR-algebra $A_\text{ext}$.
$A_\text{ext}$ carries a representation of the group $SO(3)\ltimes\bbR^3$. 
The algebra $A_\text{ext}$ is an extension of the usual CCR-algebra $A$, where also certain delocalized observables
are allowed. This already shows that $A_\text{ext}$ is more suitable to study the connection between deformed and undeformed
quantum field theories.

Analogously to above, we define the deformed extended symplectic vector space
 $(V_{\text{ext}\star},\omega_\star):=(V_\text{ext},\omega_\star)$ and 
a CCR-representation of $(V_{\text{ext}\star},\omega_\star)$ yields the deformed extended CCR-algebra $A_{\text{ext}\star}$.

The bijective map $S:C^\infty_{0,\mathscr{S}}(\mathcal{M},\bbR)\to C^\infty_{\text{img}^\prime}(\mathcal{M},\bbR)
\subset C^\infty_{0,\mathscr{S}}(\mathcal{M},\bbR)$ induces a symplectic isomorphism between
$(V_{\text{ext}\star},\omega_\star)$ and $(V_{\text{img}^\prime},\omega)$. Furthermore,
$S$ induces a symplectic embedding $(V_{\text{ext}\star},\omega_\star)\rightarrow (V_\text{ext},\omega)$.
Using Corollary \ref{cor:asthomoweyl}
we can induce a unique $\ast$-isomorphism between the CCR-algebras $A_{\text{ext}\star}$ and 
$A_{\text{img}^\prime}$, as well as a unique injective, but not surjective, $\ast$-homomorphism from $A_{\text{ext}\star}$ 
to $A_\text{ext}$. We thus have $A_{\text{ext}\star}\simeq A_{\text{img}^\prime}\subset A_{\text{ext}}$.

Physically, this means that due to the noncommutative deformation the quantum field theory looses observables. 
Since $C^\infty_{\text{img}^\prime}(\mathcal{M},\bbR)$ depends on the value of 
$\lambda$, the observables we loose also depend on the value of $\lambda$.
Note that all functions in $C^\infty_{\text{img}^\prime}(\mathcal{M},\bbR)$ have Fourier spectra which decrease
 faster than $e^{-3\lambda \vert k_1\vert/2}$ for large $\vert k_1\vert$. Increasing $\lambda$ leads 
to a sharper localization in momentum space and therefore 
a weaker localization in position space. Thus, the deformed quantum field theory looses those observables
that are strongly localized in position space, which is physically very natural.

\subsubsection*{{\bf Induction of states:}}
The injective $\ast$-homomorphism $\mathfrak{S}:A_\star\to A_{\text{ext}}$, 
or its extension $\mathfrak{S}_\text{ext}:A_{\text{ext}\star}\to A_{\text{ext}}$,
is useful for inducing states via the pull-back.
Assume that we have an $SO(3)\ltimes\bbR^3$-invariant state $\Omega_{\text{ext}}:A_\text{ext}\to \bbC$.
Defining $\Omega_\star:= \Omega_\text{ext}\circ\mathfrak{S}:A_\star\to\bbC$ induces a state
on $A_\star$, which is invariant under the unbroken symmetry group $SO(2)\ltimes\bbR^3$.
The same holds true for $\Omega_{\text{ext}\star}:= \Omega_\text{ext}\circ\mathfrak{S}_\text{ext}:A_{\text{ext}\star}\to\bbC$.

\subsubsection*{{\bf Noncommutative effects on the cosmological power spectrum:}}
In order to show that our models lead to nontrivial physical effects we briefly discuss the
cosmological power spectrum. We shall work with the extended CCR-algebras $A_{\text{ext}\star}$ and $A_{\text{ext}}$.
Let $\Omega:A_{\text{ext}}\to \bbC$ be a regular, quasi-free and translation invariant (i.e.~$\bbR^3$-invariant)
state. We denote by $\bigl(\mathcal{H},\pi,\vert 0\rangle\bigr)$ the GNS-representation of $(A_{\text{ext}},\Omega)$.
Employing the $\ast$-homomorphism $\mathfrak{S}_\text{ext}:A_{\text{ext}\star}\to A_{\text{ext}}$ we obtain
a representation $\pi_\star=\pi\circ\mathfrak{S}_\text{ext}$ of $A_{\text{ext}\star}$ on $\mathcal{H}$. The vector $\vert 0\rangle$ 
might not be cyclic with respect to this representation and we define the Hilbert subspace $\mathcal{H}_\star = 
\overline{\pi_\star[A_{\text{ext}\star}]\vert 0\rangle}$. By the GNS-Theorem the cyclic representation 
$\bigl(\mathcal{H}_\star,\pi_\star,\vert 0 \rangle\bigr)$ is unitary equivalent to the GNS-representation 
of $(A_{\text{ext}\star},\Omega_\star)$. 

Using regularity of the state $\Omega$ we can define the 
selfadjoint linear field operators $\Phi(v)\in L(\mathcal{D})$, $v\in V_\text{ext}$, 
as the generators of the Weyl operators $\pi\bigl(W(v)\bigr)\in B(\mathcal{H})$. They are related to the 
deformed field operators $\Phi_\star$ by $\Phi_\star(v)= \Phi(\mathbf{S}v)\in L(\mathcal{D})$,
which are the generators of $\pi_\star\bigl(W_\star(v)\bigr) = \pi\bigl(W(\mathbf{S}v)\bigr)$.
Thus, $n$-point functions of the deformed field operators can be related to $n$-point functions of the
undeformed ones. Due to the quasi-free assumption on the state $\Omega$ only the $2$-point function
is of interest and all higher $n$-point functions can be derived from it, in the undeformed and also deformed case.
 We define the deformed $2$-point function
\begin{flalign}
 \Omega_{\star2}(\varphi,\psi):= \langle 0\vert \Phi_\star([\varphi])\Phi_\star([\psi])\vert0\rangle=
\langle 0\vert \Phi([S\varphi])\Phi([S\psi])\vert0\rangle=\Omega_2(S\varphi,S\psi)~,
\end{flalign}
for all $\varphi,\psi\in C^\infty_0(\MM,\bbR)$, where $\Omega_2$ is the undeformed $2$-point function.

In order to calculate the power spectrum we require the kernel of $\Omega_2$ ($\Omega_{\star2}$) in Fourier space.
Making use of the translation invariance of the state $\Omega$ we have
\begin{flalign}
 \Omega_2(\varphi,\psi) = \int\limits_0^\infty dt t^3\int\limits_0^\infty d\tau \tau^3 \int\limits_{\bbR^3}
\frac{d^3k}{(2\pi)^3}~
\widehat{\Omega}_2(t,\tau,k)\,\widehat{\varphi}(t,k)\,\widehat{\psi}(\tau,-k)~.
\end{flalign}
It follows that
\begin{flalign}
\label{eqn:def2point}
\Omega_{\star2}(\varphi,\psi) = 
 \Omega_{2}(S\varphi,S\psi) =
\int\limits_0^\infty dt t^3\int\limits_0^\infty d\tau \tau^3 \int\limits_{\bbR^3}
\frac{d^3k}{(2\pi)^3}~
\frac{\widehat{\Omega}_2(t,\tau,k)}{\cosh(3\lambda k_1)}\,\widehat{\varphi}(t,k)\,\widehat{\psi}(\tau,-k)~.
\end{flalign}
The undeformed power spectrum is then given by
\begin{flalign}
 \mathcal{P}(t,k) := \widehat{\Omega}_2(t,t,k)~,
\end{flalign}
and the deformed power spectrum reads
\begin{flalign}
\label{eqn:powerspectrum} 
\mathcal{P}_\star(t,k) = \frac{\mathcal{P}(t,k)}{\cosh(3\lambda k_1)}~.
\end{flalign}
The physical feature of $\mathcal{P}_\star$ compared to $\mathcal{P}$ is an exponential loss of power for large $\vert k_1\vert$.

\subsection{\label{ex:2dmodel}A spatially compact FRW toy-model:}
The flow generated by the vector field $X_1=\partial_1$ in the model above is noncompact. We now investigate
if deformations along vector fields $X_1$ generating compact flows differ from the noncompact case.
A possible choice of toy-model would be the model of Subsection \ref{ex:frwqft} with
 $X_1=x^2\partial_3-x^3\partial_2$, i.e.~a rotation around the $x^1$-axis, and $X_2=t\partial_t$.
 
However, there is an even simpler model which we can use for our studies. Consider the manifold 
$\mathcal{M}= (0,\infty)\times S^1$, where $S^1$ is the $1$-dimensional circle, 
equipped with the metric $g=-dt\otimes_A dt+ t^2 d\phi\otimes_A d\phi$.
Here $t\in(0,\infty)$ denotes time and $\phi\in[0,2\pi)$ is the angle.  A proper homothetic Killing vector field is
$X_2=t\partial_t$ and $X_1=2\partial_\phi$ is a Killing vector field.\footnote{
The factor $2$ in the definition of $X_1$ is a convenient normalization.}  
The necessary condition $[X_1,X_2]=0$ is satisfied.
The calculation of the undeformed wave operator $P=\square_g-\xi\,\mathfrak{R}$ and the corresponding 
retarded and advanced Green's operators $\Delta_{\pm}$ is standard. We do not need to present the explicit results here.

We investigate in detail the convergent version of $S=\sqrt{\cosh(3\lambda i\partial_\phi)^{-1}}$ and
$S^{-1}=\sqrt{\cosh(3\lambda i\partial_\phi)}$. 
For this we make use of the Fourier transformation on the spatial hypersurfaces $S^1$.
Since $S^1$ is compact, the momenta $n\in\mathbb{Z}$ are discrete. 
We define in Fourier space
\begin{subequations}
\begin{flalign}
\label{eqn:frw2S}\bigl(\widehat{S}\widehat{\varphi}\bigr)(t,n)&:= \sqrt{\cosh(3\lambda n)^{-1}}\,\widehat{\varphi}(t,n)~,\\
\bigl(\widehat{S}^{-1}\widehat{\varphi}\bigr)(t,n)& = \sqrt{\cosh(3\lambda n)}\,\widehat{\varphi}(t,n)~.
\end{flalign}
\end{subequations}
It remains to discuss the domains of the maps $S$ and $S^{-1}$. Note that a function $\varphi$ on the circle $S^1$
is smooth if and only if its Fourier spectrum is of rapid decrease.
From this and (\ref{eqn:frw2S}) we infer $S: C^\infty_0(\mathcal{M})\to  C^\infty_0(\mathcal{M})$. 
However, $S^{-1}$ can not be defined on all of $C^\infty_0(\mathcal{M})$, provided we restrict the image of 
$S^{-1}$ to smooth functions. 
Since $S$ is injective, we find the isomorphism 
$S:C^\infty_0(\mathcal{M})\to C^\infty_\text{img}(\mathcal{M})\subset C^\infty_0(\mathcal{M})$.

Analogously to the situation where $X_1$ generates a noncompact flow we find that the symplectic isomorphism maps between
the deformed symplectic vector space $(V_\star,\omega_\star)$ and a nonstandard undeformed symplectic vector space
$(V_{\text{img}},\omega)$. In case of models with additional isometries $(V_{\text{img}},\omega)$ carries
only a representation of the unbroken subgroup. 
Furthermore, the map $S$ embeds $(V_\star,\omega_\star)$ into the undeformed  symplectic vector space $(V,\omega)$.
This is also analogous to the situation before, with the difference that we do not have to extend the 
symplectic vector space. All linear symplectic maps induce unique injective $\ast$-homomorphisms between
the corresponding CCR-algebras due to Corollary \ref{cor:asthomoweyl}.

We find again that the deformed quantum field theory looses observables, depending on the
value  of $\lambda$.
Since the functions in $C^\infty_{\text{img}}(\mathcal{M})$ have Fourier spectra which decrease
faster than $e^{-3\lambda \vert n\vert/2}$ for large $\vert n\vert$ we loose those observables
that are strongly localized in position space.

\subsection{Physical inequivalence of the deformed and undeformed quantum field theory:}
In this subsection we collect arguments that our deformed quantum field theory is physically inequivalent to 
standard commutative quantum field theories. 
This discussion is very important, since it has been shown \cite{0158.45703,0159.29001}
that it is not possible in the framework of Wightman quantum field theory on the
undeformed Minkowski spacetime to construct a nonlocal theory with
nonlocalities in the commutator function that fall off faster than
exponentially.
Indeed, Wightman theories with a faster than
exponentially vanishing nonlocality in the commutator function of
two fundamental fields are equivalent to local quantum field theories.  Therefore,
we have to ensure that a similar result does not apply to our deformed models.

The first sign for an inequivalence of our deformed quantum field theories to standard commutative ones comes
from the nonlocal behavior of the deformed Green's operators $\Delta_{\star\pm}$ discussed above.
Coupling our deformed theories to external sources (or introducing perturbative interactions) then leads to 
nonlocal effects which in principle can be used to measure the deformation parameter $\lambda$.

The second sign comes from the form of the deformed power spectrum (\ref{eqn:powerspectrum}).
Commutative quantum field theories on FRW spacetimes typically have a power spectrum, which goes as 
$\mathcal{P}(t,k)\propto \vert k\vert^{n_s-4}$, where $n_s$ is
the spectral index. Our deformed power spectrum (\ref{eqn:powerspectrum}) shows, additionally to this power-law behavior,
an exponential drop-off for large $\vert k_1\vert$. Such a drop-off, and with this also the value of $\lambda$,
can in principle be measured.

The third sign comes from the nonlocality of the symplectic structure (\ref{eqn:nonlocalsymp}).
We consider analogously to \cite{0158.45703,0159.29001} the correlation function
\begin{flalign}
\label{eqn:4pointcom} \langle 0\vert \bigl[\Phi_\star([\varphi_1]),\Phi_\star([\varphi_2])\bigr]
\Phi_\star([\varphi_3])\Phi_\star([\varphi_4])\vert 0\rangle~,
\end{flalign}
which, in our case, reduces due to the canonical commutation relations to
\begin{flalign}
 \nonumber \langle 0\vert \bigl[\Phi_\star([\varphi_1]),\Phi_\star([\varphi_2])\bigr]\Phi_\star([\varphi_3])
\Phi_\star([\varphi_4])\vert 0\rangle
&=i\,\omega_\star(\varphi_1,\varphi_2)\,\langle 0\vert \Phi_\star([\varphi_3])\Phi_\star([\varphi_4])\vert 0\rangle\\
&=
i\,\omega_\star(\varphi_1,\varphi_2)\,\Omega_{\star2}(\varphi_3,\varphi_4) ~.
\end{flalign}
Since, as explained above, the commutator function $\omega_\star$ can be nonzero for functions $\varphi_1,\varphi_2$ with 
 spacelike separated support, (\ref{eqn:4pointcom}) is nonzero for these $\varphi_1,\varphi_2$ and for all
 $\varphi_3,\varphi_4$ such that $\Omega_{\star2}(\varphi_3,\varphi_4)\neq 0$.
Thus, our quantum field theory is nonlocal.


\section{\label{sec:isotrop}Towards quantum field theory on the isotropically deformed de Sitter spacetime}
We make some important steps towards constructing a scalar quantum field theory on the 
isotropically deformed de Sitter spacetime. 
The wave operator for this model is given by (\ref{eqn:isotropicdesitter})
\begin{flalign}
 \widetilde{P}_\star = -\cosh\left(\frac{3\lambda}{2}\,i\mathcal{D}\right)\bigl( \partial_t^2 + 3H\partial_t +M^2 \bigr)
 +\cosh\left(\frac{5\lambda}{2}\,i\mathcal{D}\right)e^{-2Ht}\bigtriangleup~,
\end{flalign}
where $\mathcal{D} = \partial_t - H\,x^i\partial_i$. In this section only the tilded wave operator
and Green's operators will appear. In order to simplify the notation we are going to drop the tilde on
all quantities.

We follow a similar strategy as for the homothetic Killing deformations and
first rewrite formally the deformed wave operator in convenient coordinates, such that
the nonlocal operators $\cosh\left(\frac{3\lambda}{2}\,i\mathcal{D}\right)$ and 
$\cosh\left(\frac{5\lambda}{2}\,i\mathcal{D}\right)$ are multiplication operators in a suitable Fourier space.
Starting from spherical coordinates $(t,r,\zeta,\phi)$ we perform the following
coordinate transformation
\begin{flalign}
 r=\frac{1}{H} \,e^{H\rho}~,\quad \rho = \frac{1}{H}\log (H r)~.
\end{flalign}
In the coordinates $(t,\rho,\zeta,\phi)$ we have $\mathcal{D}=\partial_t - \partial_\rho$.
The second coordinate transformation is given by
\begin{flalign}
 \tau^{\pm} = t\pm\rho~,\quad t=\frac{1}{2}(\tau^+ + \tau^-)~,\quad \rho = \frac{1}{2}(\tau^+ -\tau^-)~.
\end{flalign}
This leads to $\mathcal{D} = 2\partial_-$, where $\partial_-$ is the derivative along $\tau^-$.
The deformed wave operator in the coordinates $(\tau^+,\tau^-,\zeta,\phi)$ reads
\begin{multline}
 P_\star = -\cosh(3\lambda\,i\partial_-)\,\bigl((\partial_+ +\partial_-)^2 + 3 H (\partial_++\partial_-) +M^2\bigr)\\
+\cosh(5\lambda\,i\partial_-)\, e^{-2 H \tau^+}\,\bigl( (\partial_+-\partial_-)^2 + H (\partial_+-\partial_-) + H^2 \Delta_{S^2}\bigr)~,
\end{multline}
where $\Delta_{S^2}$ is the Laplacian on the unit two-sphere $S^2$.
The volume form reads in the new coordinates
\begin{flalign}
 \vol = \frac{1}{2\,H^2}\,e^{3 H \tau^+} \,\sin\zeta~ d\tau^-\wedge d\tau^+ \wedge d\zeta \wedge d\phi~.
\end{flalign}
 We make use of the Fourier transformation 
\begin{flalign}
 \varphi(\tau^+,\tau^-,\zeta,\phi) = \int\limits_{\bbR} \frac{dp}{2\pi} ~e^{-ip\tau^{-}} ~\widehat{\varphi}(\tau^+,p,\zeta,\phi)~,
\quad \widehat{\varphi}(\tau^+,p,\zeta,\phi) = \int\limits_{\bbR} d\tau^-~e^{ip\tau^-}~ \varphi(\tau^+,\tau^-,\zeta,\phi)~.
\end{flalign}
In Fourier space the deformed wave operator is given by
\begin{multline}
 \widehat{P}_\star = \cosh(3\lambda p)\,\Bigl( -\bigl((\partial_+ -ip)^2 + 3 H (\partial_+ -ip) +M^2\bigr)\\
+\frac{\cosh(5\lambda p)}{\cosh(3\lambda p)}\, e^{-2 H \tau^+}\,\bigl( (\partial_+ +ip)^2 + H (\partial_+ +ip) + H^2 \Delta_{S^2}\bigr)\Bigr)~.
\end{multline}
Defining 
\begin{flalign}
\label{eqn:defdeltap}
 e^{2H\delta_p}:=\frac{\cosh(5\lambda p)}{\cosh(3\lambda p)}
\end{flalign}
and the translation operator 
\begin{flalign}
 (T_a\widehat{\varphi})(\tau^+,p,\zeta,\phi) := \widehat{\varphi}(\tau^+ +a ,p,\zeta,\phi)~,
\end{flalign}
we observe that the deformed wave operator is given by a $p$-dependent translation
 and rescaling of the undeformed one
\begin{flalign}
 \widehat{P}_\star = \cosh(3\lambda p)\circ T_{-\delta_p}\circ \widehat{P}\circ T_{\delta_p}~.
\end{flalign}
Correspondingly, we deform the retarded and advanced Green's operators by
\begin{flalign}
 \widehat{\Delta}_{\star\pm} = \frac{1}{\cosh(3\lambda p)}\circ T_{-\delta_p}\circ\widehat{\Delta}_{\pm}\circ T_{\delta_p}~.
\end{flalign}
Let us now focus on the deformed pre-symplectic structure
\begin{flalign}
 \omega_\star(\varphi,\psi) = \int\limits_\MM \widehat{\varphi}(\tau^+,-p,\zeta,\phi)\, \bigl(\widehat{\Delta}_\star(\widehat{\psi})\bigr)
(\tau^+,p,\zeta,\phi)\,\widehat{\vol}~,
\end{flalign}
where $\widehat{\vol}$ is the integration measure in Fourier space.
Shifting the $\tau^+$ integral by $\delta_p$ we obtain
\begin{flalign}
  \nn \omega_\star(\varphi,\psi) &= \int\limits_\MM (T_{\delta_p}\widehat{\varphi})(\tau^+,-p,\zeta,\phi)\, 
\bigl(\widehat{\Delta}(T_{\delta_p}\widehat{\psi})\bigr)(\tau^+,p,\zeta,\phi)\,\frac{1}{\cosh(3\lambda p)}\,e^{3 H \delta_p}\,
\widehat{\vol}\\
&=\int\limits_\MM (T_{\delta_p}\widehat{\varphi})(\tau^+,-p,\zeta,\phi)\, 
\bigl(\widehat{\Delta}(T_{\delta_p}\widehat{\psi})\bigr)(\tau^+,p,\zeta,\phi)\,\frac{\cosh(5\lambda p)^{3/2}}{\cosh(3\lambda p)^{5/2}}\,
\widehat{\vol}~,
\end{flalign}
where in the last equality we have used the definition of $\delta_p$ (\ref{eqn:defdeltap}).
From this we find that the deformed pre-symplectic structure
is related to the undeformed one via the map
\begin{flalign}
\label{eqn:symplectoisotropic}
 \bigl(\widehat S \widehat{\varphi}\bigr)(\tau^+,p,\zeta,\phi) :=
 \frac{\cosh(5\lambda p)^{3/4}}{\cosh(3\lambda p)^{5/4}}\,\widehat{\varphi}(\tau^++\delta_p,p,\zeta,\phi)~. 
\end{flalign}
Denoting by $S$ the position space version of the linear symplectic map we have
\begin{flalign}
 \omega_\star(\varphi,\psi) = \omega(S\varphi,S\psi)~.
\end{flalign}

Thus, we were able to explicitly construct the formal symplectic isomorphism
between the deformed and undeformed field theory.
As shown in Chapter \ref{chap:qftcon}, this symplectic isomorphism provides
a $\ast$-algebra isomorphism between the deformed and undeformed quantum field theory.
In particular, given a state on the undeformed quantum field theory (e.g.~the Bunch-Davies vacuum)
with $2$-point function $\Omega_2$,
we can pull it back to the deformed quantum field theory and obtain for the deformed
$2$-point function
\begin{flalign}
 \Omega_{\star2}(\varphi,\psi) = \Omega_2(S\varphi,S\psi)~.
\end{flalign}
The noncommutative corrections to the power spectrum are thus contained in $S$
and can in principle be calculated to the desired order in $\lambda$ by using (\ref{eqn:symplectoisotropic}).
This is not the goal of the present work.

Note that the symplectic linear map was constructed in a way also suitable for nonformal investigations.
Since in (\ref{eqn:symplectoisotropic}) the deformation parameter only
appears in multiplication operators, we can set it to a finite value $\lambda\in \bbR$
and obtain a convergent definition of $S$. However, this does not yet lead to a convergent deformation
of the quantum field theory on the isotropically deformed de Sitter spacetime.
In particular, it remains to study if the action of $S$ on {\it all} compactly supported
functions $C^\infty_0(\MM)$ is well-defined and if the image of $S$ lies in the domain of $\Delta$.
This then would provide a nonformal definition of the symplectic structure $\omega_\star$.
These studies are considerably more complicated than the corresponding investigations
we have made for the homothetic Killing deformations above.
We hope to come back to this issue in a future work.


\section{\label{sec:z=2qft}A new noncommutative Euclidean quantum field theory}
Interacting quantum field theories on the Euclidean Moyal-Weyl space $\bbR^4_\Theta$
gathered a lot of attention after it was shown by Grosse and Wulkenhaar, and later also by the 
group around Rivasseau, that the $\Phi^4$-theory with a harmonic oscillator potential on $\bbR^4_\Theta$
has improved quantum properties, compared to the commutative $\Phi^4$-theory.
In particular, this model has no Landau pole, is renormalizible to all orders and 
it is a good candidate for a rigorous interacting quantum field theory in $4$-dimensions.
For details and properties of the Grosse-Wulkenhaar model see \cite{Grosse:2003aj,Grosse:2004yu,Grosse:2004by,Grosse:2009pa,
Rivasseau:2005bh,Gurau:2005gd,Disertori:2006nq,Gurau:2009ni,Sfondrini:2010zm} and references therein.
Despite its mathematical beauty, the Grosse-Wulkenhaar model has two phenomenological
drawbacks. Firstly, it is formulated on Euclidean space and it is not clear how to generalize it to the
Minkowski spacetime. See \cite{Fischer:2008dq,Fischer:2010zg} for a recent approach in this direction,
however, there seem to be unsolved issues as pointed out in \cite{Zahn:2010yt}.
Secondly, the harmonic oscillator term of the Grosse-Wulkenhaar model
leads to a strong modification of the propagator in the infrared, which will make it very hard
to render such theories compatible with experiments in particle and astrophysics.
The phenomenologically problematic term can be traced back to an infamous feature called
UV/IR-mixing, which is typically present in quantum field theories on the Euclidean Moyal-Weyl space.
The reason for this mixing lies in the phase factors entering the Feynman rules of such theories,
which originate from the deformed $\Phi^4$-potential.

We are going to present now an alternative deformation of the Euclidean space $\bbR^N$,
on which the perturbatively interacting quantum field theory has completely different features as on $\bbR^N_\Theta$.
\subsection*{The deformation:}
Consider the Euclidean space $\bbR^N$ with metric $g=\delta_{\mu\nu}\,dx^\mu\otimes_A dx^\nu$.
We choose one direction, say $x^N$, and denote the orthogonal coordinates by $x^i$, $i=1,\dots,N-1$.
Analogously to the anti-de Sitter spacetime in Section \ref{sec:exampleswaveop} of this chapter
 we consider the following $2n$ vector fields
\begin{flalign}
 X_{2a-1} = T_a^i\partial_i~,\quad X_{2a} = \vartheta(x^N)\,T_a^i\partial_i~,\quad\text{for all }a=1,\dots,n~,
\end{flalign}
where $T_a^i$ are constant and real matrices and $\vartheta$ is a smooth and real function. Note that 
$X_{2a-1}$ and $X_{2a}$ are parallel for all $a$. At first sight one might expect
that the deformation generated by these vector fields via the abelian twist is trivial,
since every two smooth functions $h,k$ commute, i.e.~$\starcom{h}{k}=0$. But this is not the case,
because the $\star$-product acts nontrivially between functions and one-forms.
In particular, we find for the basis one-forms $dx^\mu$ and a smooth function $h$
\begin{flalign}
dx^N\star h = dx^N\,h~,\quad dx^i\star h = dx^i\,h - dx^N \frac{i\lambda}{2}\vartheta^\prime(x^N) \sum\limits_{a=1}^n T_a^iT_a^j\partial_j h~,
\end{flalign}
where $\vartheta^\prime$ denotes the derivative of $\vartheta$. Note that this is an exact expression valid to all orders in
$\lambda$.

\subsection*{The deformed action:}
We evaluate the deformed action (\ref{eqn:defactioncobas}) for our model. 
The inverse metric in the deformed basis agrees with the inverse metric in the undeformed one, 
i.e.~$g_\star^{\mu\nu} =\delta^{\mu\nu}$.
The volume form is given by $\vols=\vol=dx^1\wedge\dots\wedge dx^N$ and satisfies
$\mathcal{L}_{X_\alpha}(\vols)=0$ for all $\alpha$.
Comparing both sides of $dx^\mu\,\partial_\mu h = dx^\mu\star\partial_{\star\mu}  h$ order by order
in the deformation parameter we find for the deformed partial derivatives
\begin{flalign}
 \partial_{\star N} = \partial_N +\frac{i\lambda}{2} \vartheta^\prime(x^N)\sum\limits_{a=1}^n T_a^i T_a^j\partial_i\partial_j
=: \partial_N +\frac{i\lambda}{2}\vartheta^\prime(x^N)\,\mathbb{T}~,\quad 
\partial_{\star i} =\partial_i~.
\end{flalign}
Plugging this into the action (\ref{eqn:defactioncobas}) we find for the free and massless action
\begin{flalign}
 S^\text{free}_\star[\Phi] = \frac{1}{2}\int\limits_{\bbR^N}\Bigl(\partial_\mu\Phi\, \partial^\mu\Phi + 
\frac{\lambda^2}{4}\vartheta^{\prime2}\,\mathbb{T}\Phi\,\mathbb{T}\Phi\Bigr)\,\vol~.
\end{flalign}
Let us now also introduce a $\star$-local potential $V_\star(\Phi)$. 
Since all $\star$-products between functions reduce to the pointwise products,
the potential term is undeformed, i.e.~$V_\star(\Phi) = V(\Phi)$.
For the deformed action of an interacting scalar field we thus find
\begin{flalign}
\label{eqn:euclidaction}
 S_\star[\Phi] = \int\limits_{\bbR^N}\Bigl(\frac{1}{2}\partial_\mu\Phi\,\partial^\mu\Phi + 
\frac{\lambda^2}{8}\vartheta^{\prime2}\,\mathbb{T}\Phi\,\mathbb{T}\Phi + V(\Phi)\Bigr)\,\vol~.
\end{flalign}
Note that this is an exact result, valid to all orders in $\lambda$.

\subsection*{{\boldmath $z{=}2$}  scalar field propagator:}
The action of a scalar field with a $z{=}2$ anisotropic scaling in the sense of \cite{Horava:2009uw}
can be obtained from (\ref{eqn:euclidaction}) by choosing
$n=N-1$, $T_a^i=\delta_a^i$ and $\vartheta(x^N) = 2\gamma\,x^N$, where the factor $2$ is for later convenience.
In this case we have $\mathbb{T} = \partial_i\partial^i$ and, including a quadratic potential
term $V(\Phi) = \frac{M^2}{2}\Phi^2$, we find for the action (\ref{eqn:euclidaction})
\begin{flalign}
 S_\star[\Phi] = \frac{1}{2}\int\limits_{\bbR^N}\Bigl(\partial_\mu\Phi\,\partial^\mu\Phi + 
\beta^2~\partial_i\partial^i\Phi\,\partial_j\partial^j\Phi + M^2 \,\Phi^2\Bigr)\,\vol~,
\end{flalign}
where $\beta^2 = \lambda^2\gamma^2 \geq 0$.
The propagator denominator in momentum space corresponding to this action is given by
\begin{flalign}
\label{eqn:z=2propa}
 E^2 +\mathbf{k}^2 +\beta^2\,\mathbf{k}^4 +M^2~,
\end{flalign}
where $E/k^i$ are the momenta associated to $x^N/x^i$.

This result provides an explanation for anisotropic propagators in the sense of \cite{Horava:2009uw}
from noncommutative geometry. The noncommutativity scale $\lambda$ and the slope $\gamma$ of $\vartheta(x^N)$
set the scale $\beta$ of the propagator modification, which we assume to be the Planck scale.
At small momenta, the propagator can be approximated by $E^2+\mathbf{k}^2 +M^2$, leading
to a theory which is approximately invariant under the Euclidean group $ISO(N)$,
up to $1/M_\text{pl}$-corrections. 

\subsection*{One-loop structure of the deformed {\boldmath $\Phi^4$}-theory:}
In order to check if the approximate $ISO(N)$-symmetry of our theory is also present
at the quantum (one-loop) level, we study the potential $V(\Phi) = \frac{M^2}{2}\Phi^2 + \frac{g_4}{4!}\Phi^4$
in $N=4$ dimensions. Remember that for our deformation all $\star$-products drop out of the potential.
As a consequence, we do not obtain the infamous UV/IR-mixing. The inverse propagator 
of our model is given by (\ref{eqn:z=2propa}). Due to the additional $\mathbf{k}^4$-term 
the divergences are tamed and we find for the one-loop correction $\Pi$ to the $2$-point function
\begin{flalign}
\label{eqn:twopointeqft}
 \Pi\big\vert_{\Lambda_E/\Lambda_k=\text{const.}} \,=-\frac{g_4}{\sqrt{8}\pi^2}\sqrt{\frac{\Lambda_E}{\beta^3}} + \text{finite}~,
\end{flalign}
where $\Lambda_E$ and $\Lambda_k$ denote the momentum space cutoff for $E$ and $k=\sqrt{k_i k^i}$, respectively.
The ratio $\Lambda_E/\Lambda_k$ is kept fixed for the limit $\Lambda_E, \Lambda_k\to\infty$.
The one-loop corrections to the $4$-point function are found to be finite.
Thus, due to our noncommutative deformation we have obtained an improved one-loop structure
of the $\Phi^4$-theory. On top of that, the Landau pole, which arises through the renormalization
group running of the coupling $g_4$ in the commutative case, is absent in our noncommutative theory at the
one-loop level due to the finiteness of the $4$-point function. Power-counting suggests that this
result extends to higher orders in the perturbation theory.
The result (\ref{eqn:twopointeqft}) shows that the $SO(4)$-violating $\mathbf{k}^4$-term receives no
renormalization at the one-loop level, which means that the approximate $ISO(4)$-symmetry
at small momenta is still present in the perturbatively interacting quantum field theory. This is an advantage compared
to the $\Phi^4$-theory on the Euclidean Moyal-Weyl space, since there one-loop effects lead to
infrared modifications of the propagator.

An interesting project for future research would be to find out if there are similar improvements
for Yang-Mills and perturbative quantum gravity theories when deformed by this special choice of twist.
Furthermore, comparing the resulting theory to Ho{\v r}ava's gravity theory \cite{Horava:2009uw}
 and understanding similarities and differences would be interesting.



\chapter{\label{chap:qftproblems}Open problems}
We have proven in Chapter \ref{chap:qftdef} that, in the framework of formal deformation quantization,
the construction of a scalar quantum field theory on a large class of noncommutative curved spacetimes 
is possible. We were able to treat explicit models of field and quantum field theories
on noncommutative curved spacetimes, see Chapter \ref{chap:qftapp}, where in some cases we were even
able to go over to a convergent setting for the deformations.
Because of this success, the open problems in quantum field theory on noncommutative curved spacetimes
are not concerned with the basic formalism, but they include specific applications and extensions
of it.
\subsection*{Explicit investigations in cosmology and black hole physics:}
The examples of deformed wave operators discussed in Chapter \ref{chap:qftapp}, Section \ref{sec:exampleswaveop},
provide interesting models for noncommutative cosmology and black hole physics.
These models are distinct from the
usual Moyal-Weyl deformation and they are exact solutions of the noncommutative Einstein equations.
Making use of the general formalism of Chapter \ref{chap:qftdef} we would be able to consider
perturbative noncommutative geometry effects in the cosmological power spectrum 
of a scalar field or the Hawking radiation in presence of a noncommutative black hole.
However, for a full understanding of the noncommutative geometry effects 
in these models, including possible nonlocalities, we have to take into account 
all orders in the deformation parameter $\lambda$. For toy-models in the class of
homothetic Killing deformations, we were able to do so, see Chapter \ref{chap:qftapp}, 
Section \ref{sec:homothetic}, and we have found that the convergently deformed
quantum field theory differs drastically from the formally deformed one.
Trying to extend these investigations to more realistic models of noncommutative cosmologies and
black holes, like e.g.~the models in Chapter \ref{chap:qftapp}, Section \ref{sec:exampleswaveop},
is an important topic for future research. For the quantum field theory on the
isotropically deformed de Sitter universe we have done relevant steps towards this goal 
in Chapter \ref{chap:qftapp}, Section \ref{sec:isotrop}.

\subsection*{Extension to spinor and gauge fields:}
As usual in quantum field theory on curved spacetimes, we started with
the simplest model of a free scalar field.
However, the constituents of the standard model of particle physics are spinor and gauge fields,
and the only scalar field is the up to now undetected Higgs.
This motivates to extend our formalism to include fermions and gauge fields.
For Dirac fermions, we expect that this extension is straightforward
by combining methods from noncommutative vielbein gravity \cite{Aschieri:2009ky,Aschieri:2009mc}
with the Dirac field on commutative curved spacetimes, see e.g.~\cite{Dappiaggi:2009xj}.
Nevertheless, the explicit construction along the lines presented in Chapter \ref{chap:qftdef}
 is a nontrivial and important topic for future research.

For gauge fields this extension turns out to be more complicated. To explain why, we consider
the simple case of $\MM=\bbR^4$ deformed by the Moyal-Weyl twist.
Let $g=dx^\mu\otimes_{A_\star} g_{\mu\nu}\star dx^\nu$ be a metric field, which is not invariant under the twist.
A possible deformation of the $U(1)$ action functional is given by
\begin{flalign}
 S_\star[A]= -\frac{1}{4}\int\limits_{\bbR^N} F_{\mu\nu}\star g^{\mu\rho}\star g^{\nu\sigma}\star F_{\rho\sigma}\star \vols~,
\end{flalign}
where $g^{\mu\nu}$ is the $\star$-inverse matrix of $g_{\mu\nu}$ and 
$F_{\mu\nu} = \partial_\mu A_\nu -\partial_\nu A_\mu - e \,\starcom{A_\mu}{A_\nu}$ is the field strength.
Under $U(1)$ gauge transformations the field strength transforms in the $\star$-adjoint, i.e.~
\begin{flalign}
 \delta_\epsilon F_{\mu\nu} = -e\,\starcom{F_{\mu\nu}}{\epsilon}~.
\end{flalign}
For the action to be invariant under gauge transformations, we also have to transform the inverse metric
and volume form in the $\star$-adjoint of $U(1)$,~i.e.
\begin{flalign}
\label{eqn:metricgauge}
 \delta_{\epsilon}g^{\mu\nu} = -e\,\starcom{g^{\mu\nu}}{\epsilon}~,\quad \delta_\epsilon\vols = -e\,\starcom{\vols}{\epsilon}~.
\end{flalign}
Fixing a metric field which is noncentral, i.e.~$\starcom{g^{\mu\nu}}{\epsilon}\neq 0$,
then leads to a broken gauge symmetry. This problem does not occur on the Moyal-Weyl deformed
Minkowski spacetime, since there the metric is central.
Let us give a speculative interpretation of this feature. It seems that 
in noncommutative geometry gauge and gravitational degrees of freedom are
coupled via $\star$-gauge invariance in such a way that 
separating the gauge degrees of freedom is in general not possible. 
As a consequence, we have to focus on situations where both, the gauge and metric field, are dynamical.
This indicates that there is a fundamental connection between noncommutative gauge theory and gravity,
see also \cite{Szabo:2006wx,Steinacker:2010rh} and references therein.

From the practical point of view, the observation above shows that discussing
gauge theories on fixed noncommutative curved spacetimes is in general not possible
due to the breaking of gauge invariance by the metric. One way out is to allow
for metric perturbations, i.e.~gravitons, which are dynamical and transform under the $\star$-gauge transformations.
The gauge field-graviton system then can be analyzed perturbatively.
However, in order to do so we have to improve our understanding of noncommutative gravity
and in particular noncommutative gravitons, see Chapter \ref{chap:ncgproblems}.

\subsection*{Perturbative interactions:}
We have explicitly shown in Chapter \ref{chap:qftapp}, Section \ref{sec:homothetic}, that
convergent deformations can lead to an improved ultraviolet behavior of the Green's operators and $2$-point functions.
The hyperbolic cosine factor in (\ref{eqn:def2point}) leads to an exponential drop off in certain directions in spatial momentum space
and thus tames the singularities in the $2$-point function.
It would be very interesting to study how this ultraviolet improvement of the free field
theory modifies the structure of the quantum field theory in presence of perturbative interactions.
Of particular interest are investigations proving the presence or absence of
UV/IR-mixing in these models.
The Yang-Feldman formalism \cite{Yang:1950vi,Bahns:2002vm,Zahn:2010yc} to perturbative quantum field theory 
is expected to be suitable for these studies.


\part{\label{part:math}Noncommutative Vector Bundles, Homomorphisms and Connections}


\chapter{Motivation and outline}
The fundamental object in noncommutative geometry is a noncommutative algebra $A$,
which we interpret as the algebra of functions on the quantized manifold.
Diffeomorphisms of the noncommutative manifold are described in terms of a
Hopf algebra $H$, which may also carry additional structures, such as an $R$-matrix.
A vector bundle on the noncommutative manifold is a module $V$ over the algebra $A$,
which we interpret as the module of sections of the quantized vector bundle, i.e.~vector fields.
All algebras and modules are assumed to transform covariantly under the Hopf algebra
$H$, a property which we have called ``deformed general covariance'', and all operations 
on algebras and modules shall respect this covariance.

The main aim of this part is to understand important classes of what was called ``operations'' above
 in the deformed covariant setting. We shall focus, in particular, on 
the endomorphisms of a module $V$, the homomorphisms between two modules $V$ and $W$ and 
connections on a module $V$. This part is based on ongoing work with Paolo Aschieri [that appeared after finishing the thesis
 in \cite{AlexPaolo}].

Let us explain why these structures are important for noncommutative gravity and
also Yang-Mills theory.
Modules of primary interest in (noncommutative) gravity are
the (quantized) vector fields $\Xi$ and  one-forms $\Omega^1$.
These modules are defined to be dual to each other, i.e.~vector fields are homomorphisms
from $\Omega^1$ to the algebra $A$. A metric field can be seen as an invertible
 homomorphism between $\Xi$ and $\Omega^1$, subject to reality, symmetry or hermiticity conditions.
A covariant derivative is modeled in terms of a connection on $\Xi$ or $\Omega^1$.
In Yang-Mills theory the focus is on a representation module $V$ of the gauge group, e.g.~quark fields
being in the fundamental representation of $SU(3)$. The gauge field is described in terms of a connection
on this module and gauge transformations are module endomorphisms of $V$. The field strength is a homomorphism
between $V$ and $V\otimes_A \Omega^2$.

A nontechnical outline of the content of this part is as follows: In Chapter \ref{chap:prelim} we 
fix the notation and provide a definition of the basic algebraic objects appearing in our studies.
We review standard properties of these algebraic structures and state precisely what 
it means for an algebra or module to be covariant under the action of a Hopf algebra.

In Chapter \ref{chap:HAdef} we first review well-known results on the deformation of a Hopf algebra
and its modules by a Drinfel'd twist. In the second part of this chapter we focus on
a class of algebras carrying a particular structure, namely in addition of being covariant under
the Hopf algebra $H$ we also have a compatible right module structure with respect to the algebra structure of $H$.
It is shown that algebras of this type appear quite naturally,
in particular, the endomorphism algebra of a module and the Hopf algebra itself are of this type.
We prove that under these assumptions the deformed algebra is isomorphic to the undeformed one.

The focus of Chapter \ref{chap:modhom} is on module endomorphisms and homomorphisms.
We prove that the algebra of endomorphisms of a twist quantized module
is isomorphic to a quantization of the algebra of endomorphisms of the undeformed module.
This, in particular, gives us a prescription of how to quantize a classical endomorphism
and moreover tells us that all endomorphisms of the quantized module can be obtained by this
quantization prescription.
We extend the results to homomorphisms between two modules.
As an application we show that the module obtained by quantizing the dual of a module is isomorphic to
the one obtained by dualizing the quantized module, i.e.~that there are no ambiguities in defining
the dual of a module in a twist quantized setting.
Motivated by the example of noncommutative gravity we focus on algebras and modules
which are commutative up to the action of an $R$-matrix. In this setting, we prove that
there is an isomorphism between the homomorphisms respecting the right module structure and
homomorphisms respecting the left module structure.
Similar to the commutative case, this result allows us to restrict our studies to either right or left homomorphisms, 
since the other ones can be obtained canonically by this isomorphism.
We finish this chapter by investigating how homomorphisms can be extended to yield homomorphisms
on the tensor product of two modules. 
The quantization of the product of two homomorphisms is discussed
and we obtain that it is (up to isomorphisms) equal to the deformed product of the quantized homomorphisms.
Thus, there are no ambiguities in constructing product module homomorphisms in a twist quantized setting.

In Chapter \ref{chap:con} we focus on connections on modules.  We first provide a precise definition
of a connection on a left or right module. We then prove that
all connections on a twist quantized module can be obtained by quantizing connections 
on the undeformed module. Motivated again by the example of noncommutative gravity, we then focus
on algebras and modules which are commutative up to an $R$-matrix and show that there is an isomorphism
between the connections satisfying the left Leibniz rule and the connections satisfying the right Leibniz rule.
Thus, analogously to the homomorphisms above, we can restrict our studies to either left or right connections.
As a next step we investigate the extension of connections to the tensor product of two modules and to the dual
of a module.
This study is essential for noncommutative gravity since it allows us to express the connection
on tensor fields in terms of a fundamental connection on either vector fields or one-forms.
We show that the extension of connections to tensor products is compatible with twist quantization
in the sense that the quantization of the extended connection is (up to isomorphisms) equal to
the deformed extension of the quantized connection.


We focus in Chapter \ref{chap:curvature} on the curvature and torsion of connections.
In particular, we express the curvature and torsion of a connection in the twist quantized setting
in terms of the corresponding undeformed connection.

In Chapter \ref{chap:ncgmath} we reinvestigate noncommutative gravity solutions
in the new formalism. We can extend the results obtained before to 
exotic abelian deformations and, even more, general Drinfel'd twists.
This provides a nontrivial application of our formalism.

In Chapter \ref{chap:outlookmath} we discuss problems which are still unsolved
and give an outlook to possible applications of the formalism developed in this part.

The Appendix \ref{app:symbols} contains a list of the mathematical symbols appearing in this part.


\chapter{\label{chap:prelim}Preliminaries and notation}
In this chapter we fix the notation and recall some basic facts about Hopf algebras
and their modules. 
The intention is to provide a collection of the algebraic structures relevant for this part.
This is why we present this chapter in a rather dry and minimalistic language, focusing only
mathematical aspects.
We refer the reader to Chapter \ref{chap:basicncg} for a more elementary introduction to Hopf algebras based on examples
relevant in physics.
This chapter also contains more explanations why the structures defined below are important
for noncommutative geometry and in particular noncommutative gravity.

Let $\bfK$ be a commutative unital ring.\footnote{
In noncommutative gravity the ring $\bfK$ will be $\bbC[[\lambda]]$ or $\bbR[[\lambda]]$.
} A {\it $\bfK$-module} $A$ is an abelian group with an action
$\bfK\times A \to A$, $(\beta,a)\mapsto \beta\,a$ of the ring on the module, such that
for all $\beta,\gamma\in\bfK$ and $a,b\in A$,
\begin{flalign}
 (\beta\gamma)\,a = \beta\,(\gamma\,a)~,\quad \beta\,(a+b)=\beta\,a +\beta\,b~,\quad (\beta+\gamma)\,a = \beta\,a +\gamma\,a~,
\quad 1\,a=a~.
\end{flalign}
Note that in the special case where $\bfK$ is a field, a $\bfK$-module is typically called a vector space.
A {\it $\bfK$-module homomorphism} (or {\it $\bfK$-linear map}) $\varphi:A\to B$ between the $\bfK$-modules $A$ and $B$
is a homomorphism of the abelian groups that satisfies, for all $a\in A$ and $\beta\in \bfK$,
$\varphi(\beta\,a)=\beta\,\varphi(a)$. 

An {\it algebra $A$ over $\bfK$} is a module over the ring $\bfK$
together with a $\bfK$-linear map $\mu:A\otimes A\to A$ (product). 
We denote by $a\otimes b$ the image of $(a,b)$ under the natural $\bfK$-bilinear map $A\times A\to A\otimes A$
and write for the product $\mu(a\otimes b) = a\,b$.
The algebra is called {\it associative}, if the following diagram commutes
\begin{flalign}
\xymatrix{
~A\otimes A\otimes A~ \ar[d]_-{{\mu\otimes\id}} \ar[r]^-{{\id\otimes\mu}}  &  A\otimes A  \ar[d]^-{\mu}\\ 
A\otimes A \ar[r]_-{\mu}  &  A
}
\end{flalign}
For $a\otimes b\otimes c\in A\otimes A\otimes A$ this means that
\begin{flalign}
 (a\,b)\,c = a\,(b\,c)~.
\end{flalign}
The algebra $A$ is {\it unital}, if there exists a $\bfK$-linear map $\mathbf{e}:\bfK\to A$ (unit)
satisfying
\begin{flalign}
\xymatrix{
A\otimes A  \ar[dr]^-{\mu}&  \\
\bfK\otimes A \ar[u]^-{\mathbf{e}\otimes\id} \ar[r]_-{\simeq} & A \ar[l]
}\qquad
\xymatrix{
A\otimes A  \ar[dr]^-{\mu}&  \\
A \otimes \bfK\ar[u]^-{\id\otimes \mathbf{e}} \ar[r]_-{\simeq} & A \ar[l]
}
\end{flalign}
We denote by $1:=\mathbf{e}(1)\in A$ the element satisfying $1\,a = a\,1 = a$, for all $a\in A$.
Algebras will always be associative and unital if not stated otherwise.

``Reversing the arrows'' in the commutative diagrams above yields a coalgebra.
A {\it coalgebra} $A$ is a $\bfK$-module together with a $\bfK$-linear map
$\Delta: A\to A\otimes A$ (coproduct) and a $\bfK$-linear map 
$\epsilon:A\to \bfK$ (counit) satisfying
\begin{subequations}
\begin{flalign}
\label{eqn:haprop1}&\xymatrix{
A \ar[r]^-{\Delta}  \ar[d]_-{\Delta}&  A\otimes A \ar[d]^-{\Delta\otimes \id}\\
A\otimes A \ar[r]_-{\id\otimes\Delta} & A\otimes A\otimes A
}
\end{flalign}
\begin{flalign}
\label{eqn:haprop2}\xymatrix{
A\otimes A \ar[d]_-{\epsilon\otimes\id} & \\
\bfK\otimes A \ar[r]_-{\simeq}  &  A \ar[ul]_-{\Delta} \ar[l]
}\qquad
\xymatrix{
A\otimes A \ar[d]_-{\id\otimes \epsilon} & \\
A\otimes\bfK \ar[r]_-{\simeq}  &  A \ar[ul]_-{\Delta} \ar[l]
}
\end{flalign}
\end{subequations}
It is useful to introduce a compact notation (Sweedler's notation) for the coproduct,
for all $a\in A$, $\Delta(a)=a_1\otimes a_2$ (sum understood). 
The diagrams (\ref{eqn:haprop1}) and  (\ref{eqn:haprop2}) in this notation read, for all $a\in A$,
\begin{subequations}
\begin{flalign}
 a_{1_1}\otimes a_{1_2}\otimes a_2 &= a_1\otimes a_{2_1} \otimes a_{2_2} =: a_1\otimes a_2\otimes a_3~,\\
 \epsilon(a_1)\,a_2 &= a_1\,\epsilon(a_2) =a~.
\end{flalign}
\end{subequations}
We denote the three times iterated application of the coproduct on  $a\in A$ by
$a_1\otimes a_2\otimes a_3\otimes a_4$. 

A {\it bialgebra} $A$ is an algebra and coalgebra satisfying the compatibility conditions
 \begin{subequations}
 \begin{flalign}
 \xymatrix{
 A\otimes A \ar[rrr]^-{\mu}  \ar[d]_-{\Delta\otimes\Delta} & & & A \ar[d]^-{\Delta} \\
 A\otimes A\otimes A\otimes A \ar[rrr]_-{(\mu\otimes\mu) \circ (\id\otimes\tau\otimes\id)}& & &A\otimes A
 }
 \end{flalign}
 \begin{flalign}
 \xymatrix{
A \ar[r]^-{\epsilon} & \bfK \\
A\otimes A \ar[u]_-{\mu} \ar[ur]_-{\epsilon\otimes\epsilon} & 
 }\qquad
 \xymatrix{
 \bfK \ar[r]^-{\mathbf{e}}  \ar[dr]_-{\mathbf{e}\otimes\mathbf{e}}& A \ar[d]^-{\Delta}\\
 & A\otimes A
 }
 \end{flalign}
\end{subequations}
In the first diagram $\tau$ denotes the flip map $\tau(a\otimes b)=b\otimes a$.
On the level of elements the conditions read, for all $a,b\in A$,
\begin{subequations}
\begin{flalign}
 \Delta(a\,b) &= a_1\, b_1\otimes a_2\,b_2~,\\
 \epsilon(a\,b) &= \epsilon(a)\,\epsilon(b)~,\quad \Delta(1) = 1\otimes 1~.
\end{flalign}
\end{subequations}

\begin{defi}
 A {\it Hopf algebra} $H$ is bialgebra together with a $\bfK$-linear map $S: H\to H$ (antipode) satisfying
 \begin{flalign}\label{eqn:haprop3} 
 \xymatrix{H\ar[r]^-{\epsilon} \ar[d]_-{\Delta} & \bfK \ar[r]^{\mathbf{e}} & H\\
 H\otimes H \ar[rr]_-{\id\otimes S~,~S\otimes \id} & & H\otimes H \ar[u]^-{\mu}
 }
 \end{flalign}
On the level of elements the condition reads, for all $\xi\in H$,
\begin{flalign}
 S(\xi_1)\,\xi_2 & = \xi_1\,S(\xi_2) = \epsilon(\xi)\,1~.
\end{flalign}
\end{defi}
It can be shown that the antipode of a Hopf algebra
is unique and satisfies
the following standard properties
\begin{subequations}
\begin{flalign}
S\circ \mu  = \mu\circ (S\otimes S)\circ\tau~&,\quad S\circ \mathbf{e} =\mathbf{e}~ ,\\
(S\otimes S)\circ \Delta  =\tau\circ\Delta \circ S ~&,\quad \epsilon \circ S = \epsilon~,
\end{flalign}
\end{subequations}
which on the level of elements read, for all $\xi,\eta\in H$,
\begin{subequations}
\begin{flalign}
 S(\xi\,\eta) = S(\eta)\,S(\xi)~&,\quad S(1) =1~,\\
 S(\xi_1)\otimes S(\xi_2) = S(\xi)_2\otimes S(\xi)_1 ~&,\quad \epsilon(S(\xi)) = \epsilon(\xi)~.
\end{flalign}
\end{subequations}

\begin{defi}
 A {\it left module $V$ over an algebra $A$} (i.e.,~a {\it left $A$-module}) is a $\bfK$-module
with a $\bfK$-linear map $\cdot:A\otimes V\to V$ satisfying
\begin{flalign}
\xymatrix{
~A\otimes A\otimes V~ \ar[d]_-{{\mu\otimes\id}} \ar[r]^-{{\id\otimes\cdot}}  &  A\otimes V \ar[d]^-{\cdot}\\ 
A\otimes V \ar[r]_-{\cdot}  &  V
}\qquad 
\xymatrix{
A\otimes V  \ar[dr]^-{\cdot}&  \\
\bfK\otimes V \ar[u]^-{\mathbf{e}\otimes\id} \ar[r]_-{\simeq} & V \ar[l]
}
\end{flalign}
Using the notation $a\cdot v = \cdot(a\otimes v)$ the conditions read on the level of elements, 
for all $a,b\in A$, $v\in V$,
\begin{flalign}
 a\cdot(b\cdot v) = (a\,b)\cdot v~,\quad 1\cdot v = v~.
\end{flalign}
The map $\cdot:A\otimes V\to V$ is called an {\it action} of $A$ on $V$
or a {\it representation} of $A$ on $V$. The class of left $A$-modules is denoted by $_A\MMM$.
\end{defi}
Analogously, a {\it right $A$-module} $V$ is a $\bfK$-module with a $\bfK$-linear map
$\cdot:V\otimes A\to V$ satisfying
\begin{flalign}
\xymatrix{
~V\otimes A\otimes A~ \ar[d]_-{{\cdot\otimes\id}} \ar[r]^-{{\id\otimes\mu}}  &  V\otimes A \ar[d]^-{\cdot}\\ 
V\otimes A \ar[r]_-{\cdot}  &  V
}\qquad 
\xymatrix{
V\otimes A  \ar[dr]^-{\cdot}&  \\
V\otimes \bfK \ar[u]^-{\id\otimes \mathbf{e}} \ar[r]_-{\simeq} & V \ar[l]
}
\end{flalign}
On the level of elements we have, for all  $a,b\in A$ and $v\in V$,
\begin{flalign}
 v\cdot(a\,b) = (v\cdot a)\cdot b~,\quad v\cdot 1 = v~.
\end{flalign}
The class of right $A$-modules is denoted by $\MMM_A$.

A left and a right module structure on $V$ are compatible if the left and right actions commute:
\begin{defi}
An {\it $(A,B)$-bimodule} $V$ is a left $A$-module and a right $B$-module
($A$ and $B$ are algebras over $\bfK$) satisfying the compatibility condition
\begin{flalign}
\xymatrix{
A\otimes V\otimes B \ar[r]^-{\cdot\otimes \id}  \ar[d]_-{\id\otimes\cdot}& V\otimes B \ar[d]^-{\cdot}\\
A\otimes V \ar[r]_-{\cdot} & V
}
\end{flalign}
On the level of elements we have, for all $a\in A$, $b\in B$ and $v\in V$, 
\begin{flalign}
 (a\cdot v)\cdot b = a\cdot(v\cdot b)~.
\end{flalign}
The class of $(A,B)$-bimodules is denoted by $_A\MMM_B$.

\noindent In case of $B=A$ we call $V$ an {\it $A$-bimodule}. We denote the class of
$A$-bimodules by $_A\MMM_A$.
\end{defi}

The algebra $A$ can itself be a module over another algebra $H$. If $H$ is further a Hopf algebra
we have the notion of an $H$-module algebra, expressing covariance of $A$ under $H$.
\begin{defi}\label{defi:hmodalg}
 Let $H$ be a Hopf algebra. A {\it left $H$-module algebra} is an algebra $A$
which is also a left $H$-module (where the left action is denoted by $\ra$), such that
\begin{flalign}
\xymatrix{
H\otimes A\otimes A \ar[r]^-{\id\otimes\mu} \ar[d]_-{\Delta\otimes\id\otimes\id}&H\otimes A \ar[r]^-{\ra} & A\\
H\otimes H\otimes A\otimes A \ar[rr]_-{(\ra\otimes\ra)\circ (\id\otimes\tau\otimes \id)} & & A\otimes A \ar[u]_-{\mu}
}\qquad
\xymatrix{
H\otimes\bfK \ar[r]^-{\id\otimes \mathbf{e}} \ar[d]_-{\epsilon\otimes \id} & H\otimes A \ar[d]^-{\ra}\\
\bfK \ar[r]_-{\mathbf{e}} & A
}
\end{flalign}
On the level of elements we have, for all $\xi\in H$ and $a,b\in A$,
\begin{flalign}
 \xi\ra (a\,b) = (\xi_1\ra a)\,(\xi_2\ra b)~,\quad \xi\ra 1 = \epsilon(\xi)\,1~.
\end{flalign}
The class of left $H$-module algebras is denoted by $^{H,\ra}\AAA$. 

\noindent A {\it left $H$-module algebra homomorphism $\varphi:A\to B$}
is an algebra homomorphism that intertwines between the left action of $H$ on $A$
and the left action of $H$ on $B$, i.e.~
\begin{flalign}
\xymatrix{
 A\otimes A \ar[r]^-{\varphi\otimes \varphi} \ar[d]_-{\mu}& B\otimes B \ar[d]^-{\mu}\\
 A \ar[r]_-{\varphi} & B
} \qquad 
\xymatrix{
\bfK \ar[r]^-{\mathbf{e}} \ar[dr]_-{\mathbf{e}}& A \ar[d]^-{\varphi}\\
& B
}\qquad
\xymatrix{
H\otimes A \ar[r]^-{\ra} \ar[d]_-{\id\otimes\varphi}& A \ar[d]^-{\varphi}\\
H\otimes B \ar[r]_-{\ra} & B
}
\end{flalign}
or on the level of elements, for all $a,b\in A$ and $\xi\in H$,
\begin{flalign}
 \varphi(a\,b) = \varphi(a)\,\varphi(b)~,\quad\varphi(1) =1~,\quad \varphi(\xi\ra a) = \xi\ra\varphi(a)~.
\end{flalign}

\end{defi}
We denote, for all $\xi\in H$, the endomorphism $A\to A,~a\mapsto \xi\ra a$ by $\xi\ra\in\End_\bfK(A)$.
Furthermore, we denote the action of $H\otimes H$ on $A\otimes A$ also by $\ra$, i.e.~
\begin{flalign}
\label{eqn:tensoraction}
 \ra= (\ra\otimes\ra)\circ (\id\otimes\tau\otimes \id) : H\otimes H\otimes A\otimes A \to A\otimes A~.
\end{flalign}

We can now consider $(A,B)$-bimodules $V$, where $A,B\in{^{H,\ra}}\AAA$ and $V$ is also
a left $H$-module. Compatibility between the Hopf algebra structure
of $H$ and the $(A,B)$-bimodule structure of $V$ leads to the following covariance requirement
\begin{defi}\label{defi:hmod}
A {\it left $H$-module $(A,B)$-bimodule} $V$ is an $(A,B)$-bimodule over
$A,B\in {^{H,\ra}}\AAA$ which is also a left $H$-module, such that 
\begin{subequations}\label{eqn:HABbimodule}
\begin{flalign}\label{eqn:HABbimodule1}
\xymatrix{
H\otimes A\otimes V \ar[r]^-{\id\otimes\cdot} \ar[d]_-{\Delta\otimes\id\otimes\id} &H\otimes V \ar[r]^-{\ra}  & V\\
H\otimes H\otimes A\otimes V \ar[rr]_-{(\ra\otimes\ra)\circ(\id\otimes\tau\otimes \id)} & & A\otimes V \ar[u]_-{\cdot}
}
\end{flalign}
\begin{flalign}\label{eqn:HABbimodule2}
\xymatrix{
H\otimes V\otimes B \ar[r]^-{\id\otimes\cdot} \ar[d]_-{\Delta\otimes\id\otimes\id}& H\otimes V \ar[r]^-{\ra} & V\\
H\otimes H\otimes V\otimes B \ar[rr]_-{(\ra\otimes\ra)\circ(\id\otimes\tau\otimes \id)} & & V\otimes B \ar[u]_-{\cdot}
}
\end{flalign}
\end{subequations}
or on the level of elements, for all $a\in A$, $b\in B$ and $v\in V$,
\begin{subequations}
 \begin{flalign}
 \xi\ra (a\cdot v) &= (\xi_1\ra a)\cdot(\xi_2\ra v)~,\\
 \xi\ra (v\cdot b) &= (\xi_1\ra v)\cdot (\xi_2\ra b)~.
 \end{flalign}
\end{subequations}
The class of left $H$-module $(A,B)$-bimodules is denoted by ${^{H,\ra}_A}\MMM_B$.

\noindent In case of $B=A$ we say that $V$ is a {\it left $H$-module $A$-bimodule}
 and denote the corresponding class by ${^{H,\ra}_A}\MMM_A$.

\noindent An algebra $E$ is a {\it left $H$-module $(A,B)$-bimodule algebra}, if $E$ as a module is a left $H$-module
$(A,B)$-bimodule and if $E$ is also a left $H$-module algebra. The class of left $H$-module $(A,B)$-bimodule
algebras is denoted by ${^{H,\ra}_A}\AAA_B$.
\end{defi}
The classes ${^{H,\ra}_A}\MMM$ and $^{H,\ra}\MMM_B$ are defined analogously to ${^{H,\ra}_A}\MMM_B$, where
(\ref{eqn:HABbimodule}) is restricted to (\ref{eqn:HABbimodule1}) or (\ref{eqn:HABbimodule2}), respectively.

\begin{ex}
 Consider the universal enveloping algebra $U\Xi$ associated with the Lie algebra of vector fields
$\Xi$ on a smooth manifold $\MM$. $U\Xi$ has a natural Hopf algebra structure as explained
in Chapter \ref{chap:basicncg}.

Let $V$ be the space of one-forms $\Omega^1$ (or vector fields $\Xi$) on $\MM$, and $A=C^\infty(\MM)$
be the algebra of smooth functions on $\MM$. $\Omega^1$ ($\Xi$) is a $C^\infty(\MM)$-bimodule,
where the right module structure equals the left module structure.

$C^\infty(\MM)$ is a left $U\Xi$-module algebra, where the action $\ra$ is simply the 
Lie derivative $\mathcal{L}$ on functions. Employing the Lie derivative on vector fields and one-forms,
we also have that $\Omega^1$ ($\Xi$) is a left $U\Xi$-module $C^\infty(\MM)$-bimodule, 
i.e.~$\Omega^1,\Xi\in {^{U\Xi,\mathcal{L}}_{C^\infty(\MM)}}\MMM_{C^\infty(\MM)}$.
\end{ex}


\chapter{\label{chap:HAdef}Hopf algebras, twists and deformations}
\section{Twist deformation preliminaries}
\begin{defi}\label{defi:twist}
 Let $H$ be a Hopf algebra. A {\it twist} $\mathcal{F}$ is an element $\mathcal{F}\in H\otimes H$ that
is invertible and that satisfies
\begin{subequations}
\label{eqn:twistprop}
\begin{flalign}
\label{eqn:twistprop1} \mathcal{F}_{12}\,(\Delta\otimes \id)\mathcal{F} &= \mathcal{F}_{23}\,(\id \otimes \Delta)\mathcal{F}~,\quad\text{(2-cocycle property)}\\
\label{eqn:twistprop2}  (\epsilon\otimes \id)\mathcal{F} &= 1= (\id \otimes \epsilon)\mathcal{F} ~,\quad\text{(normalization property)}~
\end{flalign}
\end{subequations}
where $\mathcal{F}_{12} =\mathcal{F}\otimes 1$ and $\mathcal{F}_{23}=1\otimes\mathcal{F}$.
\end{defi}
We shall frequently use the notation (sum over $\alpha$ understood)
\begin{flalign}
\label{eqn:twistnotation}
 \mathcal{F}=f^\alpha\otimes f_\alpha~,\quad \mathcal{F}^{-1} = \bar f^\alpha\otimes \bar f_\alpha~.
\end{flalign}
The $f^\alpha,f_\alpha,\bar f^\alpha,\bar f_\alpha$ are elements in $H$.

In order to get familiar with this notation we rewrite (\ref{eqn:twistprop1}), 
(\ref{eqn:twistprop2}) and the inverse of (\ref{eqn:twistprop1}), i.e.~the condition
\begin{flalign}
 \label{eqn:twistprop3} \bigl((\Delta\otimes\id)\mathcal{F}^{-1}\bigr)\,\mathcal{F}^{-1}_{12} = 
\bigl((\id\otimes\Delta)\mathcal{F}^{-1}\bigr)\,\mathcal{F}_{23}^{-1}~,
\end{flalign}
using the notation (\ref{eqn:twistnotation}). Explicitly,
\begin{subequations}
\label{eqn:twistpropsimp}
 \begin{flalign}
  \label{eqn:twistpropsimp1} f^\beta f^\alpha_1\otimes f_\beta f^\alpha_2\otimes f_\alpha &= f^\alpha\otimes f^\beta f_{\alpha_1}\otimes
f_\beta f_{\alpha_2}~,\\
\label{eqn:twistpropsimp2} \epsilon(f^\alpha)f_\alpha &=1 = f^\alpha \epsilon(f_\alpha)~,\\
\label{eqn:twistpropsimp3} \bar f^\alpha_1\bar f^\beta\otimes\bar f^\alpha_2\bar f_\beta \otimes \bar f_\alpha &= 
\bar f^\alpha\otimes \bar f_{\alpha_1}\bar f^\beta \otimes \bar f_{\alpha_2}\bar f_\beta~.
 \end{flalign}
\end{subequations}

We now recall how a twist $\mathcal{F}$ induces a deformation $H^\mathcal{F}$ of the Hopf algebra
$H$ and of all its $H$-modules, that become $H^\mathcal{F}$-modules. In particular, $H$-module algebras
are deformed into $H^\mathcal{F}$-module algebras, and commutative ones are typically deformed into noncommutative
ones. In this sense $\mathcal{F}$ induces a quantization.
\begin{theo}\label{theo:HAdef}
 The twist $\FF$ of the Hopf algebra $H$ leads to a new Hopf algebra $H^\FF$, given by
\begin{flalign}
 \bigl(H,\mu,\Delta^\FF,\epsilon, S^\FF\bigr)~.
\end{flalign}
As algebras $H^\FF=H$ and they also have the same counit $\epsilon^\FF=\epsilon$. The new coproduct $\Delta^\FF$
is given by, for all $\xi\in H$,
\begin{flalign}
 \label{eqn:defcoproduct}\Delta^\FF(\xi)=\mathcal{F}\,\Delta(\xi)\FF^{-1}~.
\end{flalign}
The new antipode is, for all $\xi\in H$,
\begin{flalign}
 \label{eqn:defantipode1} S^\FF(\xi) = \chi\,S(\xi)\,\chi^{-1}~,
\end{flalign}
where 
\begin{flalign}
 \label{eqn:defantipode2}\chi:=f^\alpha\,S(f_\alpha)~,\quad \chi^{-1}=S(\bar f^\alpha)\,\bar f_\alpha~.
\end{flalign}
\end{theo}
\noindent A proof of this theorem can be found in every textbook on Hopf algebras, e.g.~\cite{Majid:1996kd}.

Note: It is easy to show that the Hopf algebra $H^\FF$ admits the twist $\FF^{-1}$, indeed:
\begin{flalign}
 \FF_{12}^{-1}\,(\Delta^\FF\otimes \id)\FF^{-1} = \FF_{23}^{-1}\,(\id\otimes\Delta^\FF)\FF^{-1}
\end{flalign}
is equivalent to (\ref{eqn:twistprop3}). From (\ref{eqn:defcoproduct}), (\ref{eqn:defantipode1}) and (\ref{eqn:defantipode2})
we see that the Hopf algebra $(H^\FF)^{\FF^{-1}}$ is canonically isomorphic to $H$.
In this respect, we call the deformation $\FF^{-1}$ of $H^\FF$ dequantization.

\begin{theo}\label{theo:algebradef}
 Given a Hopf algebra $(H,\mu,\Delta,\epsilon,S)$, a twist $\FF\in H\otimes H$ and a left
$H$-module algebra $A$ (not necessarily associative or with unit), then there exists a deformed
left $H^\FF$-module algebra $A_\star$. The algebra $A_\star$ has the same $\bfK$-module structure as $A$
and the action of $H^\FF$ on $A_\star$ is that of $H$ on $A$. The product in $A_\star$ is given by
$\mu_\star := \mu\circ\FF^{-1}\ra : A\otimes A\to A$.
In the notation (\ref{eqn:twistnotation}) the deformed product reads, for all $a,b\in A$,
\begin{flalign}
 a\star b=(\bar f^\alpha\ra a)\,(\bar f_\alpha\ra b)~.
\end{flalign}
If $A$ has a unit then $A_\star$ has the same unit element. If $A$ is associative then
$A_\star$ is an associative algebra as well.
\end{theo}
\begin{proof}
 We have to prove that the product in $A_\star$ is compatible with the Hopf algebra structure on $H^\FF$,
for all $a,b\in A$ and $\xi\in H$,
\begin{flalign}
 \nn \xi\ra(a\star b) &= \xi\ra\bigl((\bar f^\alpha\ra a)\,(\bar f_\alpha\ra b)\bigr)\\
\nn &=(\xi_1\bar f^\alpha\ra a)\,(\xi_2\bar f_\alpha\ra b)\\
\nn &=(\bar f^\gamma f^\beta \xi_1\bar f^\alpha\ra a)\,(\bar f_\gamma f_\beta \xi_2\bar f_\alpha\ra b)\\
&=(\xi_{1_\FF}\ra a)\star(\xi_{2_\FF}\ra b)~,
\end{flalign}
where in line three we have inserted $1\otimes 1 = \FF^{-1}\,\FF$ and in line four
we have used the notation $\Delta^{\FF}(\xi)=\xi_{1_\FF}\otimes \xi_{2_\FF}$.

If $A$ has a unit element $1$, then $1\star a = a\star 1 =a$, for all $a\in A$, follows from the normalization property
(\ref{eqn:twistprop2}). If $A$ is an associative algebra we have to prove associativity of the new product,
for all $a,b,c\in A$,
\begin{flalign}
 \nn (a\star b)\star c &= \bar f^\alpha\ra\bigl((\bar f^\beta\ra a)\,(\bar f_\beta\ra b)\bigr)\,(\bar f_\alpha\ra c)\\
&=(\bar f^\alpha\ra a)\,\bar f_\alpha\ra\bigl((\bar f^\beta\ra b)\,(\bar f_\beta\ra c)\bigr)=a\star(b\star c)~,
\end{flalign}
where we have used the 2-cocycle property (\ref{eqn:twistprop1}) in the notation of (\ref{eqn:twistpropsimp3}).

\end{proof}

\begin{theo}\label{theo:moduledef}
 In the hypotheses of Theorem \ref{theo:algebradef}, given a left $H$-module $(A,B)$-bimodule $V\in {^{H,\ra}_A}\MMM_B$,
then there exists a left $H^\FF$-module $(A_\star,B_\star)$-bimodule $V_\star\in {^{H^\FF,\ra}_{A_\star}}\MMM_{B_\star}$.
The module $V_\star$ has the same $\bfK$-module structure as $V$ and the left action of $H^\FF$
on $V_\star$ is that of $H$ on $V$. The $A_\star$ and $B_\star$ action on $V_\star$ are respectively given by
$\star:= \cdot\circ \FF^{-1}\ra:A\otimes V \to V$ and $\star:= \cdot\circ \FF^{-1}:V\otimes B\to V$.
In the notation (\ref{eqn:twistnotation}) the deformed actions read, for all $a\in A$, $b\in B$ and $v\in V$,
\begin{subequations}
\label{eqn:moduledef}
\begin{flalign}
 \label{eqn:moduledef1}a\star v &=  (\bar f^\alpha\ra a)\cdot (\bar f_\alpha\ra v)~,\\
 \label{eqn:moduledef2}v\star b &=  (\bar f^\alpha\ra v)\cdot (\bar f_\alpha\ra b)~.
\end{flalign}
\end{subequations}
If $V$ is further a left $H$-module $(A,B)$-bimodule algebra $V=E\in {^{H,\ra}_A}\AAA_B$, 
then $E_\star\in {^{H^\FF,\ra}_{A_\star}}\AAA_{B_\star}$,
where the $\star$-product in the algebra $E_\star$ is given in Theorem \ref{theo:algebradef}.
\end{theo}
\begin{proof}
 The left $A_\star$-module property holds, for all $a,b\in A$ and $v\in V$,
\begin{flalign}
 \nn (a\star b)\star v &= \bar f^\alpha\ra\bigl((\bar f^\beta\ra a)\,(\bar f_\beta\ra b)\bigr)\cdot(\bar f_\alpha\ra v)\\
&=(\bar f^\alpha\ra a)\cdot \bar f_\alpha\ra\bigl((\bar f^\beta\ra b)\cdot (\bar f_\beta\ra v)\bigr)=a\star(b\star v)~.
\end{flalign}
The right $B_\star$-module property and the $(A_\star,B_\star)$-bimodule property are similarly proven.

\noindent Compatibility between the left $H^\FF$ and the left $A_\star$-action is shown by
\begin{flalign}
 \xi\ra (a\star v) = (\xi_1\bar f^\alpha\ra a)\cdot(\xi_2\bar f_\alpha\ra v) =  (\xi_{1_\FF}\ra a)\star (\xi_{2_\FF}\ra v)~.
\end{flalign}
Compatibility between the left $H^\FF$ and the right $B_\star$-action is shown analogously.
In case we have $V=E\in {^{H,\ra}_A}\AAA_B$, then $E_\star\in {^{H^\FF,\ra}_{A_\star}}\AAA_{B_\star}$ because of Theorem
\ref{theo:algebradef}.

\end{proof}
Analogously to Theorem \ref{theo:moduledef} we can deform ${^{H,\ra}_A}\MMM$ and ${^{H,\ra}}\MMM_B$ modules
into ${^{H^\FF,\ra}_{A_\star}}\MMM$ and ${^{H^\FF,\ra}}\MMM_{B_\star}$ modules by restricting (\ref{eqn:moduledef})
to (\ref{eqn:moduledef1}) or (\ref{eqn:moduledef2}), respectively.


\section{The quantization isomorphism $D_\FF$}
In general the algebras $A$ and $A_\star$ of Theorem \ref{theo:algebradef} are not isomorphic
algebras, if they are we name them $\bbA$ and $\bbA_\star$ and we denote their elements by $P,Q\in \bbA$.
\begin{theo}\label{theo:algebraiso}
 Let $\bbA$ be a left $H$-module algebra (not necessarily associative or with unit) and also a right
module with respect to the algebra $(H,\mu)$ (the right action of $(H,\mu)$ on $\bbA$
is simply denoted by juxtaposition), with the compatibility conditions, for all $P,Q\in\bbA$ and $\xi,\eta\in H$,
\begin{subequations}
\label{eqn:algebraiso}
\begin{flalign}
 \label{eqn:algebraiso1} (PQ)\xi &= P(Q\xi)~,\\
\label{eqn:algebraiso2} (P\xi)Q  &= P(\xi_1\ra Q)\xi_2~,\\
\label{eqn:algebraiso3} \xi\ra(P\eta) &= (\xi_1\ra P)(\xi_2\ra \eta)~,
\end{flalign}
\end{subequations}
where $\xi\ra \eta = \xi_1\,\eta\,S(\xi_2)$ is the adjoint action of $H$ on $H$.
In this case the algebras $\bbA$ and $\bbA_\star$ are isomorphic algebras.
\end{theo}
Before proving this theorem notice that $\bbA$ in the hypotheses above is a left module
with respect to the algebra $(H,\mu)$ by defining, for all $\xi\in H$ and $P\in\bbA$,
\begin{flalign}
 \xi P := (\xi_1\ra P)\xi_2~.
\end{flalign}
Indeed, we have for all $\xi,\eta\in H$ and $P\in \bbA$
\begin{flalign}
\nn \xi(\eta P) &= \xi\bigl((\eta_1\ra P) \eta_2\bigr)= \Bigl(\xi_1\ra\bigl((\eta_1\ra P) \eta_2\bigr) \Bigr)\,\xi_2
=(\xi_1\eta_1\ra P) (\xi_2\ra \eta_2) \xi_3 \\
&= (\xi_1\eta_1\ra P) \xi_2\eta_2 S(\xi_3)\xi_4 = (\xi_1\eta_1\ra P) \xi_2\eta_2 
= (\xi\eta)P~.
\end{flalign}
Moreover, $\bbA$ is an $(H,\mu)$-bimodule, i.e.~we have for all $\xi,\eta\in H$ and $P\in\bbA$,
\begin{flalign}
 \xi(P\eta) = \bigl(\xi_1\ra(P\eta)\bigr)\xi_2 = \bigl((\xi_1\ra P)(\xi_2\ra \eta)\bigr)\xi_3
=(\xi_1\ra P)\xi_2\eta S(\xi_3) \xi_4 = (\xi P)\eta~.
\end{flalign}
The Hopf algebra action $\ra$ on $\bbA$ is just the adjoint action with respect to this bimodule structure
\begin{flalign}
\label{eqn:adjointaction}
 \xi\ra P = \xi_1 P S(\xi_2)~.
\end{flalign}
Condition (\ref{eqn:algebraiso2}) then simply reads
\begin{flalign}
 \label{eqn:algebraiso4}(P\xi)Q = P(\xi Q)~,
\end{flalign}
for all $P,Q\in\bbA$ and $\xi\in H$. We obtain for the left $(H,\mu)$-action, for all $\xi\in H$ and $P,Q\in\bbA$,
\begin{flalign}
 \label{eqn:algebraiso5} \xi(PQ) = (\xi P)Q~.
\end{flalign}
In case $\bbA$ is unital with $1_\bbA\in\bbA$ we also find, for all $\xi\in H$,
\begin{flalign}
\label{eqn:algebraiso6}
 \xi 1_\bbA = (\xi_1\ra 1_\bbA) \xi_2 = 1_\bbA \epsilon(\xi_1)\xi_2 = 1_\bbA \xi~.
\end{flalign}

Vice versa, if $\bbA$ is an algebra and an $(H,\mu)$-bimodule satisfying (\ref{eqn:algebraiso1}), 
(\ref{eqn:algebraiso4}) and (\ref{eqn:algebraiso5}) (as well as (\ref{eqn:algebraiso6}) in case $\bbA$ is unital), 
then $\xi\ra P:= \xi_1\ P S(\xi_2)$ defines a left
$H$-module structure on $\bbA$ that satisfies (\ref{eqn:algebraiso2}) and (\ref{eqn:algebraiso3}).
Hence, Theorem \ref{theo:algebraiso} equivalently reads
\begin{theo}\label{theo:algebraiso2}
 Consider a Hopf algebra $H$ and an $(H,\mu)$-bimodule $\bbA$ that is also an algebra 
(not necessarily associative or with unit). If, for all $\xi\in H$ and $P,Q\in \bbA$, the ``generalized associativity'' conditions
\begin{subequations}
\begin{flalign}
 (PQ)\xi=P(Q\xi)~,\quad (P\xi)Q= P(\xi Q)~,\quad \xi(PQ) = (\xi P)Q~,
\end{flalign}
and in case of $\bbA$ unital also the condition
\begin{flalign}
 \xi 1_\bbA = 1_\bbA \xi~,
\end{flalign}
\end{subequations}
hold true, then the adjoint action (\ref{eqn:adjointaction}) structures $\bbA$ as a left module algebra with respect to 
the Hopf algebra $H$. Given a twist $\FF$ of the Hopf algebra $H$, the twist quantized
algebra $\bbA_\star$ is isomorphic to $\bbA$, via the map 
\begin{flalign}
\label{eqn:Ddef}
D_\FF:\bbA_\star\to\bbA~,\quad P\mapsto D_\FF(P) = (\bar f^\alpha\ra P)\bar f_\alpha = \bar f^\alpha_1 P S(\bar f^\alpha_2)\bar f_\alpha~.
\end{flalign}
\end{theo}
\begin{proof}
 $D_\FF$ is obviously a $\bfK$-linear map. We prove that
\begin{flalign}
 D_\FF\circ \mu_\star = \mu\circ (D_\FF\otimes D_\FF)~,
\end{flalign}
for all $P,Q\in\bbA$,
\begin{flalign}
\nn D_\FF(P\star Q) &= D_\FF\bigl((\bar f^\beta\ra P)(\bar f_\beta\ra Q)\bigr)
 =\Bigl(\bar f^\alpha\ra \bigl((\bar f^\beta\ra P)(\bar f_\beta\ra Q)\bigr)\Bigr)\bar f_\alpha\\
\nn &=(\bar f^\alpha_1\bar f^\beta\ra P)(\bar f^\alpha_2\bar f_\beta\ra Q)\bar f_\alpha
 = (\bar f^\alpha\ra P)(\bar f_{\alpha_1}\bar f^\beta\ra Q)\bar f_{\alpha_2}\bar f_\beta\\
 &= (\bar f^\alpha\ra P)\bar f_\alpha\,(\bar f^\beta\ra Q)\bar f_\beta = D_\FF(P)\,D_\FF(Q)~,
\end{flalign}
where in the second equality in the second line we have used (\ref{eqn:twistpropsimp3})
and in the first equality in the third line we have used
\begin{flalign}
 (\bar f_{\alpha_1}\bar f^\beta\ra Q)\bar f_{\alpha_2} = \bar f_{\alpha_1}(\bar f^\beta\ra Q)S(\bar f_{\alpha_2})\bar f_{\alpha_3}
= \bar f_\alpha(\bar f^\beta\ra Q)~.
\end{flalign}
This shows that $D_\FF$ is an algebra homomorphism and it remains to prove its invertability.
We can simplify (\ref{eqn:Ddef}) by using (\ref{eqn:twistpropsimp}) as follows
\begin{flalign}
 \nn D_\FF(P) &= \bar f^\alpha_1 P S(\bar f^\alpha_2)\bar f_\alpha = 
\bar f^\alpha f^\gamma P S(\bar f_{\alpha_1}\bar f^\beta f_\gamma)\bar f_{\alpha_2}\bar f_\beta\\
\label{eqn:Dsimp}&=\bar f^\alpha f^\gamma P S(f_\gamma) S(\bar f^\beta)\epsilon(\bar f_\alpha)\bar f_\beta =
 f^\gamma P S(f_\gamma) \chi^{-1}~,
\end{flalign}
where $\chi^{-1} = S(\bar f^\alpha)\bar f_\alpha$. Therefore, $D_\FF$ is invertible and we have for all $P\in\bbA$
\begin{flalign}\label{eqn:Dinv}
 D_\FF^{-1}(P) = \bar f^\alpha P \chi S(\bar f_\alpha)~,
\end{flalign}
where $\chi = f^\beta S(f_\beta)$.

Finally, if $\bbA$ is unital, $D_\FF$ maps the unit of $\bbA_\star$ to the unit of $\bbA$ because of the normalization
property of the twist (\ref{eqn:twistprop2}),
\begin{flalign}
 D_\FF(1_\bbA) = (\bar f^\alpha\ra 1_\bbA) \bar f_\alpha = 1_\bbA \epsilon(\bar f^\alpha) \bar f_\alpha = 1_\bbA 1 = 1_\bbA~.
\end{flalign}

\end{proof}
In the hypotheses of Theorem \ref{theo:algebraiso2}, the Hopf algebra properties
immediately imply that the algebra $\bbA$ has a left $H^\FF$-module algebra structure given by
the adjoint $H^\FF$-action, for all $\xi\in H$ and $P\in\bbA$,
\begin{flalign}
 \xi\ra_\FF P := \xi_{1_\FF} P S^\FF(\xi_{2_\FF})~.
\end{flalign}
\begin{theo}\label{theo:algebraiso3}
 The algebra isomorphism $D_\FF:\bbA_\star \to \bbA$ of Theorem \ref{theo:algebraiso2}
is also an isomorphism between the left $H^\FF$-module algebra $\bbA_\star\in {^{H^\FF,\ra}}\AAA$ and
the left $H^\FF$-module algebra $\bbA\in {^{H^\FF,\ra_\FF}}\AAA$, i.e.~$D_\FF$ intertwines between the
left $H^\FF$-actions $\ra$ and $\ra_\FF$, for all $\xi\in H$ and $P\in \bbA$,
\begin{flalign}
 D_\FF(\xi\ra P) = \xi\ra_\FF D_\FF(P)~.
\end{flalign}
\end{theo}
\begin{proof} Using (\ref{eqn:Dsimp}) we obtain
 \begin{flalign}
\nn  D_\FF(\xi\ra P) &= f^\beta (\xi\ra P) S(f_\beta)\chi^{-1} = f^\beta \xi_1 P S(f_\beta \xi_2)\chi^{-1} = 
\xi_{1_\FF}f^\beta P S(\xi_{2_\FF}f_\beta)\chi^{-1}\\
&=\xi_{1_\FF} f^\beta P S(f_\beta)\chi^{-1}\chi S(\xi_{2_\FF})\chi^{-1} = 
\xi_{1_\FF} D_\FF(P)  S^\FF(\xi_{2_\FF}) = \xi\ra_\FF D_\FF(P)~.
 \end{flalign}

\end{proof}
\begin{rem}\label{rem:Dinv}
We have discussed above that $H^\FF$ admits the twist $\FF^{-1}$ leading to $(H^\FF)^{\FF^{-1}}=H$.
The associated quantization isomorphism is exactly $D_\FF^{-1}$ given in (\ref{eqn:Dinv}).
This can be shown by using (\ref{eqn:Dsimp}) and a short calculation, for all $P\in\bbA$,
\begin{flalign}
 (f^\alpha\ra_\FF P) f_\alpha = \bar f^\alpha P S^\FF(\bar f_\alpha) \chi_\FF^{-1} 
= \bar f^\alpha P \chi S(\bar f_\alpha)\chi^{-1} \chi = D_\FF^{-1}(P)~,
\end{flalign}
where we have used that $\chi_\FF^{-1} = S^\FF(f^\alpha)f_\alpha = \chi$.
\end{rem}
In order to show that algebras of the type $\bbA$ actually appear quite naturally we provide some
examples.
\begin{ex}
 Given a left $H$-module algebra $A$ (not necessarily associative or with unit) consider the
crossed product (or smash product) $A\,\sharp\, H$. By definition the underlying $\bfK$-module of
$A\,\sharp\, H$ is $A\otimes H$, and the product is given by
\begin{flalign}
 (a\otimes \xi)(b\otimes \eta) = a(\xi_1\ra b)\otimes \xi_2\eta~.
\end{flalign}
Note that in case $A$ is associative, then so is $A\,\sharp\, H$, and that in case $A$ is unital, then so is
$A\,\sharp\, H$ with unit $1_A\otimes 1_H$.
The algebra $A\,\sharp\, H$ is a left $H$-module algebra with the action $\xi\ra (a\otimes\eta) 
= (\xi_1\ra a)\otimes(\xi_2\ra \eta)$.
The right $(H,\mu)$-module structure is given by $(a\otimes\xi)\eta=a\otimes(\xi\eta)$ and the compatibility conditions
(\ref{eqn:algebraiso}) hold true. Hence, the requirements of Theorem \ref{theo:algebraiso} are satisfied.
Notice that $A$ is always a subalgebra of $A\,\sharp\, H$ via $a\mapsto a\otimes 1_H$, for all $a\in A$.
If $A$ is unital, then $H$ is also a subalgebra of $A\,\sharp\, H$ via $\xi\mapsto 1_A\otimes \xi$, for all $\xi\in H$.
\end{ex}
\begin{ex}
 Given an associative and unital algebra $\bbA$ that admits a homomorphism $\rho:H\to\bbA$,
 then the hypotheses of Theorem \ref{theo:algebraiso2}
immediately hold, just define the $(H,\mu)$-bimodule structure of $\bbA$ by, for all $\xi\in H$ and $P\in \bbA$, 
$\xi P :=\rho(\xi)P$ and $P\xi := P\rho(\xi)$. A particular case is when $\bbA =H$ and we consider the
identity homomorphism. Then we recover the (Hopf algebra) isomorphism $D:H_\star \to H^\FF$ 
discussed in Chapter \ref{chap:basicncg}, Section \ref{sec:QLA}, see also \cite{0787.17010,Aschieri:2005zs}.
\end{ex}
\begin{ex}\label{ex:endoex}
 Given a Hopf algebra $H$ over the ring $\bfK$ and a left $H$-module $V$, consider the algebra
$\End_\bfK(V)$ of $\bfK$-linear maps ($\bfK$-module homomorphisms) from $V$ to $V$. Since $H$ is a Hopf algebra
the left action of $H$ on $V$ lifts to a left action of $H$ on $\End_\bfK(V)$, defined by,
for all $\xi\in H$ and $P\in \End_\bfK(V)$,
\begin{flalign}
 \xi\RA P := \xi_1\ra\circ P\circ S(\xi_2)\ra~\,~,
\end{flalign}
where $\circ$ denotes the usual composition of morphisms and $\xi\ra\in\End_\bfK(V)$ the endomorphism
$v\mapsto \xi\ra v$.
The algebra $\End_\bfK(V)$ is thus a left $H$-module algebra, i.e.~$\End_\bfK(V)\in{^{H,\RA}}\AAA$.
The algebra homomorphism $H\to \End_\bfK(V)\,,~\xi\mapsto\xi\ra$ implies, by using the previous example,
the isomorphism $D_\FF:\End_\bfK(V)_\star \to \End_\bfK(V)$, where the composition law in
$\End_\bfK(V)_\star$ is given by the $\star$-composition, for 
all $P,Q\in \End_\bfK(V)_\star$,
\begin{flalign}\label{eqn:starcompo}
 P\circ_\star Q := (\bar f^\alpha\RA P)\circ (\bar f_\alpha\RA Q)~.
\end{flalign}

By Theorem \ref{theo:algebraiso3} we even obtain a left $H^\FF$-module algebra isomorphism
$D_\FF$ between $\End_\bfK(V)_\star \in {^{H^\FF,\RA}}\AAA$ and $\End_\bfK(V)\in {^{H^\FF,\RA_\FF}}\AAA$.
\end{ex}


\chapter{\label{chap:modhom}Module homomorphisms}
\section{Quantization of endomorphisms}
In this section we study the algebras $\End_\bfK(V)$ and $\End_B(V)$ of endomorphisms of a module
$V\in {^{H,\ra}_A}\MMM_B$. We then investigate how they behave under twist quantization.
\begin{propo}\label{propo:endomodulealgebra}
 Let $A\in {^{H,\ra}}\AAA$ and $V\in {^{H,\ra}_A}\MMM$, then the algebra $\End_\bfK(V)$ of $\bfK$-linear maps
from $V$ to $V$ is a left $H$-module $A$-bimodule algebra
\begin{flalign}
 \End_\bfK(V)\in {^{H,\RA}_A}\AAA_A~,
\end{flalign}
where $\RA$ is the adjoint $H$-action, for all $\xi\in H$ and $P\in \End_\bfK(V)$,
\begin{flalign}
\label{eqn:endomodulealgebra1}
 \xi\RA P := \xi_1\ra\circ P\circ S(\xi_2)\ra~~,
\end{flalign}
and the $A$-bimodule structure is given by, for all $a\in A$ and $P\in \End_\bfK(V)$,
\begin{subequations}
\label{eqn:endomodulealgebra2}
\begin{flalign}
\label{eqn:endomodulealgebra21}  a\cdot P &:= l_a\circ P~,\\
 \label{eqn:endomodulealgebra22} P\cdot a &:= P\circ l_a~,
\end{flalign}
\end{subequations}
where, for all $v\in V$, $l_a(v) := a\cdot v$.

\noindent If $V$ is also a right module over $B\in{^{H,\ra}}\AAA$, such that $V\in {^{H,\ra}_A}\MMM_B$, then the subalgebra
 of right $B$-linear endomorphisms of $V$,  $\End_B(V)\subseteq \End_\bfK(V)$, is still a left $H$-module
$A$-bimodule algebra
\begin{flalign}
 \End_B(V) \in {^{H,\RA}_A}\AAA_A~,
\end{flalign}
with $H$ and $A$ actions given in (\ref{eqn:endomodulealgebra1}) and (\ref{eqn:endomodulealgebra2}), respectively.
\end{propo}
\begin{proof}
 $\End_\bfK(V)$ is an $A$-bimodule, since for all $a,b\in A$ and $P\in \End_\bfK(V)$,
\begin{flalign}
 \nn a\cdot(b\cdot P) &= l_a\circ (l_b\circ P) = l_{ab}\circ P = (ab)\cdot P~,\\
 \nn (P \cdot a)\cdot b &= (P\circ l_a)\circ l_b = P\circ l_{ab} = P\cdot(ab)~,\\
 a\cdot(P\cdot b) &= l_a\circ (P\circ l_b) = (l_a\circ P)\circ l_b = (a\cdot P)\cdot b~.
\end{flalign}
It is straightforward to check that $\End_\bfK(V)$ is a left $H$-module algebra, for all $\xi,\eta\in H$ and
$P,Q\in\End_\bfK(V)$,
\begin{flalign}
 \xi\RA (\eta\RA P) = (\xi\eta)\RA P~,\quad \xi\RA\id_V = \epsilon(\xi)\,\id_V~,
\end{flalign}
and
\begin{flalign}\label{eqn:endomodulealgebra3}
 \xi\RA(P\circ Q) = (\xi_1\RA P)\circ (\xi_2\RA Q)~.
\end{flalign}
In order to show that $\End_\bfK(V)$ is a left $H$-module $A$-bimodule algebra it remains to prove
the compatibilities $\xi\RA(a\cdot P) = (\xi_1\ra a)\cdot(\xi_2\RA P)$ and $\xi\RA(P\cdot a) = (\xi_1\RA P)\cdot (\xi_2\ra a)$.
They immediately follow from (\ref{eqn:endomodulealgebra3}) and $\xi\RA l_a= l_{\xi\ra a}$, for all
$\xi\in H$ and $a\in A$. The latter property is simply shown to hold
\begin{flalign}
 \nn(\xi\RA l_a)(v) &= \xi_1\ra\bigl(l_a(S(\xi_2)\ra v)\bigr) = \xi_1\ra\bigl(a\cdot(S(\xi_2)\ra v)\bigr)\\
&=(\xi_1\ra a)\cdot (\xi_2 S(\xi_3)\ra v) = (\xi\ra a)\cdot v = l_{\xi\ra a}(v)~, 
\end{flalign}
for all $v\in V$.

Finally, let $V\in {^{H,\ra}_A}\MMM_B$, then (by definition) $V$ is an $(A,B)$-bimodule
and we in particular have $a\cdot (v \cdot b) = (a\cdot v)\cdot b$, hence $l_a\in\End_B(V)$ for all 
$a\in A$. Therefore, $a\cdot P\in \End_B(V)$ and $P\cdot a\in \End_B(V)$ for all $a\in A$ and 
$P\in \End_B(V)$. Furthermore, for all $\xi\in H$ and $P\in \End_B(V)$ we have $\xi\RA P\in \End_B(V)$, since
\begin{flalign}
 \nn(\xi\RA P)(v\cdot b) &= \xi_1\ra\Bigl( P\bigl( (S(\xi_3)\ra v)\cdot (S(\xi_2)\ra b) \bigr)\Bigr)\\
 &=\Bigl( \xi_1\ra\bigl( P(S(\xi_4)\ra v) \bigr)\Bigr)\cdot \Bigl(\xi_2\ra\bigl(S(\xi_3)\ra b\bigr)\Bigr) = (\xi\RA P)(v)\cdot b~,
\end{flalign}
for all $v\in V$ and $b\in B$.

\end{proof}

Let $H$ be a Hopf algebra with twist $\FF\in H\otimes H$, and let $A,B\in {^{H,\ra}}\AAA$, $V\in {^{H,\ra}_A}\MMM_B$. 
As a consequence,  $\End_\bfK(V),\End_B(V)\in {^{H,\RA}_A}\AAA_A$. We have two possible deformations of the endomorphisms
$\End_\bfK(V)$, $\End_B(V)$: The first, $\End_\bfK(V)_\star,\End_B(V)_\star\in  {^{H^\FF,\RA}_{A_\star}}\AAA_{A_\star}$
is obtained by applying Theorem \ref{theo:moduledef} to the left $H$-module $A$-bimodule algebras $\End_\bfK(V)$, $\End_B(V)$.
It is characterized by a deformed composition law $P\circ_\star Q = (\bar f^\alpha\RA P)\circ (\bar f_\alpha\RA Q)$.
The second option is simply to consider the $\bfK$-linear or right $B_\star$-linear endomorphisms of the deformed
module $V_\star$. From Proposition \ref{propo:endomodulealgebra} we have $\End_\bfK(V_\star),\End_{B_\star}(V_\star)\in
{^{H^\FF,\RA_\FF}_{A_\star}}\AAA_{A_\star}$. Here the product is the usual composition law $\circ$
and the action is the adjoint $H^\FF$-action, for all $\xi\in H^\FF$ and $P\in\End_\bfK(V_\star)$,
\begin{flalign}\label{eqn:adjointhfaction}
 \xi\RA_\FF P := \xi_{1_\FF}\ra\,\circ P\circ\,S^\FF(\xi_{2_\FF})\ra~.
\end{flalign}

\begin{theo}\label{theo:endodef}
 The map 
\begin{flalign}\label{eqn:endodef}
 D_\FF : \End_\bfK(V)_\star \to \End_\bfK(V_\star)~,\quad P\mapsto D_\FF(P)= (\bar f^\alpha\RA P)\circ \bar f_\alpha\ra~~
\end{flalign}
is an isomorphism between the left $H^\FF$-module $A_\star$-bimodule algebras 
$\End_\bfK(V)_\star\in {^{H^\FF,\RA}_{A_\star}}\AAA_{A_\star}$ and
$\End_\bfK(V_\star)\in {^{H^\FF,\RA_\FF}_{A_\star}}\AAA_{A_\star}$. It restricts to a
left $H^\FF$-module $A_\star$-bimodule algebra isomorphism 
\begin{flalign}
 D_\FF:\End_{B}(V)_\star \to \End_{B_\star}(V_\star)~.
\end{flalign}

\noindent We call $D_\FF(P)$ the quantization of the endomorphism $P\in \End_B(V)$.

\end{theo}
\begin{proof}
 Since $V = V_\star$ as $H^\FF$-modules, we have as left $H^\FF$-module algebras $\End_\bfK(V_\star)=\End_\bfK(V)$.
The left $H^\FF$-module algebra isomorphism (\ref{eqn:endodef}) is therefore the isomorphism
$D_\FF:\End_\bfK(V)_\star\to \End_\bfK(V)$ discussed in Example \ref{ex:endoex}.

The $A_\star$-bimodule structure of $\End_\bfK(V)_\star$ is given by (cf.~Theorem \ref{theo:moduledef}), for all
$a\in A$ and $P\in \End_\bfK(V)$,
\begin{subequations}
\begin{flalign}
 a\star P &:= l_{\bar f^\alpha\ra a}\circ (\bar f_\alpha\RA P) = (\bar f^\alpha\RA l_a)\circ (\bar f_\alpha\RA P) = l_a\circ_\star P~,\\
P\star a &:= (\bar f^\alpha\RA P)\circ l_{\bar f_\alpha\ra a}=(\bar f^\alpha\RA P)\circ (\bar f_\alpha\RA l_a) = P\circ_\star l_a~.
\end{flalign}
\end{subequations}
The $A_\star$-bimodule structure of $\End_\bfK(V_\star)$ is given by (cf.~Proposition \ref{propo:endomodulealgebra}),
for all $a\in A_\star$ and $P_\star\in \End_\bfK(V_\star)$,
\begin{subequations}
 \begin{flalign}
  a\bigstar P_\star := l_a^\star\circ P_\star~,\\
 P_\star \bigstar a := P_\star\circ l_a^\star~,
 \end{flalign}
\end{subequations}
where, for all $v\in V_\star$, $l_a^\star(v) := a\star v$.
The map $D_\FF$ is a left $H^\FF$-module $A_\star$-bimodule algebra isomorphism,
because (cf.~Theorem \ref{theo:algebraiso2}) $D_\FF(a\star P) = D_\FF(l_a\circ_\star P) = D_\FF(l_a)\circ D_\FF(P)$,
$D_\FF(P \star a) = D_\FF(P)\circ D_\FF(l_a)$ and $D_\FF(l_a) =l_a^\star$. The last property follows from a short calculation
\begin{flalign}
 D_\FF(l_a)(v) = (\bar f^\alpha\RA l_a)(\bar f_\alpha\ra v) = l_{\bar f^\alpha\ra a}(\bar f_\alpha\ra v) = a\star v = l_a^\star (v)~,
\end{flalign}
for all $v\in V$.

\noindent In order to show that $D_\FF$ restricts to an isomorphism between the left $H^\FF$-module
$A_\star$-bimodule subalgebras $\End_B(V)_\star\in {^{H^\FF,\RA}_{A_\star}}\AAA_{A_\star}$ and
$\End_{B_\star}(V_\star)\in {^{H^\FF,\RA_\FF}_{A_\star}}\AAA_{A_\star}$, we have to prove that for all
$P\in \End_B(V)$ we have $D_\FF(P)\in \End_{B_\star}(V_\star)$ and that for all $P_\star\in \End_{B_\star}(V_\star)$
we have $D_\FF^{-1}(P_\star)\in \End_B(V)$. 
Employing Remark \ref{rem:Dinv} it is sufficient to prove the first statement, since 
the second follows from dequantization with $\FF^{-1}$.
We obtain by a short calculation, for all $v\in V$ and $b\in B$,
\begin{flalign}
\nn D_\FF(P)(v\star b) &= 
(\bar f^\alpha\RA P)\bigl((\bar f_{\alpha_1}\bar f^\beta \ra v)\cdot (\bar f_{\alpha_2}\bar f_\beta\ra b)\bigr)\\
\nn&= (\bar f^\alpha_1\bar f^\beta\RA P)\bigl(\bar f^\alpha_2\bar f_\beta\ra v \bigr)\cdot (\bar f_\alpha\ra b)\\
&=\bar f^\alpha\ra \bigl(D_\FF(P)(v)\bigr)\cdot (\bar f_\alpha\ra b) = D_\FF(P)(v)\star b~,
\end{flalign}
where in the second line we have used the $2$-cocycle property of the twist (\ref{eqn:twistpropsimp3}) and 
the fact that $\xi\RA P\in \End_B(V)$ for all $\xi\in H$ and $P\in \End_B(V)$.

\end{proof}

We have so far studied right $B$-linear endomorphisms $\End_B(V)$, but we as well could have studied
left $A$-linear endomorphisms $_A\End(V)$ of $V\in {^{H,\ra}_A}\MMM_B$. 
These are related to the right linear endomorphisms by a
``mirror'' construction.

Recall that given a Hopf algebra $(H,\mu,\Delta,\epsilon,S)$ with an invertible antipode $S$,
we have the Hopf algebras $H^\cop$ and $H_\op$. With an abuse of notation denoting
also the underlying $\bfK$-module by $H$, we have explicitly
$H^\cop=(H,\mu,\Delta^\cop,\epsilon,S^{-1})$ and $H_\op=(H,\mu_\op,\Delta,\epsilon,S^{-1})$,
where the coopposite coproduct is defined by, for all $\xi\in H$, $\Delta^\cop(\xi)=\xi_{1^\cop}\otimes \xi_{2^\cop} 
= \xi_2\otimes \xi_1$, and the opposite product is given by, for all $\xi,\eta\in H$,
$\mu_\op(\xi\otimes \eta)=\mu(\eta\otimes\xi) =\eta\,\xi$. Note that $H_\op^\cop=(H,\mu_\op,\Delta^\cop,\epsilon,S)$
even exists when $S$ is not invertible. For a simplified mirror construction based on
$H^\cop$ (and not on $H^\cop_\op$) we are going to assume from now on that $S^{-1}$ exists.
This assumption is motivated by the fact that $S^{-1}$ exists, in particular, for {\it all} quasitriangular 
Hopf algebras (see e.g.~\cite{Majid:1996kd}), which will be the main subject of our studies.

We observe that for any left module $V\in {_A}\MMM$ there is a right module
$V^\op\in \MMM_{A^\op}$. As $\bfK$-modules $V=V^\op$, the algebra $A^\op$ is the algebra with opposite product
and its right action on $V^\op$ is given by $v\cdot^\op a := a\cdot v$. Similarly, for any
right module $V\in \MMM_{A}$ there is a left module $V^\op\in {_{A^\op}}\MMM$, and $(V^\op)^\op=V$.

Moreover, if we have a left $H$-module algebra $A\in {^{H,\ra}}\AAA$ then the opposite algebra
is a left $H^\cop$-module algebra
\begin{flalign}
 A\in {^{H,\ra}}\AAA \quad\Longrightarrow\quad A^\op\in{^{H^\cop,\ra}}\AAA~,
\end{flalign}
where the Hopf algebra action is unchanged.
This easily leads to the following
\begin{lem}\label{lem:opposite}
 Let $A,B\in {^{H,\ra}}\AAA$ and $V\in {^{H,\ra}_A}\MMM_B$, then $A^\op,B^\op\in {^{H^\cop,\ra}}\AAA$
and $V^\op\in {^{H^\cop,\ra}_{B^\op}}\MMM_{A^\op}$. If $E\in {^{H,\ra}_A}\AAA_B$, then 
$E^\op\in {^{H^\cop,\ra}_{B^\op}}\AAA_{A^\op}$.
\end{lem}
We apply these observations to the algebra of endomorphisms of the module $V$ and obtain
\begin{propo}\label{propo:opposite}
 Let $A,B\in {^{H,\ra}}\AAA$ and $V\in{^{H,\ra}_A}\MMM_B$, then $\bigl({_A}\End(V)\bigr)^\op \in {^{H,\RA^\cop}_B}\AAA_B$,
where for all $\xi\in H$ and $P\in {_A}\End(V)$,
\begin{flalign}\label{eqn:adjointhcopaction}
 \xi\RA^\cop P := \xi_2\ra\circ P\circ S^{-1}(\xi_1)\ra~~,
\end{flalign}
i.e.~the action $\RA^\cop$ is just the adjoint action of
 $H^\cop$, $\xi\RA^\cop P = \xi_{1^\cop}\ra\circ P \circ S^\cop(\xi_{2^\cop})\ra$~.
\end{propo}
\begin{proof}
 The hypothesis and Lemma \ref{lem:opposite} 
implies  $A^\op,B^\op\in {^{H^\cop,\ra}}\AAA$, $V^\op\in {^{H^\cop,\ra}_{B^\op}}\MMM_{A^\op}$, and hence
 $\End_{A^\op}(V^\op)\in {^{H^\cop,\RA^\cop}_{B^\op}}\AAA_{B^\op}$. Because of Lemma \ref{lem:opposite}
we have $\bigl(\End_{A^\op}(V^\op)\bigr)^\op \in {^{H,\RA^\cop}_B}\AAA_B$.
The proof follows by using the canonical algebra isomorphism
\begin{flalign}
 \End_{A^\op}(V^\op) \simeq {_A}\End(V)
\end{flalign}
given by the identity map. Indeed, if $P\in \End_{A^\op}(V^\op)$, then for all $a\in A$ and $v\in V$,
\begin{flalign}
 P(a\cdot v) = P(v\cdot^\op a) = P(v)\cdot^\op a = a\cdot P(v)~,
\end{flalign}
hence $P\in {_A}\End(V)$, and vice versa.

\end{proof}

The $B$-bimodule structure of $\bigl({_A}\End(V)\bigr)^\op$ explicitly reads, for all $b\in B$ and $P\in
{_A}\End(V)$,
\begin{flalign}
 b\cdot P = P\circ l_b^\op = P\circ r_b ~,\quad P\cdot b = l^\op_b\circ P = r_b\circ P~,
\end{flalign}
where, for all $v\in V$, $r_b(v) := v\cdot b$.

Similar to Theorem \ref{theo:endodef} we have a left quantization map.
\begin{theo}\label{theo:leftendodef}
Let $H$ be a Hopf algebra with twist $\FF\in H\otimes H$, and let $A,B\in {^{H,\ra}}\AAA$, $V\in{^{H,\ra}_A}\MMM_B$.
The map 
\begin{flalign}
\label{eqn:DFcop}
 D_\FF^\cop : \Bigl(\bigl({_A}\End(V)\bigr)^\op\Bigr)_\star\to \bigl({_{A_\star}}\End(V_\star)\bigr)^\op~,\quad P\mapsto D_\FF^\cop(P) = 
(\bar f_\alpha\RA^\cop P)\circ \bar f^\alpha\ra~
\end{flalign}
is an isomorphism between the left $H^\FF$-module $B_\star$-bimodule algebras
$\Bigl(\bigl({_A}\End(V)\bigr)^\op\Bigr)_\star\in {^{H^\FF,\RA^\cop}_{B_\star}}\AAA_{B_\star}$
and $\bigl({_{A_\star}}\End(V_\star)\bigr)^\op\in {^{H^\FF,(\RA_\FF)^\cop}_{B_\star}}\AAA_{B_\star}$.

\noindent We call $D_\FF^\cop(P)$ the left quantization of the endomorphism $P\in {_A}\End(V)$.

\end{theo}
\begin{proof}
 Recall that if $\FF$ is a twist of $H$, then $\FF^\cop = \FF_{21}$ is a twist of the Hopf algebra $H^\cop$.
From Lemma \ref{lem:opposite}, Proposition \ref{propo:endomodulealgebra}, and the Theorems 
\ref{theo:algebradef} and \ref{theo:moduledef}
we have 
\begin{subequations}
\begin{flalign}
 \bigl(\End_{A^\op}(V^\op)\bigr)_{\star^\cop}\in 
{^{ (H^\cop)^{\FF^\cop},\RA^\cop }_{(B^\op)_{\star^\cop}}}\AAA_{(B^\op)_{\star^\cop}}~,
\end{flalign}
and
\begin{flalign}
 \End_{(A^\op)_{\star^\cop}}((V^\op)_{\star^\cop})\in
 {^{ (H^\cop)^{\FF^\cop},(\RA^\cop)_{\FF^\cop} }_{(B^\op)_{\star^\cop}}}\AAA_{(B^\op)_{\star^\cop}}~,
\end{flalign}
\end{subequations}
where $\star^\cop$ is the $\star$-product of $\FF^\cop$.
Theorem \ref{theo:endodef} implies the left $(H^\cop)^{\FF^\cop}$-module $(B^\op)_{\star^\cop}$-bimodule algebra
isomorphism
\begin{flalign}
\label{eqn:oppo1}
 D_{\FF^\cop} = D_\FF^\cop : \bigl(\End_{A^\op}(V^\op)\bigr)_{\star^\cop} \to \End_{(A^\op)_{\star^\cop}}((V^\op)_{\star^\cop})~.
\end{flalign}
Observe that
\begin{subequations}
\begin{flalign}
 (A^\op)_{\star^\cop} = (A_\star)^\op~,\quad (B^\op)_{\star^\cop} = (B_\star)^\op~,\quad (V^\op)_{\star^\cop} = (V_\star)^\op~,
\end{flalign}
as well as 
\begin{flalign}
 (H^\cop)^{\FF^\cop} = (H^\FF)^\cop~,\quad (\RA^\cop)_{\FF^\cop} = (\RA_\FF)^\cop~.
\end{flalign}
\end{subequations}

Using the canonical algebra isomorphism
$ \End_{A^\op}(V^\op)\simeq {_A}\End(V)$ 
we have
\begin{subequations}
\begin{flalign}
  \bigl(\End_{A^\op}(V^\op)\bigr)_{\star^\cop} \simeq \bigl({_A}\End(V)\bigr)_{\star^\cop}
=\bigl(\bigl({_A}\End(V)\bigr)^{\op\,\op}\bigr)_{\star^\cop} = \Bigr(\bigl({_A}\End(V)\bigr)^{\op}\Bigr)_\star^{~\,~\op}
\end{flalign}
and
\begin{flalign}
 \End_{(A^\op)_{\star^\cop}}((V^\op)_{\star^\cop}) = \End_{(A_\star)^\op}((V_\star)^\op)\simeq {_{A_\star}}\End(V_\star)~.
\end{flalign}
\end{subequations}
The proof of the theorem follows by applying Lemma \ref{lem:opposite} to (\ref{eqn:oppo1}), leading to the
 left $H^\FF$-module $B_\star$-bimodule algebra isomorphism (\ref{eqn:DFcop}).

\end{proof}


\section{\label{sec:homoquant}Quantization of homomorphisms}
Let $H$ be a Hopf algebra, $A,B\in {^{H,\ra}}\AAA$ and $V,W\in {^{H,\ra}_{A}}\MMM$.
In this section we focus on $\bfK$-linear maps from $V$ to $W$, denoted by
$\Hom_\bfK(V,W)$, and their quantization. In case of
$V,W\in {^{H,\ra}_A}\MMM_{B}$, we can also consider right $B$-linear
homomorphisms $\Hom_{B}(V,W)$ from $V$ to $W$. Properties of left $A$-linear
homomorphisms ${_A}\Hom(V,W)$ and their quantization 
will follow from a mirror construction.
The results are analogous to those for the endomorphisms above. 
For completeness and for a better overview we will briefly present the main propositions and theorems.

\begin{propo}\label{propo:homomodule}
 Let $A\in {^{H,\ra}}\AAA$ and $V,W\in {^{H,\ra}_A}\MMM$, then the $\bfK$-module $\Hom_\bfK(V,W)$ of $\bfK$-linear maps
from $V$ to $W$ is a left $H$-module $A$-bimodule
\begin{flalign}
 \Hom_\bfK(V,W)\in {^{H,\RA}_A}\MMM_A~,
\end{flalign}
where $\RA$ is the adjoint $H$-action, for all $\xi\in H$ and $P\in \Hom_\bfK(V,W)$,
\begin{flalign}
\label{eqn:homomodule1}
 \xi\RA P := \xi_1\ra\circ P\circ S(\xi_2)\ra~~,
\end{flalign}
and the $A$-bimodule structure is given by, for all $a\in A$ and $P\in \Hom_\bfK(V,W)$,
\begin{subequations}
\label{eqn:homomodule2}
\begin{flalign}
\label{eqn:homomodule21}  a\cdot P &:= l_a\circ P~,\\
 \label{eqn:homomodule22} P\cdot a &:= P\circ l_a~,
\end{flalign}
\end{subequations}
where we used for notational simplicity 
the same symbol $l_a$ for the left multiplication map on $V$ and $W$.

\noindent If   $V,W$ are also right $B$-modules, such that $V,W\in {^{H,\ra}_A}\MMM_B$, then the $\bfK$-submodule
 of right $B$-linear homomorphisms, $\Hom_B(V,W)\subseteq \Hom_\bfK(V,W)$, is still a left $H$-module
$A$-bimodule
\begin{flalign}
 \Hom_B(V,W) \in {^{H,\RA}_A}\MMM_A~,
\end{flalign}
with $H$ and $A$ actions given in (\ref{eqn:homomodule1}) and (\ref{eqn:homomodule2}), respectively.
\end{propo}
\noindent The proof of this proposition is analogous to Proposition \ref{propo:endomodulealgebra} so that we can omit it here.

Let now $H$ be a Hopf algebra with twist $\FF\in H\otimes H$, and let
$A,B\in {^{H,\ra}}\AAA$, $V,W\in{^{H,\ra}_A}\MMM_B$. As a consequence,
$\Hom_{\bfK}(V,W),\Hom_B(V,W)\in {^{H,\RA}_A}\MMM_A$.
Analogously to the endomorphisms, we have two possible deformations of the homomorphisms $\Hom_\bfK(V,W)$ and $\Hom_{B}(V,W)$. 
The first option,
$\Hom_{\bfK}(V,W)_\star$, $\Hom_B(V,W)_\star\in {^{H^\FF,\RA}_{A_\star}}\MMM_{A_\star}$,
 is obtained by applying Theorem \ref{theo:moduledef}.
The second option, $\Hom_\bfK(V_\star,W_\star)$, $\Hom_{B_\star}(V_\star,W_\star) \in{^{H^\FF,\RA_\FF}_{A_\star}}\MMM_{A_\star}$,
is just to consider homomorphisms between the quantized modules $V_\star$ and $W_\star$.
Similar to Theorem \ref{theo:endodef} we can relate these two constructions.
\begin{theo}\label{theo:homodef}
 The map 
\begin{flalign}\label{eqn:homodef}
 D_\FF : \Hom_\bfK(V,W)_\star \to \Hom_\bfK(V_\star,W_\star)~,\quad P\mapsto D_\FF(P)= 
(\bar f^\alpha\RA P)\circ \bar f_\alpha\ra~~
\end{flalign}
is an isomorphism between the left $H^\FF$-module $A_\star$-bimodules 
$\Hom_\bfK(V,W)_\star\in {^{H^\FF,\RA}_{A_\star}}\MMM_{A_\star}$ and \newline
$\Hom_\bfK(V_\star,W_\star)\in {^{H^\FF,\RA_\FF}_{A_\star}}\MMM_{A_\star}$. It restricts to a
left $H^\FF$-module $A_\star$-bimodule isomorphism 
\begin{flalign}
 D_\FF:\Hom_{B}(V,W)_\star \to \Hom_{B_\star}(V_\star,W_\star)~.
\end{flalign}

\noindent We call $D_\FF(P)$ the quantization of the homomorphism $P\in \Hom_B(V,W)$.

\end{theo}
\noindent Since the proof is analogous to the one of Theorem \ref{theo:endodef} we do not have to provide it here.

For completeness, we also briefly discuss left $A$-linear homomorphisms ${_A}\Hom(V,W)$ and their quantization.
We have $\bigl({_A}\Hom(V,W)\bigr)^\op \in {^{H,\RA^\cop}_{B}}\MMM_B$, where the left $H$-action is given by,
for all $\xi\in H$ and $P\in {_A}\Hom(V,W)$,
\begin{subequations}
\begin{flalign}
 \xi\RA^\cop P := \xi_2\ra\circ P\circ S^{-1}(\xi_1)\ra~,
\end{flalign}
and the $B$-bimodule structure reads, for all $b\in B$ and $P\in {_A}\Hom(V,W)$,
\begin{flalign}
 b\cdot P := P\circ r_b~,\quad P\cdot b := r_b\circ P~.
\end{flalign}
\end{subequations}
The quantization of left $A$-linear homomorphisms is understood analogously to Theorem \ref{theo:homodef}.
\begin{theo}\label{theo:lefthomodef}
 Let $H$ be a Hopf algebra with twist $\FF\in H\otimes H$, and let $A,B\in {^{H,\ra}}\AAA$, $V,W\in{^{H,\ra}_A}\MMM_B$.
The map
\begin{flalign}
 D_\FF^\cop : \Bigl(\bigl({_A}\Hom(V,W)\bigr)^\op\Bigr)_\star\to \bigl({_{A_\star}}\Hom(V_\star,W_\star)\bigr)^\op~,
\quad P\mapsto D_\FF^\cop(P) = 
(\bar f_\alpha\RA^\cop P)\circ \bar f^\alpha\ra~
\end{flalign}
is an isomorphism between the left $H^\FF$-module $B_\star$-bimodules
$\Bigl(\bigl({_A}\Hom(V,W)\bigr)^\op\Bigr)_\star\in {^{H^\FF,\RA^\cop}_{B_\star}}\MMM_{B_\star}$
and $\bigl({_{A_\star}}\Hom(V_\star,W_\star)\bigr)^\op\in {^{H^\FF,(\RA_\FF)^\cop}_{B_\star}}\MMM_{B_\star}$.

\noindent We call $D_\FF^\cop(P)$ the left quantization of the homomorphism $P\in{_A}\Hom(V,W)$.
\end{theo}
\noindent This theorem can be proven with a mirror construction (analogously to Theorem \ref{theo:leftendodef}),
 or equivalently by direct a calculation.

\begin{ex}
 Let $A\in {^{H,\ra}}\AAA$ and $V\in {^{H,\ra}_A}\MMM_A$. The {\it dual module} of $V$
is defined by $V^\prime := \Hom_A(V,A)$. Since $A$ can be regarded as a left $H$-module $A$-bimodule,
we have by Proposition \ref{propo:homomodule} $V^\prime \in {^{H,\RA}_A}\MMM_A$.
Let $\FF\in H\otimes H$ be a twist of $H$ and consider the deformed $H^\FF$-modules 
$A_\star\in{^{H^\FF,\ra}}\AAA$ and $V_\star\in{^{H^\FF,\ra}_{A_\star}}\MMM_{A_\star}$.
We have two possible deformations of the dual module, $(V^\prime)_\star = \Hom_{A}(V,A)_\star
\in {^{H^\FF,\RA}_{A_\star}}\MMM_{A_\star}$
and $(V_\star)^\prime = \Hom_{A_\star}(V_\star,A_\star)\in {^{H^\FF,\RA_\FF}_{A_\star}}\MMM_{A_\star}$.
Due to Theorem \ref{theo:homodef} there is a left $H^\FF$-module $A_\star$-bimodule isomorphism
$D_\FF$, such that $(V_\star)^\prime\simeq (V^\prime)_\star$.
In words, dualizing the quantized module is (up to isomorphism) equivalent to quantizing the dual one.
Note that by Theorem \ref{theo:lefthomodef} similar statements hold true for the left $A$-linear dual 
$^\prime V= \bigl({_A}\Hom(V,A)\bigr)^\op$ and its quantization.
\end{ex}


\section{\label{sec:modhomqc}Quasi-commutative algebras and bimodules}
In many examples Hopf algebras come with the additional structure of an $R$-matrix.
In particular, all models from Part \ref{part:ncg} and Part \ref{part:qft} fall into this class.
This justifies to study properties of endomorphisms and homomorphisms in presence of this extra structure.
\begin{defi}\label{defi:rmatrix}
A {\it quasi-cocommutative} Hopf algebra $(H,R)$ is a Hopf algebra $H$ together with an
invertible element $R\in H\otimes H$ (called {\it universal $R$-matrix}) such that, for all $\xi\in H$,
\begin{flalign}
\label{eqn:Rproper1}
 \Delta^\cop(\xi) = R \,\Delta(\xi)\, R^{-1}~.
\end{flalign}
 A quasi-cocommutative Hopf algebra is called {\it quasitriangular} if 
\begin{flalign}
\label{eqn:Rproper2}
 (\Delta\otimes\id)R = R_{13}\,R_{23}~,\quad (\id\otimes\Delta)R = R_{13}\,R_{12}~,
\end{flalign}
and {\it triangular} if additionally
\begin{flalign}
 R_{21}=R^{-1}~.
\end{flalign}
\end{defi}
\begin{ex}
 The Hopf algebra $U\Xi$ of diffeomorphisms of a manifold $\MM$ (see Chapter \ref{chap:basicncg}) is triangular
with trivial $R$-matrix $R=1\otimes 1$. The twist deformation $U\Xi^\FF$ of $U\Xi$ is triangular
with $R$-matrix $R^\FF=\FF_{21}\,\FF^{-1}$. Moreover, the twist deformation of any (quasi)triangular Hopf algebra $(H,R)$ is
(quasi)triangular with $R$-matrix $R^\FF = \FF_{21}\, R\, \FF^{-1}$.
\end{ex}
From the definition of a quasitriangular Hopf algebra the following standard properties follow
(see e.g.~\cite{Majid:1996kd,Kassel:1995xr})
\begin{subequations}\label{eqn:Rproper3}
\begin{flalign}
 \label{eqn:Rproper31} \qquad\,\qquad\,\qquad\,\qquad&(\epsilon\otimes\id)R =1~,\quad& &(\id\otimes \epsilon)R=1~,&\qquad\,\qquad\\
 \label{eqn:Rproper32} \qquad\,\qquad\,\qquad\,\qquad&(S\otimes\id)R =R^{-1}~,\quad& &(\id\otimes S)R^{-1}=R~,&\qquad\,\qquad
\end{flalign}
and
\begin{flalign}
  \label{eqn:Rproper33} R_{12}\,R_{13}\,R_{23} = R_{23}\,R_{13}\,R_{12}~.
\end{flalign}
\end{subequations}
Notice that the properties (\ref{eqn:Rproper2}) and (\ref{eqn:Rproper33}) imply that $R$ is a twist of the Hopf algebra
$H$, and from (\ref{eqn:Rproper1}) we obtain $H^R = H^\cop$. The Hopf algebra $H^\cop$ is quasitriangular with
$R$-matrix $R_{21}$.

Similar to the twist we introduce the notations $R=R^\alpha\otimes R_\alpha$ and $R^{-1}=\bar R^\alpha\otimes \bar R_\alpha$
(sum over $\alpha$ understood). 
\begin{defi}
 Let $(H,R)$ be a quasitriangular Hopf algebra. A left $H$-module algebra $A\in {^{H,\ra}}\AAA$ is called
{\it quasi-commutative} if, for all $a,b\in A$,
\begin{flalign}
 a\,b = (\bar R^\alpha\ra b)\,(\bar R_\alpha\ra a)~.
\end{flalign}
Similarly, a left $H$-module $A$-bimodule $V\in{^{H,\ra}_A}\MMM_A$ is {\it quasi-commutative} if, for all $a\in A$ and $v\in V$,
\begin{flalign}
\label{eqn:qcra}
 v\cdot a = (\bar R^\alpha\ra a)\cdot (\bar R_\alpha\ra v)~.
\end{flalign}
\end{defi}
Note that as a consequence of (\ref{eqn:qcra}) we have
\begin{flalign}
 a\cdot v = (R_\alpha\ra v)\cdot (R^\alpha\ra a)~.
\end{flalign}
\begin{ex}
 Consider the Hopf algebra of diffeomorphisms $U\Xi$ of a manifold $\MM$.
Then trivially the algebra of smooth functions $C^\infty(\MM)$ is quasi-commutative, 
since it is commutative and $R=1\otimes 1$.
Furthermore, the $C^\infty(\MM)$-bimodules $\Omega^1$ and $\Xi$ of one-forms and vector fields 
are quasi-commutative bimodules.
Let $\FF$ be a twist of $U\Xi$. Then $U\Xi^\FF$ is triangular with $R=\FF_{21}\FF^{-1}$
and the deformed algebra $(C^\infty(\MM),\star)$, as well as the deformed bimodules $\Omega^1_\star$ and $\Xi_\star$,
are quasi-commutative. In general, the twist quantization of any quasi-commutative algebra or bimodule
is quasi-commutative.
\end{ex}

\begin{propo}
 Let $(H,R)$ be a quasitriangular Hopf algebra and let $A\in{^{H,\ra}}\AAA$ and $V,W\in {^{H,\ra}_A}\MMM_A$ be quasi-commutative.
 Then for all $P\in \Hom_A(V,W)$, $a\in A$ and $v\in V$
\begin{flalign}
\label{eqn:righthomleft}
 P(a\cdot v) = (\bar R^\alpha\ra a)\cdot (\bar R_\alpha\RA P)(v)~.
\end{flalign}
Similarly, for all $P\in{_A}\Hom(V,W)$, $a\in A$ and $v\in V$
\begin{flalign}
 \label{eqn:lefthomright}
P(v\cdot a) = (R_\alpha\RA^\cop P)(v)\cdot(R^\alpha\ra a)~.
\end{flalign}
\end{propo}
\begin{proof}
 We show (\ref{eqn:righthomleft}), for all $P\in \Hom_A(V,W)$, $a\in A$ and $v\in V$,
\begin{flalign}
 \nn P(a\cdot v) &= P\bigl((R_\alpha\ra v)\cdot (R^\alpha\ra a)\bigr) = P(R_\alpha\ra v)\cdot (R^\alpha\ra a)\\
\nn &=(\bar R^\beta R^\alpha\ra a)\cdot \bar R_\beta\ra \bigl(P(R_\alpha\ra v)\bigr)\\
\nn &=(\bar R^\beta \bar R^\alpha\ra a)\cdot \bar R_\beta\ra \bigl(P(S(\bar R_\alpha)\ra v)\bigr)\\
\nn&=(\bar R^\alpha\ra a)\cdot \bar R_{\alpha_1}\ra \bigl(P(S(\bar R_{\alpha_2})\ra v)\bigr)\\
&= (\bar R^\alpha\ra a)\cdot (\bar R_\alpha\RA P)(v)~,
\end{flalign}
where in line three we have used (\ref{eqn:Rproper32}) and in line four (\ref{eqn:Rproper2}).
The property (\ref{eqn:lefthomright}) is proven analogously.

\end{proof}
This immediately leads us to the following
\begin{cor}\label{cor:homqc}
  Let $(H,R)$ be a quasitriangular Hopf algebra and let $A\in{^{H,\ra}}\AAA$ and $V,W\in {^{H,\ra}_A}\MMM_A$ be quasi-commutative.
Then $\Hom_A(V,W)\in {^{H,\RA}_A}\MMM_A$ and $\bigl({_A}\Hom(V,W)\bigr)^\op\in{^{H,\RA^\cop}_A}\MMM_A$ are quasi-commutative.
\end{cor}
\begin{proof}
 We first show that $\Hom_A(V,W)$ is quasi-commutative, for all $P\in \Hom_A(V,W)$ and $a\in A$,
\begin{flalign}
P\cdot a = P\circ l_a = l_{\bar R^\alpha\ra a}\circ (\bar R_\alpha\RA P) = (\bar R^\alpha\ra a) \cdot (\bar R_\alpha\RA P)~,
\end{flalign}
where in the second equality we have used (\ref{eqn:righthomleft}).

\noindent Similarly, we show that $\bigl({_A}\Hom(V,W)\bigr)^\op$ is quasi-commutative, for all $P\in {_A}\Hom(V,W)$ and $a\in A$,
\begin{flalign}
 P\cdot a = r_a\circ P = (\bar R_\alpha\RA^\cop P)\circ r_{\bar R^\alpha\ra a} =(\bar R^\alpha\ra a)\cdot (\bar R_\alpha\RA^\cop P)~,
\end{flalign}
where in the second equality we have used (\ref{eqn:lefthomright}).

\end{proof}
\begin{rem}
 Let $(H,R)$ be a quasitriangular Hopf algebra and let $A\in{^{H,\ra}}\AAA$ and $V\in {^{H,\ra}_A}\MMM_A$ be quasi-commutative.
Due to the corollary above we know that $\End_A(V)\in {^{H,\RA}_A}\AAA_A$ and
$\bigl({_A}\End(V)\bigr)^\op\in {^{H,\RA^\cop}_A}\AAA_A$ are quasi-commutative left $H$-module $A$-bimodules.
However, $\End_A(V)$ and $\bigl({_A}\End(V)\bigr)^\op$ are in general not quasi-commutative as algebras.
\end{rem}
\begin{ex}
 Let $A\in{^{H,\ra}}\AAA$ and $V\in {^{H,\ra}_A}\MMM_{A}$ be quasi-commutative.
Then the right dual module $V^\prime = \Hom_A(V,A)\in{^{H,\RA}_A}\MMM_A$ and the left
 dual module $^\prime V=\bigl({_A}\Hom(V,A)\bigr)^\op\in{^{H,\RA^\cop}_A}\MMM_A$ are quasi-commutative.
\end{ex}

Let us now focus on the quantization with twist $\FF=R$ of the quasi-commutative left $H$-modules $A\in {^{H,\ra}}\AAA$ and
$V\in {^{H,\ra}_A}\MMM_A$ in case of a triangular $R$-matrix. Note that in this case we have in addition
to (\ref{eqn:qcra}), for all $a\in A$ and $v\in V$,
\begin{flalign}
\label{eqn:qcla}
 a\cdot v= (R_\alpha\ra v)\cdot (R^\alpha\ra a) = (\bar R^\alpha\ra v)\cdot (\bar R_\alpha\ra a)~,
\end{flalign}
where the first equality always holds (as a consequence of (\ref{eqn:qcra})) and 
in the second equality we have used triangularity $R_{21} = R^{-1}$.
\begin{lem}\label{lem:copRdef}
 Let $(H,R)$ be a triangular Hopf algebra with twist $\FF=R$ and let $A\in {^{H,\ra}}\AAA$ and
$V\in {^{H,\ra}_A}\MMM_A$ be quasi-commutative. Then $H^R = H^\cop$, $A_{\star}=A^\op\in {^{H^\cop,\ra}}\AAA$ 
and $V_\star=V^\op\in{^{H^\cop,\ra}_{A^\op}}\MMM_{A^\op}$, where $\star$ denotes the deformation
associated to the twist $\FF=R$ (cf.~Theorems \ref{theo:algebradef} and \ref{theo:moduledef}).
\end{lem}
\begin{proof}
 The proof that $H^R=H^\cop$ is standard. Let us consider $A_\star$. We have, for all $a,b\in A$,
\begin{flalign}
 a\star b = (\bar R^\alpha\ra a)\,(\bar R_\alpha\ra b) = b\,a=\mu_\op(a\otimes b)~,
\end{flalign}
where in the second equality we have used quasi-commutativity of $A$.
Similarly we have for all $a\in A$ and $v\in V$
\begin{subequations}
\begin{flalign}
 a\star v &= (\bar R^\alpha\ra a)\cdot (\bar R_\alpha \ra v) = v\cdot a = a\cdot^\op v~,\\
 v\star a &= (\bar R^\alpha\ra v)\cdot (\bar R_\alpha \ra a) = a\cdot v = v\cdot^\op a~,
\end{flalign}
\end{subequations}
where in the second equality we have used (\ref{eqn:qcra}) and (\ref{eqn:qcla}), respectively.
Thus, $A_\star =A^\op$ and $V_\star = V^\op$.

\end{proof}
Using this lemma and Theorems \ref{theo:endodef} and \ref{theo:homodef} we
find an isomorphism between left and right $A$-linear endomorphisms and homomorphisms.
\begin{theo}\label{theo:qclriso}
 Let $(H,R)$ be a triangular Hopf algebra with twist $\FF=R$ and let $A\in {^{H,\ra}}\AAA$ and
$V,W\in {^{H,\ra}_A}\MMM_A$ be quasi-commutative. 
Then there is an isomorphism
\begin{flalign}
 D_R: (\End_{A}(V)_\star)^\op \to \bigl({_A}\End(V)\bigr)^\op~,\quad P\mapsto D_R(P) 
 = (\bar R^\alpha\RA P)\circ \bar R_\alpha\ra~
\end{flalign}
between the left $H$-module $A$-bimodule algebras $(\End_{A}(V)_\star)^\op\in {^{H,\RA}_A}\AAA_A$ and
$({_A}\End(V))^\op\in {^{H,\RA^\cop}_A}\AAA_A$. 

\noindent Similarly, there is an isomorphism (denoted by the same symbol)
\begin{flalign}
 D_R: \Hom_{A}(V,W) \to \bigl({_A}\Hom(V,W)\bigr)^\op~,\quad 
P\mapsto D_R(P) = (\bar R^\alpha\RA P)\circ \bar R_\alpha\ra~
\end{flalign}
between the left $H$-module $A$-bimodules $\Hom_{A}(V,W)\in {^{H,\RA}_A}\MMM_A$ and
$({_A}\Hom(V,W))^\op\in {^{H,\RA^\cop}_A}\MMM_A$. 
\end{theo}
\begin{proof}
We first prove the statement for endomorphisms.
From Proposition \ref{propo:endomodulealgebra} we know that
$\End_A(V)\in {^{H,\RA}_A}\AAA_A$ and from Theorem \ref{theo:algebradef}
and Theorem \ref{theo:moduledef} we obtain $\End_A(V)_\star\in {^{H^R,\RA}_{A_\star}}\AAA_{A_\star}$.
Lemma \ref{lem:copRdef} leads to
$\End_A(V)_\star\in {^{H^\cop,\RA}_{A^\op}}\AAA_{A^\op}$ and Lemma
\ref{lem:opposite} implies $\bigl(\End_A(V)_\star\bigr)^\op \in {^{H,\RA}_{A}}\AAA_{A}$.
By Proposition \ref{propo:opposite} we also have $\bigl({_A}\End(V)\bigr)^\op\in {^{H,\RA^\cop}_A}\AAA_A$.
Thus, the module structure is as claimed above.

\noindent From Theorem \ref{theo:endodef} we have a left $H^\cop$-module $A^\op$-bimodule algebra
isomorphism
\begin{flalign}
 D_R:\End_A(V)_\star \to \End_{A_\star}(V_\star)~.
\end{flalign}
Using Lemma \ref{lem:copRdef} we find $\End_{A_\star}(V_\star) = \End_{A^\op}(V^\op)\simeq {_A}\End(V)$, where
the last isomorphism is canonical. This observation leads to a
 left $H^\cop$-module $A^\op$-bimodule algebra isomorphism
\begin{flalign}
 D_R:\End_A(V)_\star \to {_A}\End(V)~,
\end{flalign}
and the proof follows since this map canonically induces a left $H$-module $A$-bimodule algebra isomorphism
(denoted by the same symbol)
\begin{flalign}
 D_R:\bigl(\End_A(V)_\star\bigr)^\op \to \bigl({_A}\End(V)\bigr)^\op~.
\end{flalign}

The construction of $D_R$ for the homomorphisms is analogous, leading to a left
$H$-module $A$-bimodule isomorphism
\begin{flalign}
D_R: \bigl(\Hom_{A}(V,W)_\star\bigr)^\op \to \bigl({_A}\Hom(V,W)\bigr)^\op~.
\end{flalign}
It remains to show that $(\Hom_{A}(V,W)_\star)^\op = \Hom_{A}(V,W)$ as left $H$-module $A$-bimodules.
Since $\Hom_{A}(V,W)$ is quasi-commutative (cf.~Corollary \ref{cor:homqc}),
Lemma \ref{lem:copRdef} implies $ (\Hom_{A}(V,W)_\star)^\op = (\Hom_{A}(V,W))^{\op\,\op} = \Hom_{A}(V,W)$.

\end{proof}

By Remark \ref{rem:Dinv} and triangularity $R_{21}=R^{-1}$ we obtain for the inverse of $D_R$, respectively,
\begin{subequations}
\begin{flalign}
  D_R^{-1}: \bigl({_A}\End(V)\bigr)^\op\to (\End_{A}(V)_\star)^\op~,\quad P\mapsto D^{-1}_R(P) 
 = (\bar R_\alpha\RA^\cop P)\circ \bar R^\alpha\ra~
\end{flalign}
and
\begin{flalign}
 D^{-1}_R: \bigl({_A}\Hom(V,W)\bigr)^\op \to \Hom_{A}(V,W)~,\quad 
P\mapsto D^{-1}_R(P) = (\bar R_\alpha\RA^\cop P)\circ \bar R^\alpha\ra~.
\end{flalign}
\end{subequations}

\begin{ex}
 Let $(H,R)$ be a triangular Hopf algebra and let $A\in{^{H,\ra}}\AAA$ and $V\in {^{H,\ra}_A}\MMM_{A}$ be quasi-commutative.
 Then by Theorem \ref{theo:qclriso} there is a left $H$-module $A$-bimodule isomorphism between the right dual $V^\prime = \Hom_A(V,A)$
and the left dual $^\prime V = \bigl({_A}\Hom(V,A)\bigr)^\op$.
\end{ex}

\section{\label{sec:productmodhom}Product module homomorphisms}
Given two $\bfK$-modules $V,W\in\MMM$ we can consider the tensor product
module (over $\bfK$) $V\otimes W\in\MMM$. 
Provided $\bfK$-linear maps $P\in \Hom_\bfK(V,\widetilde{V})$ and $Q\in\Hom_\bfK(W,\widetilde{W})$
between the $\bfK$-modules $V,W,\widetilde{V},\widetilde{W}\in\MMM$ we define
a $\bfK$-linear map $P\otimes Q\in\Hom_\bfK(V\otimes W,\widetilde{V}\otimes \widetilde{W})$
by
\begin{flalign}
\label{eqn:Ktensorhom}
(P\otimes Q)(v\otimes w):=P(v)\otimes Q(w)~.
\end{flalign}
If $\widehat{V},\widehat{W}\in \MMM$ are two further $\bfK$-modules
and $\widetilde{P}\in \Hom_\bfK(\widetilde{V},\widehat{V})$, $\widetilde{Q}\in \Hom_\bfK(\widetilde{W},\widehat{W})$,
then
\begin{flalign}
\label{eqn:Ktensorhomcomp}
\bigl(\widetilde{P}\otimes \widetilde{Q}\bigr)\circ\bigl(P\otimes Q  \bigr) = (\widetilde{P}\circ P)\otimes (\widetilde{Q}\circ Q)~.
\end{flalign}

Let us now study left $H$-modules $V,W,\widetilde{V},\widetilde{W}\in{^{H,\ra}}\MMM$ over a Hopf algebra $H$.
Employing the coproduct on $H$ we have $V\otimes W\in{^{H,\ra}}\MMM$ by defining,
for all $\xi\in H$, $v\in V$ and $w\in W$,
\begin{flalign}
\label{eqn:productaction}
\xi\ra (v\otimes w) := (\xi_1\ra v)\otimes (\xi_2\ra w)~.
\end{flalign}
The $\bfK$-modules $\Hom_\bfK(V,\widetilde{V})$, $\Hom_\bfK(W,\widetilde{W})$ and 
$\Hom_\bfK(V\otimes W,\widetilde{V}\otimes \widetilde{W})$ 
can be equipped with a left $H$-module structure by employing the adjoint action.
We consider now the action of $\xi\in H$ on the tensor product map (\ref{eqn:Ktensorhom}).
Using (\ref{eqn:Ktensorhomcomp}) and (\ref{eqn:productaction}) we obtain
\begin{flalign}
\nn \xi\RA (P\otimes Q) &= (\xi_1\ra\,\otimes\xi_2\ra\,) \circ (P\otimes Q)\circ (S(\xi_4)\ra\,\otimes S(\xi_3)\ra)\\
\nn &= \bigl( \xi_1\ra\,\circ P\circ S(\xi_4)\ra\, \bigr)\otimes \bigl(\xi_2\ra\,\circ Q\circ S(\xi_3)\ra\,\bigr)\\
 \label{eqn:helpprodhom}&=\bigl( \xi_1\ra\,\circ P\circ S(\xi_3)\ra\, \bigr)\otimes  (\xi_2\RA Q )~.
\end{flalign}
This is for a non-cocommutative Hopf algebra different to the map $(\xi_1\RA P)\otimes (\xi_2\RA Q)$,
i.e.~the tensor product of $\bfK$-linear maps (\ref{eqn:Ktensorhom}) is in general incompatible with the
left $H$-module structure induced by (\ref{eqn:productaction}).
This incompatibility can be understood as follows: Thinking of $\bfK$-linear 
maps as acting from left to right, the ordering on the left hand side of
(\ref{eqn:Ktensorhom})  is $P,Q,v,w$, while the ordering on the right hand side is $P,v,Q,w$,
i.e.~$P$ and $v$ do not appear properly ordered in the definition (\ref{eqn:Ktensorhom}).
For a quasitriangular Hopf algebra $(H,R)$ this can be improved by redefining the tensor product of
$\bfK$-linear maps, see also \cite{Majid:1996kd} Chapter 9.3.
\begin{defi}\label{defi:Rtensor}
Let $(H,R)$ be a quasitriangular Hopf algebra and $V,W,\widetilde{V},\widetilde{W}\in{^{H,\ra}}\MMM$
be left $H$-modules. The {\it $R$-tensor product} of $\bfK$-linear maps is defined by, for
all $P\in\Hom_\bfK(V,\widetilde{V})$ and $Q\in\Hom_\bfK(W,\widetilde{W})$,
\begin{flalign}
\label{eqn:Rtensor}
P\otimes_R Q := (P\circ \bar R^\alpha\ra\,)\otimes (\bar R_\alpha\RA Q)\in \Hom_\bfK(V\otimes W,\widetilde{V}\otimes\widetilde{W})~,
\end{flalign}
where $\otimes$ is defined in (\ref{eqn:Ktensorhom}).
\end{defi}
We can rewrite (\ref{eqn:Rtensor}) in a way convenient for the further investigations
\begin{flalign}
\nn P\otimes_R Q &= (P\circ \bar R^\alpha\ra\,)\otimes (\bar R_\alpha\RA Q)\\
 \nn &= (P\circ \bar R^\alpha\ra\,)\otimes \bigl(\bar R_{\alpha_1}\ra\,\circ Q\circ S(\bar R_{\alpha_{2}})\ra\, \bigr)\\
 \nn &= (P\circ \bar R^\alpha\bar R^\beta\ra\,)\otimes \bigl(\bar R_{\alpha}\ra\,\circ Q\circ S(\bar R_{\beta})\ra\, \bigr)\\
 \nn &= (P\circ \bar R^\alpha R^\beta\ra\,)\otimes \bigl(\bar R_{\alpha}\ra\,\circ Q\circ R_{\beta}\ra\, \bigr)\\
\label{eqn:Rtensorsimple} &= (P\otimes \id)\circ \tau_{R}\circ (Q\otimes \id ) \circ \tau^{-1}_{R}~,
\end{flalign}
where $\tau_R$ is the $R$-flip map and $\tau_R^{-1}$ its inverse, for all $v\in V$ and $w\in W$,
\begin{subequations}\label{eqn:Rflipmap}
\begin{flalign}
 \tau_{R}(w\otimes v)= (\bar R^\alpha\ra v) \otimes (\bar R_\alpha\ra w)~,\\
\tau_{R}^{-1}(v\otimes w) = ( R_\alpha\ra w)\otimes ( R^\alpha\ra v)~.
\end{flalign}
\end{subequations}

We now show that the $R$-tensor product is compatible with the left $H$-module structure, associative and
satisfies a generalized composition law.
\begin{propo}
Let $(H,R)$ be a quasitriangular Hopf algebra and $V,W,Z,\widetilde{V},\widetilde{W},\widetilde{Z},\widehat{V},\widehat{W}
\in{^{H,\ra}}\MMM$ be left $H$-modules.
The $R$-tensor product is compatible with the left $H$-module structure, i.e.~for all
$\xi\in H$, $P\in\Hom_\bfK(V,\widetilde{V})$ and $Q\in\Hom_\bfK(W,\widetilde{W})$,
\begin{subequations}
\begin{flalign}
\label{eqn:RtensorHmod}
\xi\RA (P\otimes_R Q) = (\xi_1\RA P)\otimes_R (\xi_2\RA Q)~.
\end{flalign}
Furthermore, the $R$-tensor product is associative, i.e.~for all
$P\in\Hom_\bfK(V,\widetilde{V})$, $Q\in\Hom_\bfK(W,\widetilde{W})$ and $T\in\Hom_\bfK(Z,\widetilde{Z})$,
\begin{flalign}
\label{eqn:Rtensorass}
\bigl(P\otimes_R Q\bigr)\otimes_R T = P\otimes_R \bigl(Q\otimes_R T\bigr)~,
\end{flalign}
and satisfies the composition law, for all $P\in\Hom_\bfK(V,\widetilde{V})$,  $Q\in\Hom_\bfK(W,\widetilde{W})$,
$\widetilde{P}\in\Hom_{\bfK}(\widetilde{V},\widehat{V})$ and $\widetilde{Q}\in\Hom_\bfK(\widetilde{W},\widehat{W})$,
\begin{flalign}
\label{eqn:Rtensorcirc}
\bigl(\widetilde{P}\otimes_R\widetilde{Q}\bigr)\circ \bigl(P\otimes_R Q\bigr) = \bigl(\widetilde{P}\circ (\bar R^\alpha\RA P)\bigr)
\otimes_R\bigl((\bar R_\alpha\RA \widetilde{Q})\circ Q\bigr)~.
\end{flalign}
\end{subequations}
\end{propo}
\begin{proof}
The $R$-flip map $\tau_R$ and its inverse $\tau_R^{-1}$ are left $H$-module isomorphisms, i.e.~for all $\xi\in H$,
\begin{flalign}
\xi\RA\tau_R = \epsilon(\xi)\,\tau_R~,\quad \xi\RA\tau_R^{-1} = \epsilon(\xi)\,\tau_R^{-1}~. 
\end{flalign}
Using this and (\ref{eqn:Rtensorsimple}) we can prove (\ref{eqn:RtensorHmod})
\begin{flalign}
\nn \xi\RA (P\otimes_R Q) &= \xi\RA\bigl((P\otimes \id)\circ \tau_R\circ (Q\otimes \id)\circ \tau_R^{-1}\bigr)\\
\nn &= (\xi_1\RA P\otimes \id)\circ \tau_R \circ (\xi_2\RA Q\otimes \id)\circ \tau_R^{-1}\\
&=(\xi_1\RA P)\otimes_R(\xi_2\RA Q)~,
\end{flalign}
where in the second line we have used that $\xi\RA(P\otimes \id) = \xi\RA P\otimes \id$, for all $\xi\in H$, 
which follows from (\ref{eqn:helpprodhom}).

We now prove (\ref{eqn:Rtensorass}). The left hand side of (\ref{eqn:Rtensorass}) can be simplified
as follows
\begin{subequations}
\begin{flalign}
 \nn \bigl(P\otimes_R Q\bigr)\otimes_R T &= 
\Bigl(\bigl(P\otimes_R Q\bigr)\circ \bigl(\bar R^\alpha_1\ra\,\otimes \bar R^\alpha_2\ra\,\bigr)\Bigr)\otimes \Bigl(\bar R_\alpha\RA T\Bigr)\\
\nn &= \Bigl(P\circ \bar R^\beta\bar R^\alpha_1 \ra\,\Bigr)\otimes \Bigl((\bar R_\beta\RA Q)\circ \bar R^\alpha_2\ra\,\Bigr)\otimes \Bigl(\bar R_\alpha\RA T\Bigr)\\
&= \Bigl(P\circ \bar R^\beta\bar R^\alpha \ra\,\Bigr)\otimes \Bigl((\bar R_\beta\RA Q)\circ \bar R^\gamma\ra\,\Bigr)\otimes\Bigl( \bar R_\gamma\bar R_\alpha\RA T\Bigr)~,
\end{flalign}
where in the third line we have used (\ref{eqn:Rproper2}). This is equal to the right hand side of
(\ref{eqn:Rtensorass})
\begin{flalign}
\nn P\otimes_R \bigl(Q\otimes_R T\bigr) 
\nn &= \Bigl(P\circ \bar R^\alpha\ra\,\Bigr)\otimes \Bigl(\bar R_{\alpha_1}\RA Q\otimes_R \bar R_{\alpha_2}\RA T\Bigr)\\
\nn &= \Bigl(P\circ \bar R^\alpha\ra\,\Bigr)\otimes \Bigl((\bar R_{\alpha_1}\RA Q)\circ \bar R^\gamma\ra\,\Bigr)\otimes \Bigl(\bar R_\gamma \bar R_{\alpha_2}\RA T\Bigr)\\
 &= \Bigl(P\circ \bar R^\beta \bar R^\alpha\ra\,\Bigr)\otimes \Bigl((\bar R_{\beta}\RA Q)\circ \bar R^\gamma\ra\,\Bigr)\otimes \Bigl(\bar R_\gamma \bar R_{\alpha}\RA T\Bigr)~,
\end{flalign}
\end{subequations}
where in the third line we have used (\ref{eqn:Rproper2}). A diagrammatic proof of the associativity property is given
in the Appendix \ref{app:diagram}.

Finally, we show (\ref{eqn:Rtensorcirc}). The left hand side of (\ref{eqn:Rtensorcirc}) can be written with the help
of (\ref{eqn:Rtensorsimple}) as follows
\begin{flalign}
\label{eqn:Rcircstep1}
 \bigl(\widetilde{P}\otimes_R\widetilde{Q}\bigr)\circ \bigl(P\otimes_R Q\bigr)=
(\widetilde{P}\otimes \id)\circ\tau_R \circ (\widetilde{Q}\otimes\id)\circ\tau_R^{-1}\circ (P\otimes\id)\circ \tau_R \circ (Q\otimes \id)\circ \tau_R^{-1}~.
\end{flalign}
For the middle part we obtain
\begin{flalign}
 \nn \tau_R \circ (\widetilde{Q}\otimes\id)\circ\tau_R^{-1}\circ (P\otimes\id) &= 
\bigl(\bar R^\beta R^\alpha\ra\,\circ P \bigr)\otimes \bigl(\bar R_\beta\ra\,\circ \widetilde{Q}\circ R_\alpha\ra\,\bigr)\\
\nn &= \bigl(\bar R^\beta \bar R^\alpha\ra\,\circ P\bigr) \otimes \bigl(\bar R_\beta\ra\,\circ \widetilde{Q}\circ S(\bar R_\alpha)\ra\,\bigr)\\
\nn &= \bigl(\bar R^\alpha\ra\,\circ P \bigr)\otimes\bigl( \bar R_{\alpha_1}\ra\,\circ \widetilde{Q}\circ S(\bar R_{\alpha_2})\ra\,\bigr)\\
\nn &= \bigl(\bar R^\alpha\ra\,\circ P \bigr)\otimes \bigl(\bar R_\alpha\RA \widetilde{Q}\bigr)~\\
\nn &= \bigl((\bar R^\alpha_1\RA P)\circ \bar R^\alpha_2\ra\,\bigr)\otimes \bigl(\bar R_\alpha\RA\widetilde{Q}\bigr)~\\
\nn &= \bigl((\bar R^\alpha\RA P)\circ \bar R^\beta\ra\,\bigr)\otimes \bigl(\bar R_\beta\bar R_\alpha\RA\widetilde{Q}\bigr)~\\
\label{eqn:Rcircstep2}&=(\bar R^\alpha\RA P)\otimes_R (\bar R_\alpha\RA \widetilde{Q})~,
\end{flalign}
where in line two we have used (\ref{eqn:Rproper32}) and in line three and six (\ref{eqn:Rproper2}).
Inserting (\ref{eqn:Rcircstep2}) into (\ref{eqn:Rcircstep1}) and using again (\ref{eqn:Rtensorsimple})
we find
\begin{flalign}
 \nn \bigl(\widetilde{P}\otimes_R\widetilde{Q}\bigr)\circ \bigl(P\otimes_R Q\bigr)&=
(\widetilde{P}\otimes \id)\circ (\bar R^\alpha\RA P\otimes \id) \circ \tau_R \circ (\bar R_\alpha\RA \widetilde{Q}\otimes \id)\circ 
(Q\otimes\id)\circ\tau_R^{-1}\\
&= \bigl(\widetilde{P}\circ (\bar R^\alpha\RA P)\bigr)
\otimes_R\bigl((\bar R_\alpha\RA \widetilde{Q})\circ Q\bigr)~,
\end{flalign}
where in the second line we have also used (\ref{eqn:Ktensorhomcomp}).

\end{proof}

Let us now consider the case $V,\widetilde{V}\in{^{H,\ra}}\MMM$ and
$W,\widetilde{W}\in{^{H,\ra}}\MMM_A$, where $A\in {^{H,\ra}}\AAA$.
Then we can equip $V\otimes W$ (as well as $\widetilde{V}\otimes \widetilde{W}$) 
with a right $A$-module structure by defining $(v\otimes w)\cdot a := v\otimes(w\cdot a)$,
for all $v\in V$, $w\in W$ and $a\in A$. Moreover, we have
$V\otimes W \in {^{H,\ra}}\MMM_A$ and $\widetilde{V}\otimes \widetilde{W} \in {^{H,\ra}}\MMM_A$ 
by employing the left $H$-action (\ref{eqn:productaction}).
We obtain the following
\begin{propo}
 Let $(H,R)$ be a quasitriangular Hopf algebra, $A\in {^{H,\ra}}\AAA$, $V,\widetilde{V}\in{^{H,\ra}}\MMM$
and $W,\widetilde{W}\in{^{H,\ra}}\MMM_A$. Then we have for all
$P\in\Hom_\bfK(V,\widetilde{V})$ and $Q\in\Hom_A(W,\widetilde{W})$
\begin{flalign}
 P\otimes_R Q \in\Hom_A(V\otimes W,\widetilde{V}\otimes\widetilde{W})~.
\end{flalign}
\end{propo}
\begin{proof}
The proof follows from a short calculation, for all $a\in A$, $v\in V$ and $w\in W$, 
\begin{flalign}
\nn (P\otimes_R Q)\bigl( (v\otimes w)\cdot a \bigr) &=(P\otimes_R Q)\bigl( v\otimes (w\cdot a) \bigr) 
= P\bigl(\bar R^\alpha\ra v\bigr)\otimes (\bar R_\alpha\RA Q)(w\cdot a)\\
\nn &=P\bigl(\bar R^\alpha\ra v\bigr)\otimes \bigl((\bar R_\alpha\RA Q)(w)\cdot a\bigr)\\
&=\bigl((P\otimes_R Q)(v\otimes w)\bigr)\cdot a~,
\end{flalign}
where in the second line we have used that
$\xi\RA Q\in\Hom_A(W,\widetilde{W})$, for all $\xi\in H$.

\end{proof}

We now consider the case $V,\widetilde{V}\in{^{H,\ra}}\MMM_A$ and
$W,\widetilde{W}\in{^{H,\ra}_A}\MMM_A$, where $A\in {^{H,\ra}}\AAA$.
Let also $A$, $W$ and $\widetilde{W}$ be quasi-commutative.
We consider the tensor product over $A$ and have
$V\otimes_A W \in {^{H,\ra}}\MMM_A$ and $\widetilde{V}\otimes_A \widetilde{W} \in {^{H,\ra}}\MMM_A$.
The image of $(v,w)$ under the natural map $V\times W\to V\otimes_A W$ is denoted by $v\otimes_A w$.
We obtain the following
\begin{propo}\label{propo:RtensorAhom}
 Let $(H,R)$ be a quasitriangular Hopf algebra, $A\in {^{H,\ra}}\AAA$, $V,\widetilde{V}\in{^{H,\ra}}\MMM_A$
and $W,\widetilde{W}\in{^{H,\ra}_A}\MMM_A$. Let also $A$, $W$ and $\widetilde{W}$ be quasi-commutative.
 Then we have for all $P\in\Hom_A(V,\widetilde{V})$ and $Q\in\Hom_A(W,\widetilde{W})$
\begin{flalign}
 P\otimes_R Q \in\Hom_A(V\otimes_A W,\widetilde{V}\otimes_A\widetilde{W})~.
\end{flalign}
\end{propo}
\begin{proof}
 It remains to show that $P\otimes_R Q$ is compatible with middle $A$-linearity,
for all $a\in A$, $v\in V$ and $w\in W$,
\begin{flalign}
 \nn(P\otimes_R Q)\bigl((v\cdot a)\otimes_A w\bigr) &=  P\bigl(\bar R^\alpha\ra (v\cdot a)\bigr)\otimes_A (\bar R_\alpha\RA Q)(w)\\
\nn &=P(\bar R^\alpha_1\ra v)\cdot(\bar R^\alpha_2 \ra a) \otimes_A (\bar R_\alpha\RA Q)(w)\\
\nn &=P(\bar R^\alpha\ra v)\otimes_A (\bar R^\beta \ra a) \cdot (\bar R_\beta \bar R_\alpha\RA Q)(w)\\
\nn &=P(\bar R^\alpha\ra v)\otimes_A (\bar R_\alpha\RA Q)(a\cdot w)\\
&=(P\otimes_R Q)\bigl(v\otimes_A (a\cdot w)\bigr)~,
\end{flalign}
where in the third line we have used (\ref{eqn:Rproper2}) and in line four (\ref{eqn:righthomleft}).

\end{proof}

We finish this section by studying the behavior of $P\otimes_R Q$ under twist quantization.
\begin{theo}\label{theo:promodhomdef}
 Let $(H,R)$ be a quasitriangular Hopf algebra with twist $\FF\in H\otimes H$ and  
$V,W,\widetilde{V},\widetilde{W}\in{^{H,\ra}}\MMM$. Then for all $P\in \Hom_\bfK(V,\widetilde{V})$
and $Q\in\Hom_\bfK(W,\widetilde{W})$ we have
\begin{flalign}
\label{eqn:prodhomdef}
D_\FF(P)\otimes_{R^\FF}D_\FF(Q) = 
\iota^{-1}\circ D_\FF\bigl((\bar f^\alpha\RA P)\otimes_R (\bar f_\alpha\RA Q)\bigr)\circ \iota\,~,
\end{flalign}
where $R^\FF = \FF_{21}\,R\,\FF^{-1}$ and $\iota = \FF^{-1}\ra\,$.
\end{theo}
\begin{proof}
 Firstly, note that
\begin{flalign}
 \nn (\bar f^\alpha\RA P)\otimes_R (\bar f_\alpha\RA Q) &= 
\bigl(\bar f^\alpha\RA P\otimes \id\bigr)\circ \tau_R \circ \bigl(\bar f_\alpha\RA Q\otimes \id\bigr)\circ\tau_R^{-1}\\
\nn &= \bigl(\bar f^\alpha\RA(P\otimes \id)\bigr)\circ \tau_R \circ \bigl(\bar f_\alpha\RA (Q\otimes \id)\bigr)\circ\tau_R^{-1}\\
 &= (P\otimes \id)\circ_\star \tau_R\circ_\star (Q\otimes\id)\circ_\star\tau_R^{-1}~,
\end{flalign}
where in the last line we have used that $\tau_R$ and $\tau_R^{-1}$ are $H$-invariant, i.e.~$\xi\RA\tau_R =\epsilon(\xi)\,\tau_R$
and $\xi\RA\tau_R^{-1} =\epsilon(\xi)\,\tau_R^{-1}$ for all $\xi\in H$.
Acting with $D_\FF$ on this expression we obtain
\begin{flalign}
 D_\FF\bigl((\bar f^\alpha\RA P)\otimes_R (\bar f_\alpha\RA Q)\bigr) = D_\FF(P\otimes\id) \circ \tau_R \circ D_\FF(Q\otimes\id)\circ\tau_R^{-1}~,
\end{flalign}
where we again used that $\tau_R$ and $\tau_R^{-1}$ are $H$-invariant and thus $D_\FF(\tau_R) =\tau_R$ 
and $D_\FF(\tau_R^{-1})=\tau_R^{-1}$.
The quantization of $P\otimes \id$ (and also $Q\otimes \id$) can be simplified as follows
\begin{flalign}
 \nn D_\FF(P\otimes \id) &= \bigl((\bar f^\alpha\RA P)\circ \bar f_{\alpha_1}\ra\,\bigr)\otimes \bar f_{\alpha_2}\ra\,\\
\nn &=\bigl((\bar f^\alpha_1\bar f^\beta\RA P)\circ \bar f^\alpha_2\bar f_\beta f^\gamma \ra\,\bigr)\otimes \bar f_{\alpha}f_\gamma\ra\,\\
\nn&= \bigl(\bar f^\alpha\ra\,\circ D_\FF(P)\circ f^\gamma\ra\,\bigr)\otimes \bar f_\alpha f_\gamma\ra\, \\
&= \iota\circ (D_\FF(P)\otimes \id)\circ\iota^{-1}~,
\end{flalign}
where we have used in line two the $2$-cocycle property (\ref{eqn:twistprop1}) of the twist.
Defining $\tau_{R^\FF} := \iota^{-1}\circ \tau_R \circ\iota$ we obtain
\begin{flalign}
  D_\FF\bigl((\bar f^\alpha\RA P)\otimes_R (\bar f_\alpha\RA Q)\bigr) = 
\iota\circ (D_\FF(P)\otimes \id) \circ \tau_{R^\FF} \circ (D_\FF(Q)\otimes \id)\circ \tau_{R^\FF}^{-1}\circ \iota^{-1}~.
\end{flalign}
Note that $\tau_{R^\FF}(w\otimes v) = (\bar R^{\FF\alpha}\ra v) \otimes (\bar R^\FF_{\alpha}\ra w)$ and
$\tau_{R^\FF}^{-1}(v\otimes w) = (R_\alpha^\FF\ra w)\otimes(R^{\FF\alpha}\ra v)$,
which analogously to (\ref{eqn:Rtensorsimple}) leads to 
\begin{flalign}
 (D_\FF(P)\otimes \id) \circ \tau_{R^\FF} \circ (D_\FF(Q)\otimes \id)\circ \tau_{R^\FF}^{-1}
= \bigl( D_\FF(P)\circ \bar R^{\FF\alpha}\ra\,\bigr)\otimes \bigl(\bar R_\alpha^\FF\RA_\FF D_\FF(Q)\bigr)~.
\end{flalign}

\end{proof}

Let us understand (\ref{eqn:prodhomdef}) in more detail. 
Firstly, note that each left $H$-module $V\in{^{H,\ra}}\MMM$ is also a left $H^{\FF}$-module
$V\in{^{H^\FF,\ra}}\MMM$, since $H$ and $H^\FF$ are equal as algebras.
However, for defining the left $H$-module structure on the $\bfK$-module $V\otimes W\in\MMM$
(\ref{eqn:productaction}) we have made use of the coproduct in $H$. Analogously, we can equip
the $\bfK$-module $V\otimes W \in\MMM$ with a left $H^\FF$-module structure (thus also another $H$-module structure)
by using the coproduct in $H^\FF$. We denote this module by $V\otimes_\star W\in{^{H,\ra_\FF}}\MMM$
and the image of $(v,w)$ under the natural map $V\times W\to V\otimes_\star W$ by $v\otimes_\star w$.
The $\bfK$-linear map $\iota: V\otimes_\star W \to V\otimes W$ defined by
$\iota(v\otimes_\star w) =(\bar f^\alpha\ra v) \otimes (\bar f_\alpha\ra w)$ provides a
left $H$-module (and therewith also a left $H^\FF$-module) isomorphism $V\otimes W\simeq V\otimes_\star W$, since
\begin{flalign}
 \iota\bigl(\xi\ra_\FF(v\otimes_\star w)\bigr) = \iota\bigl((\xi_{1_\FF}\ra v)\otimes_\star (\xi_{2_\FF}\ra w)  \bigr)
= (\xi_{1}\bar f^\alpha\ra v)\otimes (\xi_{2}\bar f_\alpha\ra w) = \xi\ra\iota(v\otimes_\star w)~.
\end{flalign}
The inverse is given by $\iota^{-1}(v\otimes w)  = (f^\alpha\ra v)\otimes_\star(f_\alpha\ra w)$.

Note that $D_\FF\bigl((\bar f^\alpha\RA P)\otimes_R (\bar f_\alpha\RA Q)\bigr)$ is by construction an element in $\Hom_\bfK(V\otimes W,\widetilde{V}\otimes \widetilde{W})$.
Employing the isomorphism $\iota$, we can induce a homomorphism in
 $\Hom_\bfK(V\otimes_\star W,\widetilde{V}\otimes_\star \widetilde{W})$.
Theorem \ref{theo:promodhomdef} then simply states that the following diagram commutes
\begin{flalign}\label{eqn:promodhomdef}
 \xymatrix{
V\otimes_\star W \ar[d]_-{\iota}\ar[rrrr]^-{D_\FF(P)\otimes_{R^\FF}D_\FF(Q)} & & & &\widetilde{V}\otimes_\star\widetilde{W} \\
V\otimes W \ar[rrrr]_-{D_\FF\bigl((\bar f^\alpha\RA P)\otimes_R (\bar f_\alpha\RA Q)\bigr)}& & & &\widetilde{V}\otimes\widetilde{W}\ar[u]_-{\iota^{-1}}
}
\end{flalign}

We omit a discussion of left $A$-linear homomorphisms and their tensor products, since we do not require
these structures in the following.


\chapter{\label{chap:con}Bimodule connections}
\section{\label{sec:conbas}Connections on right and left modules}
We briefly review the notion of a connection on a right or left module, 
see e.g.~\cite{1038.58004,Madore:2000aq} for an introduction.

Let $A$ be a unital and associative algebra over $\bfK$. 
A {\it differential calculus} $\bigl(\Omega^\bullet,\wedge,\dd\bigr)$ over $A$ is an $\bbN^0$-graded algebra 
$\bigl(\Omega^\bullet = \bigoplus_{n\geq0} \Omega^n,\wedge\bigr)$ over $\bfK$,\footnote{
In order to stress the analogy to classical differential geometry we denote the 
product in $\Omega^\bullet$ by a wedge $\wedge$.
However, one has to be a bit careful with this notation, since in contrast to the
$\wedge$-product in differential geometry our wedge
 is not necessarily graded commutative.
} where $\Omega^0=A$, together with a
$\bfK$-linear map $\dd:\Omega^\bullet \to \Omega^\bullet$ of degree one, satisfying $\dd\circ\dd=0$ and
\begin{flalign}
 \dd(\omega\wedge\omega^\prime) = (\dd\omega)\wedge\omega^\prime + (-1)^{\deg(\omega)}\,\omega\wedge(\dd\omega^\prime)~,
\end{flalign}
for all $\omega,\omega^\prime\in\Omega^\bullet$.
The {\it differential} $\dd$ and product $\wedge$ give rise to $\bfK$-linear maps (denoted by the same symbols) 
$\dd:\Omega^n\to\Omega^{n+1}$ and $\wedge:\Omega^n\otimes \Omega^m\to\Omega^{n+m}$.
Note that in the hypotheses above the $\bfK$-modules $\Omega^n$, $n>0$, are 
$A$-bimodules, i.e.~$\Omega^n\in{_A}\MMM_A$.
\begin{ex}
 Let $\MM$ be an $N$-dimensional smooth manifold and $A=C^\infty(\MM)$ be the smooth 
and complex valued functions on $\MM$.
The exterior algebra of differential forms $\bigl(\Omega^\bullet=\bigoplus_{n\geq 0}\Omega^n,\wedge\bigr)$ 
is an $\bbN^0$-graded algebra over $\bbC$,
where $\Omega^0=A$ and $\Omega^{N+n}=0$, for all $n>0$. 
The exterior differential $\dd$ is a differential on $\bigl(\Omega^\bullet,\wedge\bigr)$,
leading to the differential calculus $\bigl(\Omega^\bullet,\wedge,\dd\bigr)$.
In this special case $\Omega^\bullet$ is graded commutative.

\noindent Another example is given by the twist deformed differential calculus
$\bigl(\Omega^\bullet[[\lambda]],\wedge_\star,\dd\bigr)$, see Chapter \ref{chap:basicncg}.
There, the algebra $\bigl(\Omega^\bullet[[\lambda]],\wedge_\star\bigr)$ is graded quasi-commutative, 
i.e.~
\begin{flalign}
 \omega\wedge_\star \omega^\prime = 
(-1)^{\deg(\omega)\,\deg(\omega^\prime)} ~(\bar R^\alpha\ra \omega^\prime)\wedge_\star (\bar R_\alpha\ra \omega)~.
\end{flalign}
\end{ex}

\begin{defi}\label{defi:connection}
 Let $A$ be a unital and associative algebra over $\bfK$ and $\bigl(\Omega^\bullet,\wedge,\dd\bigr)$
be a differential calculus over $A$.
A {\it connection on a right $A$-module} $V\in\MMM_A$ is a $\bfK$-linear map $\nab:V\to V\otimes_A\Omega^1$, satisfying
the right Leibniz rule, for all $v\in V$ and $a\in A$,
\begin{flalign}
\label{eqn:rightcon}
 \nab(v\cdot a) = (\nab v)\cdot a + v\otimes_A \dd a~.
\end{flalign}
Similarly, a {\it connection on a left $A$-module} $V\in{_A}\MMM$ is a $\bfK$-linear map 
$\nab:V\to \Omega^1\otimes_A V$, satisfying the left Leibniz rule, for all $v\in V$ and $a\in A$,
\begin{flalign}
\label{eqn:leftcon}
\nab(a\cdot v) = a\cdot (\nab v)  + \dd a\otimes_A v~.
\end{flalign}
In case $V\in {_A}\MMM_A$ is an $A$-bimodule, we say that a $\bfK$-linear map
$\nab:V\to V\otimes_A\Omega^1$ is a {\it right connection} on $V$, if (\ref{eqn:rightcon})
is satisfied. Analogously, we say that a $\bfK$-linear map $\nab: V\to \Omega^1\otimes_A V$
is a {\it left connection} on $V$, if (\ref{eqn:leftcon}) is satisfied.
\end{defi}
 
We denote by $\Con_A(V)$ the set of all connections on a right $A$-module $V\in\MMM_A$
and by ${_A}\Con(V)$ the set of all connections on a left $A$-module $V\in{_A}\MMM$.
Similarly, we denote by $\Con_A(V)$ and ${_A}\Con(V)$ respectively the set of all right and left
connections on an $A$-bimodule $V\in{_A}\MMM_A$.
Note that given any connection $\nab \in \Con_A(V)$ and any homomorphism
$P\in\Hom_A(V,V\otimes_A\Omega^1)$, then $\nab^\prime = \nab + P \in\Con_A(V)$ is again a connection, since
\begin{flalign}
 \nab^\prime(v\cdot a) = \nab(v\cdot a) + P(v\cdot a) = (\nab v)\cdot a + P(v)\cdot a + v\otimes_A\dd a
= (\nab^\prime v)\cdot a + v\otimes_A \dd a~.
\end{flalign}
Similarly, let $\nab\in {_A}\Con(V)$ and $P\in{_A}\Hom(V,\Omega^1\otimes_A V)$,
then $\nab^\prime = \nab+P\in{_A}\Con(V)$, since
\begin{flalign}
 \nab^\prime(a\cdot v) = \nab(a\cdot v) + P(a\cdot v) = a\cdot (\nab v) + a\cdot P(v) + \dd a\otimes_A v
=a\cdot(\nab^\prime v) + \dd a\otimes_A v ~.
\end{flalign}
This means that $\Con_A(V)$ and ${_A}\Con(V)$ are affine spaces over the homomorphisms
 $\Hom_A(V,V\otimes_A\Omega^1)$ or ${_A}\Hom(V,\Omega^1\otimes_A V)$, respectively.
\begin{rem}
In this part the focus is on connections as defined above. A covariant derivative, as defined in 
Chapter \ref{chap:basicncg}, can be derived canonically from a connection:
Let $\nab:V\to \Omega^1\otimes_A V $ be a left connection.
A covariant derivative $\nab_u:V\to V$ along $u\in\Omega^{1\prime}=\Hom_A(\Omega^1,A)$ 
is obtained by the following composition of maps
\begin{flalign}
 \xymatrix{
V \ar[d]_-{\nab}\ar[rr]^-{\nab_u} & & \ar[d] V\\
\Omega^1\otimes_A V \ar[rr]_-{u\otimes\id} & & A\otimes_A V\ar[u]_-{\simeq}
}
\end{flalign} 
An analogous statement holds for right connections.
\end{rem}


\section{Quantization of connections}
We study the twist quantization of connections on left and right modules.
Let $H$ be a Hopf algebra and $A\in{^{H,\ra}}\AAA$ be a left $H$-module algebra.
Let further $\bigl(\Omega^\bullet,\wedge,\dd\bigr)$ be a {\it left $H$-covariant differential calculus} over $A$,
i.e.~$\Omega^\bullet\in {^{H,\ra}}\AAA$ is a left $H$-module algebra, the action $\ra$ is of degree zero and 
the differential is equivariant, for all $\xi\in H$ and $\omega\in \Omega^\bullet$,
\begin{flalign}
\label{eqn:equivar}
 \xi\ra (\dd\omega) = \dd(\xi\ra \omega)~.
\end{flalign}
As a consequence, we have for all $n\geq 0$, $\Omega^n\in{^{H,\ra}_A}\MMM_A$.
This is exactly the setting we face in (noncommutative) gravity. A left $H$-covariant
differential calculus can be quantized to yield a left $H^\FF$-covariant differential calculus.
\begin{lem}\label{lem:dcdef}
 Let $H$ be a Hopf algebra with twist $\FF\in H\otimes H$, $A\in{^{H,\ra}}\AAA$ be a left $H$-module algebra
and $\bigl(\Omega^\bullet,\wedge,\dd\bigr)$ be a left $H$-covariant differential calculus over $A$.
Then $\bigl(\Omega^\bullet,\wedge_\star,\dd\bigr)$ is a left $H^\FF$-covariant differential calculus
over $A_\star$.
\end{lem}
\begin{proof}
 By Theorem \ref{theo:algebradef} $\bigl(\Omega^\bullet,\wedge_\star\bigr)$ is a left $H^\FF$-module
algebra. It is $\bbN^0$-graded and we have $(\Omega^0,\wedge_\star) = A_\star$.
Due to the equivariance of the differential, $\dd$ is also a differential on $\bigl(\Omega^\bullet,\wedge_\star\bigr)$.

\end{proof}

Let $V\in {^{H,\ra}}\MMM_A$ be a left $H$-module right $A$-module. We consider the case of left $H$-modules left
$A$-modules ${^{H,\ra}_A}\MMM$ later.
Since $\Con_A(V)\subseteq \Hom_\bfK(V,V\otimes_A\Omega^1)$ we can act with the adjoint 
action $\RA$ on $\Con_A(V)$, for all $\xi\in H$ and $\nab\in \Con_A(V)$,
\begin{flalign}
\label{eqn:conadjoint}
 \xi\RA\nab := \xi_1\ra\,\circ \nab \circ S(\xi_2)\ra\,~.
\end{flalign}
Note that for general $\xi\in H$,  $\xi\RA\nab\not\in\Con_A(V)$, since
\begin{flalign}
 \nn(\xi\RA\nab)(v\cdot a) &= \xi_1\ra\bigl(\nab(S(\xi_2)_1\ra v\cdot S(\xi_2)_2\ra a)\bigr)\\
\nn&= \xi_1\ra\Bigl(\bigl(\nab (S(\xi_3)\ra v)\bigr) \cdot (S(\xi_2)\ra a) + (S(\xi_3)\ra v)\otimes_A \dd (S(\xi_2)\ra a)\Bigr)\\
\label{eqn:Hcon}&=\bigl((\xi\RA \nab) v\bigr)\cdot a + \epsilon(\xi)\,v\otimes_A \dd a~.
\end{flalign}
In particular, for $\xi\in H$ with $\epsilon(\xi)=0$ we obtain $\xi\RA \nab \in\Hom_A(V,V\otimes_A \Omega^1)$.
However, for $\epsilon(\xi)=1$ we have $\xi\RA\nab\in\Con_A(V)$.

We now show that given a twist $\FF\in H\otimes H$ of the Hopf algebra $H$, then there is an isomorphism
$\Con_A(V)\simeq \Con_{A_\star}(V_\star)$ between connections on the undeformed module $V\in {^{H,\ra}}\MMM_A$
and the deformed module $V_\star\in{^{H^\FF,\ra}}\MMM_{A_\star}$.
For this we make use of the quantization map $D_\FF$ on module homomorphisms, see Theorem \ref{theo:homodef}.
Since $\Con_A(V)\subseteq \Hom_\bfK(V,V\otimes_A\Omega^1)$ we have $D_\FF(\nab) \in \Hom_\bfK(V_\star,(V\otimes_A\Omega^1)_\star)$.
In order to relate $D_\FF(\nab)$ to a connection on $V_\star$, 
i.e.~a $\bfK$-linear map $\nab_\star:V_\star\mapsto V_\star\otimes_{A_\star}\Omega^1_\star$, 
we have to establish first an isomorphism $(V\otimes_A\Omega^1)_\star\simeq V_\star\otimes_{A_\star}\Omega^1_\star$.
This is the aim of the following
\begin{lem}\label{lem:iotaiso}
 Let $H$ be a Hopf algebra with twist $\FF\in H\otimes H$ and let $A\in{^{H,\ra}}\AAA$,
$V\in{^{H,\ra}}\MMM_A$ and $W\in{^{H,\ra}_A}\MMM_A$. Then
the $\bfK$-linear map $\iota: V_\star\otimes_{A_\star} W_\star \to (V\otimes_A W)_\star$
defined by $\iota(v\otimes_{A_\star}w) = (\bar f^\alpha\ra v)\otimes_A (\bar f_\alpha\ra w)$
is a left $H^\FF$-module right $A_\star$-module isomorphism.

\noindent In case of $V\in{^{H,\ra}_A}\MMM_A$ the map $\iota$ is a left $H^\FF$-module
$A_\star$-bimodule isomorphism.
\end{lem}
\begin{proof}
The map $\iota$ is well-defined as it is compatible with middle $A_\star$-linearity, for all
$a\in A$, $v\in V$ and $w\in W$,
\begin{flalign}
 \nn \iota\bigl((v\star a) \otimes_{A_\star} w\bigr)&=
\bigl((\bar f^\alpha_1\bar f^\beta\ra v) \cdot (\bar f^\alpha_2\bar f_\beta\ra a)\bigr)\otimes_A (\bar f_\alpha\ra w)\\
&= (\bar f^\alpha\ra v)\otimes_A \bigl((\bar f_{\alpha_1}\bar f^\beta\ra a)\cdot (\bar f_{\alpha_2}\bar f_\beta\ra w)\bigr)
= \iota\bigl(v\otimes_{A_\star} (a\star w)\bigr)~,
\end{flalign}
where in line two we have used (\ref{eqn:twistpropsimp3}). 
Furthermore,
$\iota$ is a right $A_\star$-module homomorphism, for all $a\in A$, $v\in V$ and $w\in W$,
\begin{flalign}
\nn \iota\bigl(v\otimes_{A_\star} w\star a\bigr) &= 
\bigl((\bar f^\alpha\ra v)\otimes_A (\bar f_{\alpha_1}\bar f^\beta\ra w)\bigr)\cdot (\bar f_{\alpha_2}\bar f_\beta\ra a)\\
&=\bar f^\alpha\ra \bigl((\bar f^\beta\ra v)\otimes_A (\bar f_\beta\ra w )\bigr) \cdot (\bar f_\alpha\ra a) = 
\iota(v\otimes_{A_\star} w)\star a~,
\end{flalign}
and also a left $H^\FF$-module homomorphism, for all $\xi\in H$, $v\in V$ and $w\in W$,
\begin{flalign}
 \nn \iota\bigl(\xi\ra(v\otimes_{A_\star}w)\bigr) &= \iota\bigl((\xi_{1_\FF}\ra v)\otimes_{A_\star} (\xi_{2_\FF}\ra w)\bigr)\\
&=(\xi_1\bar f^\alpha\ra v)\otimes_A (\xi_2\bar f_\alpha\ra w) = \xi\ra\iota(v\otimes_{A_\star}w)~.
\end{flalign}
The inverse $\iota^{-1}: (V\otimes_A W)_\star \to  V_\star \otimes_{A_\star} W_\star$ is given by
$\iota^{-1}(v\otimes_A w) = (f^\alpha\ra v)\otimes_{A_\star}(f_\alpha\ra w)$.

In case of  $V\in{^{H,\ra}_A}\MMM_A$, $\iota$ is also a left $A_\star$-module homomorphism, for
all $a\in A$, $v\in V$ and $w\in W$,
\begin{flalign}
\nn \iota\bigl(a\star v\otimes_{A_\star} w\bigr) &= 
(\bar f^\alpha_1\bar f^\beta\ra a)\cdot \bigl((\bar f^\alpha_2\bar f_\beta\ra v)\otimes_A (\bar f_\alpha\ra w)  \bigr)\\
&= (\bar f^\alpha\ra a)\cdot \bar f_\alpha \ra \bigl((\bar f^\beta \ra v) \otimes_A (\bar f_\beta\ra w)\bigr) 
= a\star \iota(v\otimes_{A_\star} w)~.
\end{flalign}

\end{proof}

We obtain a quantization map $\widetilde{D}_\FF:\Con_A(V)\to\Con_{A_\star}(V_\star)$ for 
right module connections by composing 
$D_\FF(\nab):V_\star \to (V\otimes_A \Omega^1)_\star$ with the 
isomorphism $\iota^{-1}:(V\otimes_A\Omega^1)_\star \to V_\star\otimes_{A_\star}\Omega^1_\star$
according to the following commuting diagram
\begin{flalign}
\xymatrix{
 V_\star \ar[rrd]_-{D_\FF(\nab)}\ar[rr]^-{\widetilde{D}_\FF(\nab)} & & V_\star\otimes_{A_\star}\Omega^1_\star \\
 & & (V\otimes_A\Omega^1)_\star \ar[u]_-{\iota^{-1}}
}
\end{flalign}

\begin{theo}\label{theo:condef}
 Let $H$ be a Hopf algebra with twist $\FF\in H\otimes H$, $A\in{^{H,\ra}}\AAA$ be a left $H$-module
algebra, $\bigl(\Omega^\bullet,\wedge,\dd\bigr)$ a left $H$-covariant differential calculus 
and $V\in {^{H,\ra}}\MMM_A$ a left $H$-module right $A$-module.
Then the $\bfK$-linear map
\begin{flalign}
 \widetilde{D}_\FF:\Con_A(V)\to \Con_{A_\star}(V_\star)~,\quad \nab\mapsto 
\widetilde{D}_\FF(\nab)= \iota^{-1}\circ\bigl( D_\FF(\nab)\bigr) = \iota^{-1}\circ (\bar f^\alpha\RA \nab)\circ \bar f_\alpha\ra~
\end{flalign}
is an isomorphism.
\end{theo}
\begin{proof}
 The proof follows from a short calculation. Let $\nab\in\Con_A(V)$, then for all $v\in V$ and $a\in A$,
\begin{flalign}
\nn D_\FF(\nab)(v\star a) 
&= (\bar f^\alpha\RA\nab)\bigl((\bar f_{\alpha_1}\bar f^\beta\ra v)\cdot (\bar f_{\alpha_2}\bar f_\beta\ra a)\bigr)\\
\nn &=\bigl((\bar f^\alpha\RA\nab)(\bar f_{\alpha_1}\bar f^\beta\ra v)\bigr) \cdot (\bar f_{\alpha_2}\bar f_\beta\ra a)
+\epsilon(\bar f^\alpha)\,(\bar f_{\alpha_1}\bar f^\beta\ra v) \otimes_A \dd(\bar f_{\alpha_2}\bar f_\beta\ra a)~\\
\nn &=\bigl((\bar f^\alpha_1\bar f^\beta\RA\nab)(\bar f^{\alpha}_2\bar f_\beta\ra v)\bigr) \cdot (\bar f_{\alpha}\ra a)
+(\bar f^\beta\ra v) \otimes_A \dd(\bar f_\beta\ra a)\\
&= D_\FF(\nab)(v)\star a +(\bar f^\beta\ra v) \otimes_A (\bar f_\beta\ra \dd a)~.
\end{flalign}
In the second line we have used (\ref{eqn:Hcon}),
in the third line (\ref{eqn:twistprop}) and in the last line (\ref{eqn:equivar}) and (\ref{eqn:conadjoint}).
Applying $\iota^{-1}$ we obtain
\begin{flalign}
 \widetilde{D}_\FF(v\star a) = \widetilde{D}_\FF(\nab)(v)\star a +v\otimes_{A_\star} \dd a~.
\end{flalign}
$\widetilde{D}_\FF$ is invertible via the map $\nab_\star\mapsto 
\widetilde{D}_\FF^{-1}(\nab_\star) = \iota\circ D_\FF^{-1}(\nab_\star)$, since for all $\nab\in\Con_A(V)$,
\begin{flalign}
 \nn \widetilde{D}^{-1}_\FF\bigl(\widetilde{D}_\FF(\nab)\bigr) &= 
\iota\circ \bigl( f^\alpha\RA_\FF (\iota^{-1}\circ D_\FF(\nab)) \bigr)\circ f_\alpha\ra
=\iota\circ\iota^{-1}\circ \bigl( f^\alpha\RA_\FF D_\FF(\nab) \bigr)\circ f_\alpha\ra\\
&= D_\FF^{-1}\bigl(D_\FF(\nab)\bigr)=\nab~,
\end{flalign}
where for the second equality we have used that $\iota$ is a left $H^\FF$-module isomorphism, thus
$\xi\RA_\FF\iota^{-1}  =\epsilon(\xi)\,\iota^{-1}$, for all $\xi\in H$.

The property $\widetilde{D}_\FF^{-1}(\nab_\star)\in\Con_A(V)$,
for all $\nab_\star\in\Con_{A_\star}(V_\star)$, follows from dequantization of $H^\FF$ and all its modules 
via the twist $\mathcal{F}^{-1}$.

\end{proof}
An analogous statement holds true for the quantization of connections on left $H$-module left $A$-modules.
\begin{theo}\label{theo:condefleft}
 Let $H$ be a Hopf algebra with twist $\FF\in H\otimes H$, $A\in{^{H,\ra}}\AAA$ be a left $H$-module
algebra, $\bigl(\Omega^\bullet,\wedge,\dd\bigr)$ a left $H$-covariant differential calculus 
and $V\in {^{H,\ra}_A}\MMM$ a left $H$-module left $A$-module.
Then the $\bfK$-linear map
\begin{flalign}
 \widetilde{D}^\cop_\FF:{_A}\Con(V)\to {_{A_\star}}\Con(V_\star)~,\quad \nab\mapsto \widetilde{D}^\cop_\FF(\nab)
=\iota^{-1}\circ \bigl(D_\FF^\cop(\nab)\bigr)=\iota^{-1}\circ(\bar f_\alpha\RA^\cop \nab)\circ \bar f^\alpha\ra~
\end{flalign}
is an isomorphism.
\end{theo}
Since the proof is analogous to Theorem \ref{theo:condef} we can omit it here.


\section{Quasi-commutative algebras and bimodules}
Let $(H,R)$ be a triangular Hopf algebra and $A\in{^{H,\ra}}\AAA$ be quasi-commutative.
We obtained in Theorem \ref{theo:qclriso} that the map $D_R$, i.e.~the quantization isomorphism
with $\FF=R$, provides an isomorphism of right and left 
module homomorphisms between quasi-commutative bimodules.
The aim of this section is to prove a similar statement for right and left connections
on a quasi-commutative bimodule $V\in {^{H,\ra}_A}\MMM_A$.

\begin{theo}\label{theo:leftrightconiso}
  Let $(H,R)$ be a triangular Hopf algebra and $A\in{^{H,\ra}}\AAA$, $V\in {^{H,\ra}_A}\MMM_A$ be quasi-commutative. 
Let further $\bigl(\Omega^\bullet,\wedge,\dd\bigr)$ be a graded quasi-commutative left $H$-covariant differential calculus over $A$.
Then the $\bfK$-linear map
\begin{flalign}
 \widetilde{D}_R: \Con_A(V)\to {_A}\Con(V)~,\quad \nab\mapsto \widetilde{D}_R(\nab) =\tau_R \circ  \bigl(D_R(\nab)\bigr)
\end{flalign}
is an isomorphism, where $D_R$ is the quantization isomorphism with $\FF=R$ and 
$\tau_R: V\otimes_A\Omega^1 \mapsto \Omega^1\otimes_A V$ is the $R$-flip map, i.e.~$\tau_R(v\otimes_A\omega)=
(\bar R^\alpha\ra \omega)\otimes_A(\bar R_\alpha\ra v)$.

\end{theo}
\begin{proof}
 By Lemma \ref{lem:copRdef} we have $A_\star = A^\op$, $\Omega_\star = \Omega^\op$ and $V_\star=V^\op$.
 Using this and Theorem \ref{theo:condef} we obtain, for all $\nab\in\Con_A(V)$, $a\in A$ and $v\in V$,
\begin{flalign}
 \nn D_R(\nab)(a\cdot v) &= D_R(\nab)(v\star a) = D_R(\nab)(v)\star a + (\bar R^\alpha\ra v)\otimes_A\dd (\bar R_\alpha\ra a)\\
&= a\cdot D_R(\nab)(v) + (\bar R^\alpha\ra v)\otimes_A(\bar R_\alpha\ra \dd a)~.
\end{flalign}
Applying the bimodule isomorphism $\tau_R$ and using triangularity $R_{21}=R^{-1}$ we find
\begin{flalign}
 \widetilde{D}_R(\nab)(a\cdot v) = a\cdot \widetilde{D}_R(\nab)(v) + \dd a\otimes_A v~.
\end{flalign}
Thus, $\widetilde{D}_R(\nab)\in{_A}\Con(V)$. The map $\widetilde{D}_R$ is invertible via
\begin{flalign}
 \widetilde{D}_R^{-1}: {_A}\Con(V)\to \Con_A(V)~,\quad \nab\mapsto \widetilde{D}_R^{-1}(\nab) =
 \tau_R\circ \bigl(D_R^{-1}(\nab)\bigr)~,
\end{flalign}
where $D_R^{-1}(\nab) = D_R^\cop(\nab)=(\bar R_\alpha\RA^\cop \nab)\circ\bar R^\alpha\ra$.
(Note that $\tau_R = \tau_R^{-1}$ since $H$ is triangular.)

\end{proof}


\section{Extension to product modules}
Motivated by the studies on product module homomorphisms in Chapter \ref{chap:modhom}, Section \ref{sec:productmodhom},
we now investigate the extension of connections to product modules.
This is of major importance for noncommutative gravity, since it provides a construction principle
for connections on deformed tensor fields in terms of a fundamental connections on the deformed vector fields or one-forms.
For the construction presented below it is essential to assume a triangular Hopf algebra 
and quasi-commutative algebras and modules.

\begin{theo}\label{theo:conplus}
 Let $(H,R)$ be a triangular Hopf algebra and $A\in{^{H,\ra}}\AAA$, $W\in {^{H,\ra}_A}\MMM_A$ be quasi-commutative. Let further
 $\bigl(\Omega^\bullet,\wedge,\dd\bigr)$ be a graded quasi-commutative left $H$-covariant differential calculus over $A$ 
and $V\in{^{H,\ra}}\MMM_A$, $\nab_V\in \Con_A(V)$ and $\nab_W\in\Con_A(W)$. Then the $\bfK$-linear map
$\nab_V \oplus_R \nab_W : V\otimes_A W\to V\otimes_A W\otimes_A\Omega^1$ defined by
\begin{flalign}
\label{eqn:tensorcon}
 \nab_V \oplus_R \nab_W := \tau_{R\,23}\circ(\nab_V\otimes_R\id) + \id\otimes_R\nab_W
\end{flalign}
is a connection on $V\otimes_A W\in{^{H,\ra}}\MMM_A$. Here $\tau_{R\,23}=\id\otimes\tau_R$ is
the $R$-flip map acting on the second and third leg of the tensor product and
 the $R$-tensor product of $\bfK$-linear maps was defined in (\ref{eqn:Rtensor}).
\end{theo}
\begin{proof}
Firstly, note that $\nab_V\otimes_R\id =\nab_V\otimes \id$, since $\id$ is $H$-invariant.
 We have to show that (\ref{eqn:tensorcon}) satisfies the right Leibniz rule (\ref{eqn:rightcon})
and that it is compatible with $A$-middle linearity on $V\otimes_A W$. 
For the proof we denote both connections $\nab_V$ and $\nab_W$ by $\nab$ in order to simplify the notation.
We start with the former property, for all $v\in V$, $w\in W$ and $a\in A$,
\begin{flalign}
\nn (\nab\oplus_R\nab)(v\otimes_A w\cdot a) &= \tau_{R\,23}\bigl((\nab v)\otimes_A w\cdot a\bigr)
+ (\bar R^\alpha\ra v)\otimes_A (\bar R_\alpha\RA\nab)(w\cdot a)\\
\nn&=(\nab\oplus_R\nab)(v\otimes_A w)\cdot a + \epsilon(\bar R_\alpha)~(\bar R^\alpha\ra v)\otimes_A w\otimes_A \dd a\\
&=(\nab\oplus_R\nab)(v\otimes_A w)\cdot a + v\otimes_A w\otimes_A \dd a~,
\end{flalign}
where in the second line we have used that $\tau_R$ is a bimodule isomorphism and (\ref{eqn:Hcon}),
and in the last line the normalization property of the $R$-matrix (\ref{eqn:Rproper31})

We now prove compatibility of (\ref{eqn:tensorcon}) with $A$-middle linearity,
for all $v\in V$, $w\in W$ and $a\in A$,
\begin{flalign}
\nn  (\nab\oplus_R\nab)&(v\cdot a\otimes_A w) = \tau_{R\,23}\bigl(\nab (v\cdot a)\otimes_A w\bigr)
+ (\bar R^\alpha\ra (v\cdot a))\otimes_A (\bar R_\alpha\RA\nab)(w)\\
\nn &= \tau_{R\,23}\bigl((\nab v)\otimes_A a\cdot w + v\otimes_A \dd a\otimes_A w\bigr)
+ (\bar R^\alpha\ra v) \otimes_A (\bar R^\beta\ra a)\cdot (\bar R_\beta \bar R_\alpha\RA\nab)(w)\\
\nn&=\tau_{R\,23}\bigl((\nab v)\otimes_A a\cdot w\bigr)
+ (\bar R^\alpha\ra v) \otimes_A (\bar R_\alpha\RA\nab)(a\cdot w)\\
&= (\nab\oplus_R\nab)(v\otimes_A a\cdot w)~,
\end{flalign}
where in the second line we have used the right Leibniz rule (\ref{eqn:rightcon}) and the property 
(\ref{eqn:Rproper2}) of the $R$-matrix. The third line follows by using quasi-commutativity, triangularity,
the properties (\ref{eqn:Rproper2}) and (\ref{eqn:Rproper3}) of the $R$-matrix and (\ref{eqn:Hcon}).
(Hint: It is easier to show that line two follows from line three.)

\end{proof}

The Hopf algebra $H$ acts on the connection $\nab_V\oplus_R \nab_W$ via the adjoint action $\RA$.
We obtain for all $\xi\in H$
\begin{flalign}
 \nn\xi\RA (\nab_V\oplus_R\nab_W) &= \xi\RA\bigl(\tau_{R\,23}\circ(\nab_V\otimes_R\id) + \id\otimes_R\nab_W\bigr)\\
\nn&=\tau_{R\,23}\circ\bigl((\xi\RA\nab_V)\otimes_R\id\bigr) + \id\otimes_R(\xi\RA\nab_W) \\
&= (\xi\RA\nab_V) \oplus_R (\xi\RA\nab_W)~,
\end{flalign}
where we have used that $\tau_R$ and $\id$ are invariant under $H$ and (\ref{eqn:RtensorHmod}).

The sum $\oplus_R$ of connections canonically 
extends to arbitrary tensor products and 
is compatible with the bimodule isomorphism $\tau_R$
due to the following
\begin{theo}\label{theo:tensorcon}
  Let $(H,R)$ be a triangular Hopf algebra and $A\in{^{H,\ra}}\AAA$, $W,Z\in {^{H,\ra}_A}\MMM_A$ be quasi-commutative.
 Let further $\bigl(\Omega^\bullet,\wedge,\dd\bigr)$ be a graded quasi-commutative left $H$-covariant differential calculus over $A$
and $V\in{^{H,\ra}}\MMM_A$, $\nab_V\in \Con_A(V)$, $\nab_W\in\Con_A(W)$ and $\nab_Z\in\Con_A(Z)$.
Then 
\begin{flalign}\label{eqn:tensorconass}
 \bigl(\nab_V\oplus_R \nab_W\bigr)\oplus_R \nab_Z = \nab_V\oplus_R\bigl(\nab_W\oplus_R\nab_Z\bigr)
\end{flalign}
and
\begin{flalign}\label{eqn:tensorconcom}
 \nab_Z\oplus_R \nab_W = \tau_{R\,12}\circ\bigl(\nab_W\oplus_R\nab_Z\bigr)\circ\tau_R~,
\end{flalign}
where $\tau_{R\,12}=\tau_R\otimes\id$ is the $R$-flip map acting on the first and second leg of the tensor product.
\end{theo}
\begin{proof}
We denote by $\tau_{R\,i~i+1} = \id\otimes\dots\otimes\tau_R\otimes\dots\otimes\id$ the
$R$-flip map acting on the $i$-th and $(i+1)$-th leg of a tensor product.
The $R$-flip map exchanging the $i$-th leg with the $(i+1)$-th and $(i+2)$-th leg is denoted by
$\tau_{R\,i~(i+1~i+2)}$ and similarly $\tau_{R\,(i~i+1)~i+2}$ is the $R$-flip map exchanging
the $i$-th and $(i+1)$-th leg with the $(i+2)$-th leg.
For example, $\tau_{R\,(12)3}(a\otimes b\otimes c) = (\bar R^\alpha\ra c) \otimes \bar R_\alpha\ra (a\otimes b)$
and $\tau_{R\,1(23)}(a\otimes b \otimes c) = \bar R^\alpha\ra(b\otimes c)\otimes(\bar R_\alpha\ra a)$.
 Let us first show the associativity property (\ref{eqn:tensorconass}). To simplify notation
we denote all connections by $\nab$.
\begin{flalign}
\nn (\nab\oplus_R\nab)\oplus_R \nab &= \tau_{R\,34}\circ \bigl((\nab\oplus_R\nab)\otimes_R\id\bigr) 
+ \id\otimes_R \id \otimes_R\nab\\
\nn&=\tau_{R\,34}\circ \bigl(\tau_{R\,23}\circ (\nab\otimes_R\id\otimes_R\id) + \id\otimes_R\nab\otimes_R\id \bigr) 
+ \id\otimes_R \id \otimes_R\nab\\
\nn &=\tau_{R\,34}\circ \tau_{R\,23}\circ (\nab\otimes_R\id\otimes_R\id) + \tau_{R\,34}\circ (\id\otimes_R\nab\otimes_R\id ) 
+ \id\otimes_R \id \otimes_R\nab\\
\nn &=\tau_{R\,2 (34)}\circ (\nab\otimes_R\id\otimes_R\id) + \id\otimes_R \bigl(\tau_{R\,23}\circ (\nab\otimes_R\id) + \id\otimes_R \nab \bigr) \\
\nn&=\tau_{R\,2 (34)}\circ (\nab\otimes_R\id\otimes_R\id) + \id\otimes_R (\nab\oplus_R\nab ) \\
&= \nab\oplus_R(\nab\oplus_R\nab)~.
\end{flalign}
We have frequently used (\ref{eqn:Rtensorcirc}) together with the fact that
$\id$ and $\tau_R$ are $H$-invariant. In line four we have also used
$\tau_{R\,2(34)} = \tau_{R\,34}\circ\tau_{R\,23}$, which follows from the properties 
of the $R$-matrix (\ref{eqn:Rproper2}).

We now show the property (\ref{eqn:tensorconcom}). The strategy is to express all $R$-tensor products
by usual ones using (\ref{eqn:Rtensorsimple}). Note that $\tau_R^{-1}=\tau_R$ for a triangular Hopf algebra.
\begin{flalign}
\nn \nab_Z\oplus_R\nab_W &= \tau_{R\,23}\circ (\nab_Z\otimes\id) + \id\otimes_R \nab_W\\
\nn&=\tau_{R\,23}\circ (\nab_Z\otimes\id) + \tau_{R\,(12)3} \circ (\nab_W\otimes \id)\circ \tau_R\\
\nn&=\tau_{R\,23}\circ (\nab_Z\otimes\id) + \tau_{R\,12}\circ\tau_{R\,23} \circ (\nab_W\otimes \id)\circ \tau_R\\
\nn &=\tau_{R\,12}\circ \bigl( \tau_{R\,12}\circ\tau_{R\,23}\circ(\nab_Z\otimes\id)\circ\tau_R + \tau_{R\,23}\circ (\nab_W\otimes\id)  \bigr)\circ\tau_R\\
&=\tau_{R\,12}\circ (\nab_{W}\oplus_R\nab_Z)\circ\tau_R~.
\end{flalign}
In line three we have used $\tau_{R\,(12)3} = \tau_{R\,12}\circ\tau_{R\,23}$, which follows from the properties 
of the $R$-matrix (\ref{eqn:Rproper2}).

\end{proof}

It remains to study the quantization of the sum of two connections (\ref{eqn:tensorcon}).
By construction, $D_\FF(\nab_V\oplus_R\nab_W):(V\otimes_A W)_\star \to (V\otimes_A W\otimes_A\Omega^1)_\star$
is a $\bfK$-linear map. Analogously to Lemma \ref{lem:iotaiso} there is a left $H^\FF$-module
right $A_\star$-module isomorphism $\iota_{123}: V_\star\otimes_{A_\star} W_\star\otimes_{A_\star} \Omega^1_\star
\to (V\otimes_A W\otimes_A \Omega^1)_\star$ given by the following commuting diagram\footnote{
The indices on $\iota$ label the legs it acts on.
}
\begin{flalign}\label{eqn:higheriota}
 \xymatrix{
V_\star\otimes_{A_\star} W_\star\otimes_{A_\star} \Omega^1_\star \ar[drr]_-{\iota_{123}}\ar[d]_-{\iota_{23}}\ar[rr]^-{\iota_{12}} & &(V\otimes_A W)_\star\otimes_{A_\star} \Omega^1_\star \ar[d]^-{\iota_{(12)3}}\\
V_\star\otimes_{A_\star} (W\otimes_A \Omega^1)_\star \ar[rr]_-{\iota_{1(23)}} & &(V\otimes_A W\otimes_A \Omega^1)_\star
}
\end{flalign}

\begin{theo}\label{theo:prodcondef}
  Let $(H,R)$ be a triangular Hopf algebra with twist $\FF\in H\otimes H$ and 
$A\in{^{H,\ra}}\AAA$, $W\in {^{H,\ra}_A}\MMM_A$ be quasi-commutative. Let further 
 $\bigl(\Omega^\bullet,\wedge,\dd\bigr)$ be a graded quasi-commutative left $H$-covariant differential calculus 
and $V\in{^{H,\ra}}\MMM_A$,
$\nab_V\in \Con_A(V)$ and $\nab_W\in\Con_A(W)$. Then the following diagram commutes
\begin{flalign}\label{eqn:quantprodcon}
\xymatrix{
V_\star \otimes_{A_\star} W_\star \ar[rrr]^-{\widetilde{D}_\FF(\nab_V)\oplus_{R^\FF} \widetilde{D}_\FF(\nab_W)} \ar[d]_-{\iota}& & &V_\star \otimes_{A_\star} W_\star \otimes_{A_\star} \Omega^1_\star\\
(V\otimes_A W)_\star \ar[rrr]_-{D_\FF(\nab_V\oplus_R\nab_W)} & & & (V\otimes_A W\otimes_A\Omega^1)_\star \ar[u]_-{\iota_{123}^{-1}}
} 
\end{flalign}
\end{theo}
\begin{proof}
 To simplify notation we denote all connections by $\nab$. Using $\bfK$-linearity
of $D_\FF$ and that $\tau_R$ and $\id$ are invariant under $H$, we obtain
\begin{flalign}
 D_\FF(\nab\oplus_R\nab) = D_\FF\bigl(\tau_{R\,23}\circ (\nab\otimes_R \id) + \id\otimes_R \nab\bigr)
= \tau_{R\,23}\circ D_\FF(\nab\otimes_R \id) + D_\FF(\id\otimes_R\nab)~.
\end{flalign}
Theorem \ref{theo:promodhomdef} and $H$-invariance of $\id$ implies
\begin{flalign}
  \nn D_\FF(\nab\oplus_R\nab) &=  \tau_{R\,23}\circ \iota_{(12)3} \circ \bigl(D_\FF(\nab)\otimes_{R^\FF}\id\bigr)\circ \iota^{-1}+
 \iota_{1(23)}\circ \bigl(\id\otimes_{R^\FF}D_\FF(\nab)\bigr)\circ \iota^{-1}~\\
&=\tau_{R\,23}\circ \iota_{(12)3}\circ \iota_{12} \circ \bigl(\widetilde{D}_\FF(\nab)\otimes_{R^\FF}\id\bigr)\circ \iota^{-1}+
 \iota_{1(23)}\circ\iota_{23}\circ \bigl(\id\otimes_{R^\FF}\widetilde{D}_\FF(\nab)\bigr)\circ \iota^{-1}~,
\end{flalign}
where in the last line we have also used (\ref{eqn:Rtensorcirc}) and that $\iota$ and $\id$ are $H$-invariant.

We now compose this with $\iota$ from the right and with $\iota_{123}^{-1}$ from the left. Using the diagram (\ref{eqn:higheriota}), 
$\tau_{R\,23}\circ\iota_{1(23)}=\iota_{1(23)}\circ\tau_{R\,23}$ and $\tau_{R^\FF\,23} = \iota_{23}^{-1}\circ\tau_{R\,23}\circ\iota_{23}$,
we find
\begin{flalign}
 \nn \iota_{123}^{-1} \circ D_\FF(\nab\oplus_R\nab) \circ \iota &= 
\tau_{R^\FF\,23}\circ \bigl(\widetilde{D}_\FF(\nab)\otimes_{R^\FF}\id\bigr) + \bigl(\id\otimes_{R^\FF}\widetilde{D}_\FF(\nab)\bigr)\\
&= \widetilde{D}_\FF(\nab)\oplus_{R^\FF} \widetilde{D}_\FF(\nab)~.
\end{flalign}

\end{proof}


\section{Extension to the dual module}
It remains to provide an extension of connections $\nab\in\Con_A(V)$ to the dual module
$V^\prime=\Hom_A(V,A)$. This construction is of major importance in noncommutative gravity, since 
given a connection on e.g.~vector fields $\Xi$, we have to know how to obtain from it a
connection on one-forms $\Omega^1$, the dual of $\Xi$. We focus only on right connections $\Con_A(V)$
and the right dual $V^\prime=\Hom_A(V,A)$ of a module $V$, since  left connections and left duals 
follow from a mirror construction as above.

For the discussion we have to assume that the module $V$ is finitely generated and projective, 
a property which is also required in every other approach to connections on modules known to us, see
e.g.~\cite{1038.58004,Madore:2000aq}.
For describing noncommutative vector bundles this restriction is motivated by the Serre-Swan theorem 
\cite{0067.16201,0109.41601}, stating that vector bundles over commutative, smooth and compact manifolds
are equivalent to finitely generated and projective modules over the algebra of smooth functions.
 Before defining and discussing finitely generated and projective modules in detail, 
let us briefly explain why this property is essential for describing dual connections.
Using a connection $\nab\in \Con_A(V)$ and the differential $\dd$ on $A$, we can induce a $\bfK$-linear
map $V^\prime\to\Hom_A(V,\Omega^1)$ by $v^\prime \mapsto \dd\circ v^\prime - \wedge\circ
(v^\prime\otimes\id) \circ \nab$.
For a finitely generated and projective module $V$ there is an isomorphism $\Hom_A(V,\Omega^1)\simeq \Omega^1\otimes_A V^\prime$.
Composing the induced map $V^\prime\to\Hom_A(V,\Omega^1)$ with this isomorphism we obtain a $\bfK$-linear map
$V^\prime \to \Omega^1\otimes_A V^\prime$. We are going to prove in detail that this map is a left connection on $V^\prime$.
Using Theorem \ref{theo:leftrightconiso} we can induce from this left connection a right connection on $V^\prime$
in the quasi-commutative setting.

There are various equivalent definitions of a finitely generated and projective module over a ring, see e.g.~the book
\cite{0911.16001}. We do not have to go into the details and use the very convenient 
characterization of this type of module in terms of a pair of dual bases.
\begin{lem}[Dual Basis Lemma]
 Let $A\in\AAA$ be a unital and associative algebra. A right $A$-module $V\in\MMM_A$ is
finitely generated and projective, if and only if there exists a family of elements
$\lbrace v_i\in V:i=1,\dots,n\rbrace$ and $A$-linear functionals $\lbrace v_i^\prime\in V^\prime=\Hom_A(V,A):i=1,\dots,n \rbrace$
with $n\in\bbN$, such that for any $v\in V$ we have
\begin{flalign}
 v=\sum\limits_{i=1}^n v_i\cdot v_i^\prime(v)~.
\end{flalign}

\noindent Analogously, a left $A$-module $V\in{_A}\MMM$ is
finitely generated and projective, if and only if there exists a family of elements
$\lbrace v_i\in V:i=1,\dots,n\rbrace$ and $A$-linear functionals $\lbrace v^\prime_i \in {^\prime}V={_A}\Hom(V,A):i=1,\dots,n \rbrace$
with $n\in\bbN$, such that for any $v\in V$ we have
\begin{flalign}
 v=\sum\limits_{i=1}^n v^\prime_i(v)\cdot v_i~.
\end{flalign}
\end{lem}
\noindent A proof of this lemma can be found e.g.~in \cite{0911.16001}.
The set $\lbrace v_i,v^\prime_i:i=1,\dots,n \rbrace$ is loosely referred to ``pair of dual bases'' for $V$,
even though $\lbrace v_i\rbrace$ is just a generating set of $V$ and not a basis.

Finitely generated and projective modules enjoy the following important properties.
\begin{propo}
 Let $V\in\MMM_A$ be a finitely generated and projective right $A$-module with
a pair of dual bases $\lbrace v_i,v_i^\prime:i=1,\dots,n \rbrace$. For any $v\in V$ let
$v^{\prime\prime} \in V^{\prime\prime}:= {_A}\Hom(V^\prime,A)$ be defined by
$v^{\prime\prime}(v^\prime) := v^\prime(v)$, for all $v^\prime \in V^\prime$. We have
\begin{enumerate}
 \item $\lbrace v_i^\prime,v_i^{\prime\prime}:i=1,\dots,n\rbrace$ is a pair of dual bases for $V^\prime$
\item $V^\prime$ is a finitely generated and projective left $A$-module
\item The natural map $V\to V^{\prime\prime}\,,~v\mapsto v^{\prime\prime}$ is an isomorphism
of right $A$-modules
\end{enumerate}
 
\end{propo}
\begin{proof}
 {\it Proof of (1):}\\
We have to show that $v^\prime = \sum_{i=1}^n v_i^{\prime\prime}(v^\prime)\cdot v_i^\prime$, for all $v^\prime\in V^\prime$.
Evaluating the right hand side on an arbitrary $v\in V$ we obtain
\begin{flalign}
 \sum\limits_{i=1}^n v^{\prime\prime}_i(v^\prime) \,v_i^{\prime}(v)=
\sum\limits_{i=1}^n v^\prime(v_i) \,v_i^{\prime}(v) = v^\prime\left(\sum\limits_{i=1}^n v_i\cdot v_i^\prime(v)\right) = v^\prime(v)~.
\end{flalign}

\noindent {\it Proof of (2):}\\
Follows directly from the dual basis lemma.

\noindent {\it Proof of (3):}\\
The map $V\to V^{\prime\prime}\,,~v\mapsto v^{\prime\prime}$ is injective, since let
$0=v^{\prime\prime}$ we have $v^{\prime\prime}(v^\prime) = v^\prime(v) =0$, for all
$v^\prime\in V^\prime$, and thus by the dual basis lemma
$v=\sum_{i=1}^n v_i\cdot v_i^\prime(v) = 0$.
It is also surjective, since $\lbrace v_i^{\prime\prime}\rbrace$ is a generating set of $V^{\prime\prime}$.

\end{proof}

We now come to the property of finitely generated and projective modules which is most relevant for our studies.
\begin{propo}\label{propo:projhomiso}
 Let $V\in \MMM_A$ be a finitely generated and projective right $A$-module and let $W\in\MMM_A$. 
Then there exists an isomorphism $\varphi: W\otimes_A V^\prime \to \Hom_A(V,W)$.

\noindent If in addition $V,W\in{_A}\MMM_A$ are $A$-bimodules, then $\varphi$
is an isomorphism between the $A$-bimodules $W\otimes_A V^\prime\in{_A}\MMM_A$ and $\Hom_A(V,W)\in{_A}\MMM_A$.

\noindent For left $H$-module $A$-bimodules $V,W\in{^{H,\ra}_A}\MMM_A$, $\varphi$ is an isomorphism
between the left $H$-module $A$-bimodules $W\otimes_A V^\prime\in{^{H,\ra}_A}\MMM_A$ and $\Hom_A(V,W)\in{^{H,\RA}_A}\MMM_A$.
\end{propo}
\begin{proof}
We define a $\bfK$-linear map $\varphi:W\otimes_A V^\prime\to \Hom_A(V,W)$ by, for all $v\in V$,
$\bigl(\varphi(w\otimes_A v^\prime)\bigr)(v):= w\cdot v^\prime(v)$.
Employing the pair of dual bases this map is invertible via
\begin{flalign}\label{eqn:projhomiso}
\varphi^{-1}:\Hom_A(V,W)\to W\otimes_A V^\prime~,\quad P\mapsto \varphi^{-1}(P) = \sum\limits_{i=1}^n P(v_i)\otimes_A v_i^\prime~. 
\end{flalign}
Indeed, we obtain for all $v\in V$ and $P\in \Hom_A(V,W)$,
\begin{flalign}
 \nn \bigl(\varphi(\varphi^{-1}(P))\bigr)(v) &= \sum\limits_{i=1}^n \bigl(\varphi(P(v_i)\otimes_A v_i^\prime) \bigr)(v)
= \sum\limits_{i=1}^n P(v_i)\cdot v_i^\prime(v) \\
&= P\left(\sum\limits_{i=1}^n v_i\cdot v_i^\prime(v)\right) = P(v)~,
\end{flalign}
and for all $w\in W$ and $v^\prime\in V^\prime$
\begin{flalign}
  \varphi^{-1}\bigl(\varphi(w\otimes_A v^\prime)\bigr) &= \sum\limits_{i=1}^n w\cdot v^\prime(v_i)\otimes_A v_i^\prime 
= w\otimes_A \sum\limits_{i=1}^n v_i^{\prime\prime}(v^\prime)\cdot v_i^\prime = w\otimes_A v^\prime~.
\end{flalign}

Let now $V,W\in{_A}\MMM_A$ be $A$-bimodules, then $V^\prime,W\otimes_A V^\prime$ and $\Hom_A(V,W)$ are
also $A$-bimodules. The isomorphism $\varphi$ respects the bimodule structure, since for all $a,b\in A$,
$w\in W$, $v^\prime\in V^\prime$ and $v\in V$,
\begin{flalign}
\nn\bigl(\varphi(a\cdot w\otimes_A v^\prime \cdot b)\bigr)(v) &= a\cdot w \cdot v^\prime(b\cdot v) =  
a\cdot \bigl(\varphi(w\otimes_A v^\prime)\bigr)(b\cdot v) \\
&= \bigl(a\cdot \varphi(w\otimes_A v^\prime)\cdot b\bigr)(v)~.
\end{flalign}

Let $V,W\in{^{H,\ra}_A}\MMM_A$ be left $H$-module $A$-bimodules, then $V^\prime,W\otimes_A V^\prime$ and $\Hom_A(V,W)$ are
also left $H$-module $A$-bimodules. The isomorphism $\varphi$ respects the left $H$-module structure,
since for all $\xi\in H$, $w\in W$, $v^\prime\in V^\prime$ and $v\in V$,
\begin{flalign}
 \nn \bigl(\varphi(\xi\ra (w\otimes_A v^\prime))\bigr)(v) &= (\xi_1\ra w)\cdot (\xi_2\RA v^\prime)(v) 
= (\xi_1\ra w)\cdot \bigl(\xi_2\ra v^\prime(S(\xi_3)\ra v)\bigr)\\
&= \xi_1\ra\bigl(w\cdot v^\prime(S(\xi_2)\ra v)\bigr) = \bigl(\xi\RA \bigl(\varphi(w\otimes_A v^\prime)\bigr)\bigr)(v)~.
\end{flalign}

\end{proof}

This property of finitely generated and projective modules now allows us to induce a connection on the dual module
$V^\prime = \Hom_A(V,A)$, see also \cite{1038.58004,Madore:2000aq} for similar discussions. 
Let $A\in \AAA$ be a unital and associative algebra and let $\bigl(\Omega^\bullet,\wedge,\dd\bigr)$
be a differential calculus over $A$.
Let further $V\in \MMM_A$ be a finitely generated and projective right $A$-module.
Given a connection $\nab\in\Con_A(V)$ we define the following $\bfK$-linear map
\begin{flalign}\label{eqn:dualcon1}
 \nab^\Hom: V^\prime\to \Hom_A(V,\Omega^1)~,\quad v^\prime\mapsto \nab^\Hom(v^\prime) = \dd\circ v^\prime - \wedge\circ
(v^\prime\otimes\id) \circ \nab~.
\end{flalign}
In the second term we have canonically extended $v^\prime\in\Hom_A(V,A)$ to
 a right $A$-linear map $v^\prime\otimes\id: V\otimes_A\Omega^1\to A\otimes_A\Omega^1$ 
by $(v^\prime\otimes\id)(v\otimes_A\omega) := v^\prime(v)\otimes_A\omega$.
The map $\wedge: A\otimes_A \Omega^1 \to\Omega^1$ is the $A$-bimodule homomorphism associated to the
product in $\Omega^\bullet$, i.e.~$\wedge(a\otimes_A \omega) = a\wedge \omega = a\cdot \omega$.
The right $A$-linearity of $\nab^\Hom(v^\prime)\in\Hom_A(V,\Omega^1)$ is easily shown, for all
$a\in A$, $v^\prime\in V^\prime$ and $v\in V$,
\begin{flalign}
\nn\left(\nab^\Hom(v^\prime)\right)(v\cdot a) &= \dd\bigl(v^\prime(v\cdot a)\bigr) -  \wedge\Bigl((v^\prime\otimes\id)\bigl(\nab(v\cdot a) \bigr)\Bigr)\\
\nn&=\dd (v^\prime(v))\cdot a + v^\prime(v)\cdot \dd a - \wedge\Bigl((v^\prime\otimes\id)\bigl(\nab v \bigr)\Bigr)\cdot a - \wedge\Bigl((v^\prime\otimes\id)(v\otimes_A \dd a)\Bigr)\\
&=\left(\nab^\Hom(v^\prime)\right)(v) \cdot a~.
\end{flalign}

Employing the left $A$-module isomorphism $\varphi^{-1}:\Hom_A(V,\Omega^1)\to \Omega^1\otimes_A V^\prime$ (\ref{eqn:projhomiso})
we can induce a $\bfK$-linear map
$\nab^\prime :V^\prime \to \Omega^1\otimes_A V^\prime$ via the following diagram
\begin{flalign}
\label{eqn:dualcon2}
 \xymatrix{
V^\prime \ar[rr]^-{\nab^\Hom} \ar[rrd]_-{\nab^\prime}& &\Hom_A(V,\Omega^1)\ar[d]^-{\varphi^{-1}}\\
 & & \Omega^1\otimes_A V^\prime 
}
\end{flalign}
The map $\nab^\prime$ acting on $v^\prime\in V^\prime$ reads explicitly
\begin{flalign}
\label{eqn:nabprimeexplicit}
 \nab^\prime(v^\prime) =\varphi^{-1}\left(\nab^\Hom(v^\prime)\right) 
= \sum\limits_{i=1}^n\Bigl(\dd (v^\prime(v_i)) - \wedge\bigl((v^\prime\otimes\id)(\nab v_i)\bigr) \Bigr)\otimes_Av_i^\prime ~.
\end{flalign}

We can now show that $\nab^\prime$ is a connection on $V^\prime$, see also \cite{1038.58004,Madore:2000aq}.
\begin{theo}
In the hypotheses above,  the map $\nab^\prime$ (\ref{eqn:dualcon2}) is a connection on the left $A$-module $V^\prime\in{_A}\MMM$.
\end{theo}
\begin{proof}
Using (\ref{eqn:nabprimeexplicit}) we have for all $a\in A$ and $v^\prime\in V^\prime$
\begin{flalign}
 \nn \nab^\prime(a\cdot v^\prime) &= \sum\limits_{i=1}^n \Bigl(\dd(a\,v^\prime(v_i)) - a\cdot \wedge\bigl((v^\prime\otimes\id)(\nab v_i)\bigr)\Bigr)\otimes_A v_i^\prime\\
\nn &=\sum\limits_{i=1}^n \Bigl(\dd a \cdot v^\prime(v_i) + a\cdot \dd( v^\prime(v_i)) - a\cdot \wedge\bigl((v^\prime\otimes\id)(\nab v_i)\bigr)\Bigr)\otimes_A v_i^\prime\\
\nn&=\dd a\otimes_A \sum\limits_{i=1}^n v_i^{\prime\prime}(v^\prime)\cdot v_i^\prime + a\cdot \nab^\prime(v^\prime) \\
&= \dd a\otimes_A v^\prime + a\cdot \nab^\prime(v^\prime)~.
\end{flalign}

\end{proof}

This well-known theorem serves as a basis for discussing the twist quantization of dual connections.
For this we require the following technical 
\begin{lem}\label{lem:varphistar}
 Let $H$ be a Hopf algebra with twist $\FF\in H\otimes H$, $A\in{^{H,\ra}}\AAA$ and $V,W\in {^{H,\ra}}\MMM_A$.
Let further $V$ be finitely generated and projective. We denote by 
$\varphi:W\otimes_A V^\prime \to \Hom_A(V,W)$ the left $H$-module isomorphism of Proposition \ref{propo:projhomiso}
and by $\varphi_\star: W_\star \otimes_{A_\star} (V_\star)^\prime\to\Hom_{A_\star}(V_\star,W_\star)$ 
the $\bfK$-linear map defined by, for all $v\in V_\star$,
$\bigl(\varphi_\star(w\otimes_{A_\star}v_\star^\prime)\bigr)(v):= w\star v_\star^\prime(v)$.
 Then the following diagram commutes
\begin{flalign}\label{eqn:varphistar}
 \xymatrix{
W_\star\otimes_{A_\star} (V_\star)^\prime \ar[rr]^-{\varphi_\star} \ar[d]_-{\id\otimes D_\FF^{-1}} & &\Hom_{A_\star}(V_\star,W_\star)\ar[dd]^-{D_\FF^{-1}}\\
 W_\star\otimes_{A_\star} (V^\prime)_\star \ar[d]_-{\iota} & &\\
 (W\otimes_A V^\prime)_\star \ar[rr]_-{\varphi} & & (\Hom_A(V,W))_\star 
}
\end{flalign}
\end{lem}
\begin{proof}
 Using the explicit expression for $D_\FF^{-1}$ (\ref{eqn:Dinv}), we find when going the upper path,
for all $v\in V_\star$,
\begin{flalign}
\nn \Bigl(D_\FF^{-1}\bigl(\varphi_\star(w\otimes_{A_\star}v_\star^\prime)\bigr)\Bigr)(v) &= \bar f^\alpha\ra \Bigl(w\star v_\star^\prime \bigl(\chi S(\bar f_\alpha)\ra v\bigr)\Bigr)\\
\nn &= (\bar f^\alpha_1\bar f^\beta\ra w) \cdot \bar f^\alpha_2\bar f_\beta\ra v_\star^\prime\bigl(\chi S(\bar f_\alpha)\ra v\bigr)\\
\nn &= (\bar f^\alpha\ra w) \cdot \bar f_{\alpha_1}\bar f^\beta\ra v_\star^\prime\bigl(\chi S(\bar f_\beta) S(\bar f_{\alpha_2})\ra v\bigr)\\
&= (\bar f^\alpha\ra w) \cdot \bigl(\bar f_\alpha\RA D_\FF^{-1}(v_\star^\prime)\bigr)(v)~,
\end{flalign}
where in line three we have used the cocycle condition (\ref{eqn:twistpropsimp3}).
Following the lower path we find, for all $v\in V_\star$,
\begin{flalign}
\nn \Bigl(\varphi\Bigl(\iota\bigl(w\otimes_{A_\star}D_\FF^{-1}(v_\star^\prime)\bigr)  \Bigr)\Bigr) (v) 
&=\Bigl(\varphi\Bigl((\bar f^\alpha\ra w)\otimes_{A}\bigl(\bar f_\alpha\RA D_\FF^{-1}(v_\star^\prime)\bigr)  \Bigr)\Bigr) (v) \\
&=(\bar f^\alpha\ra w) \cdot \bigl(\bar f_\alpha\RA D_\FF^{-1}(v_\star^\prime)\bigr)(v)~.
\end{flalign}

\end{proof}

We now prove that quantizing the dual connection is equivalent to dualizing the quantized connection.
\begin{theo}\label{theo:dualcondef}
 Let $H$ be a Hopf algebra with twist $\FF\in H\otimes H$, $A\in{^{H,\ra}}\AAA$ and $V\in {^{H,\ra}}\MMM_A$.
Let further $\bigl(\Omega^\bullet,\wedge,\dd\bigr)$ be a left $H$-covariant differential calculus
and let $V$ be finitely generated and projective.
Then the following diagram commutes
\begin{flalign}
\xymatrix{
\Con_A(V) \ar[rr]^-{^\prime}\ar[d]_-{\widetilde{D}_\FF} & & {_A}\Con(V^\prime) \ar[d]^-{\widetilde{D}_\FF^\cop} \\
\Con_{A_\star}(V_\star) \ar[rr]_-{^\prime} & & {_{A_\star}}\Con((V_\star)^\prime)\simeq  {_{A_\star}}\Con((V^\prime)_\star)
} 
\end{flalign}

\end{theo}
\begin{proof}
Let $\nab\in\Con_A(V)$ be a connection on $V$. Following the upper path in the diagram, we first construct
its dual $\nab^\prime\in{_A}\Con(V^\prime)$ and then apply the quantization map $\widetilde{D}_\FF^\cop$
yielding the connection $\widetilde{D}_\FF^\cop(\nab^\prime)\in {_{A_\star}}\Con((V^\prime)_\star)$.
In particular, $\widetilde{D}_\FF^\cop(\nab^\prime)$ is a map 
$(V^\prime)_\star\to \Omega_\star^1\otimes_{A_\star} (V^\prime)_\star$.
Using that $\varphi$ is a left $H$-module isomorphism, we have
\begin{flalign}
 \widetilde{D}_\FF^\cop(\nab^\prime) = \iota^{-1} \circ D_\FF^\cop(\varphi^{-1}\circ\nab^\Hom) = 
\iota^{-1}\circ \varphi^{-1} \circ D_\FF^\cop(\nab^\Hom)~.
\end{flalign}
When acting on an element $v^\prime\in (V^\prime)_\star$ we get
\begin{flalign}
 \label{eqn:upper}\bigl(\widetilde{D}_\FF^\cop(\nab^\prime)\bigr)(v^\prime) = 
\iota^{-1}\Bigl(\varphi^{-1}\Bigl(\dd\circ_\star v^\prime - \wedge\circ_\star (v^\prime\otimes\id)\circ_\star\nab\Bigr)\Bigr)~.
\end{flalign}

Going the lower path,  we first quantize the connection $\nab$ yielding the
connection $\widetilde{D}_\FF(\nab)\in \Con_{A_\star}(V_\star)$ and then construct its dual
$\bigl(\widetilde{D}_\FF(\nab)\bigr)^\prime\in{_{A_\star}}\Con((V_\star)^\prime)$.
In particular, $\bigl(\widetilde{D}_\FF(\nab)\bigr)^\prime$ is a map
$(V_\star)^\prime\to \Omega^1_\star\otimes_{A_\star} (V_\star)^\prime$, which reads when acting on
$v_\star^\prime \in (V_\star)^\prime$
\begin{flalign}
 \label{eqn:lower}\bigl(\widetilde{D}_\FF(\nab)\bigr)^\prime(v_\star^\prime) = \varphi_\star^{-1}\Bigl(\dd\circ v_\star^\prime 
- \wedge_\star\circ (v_\star^\prime\otimes\id)\circ \widetilde{D}_\FF(\nab)\Bigr)~.
\end{flalign}

Since $(V^\prime)_\star\simeq (V_\star)^\prime$ via the map $D_\FF$, equality
of the connections $\widetilde{D}_\FF^\cop(\nab^\prime)$ and $\bigl(\widetilde{D}_\FF(\nab)\bigr)^\prime$
means that the following diagram has to commute
\begin{flalign}
\xymatrix{
 (V^\prime)_\star \ar[rrr]^-{\widetilde{D}_\FF^\cop(\nab^\prime)} \ar[d]_-{D_\FF} & & &\Omega_\star^1\otimes_{A_\star} (V^\prime)_\star \ar[d]^-{\id\otimes D_\FF} \\
(V_\star)^\prime \ar[rrr]_-{\bigl(\widetilde{D}_\FF(\nab)\bigr)^\prime}  & & &\Omega_\star^1\otimes_{A_\star} (V_\star)^\prime 
}
\end{flalign}
Going the upper path, using (\ref{eqn:upper}) and acting with
 $\varphi_\star=D_\FF\circ\varphi\circ \iota\circ(\id\otimes D_\FF^{-1})$ (see (\ref{eqn:varphistar})) on the result, we obtain
\begin{flalign}
\nn \varphi_\star\Bigl((\id\otimes D_\FF)\bigl(\widetilde{D}_\FF^\cop(\nab^\prime)(v^\prime)\bigr)\Bigr) 
&=D_\FF\Bigl(\dd\circ_\star v^\prime -\wedge\circ_\star (v^\prime\otimes\id) \circ_\star \nab\Bigr) \\
&= \dd \circ D_\FF(v^\prime) -\wedge_\star \circ \iota^{-1}\circ D_\FF\bigl(v^\prime\otimes\id\bigr)\circ\iota \circ\widetilde{D}_\FF(\nab)~,
\end{flalign}
where we have used that $\dd$ and $\wedge$ are $H$-invariant and $\wedge_\star = \wedge\circ\iota$.
Going the lower path, using (\ref{eqn:lower}) and acting also with $\varphi_\star$ on the result, we obtain
\begin{flalign}
 \varphi_\star\Bigl(\bigl(\widetilde{D}_\FF(\nab)\bigr)^\prime(D_\FF(v^\prime))  \Bigr)= 
\dd \circ D_\FF(v^\prime) - \wedge_\star \circ (D_\FF(v^\prime)\otimes \id)\circ \widetilde{D}_\FF(\nab)~.
\end{flalign}
These two expressions are equal, since $(D_\FF(v^\prime)\otimes \id) = \iota^{-1}\circ D_\FF\bigl(v^\prime\otimes\id\bigr)\circ\iota$
due to Theorem \ref{theo:promodhomdef}.

\end{proof}

In the quasi-commutative setting, i.e.~assuming $(H,R)$ to be a triangular Hopf algebra, and
$A\in{^{H,\ra}}\AAA$, $(\Omega^\bullet,\wedge,\dd)$ and $V\in{^{H,\ra}_A}\MMM_A$ to be (graded) quasi-commutative,
we can derive a right connection on $V^\prime$ by employing Theorem \ref{theo:leftrightconiso}.
The right connections $\nab\in\Con_A(V)$  and $\widetilde{D}^{-1}_R(\nab^\prime)\in\Con_A(V^\prime)$ then can be extended
to arbitrary tensor products of $V$ and $V^\prime$ as shown in Theorem \ref{theo:tensorcon}.
In the case of (noncommutative) gravity, this provides us with a prescription to construct a connection
on arbitrary tensor fields.

\chapter{\label{chap:curvature}Curvature and torsion}
The aim of this chapter is to introduce curvature and torsion into the framework developed above. 
This is an essential step towards noncommutative gravity, see Chapter \ref{chap:ncgmath},
and also noncommutative Yang-Mills theory.
For an introduction to connections, curvature and torsion, see \cite{1038.58004,Madore:2000aq}.

Let $A\in\AAA$ be a unital and associative algebra, $(\Omega^\bullet,\wedge,\dd)$ be a differential calculus
over $A$ and $V\in\MMM_A$ a right $A$-module.
We can extend any connection $\nab\in\Con_A(V)$ to a $\bfK$-linear map
$\widetilde{\nab}: V\otimes_A \Omega^1\to V\otimes_A\Omega^2$ by defining
\begin{flalign}
 \widetilde{\nab} := (\id\otimes\wedge)\circ (\nab\otimes\id) + (\id\otimes\dd)~,
\end{flalign}
where the $A$-bimodule homomorphism $\wedge:\Omega^1\otimes_A\Omega^1\to\Omega^2$ is given by the product in $\Omega^\bullet$.
The map $\widetilde{\nab}$ is well-defined, since it is compatible with middle $A$-linearity, for all
$v\in V$, $\omega\in\Omega^1$ and $a\in A$,
\begin{flalign}
\nn \widetilde{\nab}\bigl(v\cdot a\otimes_A\omega \bigr) &= (\id\otimes\wedge)\bigl((\nab v)\cdot a \otimes_A \omega+ 
v\otimes_A \dd a\otimes_A\omega\bigr) + v\cdot a \otimes_A \dd\omega\\
\nn&= (\id\otimes\wedge)\bigl((\nab v)\otimes_A a\cdot \omega\bigr) 
 + v\otimes_A \bigl(\dd a\wedge \omega + a\cdot \dd\omega\bigr) \\
&= \widetilde{\nab}\bigl(v\otimes_A a\cdot\omega \bigr)~.
\end{flalign}

\begin{defi}\label{def:curvature}
 The {\it curvature} of a connection $\nab\in \Con_A(V)$ is defined by the right $A$-linear map
\begin{flalign}
 \mathrm{R}:= \widetilde{\nab}\circ\nab: V\to V\otimes_A\Omega^2~.
\end{flalign}
\end{defi}

We show that $\mathrm{R}$ is right $A$-linear. First, note that for all $v\in V$, $\omega\in\Omega^1$ and $a\in A$,
\begin{flalign}
\nn \widetilde{\nab}\bigl(v\otimes_A\omega\cdot a \bigr) &= (\id\otimes\wedge)\bigl((\nab v)\otimes_A\omega\cdot a\bigr) + v\otimes_A \dd (\omega \cdot a)\\
\nn&=(\id\otimes\wedge)\bigl((\nab v)\otimes_A\omega\bigr)\cdot a + v\otimes_A \bigl( (\dd \omega) \cdot a - \omega\wedge \dd a\bigr)\\
&=\widetilde{\nab}\bigl(v\otimes_A\omega\bigr)\cdot a - (\id\otimes\wedge)\bigl(v\otimes_A \omega\otimes_A\dd a\bigr)~.
\end{flalign}
Using this we obtain, for all $v\in V$ and $a\in A$,
\begin{flalign}
 \nn \mathrm{R}(v\cdot a) &= \widetilde{\nab}\bigl(\nab(v\cdot a)\bigr) = \widetilde{\nab}\bigl((\nab v)\cdot a + v\otimes_A \dd a\bigr)\\
\nn &=\widetilde{\nab}\bigl(\nab v\bigr) \cdot a - (\id\otimes \wedge)\bigl( (\nab v)\otimes_A \dd a\bigr)+  (\id\otimes \wedge)\bigl( (\nab v)\otimes_A \dd a\bigr) + v\otimes_A\dd \dd a\\
&=\mathrm{R}(v)\cdot a~.
\end{flalign}
\begin{rem}
In classical pseudo-Riemannian geometry the usual definition of curvature
is as a map $\text{Riem} : \Xi\otimes_A\Xi\otimes_A\Xi\to \Xi$, see also Chapter \ref{chap:basicncg},
which is antisymmetric in the first two legs. This map can be equivalently described
by a tensor in $\Xi\otimes_A \Omega^1\otimes_A\Omega^2$, the Riemann tensor,
which, due to duality of $\Xi$ and $\Omega^1$, canonically gives rise to a map $\Omega^1\to \Omega^1\otimes_A\Omega^2$.
In classical pseudo-Riemannian geometry this map coincides with the curvature defined above.
For our purpose it is more convenient to work with the Definition \ref{def:curvature}.
\end{rem}

We now study the behavior of the curvature $\mathrm{R}$ under twist quantization.
For this let $H$ be a Hopf algebra with twist $\FF\in H\otimes H$ and let $A\in{^{H,\ra}}\AAA$,
$V\in{^{H,\ra}}\MMM_A$ and  $\bigl(\Omega^\bullet,\wedge,\dd\bigr)$ be a left $H$-covariant differential calculus over $A$.
Twist quantization yields the deformed Hopf algebra $H^\FF$, the left $H^\FF$-modules
$A_\star\in {^{H^\FF,\ra}}\AAA$, $V_\star\in {^{H^\FF,\ra}}\MMM_{A_\star}$ and the left $H^\FF$-covariant differential
calculus $\bigl(\Omega^\bullet,\wedge_\star,\dd\bigr)$, see Theorems \ref{theo:HAdef}, \ref{theo:algebradef},
\ref{theo:moduledef} and Lemma \ref{lem:dcdef}.
Let $\nab_\star\in\Con_{A_\star}(V_\star)$ be any connection on deformed module $V_\star$, then its curvature
is given by the right $A_\star$-linear map $\mathrm{R}_\star:V_\star\to V_\star\otimes_{A_\star}\Omega^2_\star$
\begin{flalign}\label{eqn:defcurvature}
 \mathrm{R}_\star = \widetilde{\nab_\star} \circ\nab_\star = (\id\otimes\wedge_\star)\circ (\nab_\star\otimes\id)\circ \nab_\star
+ (\id\otimes\dd)\circ \nab_\star~.
\end{flalign}
Due to Theorem \ref{theo:condef} there is a unique connection $\nab\in\Con_A(V)$, such that
$\nab_\star = \widetilde{D}_\FF(\nab) = \iota^{-1}\circ D_\FF(\nab)$. We find the following
expression for the curvature $\mathrm{R}_\star$ in terms of the undeformed connection $\nab$.
\begin{propo}\label{propo:defcurvature}
 In the hypotheses above, the curvature $\mathrm{R}_\star$ of any connection $\nab_\star\in\Con_{A_\star}(V_\star)$
 can be expressed in terms of the corresponding undeformed connection $\nab\in \Con_A(V)$ by
\begin{flalign}\label{eqn:defcurvaturesimpl}
 \mathrm{R}_\star = \iota^{-1}\circ D_\FF\Bigl( (\id\otimes\wedge) \circ_\star (\nab\otimes\id)\circ_\star \nab
+ (\id\otimes \dd)\circ_\star \nab\Bigr)~,
\end{flalign}
where the $\star$-composition was defined in (\ref{eqn:starcompo}) and $\iota^{-1}$ in Lemma \ref{lem:iotaiso}.
\end{propo}
\begin{proof}
We investigate the individual terms in (\ref{eqn:defcurvature}) and make use of Theorem \ref{theo:promodhomdef}.

\noindent Firstly, the following diagram commutes
\begin{flalign}
\xymatrix{
 V_\star\otimes_{A_\star}\Omega^1_\star\otimes_{A_\star}\Omega^1_\star \ar[rrr]^-{\id\otimes\wedge_\star} \ar[d]_-{\iota_{123}} & & & V_\star \otimes_{A_\star}\Omega^2_\star\\
(V\otimes_A \Omega^1\otimes_A\Omega^1)_\star \ar[rrr]_-{D_\FF(\id\otimes\wedge)} & & & (V\otimes_A\Omega^2)_\star \ar[u]_-{\iota^{-1}}
}
\end{flalign}
Here it is essential to note that $\wedge$ is a left $H$-module homomorphism, 
i.e.~$\xi\RA \wedge =\epsilon(\xi)\,\wedge$ for all $\xi\in H$.

\noindent For the second term, $\nab_\star\otimes\id$, we find the following commuting diagram
\begin{flalign}
 \xymatrix{
V_\star\otimes_{A_\star}\Omega^1_\star \ar[rrr]^-{\nab_\star\otimes \id} \ar[d]_-{\iota}& & & V_\star\otimes_{A_\star}\Omega^1_\star\otimes_{A_\star}\Omega^1_\star\\
(V\otimes_A\Omega^1)_\star \ar[rrr]_-{D_\FF(\nab\otimes\id)} & & & (V\otimes_A \Omega^1\otimes_A\Omega^1)_\star \ar[u]_-{\iota^{-1}_{123}}
}
\end{flalign}

\noindent Since the differential $\dd$ is equivariant, we obtain the commuting diagram
\begin{flalign}
 \xymatrix{
V_\star\otimes_{A_\star}\Omega^1_\star \ar[rrr]^-{\id\otimes\dd} \ar[d]_-{\iota}& & & V_\star\otimes_{A_\star}\Omega^2_\star\\
(V\otimes_A\Omega^1)_\star \ar[rrr]_{D_\FF(\id\otimes\dd)} & & & (V\otimes_A\Omega^2)_\star \ar[u]_-{\iota^{-1}}
}
\end{flalign}
 
\noindent Combining these results, we can express the curvature $\mathrm{R}_\star$ in terms of $\nab$ by
(\ref{eqn:defcurvaturesimpl}).

\end{proof}

In (noncommutative) gravity a module $V$ of particular interest is the module of one-forms $\Omega^1$.
A right connection $\nab\in \Con_A(\Omega^1)$ is also called a {\it linear connection}.
Choosing $V = \Omega^1$ allows us to define and study torsion and the Ricci curvature,
a particular contraction of the curvature $\mathrm{R}$ above. Since these objects
are important for noncommutative gravity we are now going to investigate their properties
and behavior under twist quantization.

Let $A\in\AAA$ be a unital and associative algebra and $(\Omega^\bullet,\wedge,\dd)$ be a differential calculus
over $A$.
\begin{defi}\label{def:torsion}
 The {\it torsion} of a linear connection $\nab\in\Con_A(\Omega^1)$ is defined by the
right $A$-linear map
\begin{flalign}
 \mathrm{T}:\Omega^1\to\Omega^2~,\quad \omega\mapsto \mathrm{T}(\omega)= \wedge\bigl(\nab \omega\bigr) + \dd \omega~.
\end{flalign}
\end{defi}
We now show that $\mathrm{T}$ is right $A$-linear, for all $\omega\in\Omega^1$ and $a\in A$,
\begin{flalign}
 \nn\mathrm{T}(\omega\cdot a) &= \wedge\bigl((\nab \omega)\cdot a + \omega\otimes_A \dd a\bigr) + \dd(\omega\cdot a)\\
\nn&= \wedge\bigl((\nab \omega)\bigr)\cdot a + \omega\wedge \dd a + (\dd\omega)\cdot a - \omega\wedge \dd a\\
&=\mathrm{T}(\omega)\cdot a~.
\end{flalign}
\begin{rem}
 In classical pseudo-Riemannian geometry the usual definition of torsion is as a map $\text{Tor}:\Xi\otimes_A\Xi\to \Xi$,
see also Chapter \ref{chap:basicncg}, which is antisymmetric.
This map can be equivalently described
by a tensor in $\Xi\otimes_A\Omega^2$, the torsion tensor,
which, due to duality of $\Xi$ and $\Omega^1$, canonically gives rise to a map $\Omega^1\to \Omega^2$.
In classical pseudo-Riemannian geometry this map coincides with the torsion defined above.
\end{rem}

We now study the behavior of the torsion $\mathrm{T}$ under twist quantization.
For this let $H$ be a Hopf algebra with twist $\FF\in H\otimes H$, $A\in{^{H,\ra}}\AAA$
and let $\bigl(\Omega^\bullet,\wedge,\dd\bigr)$ be a left $H$-covariant differential calculus over $A$.
Let $\nab_\star\in\Con_{A_\star}(\Omega^1_\star)$ be any linear connection on deformed module $\Omega^1_\star$, 
then its torsion is given by the right $A_\star$-linear map $\mathrm{T}_\star:\Omega^1_\star\to \Omega^2_\star$
\begin{flalign}\label{eqn:deftorsion}
 \mathrm{T}_\star = \wedge_\star\circ\nab_\star + \dd~.
\end{flalign}
Due to Theorem \ref{theo:condef} there is a unique linear connection $\nab\in\Con_A(\Omega^1)$, such that
$\nab_\star = \widetilde{D}_\FF(\nab) = \iota^{-1}\circ D_\FF(\nab)$. We find the following
expression for the torsion $\mathrm{T}_\star$ in terms of the undeformed linear connection $\nab$.
\begin{propo}\label{propo:deftorsion}
 In the hypotheses above, the torsion $\mathrm{T}_\star$ of any linear connection $\nab_\star\in\Con_{A_\star}(\Omega^1_\star)$
 can be expressed in terms of the corresponding undeformed linear connection $\nab\in \Con_A(\Omega^1)$ by
\begin{flalign}
 \mathrm{T}_\star = D_\FF(\mathrm{T})~.
\end{flalign}
\end{propo}
\begin{proof}
Using $\wedge_\star= \wedge\circ\iota$ and $\nab_\star = \iota^{-1}\circ D_\FF(\nab)$
we obtain for (\ref{eqn:deftorsion})
\begin{flalign}
 \mathrm{T}_\star = \wedge\circ D_\FF(\nab) + \dd = D_\FF(\mathrm{T})~,
\end{flalign}
where in the second equality we have used that $\xi\RA \wedge =\epsilon(\xi)\,\wedge$ and $\xi\RA \dd = \epsilon(\xi)\,\dd$,
for all $\xi\in H$.

\end{proof}

These results on the curvature and torsion hold in a very general setting. In particular, we did not
have to assume the Hopf algebra to be (quasi)triangular and the modules to be finitely generated and projective
or quasi-commutative.
However, for constructing the Ricci curvature as a contraction of the curvature $\mathrm{R}$ 
we have to add some additional assumptions. 
In the following, let $(H,R)$ be a triangular Hopf algebra and let $A\in {^{H,\ra}}\AAA$ be quasi-commutative.
Let further $\bigl(\Omega^\bullet,\wedge,\dd\bigr)$ be a graded quasi-commutative left $H$-covariant differential calculus over $A$
and $\Omega^1\in{^{H,\ra}_A}\MMM_A$ be finitely generated and projective as a right $A$-module.
Moreover, we assume that there is an injective left $H$-module $A$-bimodule homomorphism
$\Omega^2\to \Omega^1\otimes_A\Omega^1$.\footnote{
In classical differential geometry this homomorphism exists, since two-forms can be regarded
as antisymmetric tensor fields in $\Omega^1\otimes_A\Omega^1$. The same hold true
in twist deformed differential geometry, since there two-forms can be regarded as $R$-antisymmetric
tensor fields.
}

In order to construct the Ricci curvature we first have to define a contraction map.
In the hypotheses above, let $V\in{^{H,\ra}_A}\MMM_A$ be quasi-commutative, and finitely generated and projective as a 
right $A$-module. We define the contraction map $\mathrm{contr}:V\otimes_A V^\prime \to A$ by the following diagram
\begin{flalign}\label{eqn:contrdiag}
 \xymatrix{
V\otimes_A V^\prime\ar[drr]_-{\tau_R} \ar[rr]^-{\mathrm{contr}} & & A\\
& & V^\prime\otimes_A V \ar[u]_-{\mathrm{ev}} 
}
\end{flalign}
where $\mathrm{ev}:V^\prime \otimes_A V \to A$ is the evaluation map defined by
$\mathrm{ev}(v^\prime\otimes_A v) = v^\prime(v)$. Note that $\mathrm{ev}$ is a left $H$-module
$A$-bimodule homomorphism, and hence $\mathrm{contr}$ is also a left $H$-module $A$-bimodule homomorphism
as a composition of maps with this property.

Due to the assumed injective left $H$-module $A$-bimodule homomorphism
$\Omega^2\to \Omega^1\otimes_A\Omega^1$, we can regard the curvature as
a right $A$-linear map $\Omega^1\to \Omega^1\otimes_A\Omega^1\otimes_A \Omega^1$, simply by
composing $\mathrm{R}$ with this homomorphism. 
Applying also the map $\varphi^{-1}:\Hom_A(\Omega^1,\Omega^1\otimes_A\Omega^1\otimes_A\Omega^1)\to 
\Omega^1\otimes_A\Omega^1\otimes_A\Omega^1\otimes_A\Omega^{1\prime}$ of 
Proposition \ref{propo:projhomiso}, we obtain from the curvature $\mathrm{R}$ the curvature tensor
$\mathbf{R}\in\Omega^1\otimes_A\Omega^1\otimes_A\Omega^1\otimes_A\Omega^{1\prime}$.
\begin{defi}
In the hypotheses above, the {\it Ricci curvature} of a linear connection $\nab\in\Con_A(\Omega^1)$ is defined 
by the application of the following contraction on the curvature tensor
\begin{subequations}
\begin{flalign}
\xymatrix{
 \Omega^1\otimes_A\Omega^1\otimes_A\Omega^1\otimes_A\Omega^{1\prime} \ar[rr]^-{\id\otimes\id\otimes\mathrm{contr}} && \Omega^1\otimes_A\Omega^1\otimes_A A \ar[rr]^-{\id\otimes\wedge}& &\Omega^1\otimes_A\Omega^1~,
}
\end{flalign}
i.e.~
\begin{flalign}
 \mathrm{Ric} := (\id\otimes\wedge)\Bigl((\id\otimes\id\otimes\mathrm{contr})\bigl(\mathbf{R}\bigr)\Bigr)~.
\end{flalign}
\end{subequations}
\end{defi}

We are now going to discuss the behavior of the contraction map and the Ricci curvature under twist quantization.
\begin{propo}
Let in the hypotheses above $\FF\in H\otimes H$ be a twist. 
Then the following diagram commutes
\begin{flalign}\label{eqn:contrstar}
\xymatrix{
 V_\star\otimes_{A_\star} (V_\star)^\prime \ar[d]_-{\id\otimes D_\FF^{-1}}\ar[rr]^-{\mathrm{contr}_\star} & & A_\star\\
 V_\star\otimes_{A_\star} (V^\prime)_\star \ar[rr]_-{\iota} & & (V\otimes_A V^\prime)_\star \ar[u]_-{\mathrm{contr}}
}
\end{flalign}

\end{propo}
\begin{proof}
 Our strategy is to extend the diagram (\ref{eqn:contrdiag}) for the deformed contraction 
by the natural isomorphisms and to show that all subdiagrams commute.
\begin{flalign}
 \xymatrix{
V_\star\otimes_{A_\star}(V^\prime)_\star \ar[d]_-{\iota}	& & V_\star\otimes_{A_\star} (V_\star)^\prime \ar[ll]_-{\id\otimes D_\FF^{-1}} \ar[rr]^-{\mathrm{contr}_\star} \ar[drr]_-{\tau_{R^\FF}}		& & A_\star 													& &  \ar[ll]_-{\mathrm{ev}} (V^\prime\otimes_A V)_\star\\
(V\otimes_A V^\prime)_\star     \ar[d]_-{\tau_R}              	& & 																		& & (V_\star)^\prime \otimes_{A_\star}V_\star \ar[u]_-{\mathrm{ev}_\star} \ar[rr]^-{D_\FF^{-1}\otimes\id}  	& &  (V^\prime)_\star\otimes_{A_\star}V_\star \ar[u]_-{\iota}\\
(V^\prime\otimes_A V)_\star	\ar[rr]_-{\iota^{-1}}		& &	(V^\prime)_\star\otimes_{A_\star}V_\star	\ar[urr]_-{D_\FF\otimes\id}								& &														& &		
}
\end{flalign}
The middle subdiagram commutes by definition of $\mathrm{contr}_\star$. 
The right subdiagram commutes as shown by this small calculation 
\begin{flalign}
\nn\mathrm{ev}\Bigl(\iota\bigl( (D_\FF^{-1}\otimes\id) (v_\star^\prime\otimes_{A_\star} v)\bigr) \Bigr) &=
\mathrm{ev}\Bigl(\bar f^\alpha\RA \bigl(D_\FF^{-1}(v_\star^\prime)\bigr)\otimes_{A} \bar f_\alpha\ra v \Bigr)\\
&= \Bigl(\bar f^\alpha\RA \bigl(D_\FF^{-1}(v_\star^\prime)\bigr)\Bigr)(\bar f_\alpha\ra v) 
= \Bigl(D_\FF\bigl(D_\FF^{-1}(v_\star^\prime)\bigr)\Bigr)(v) = v_\star^\prime(v)~.
\end{flalign}
Commutativity of the left subdiagram
is easily shown by a similar calculation, which we do not have to present here explicitly.

The proof follows by going the long path in the diagram above.

\end{proof}
As a consequence, we obtain for the deformed Ricci curvature 
of a deformed curvature tensor $\mathbf{R}_\star\in \Omega_\star^1\otimes_{A_\star}\Omega_\star^1\otimes_{A_\star}\Omega_\star^1
\otimes_{A_\star}(\Omega_\star^{1})^{\prime}$, i.e.~
\begin{flalign}
 \mathrm{Ric}_\star = (\id\otimes\wedge_\star)\Bigl((\id\otimes\id\otimes\mathrm{contr}_\star)\bigl(\mathbf{R}_\star\bigr)\Bigr)~,
\end{flalign}
the following property.
\begin{cor}
For all  $\mathbf{R}_\star\in \Omega_\star^1\otimes_{A_\star}\Omega_\star^1\otimes_{A_\star}\Omega_\star^1
\otimes_{A_\star}(\Omega_\star^{1})^{\prime}$,
\begin{flalign}\label{eqn:riccidef}
 \mathrm{Ric}_\star = \Bigl(\iota^{-1}\circ (\id\otimes\wedge)\circ (\id\otimes\id\otimes\mathrm{contr})\circ
 \iota_{1234}\circ (\id\otimes\id\otimes\id\otimes D_\FF^{-1})\Bigr)\bigl(\mathbf{R}_\star\bigr)~.
\end{flalign}

\end{cor}
\begin{proof}
The proof follows by using $\wedge_\star = \wedge\circ \iota$,  (\ref{eqn:contrstar})
and that $\wedge$ and $\mathrm{contr}$ are left $H$-module homomorphisms.
\begin{flalign}
\nn (\id\otimes\wedge_\star)&\circ (\id\otimes\id\otimes\mathrm{contr}_\star) 
= (\id\otimes\wedge)\circ\iota_{23}\circ (\id\otimes\id\otimes\mathrm{contr}) \circ\iota_{34} \circ (\id\otimes \id\otimes\id\otimes D_\FF^{-1})\\
\nn&=\iota^{-1}\circ\iota \circ (\id\otimes\wedge)\circ(\id\otimes\id\otimes\mathrm{contr})\circ \iota_{2(34)} \circ\iota_{34}
 \circ (\id\otimes \id\otimes\id\otimes D_\FF^{-1})\\
\nn&= \iota^{-1}\circ (\id\otimes\wedge)\circ(\id\otimes\id\otimes\mathrm{contr})\circ\iota_{1(234)}\circ \iota_{2(34)} \circ\iota_{34} \circ (\id\otimes \id\otimes\id \otimes D_\FF^{-1})\\
&=\iota^{-1}\circ (\id\otimes\wedge)\circ(\id\otimes\id\otimes\mathrm{contr})\circ\iota_{1234} \circ (\id\otimes \id\otimes\id \otimes D_\FF^{-1})~.
\end{flalign}

\end{proof}
We briefly study the double contraction of an element of $\Omega^1\otimes_{A}\Omega^1$ with an element of 
$\Omega^{1\prime}\otimes_{A}\Omega^{1\prime}$ defined by the map
\begin{flalign}
 \mathrm{contr}\circ(\wedge\otimes\id)\circ(\id\otimes\mathrm{contr}\otimes\id):\Omega^1\otimes_{A}\Omega^1\otimes_A\Omega^{1\prime}\otimes_A\Omega^{1\prime} \to A~.
\end{flalign}
This will be required in the next chapter to study the curvature scalar.
In the twist quantized setting, this map can be expressed in terms of the undeformed maps 
$\wedge$ and $\mathrm{contr}$ by using (\ref{eqn:contrstar}) as follows
\begin{multline}\label{eqn:rsdef}
  \mathrm{contr}_\star\circ(\wedge_\star\otimes\id)\circ(\id\otimes\mathrm{contr}_\star\otimes\id) = \\
\mathrm{contr}\circ(\wedge\otimes\id)\circ(\id\otimes\mathrm{contr}\otimes\id)\circ \iota_{1234} \circ (\id\otimes\id \otimes
D_\FF^{-1}\otimes D_\FF^{-1})~.
\end{multline}


\chapter{\label{chap:ncgmath}Noncommutative gravity solutions revisited}
In this chapter we provide a nontrivial application of the formalism developed above to noncommutative gravity solutions.
Remember that in Chapter \ref{chap:ncgsol} we have made use of a local nice basis of vector fields
in order to discuss when a commutative metric field also solves the noncommutative Einstein equations.
Even though we were able to provide interesting models for noncommutative cosmology and black hole physics,
there are two problems associated to this approach:
Firstly, a nice basis of vector fields is only ensured to exist for the class of nonexotic abelian twists,
i.e.~twists which are locally equivalent to the Moyal-Weyl twist. 
Secondly, even if a nice basis exists, it will in general be a local
basis defined on some coordinate patch. If we would go over from formal deformations to convergent ones,
these patches are in general not preserved under the action of the twist, so that we can not rely on
local constructions.

In this chapter we resolve these two problems by making use of the improved understanding of
homomorphisms and connections developed above. We discuss noncommutative gravity solutions 
in a completely global and basis free formulation.
In addition to providing a more elegant discussion of the solutions found before, we can extend
these results to general twists.

Let $(\MM,g)$ be an $N$-dimensional pseudo-Riemannian manifold and $A=C^\infty(\MM)$ the algebra
of complex valued smooth functions on $\MM$. 
The metric tensor $g\in\Omega^1\otimes_A\Omega^1$ gives rise to a right $A$-linear map $\mathrm{g}:\Xi\to \Omega^1$, where
$\Xi=\Omega^{1\prime}:=\Hom_A(\Omega^1,A)$ denotes the complexified and smooth vector fields and $\Omega^1$ the complexified 
and smooth one-forms on $\MM$. The diffeomorphisms of $\MM$ are described by the Hopf algebra
 $U\Xi$, see Chapter \ref{chap:basicncg}, which acts on $A,\Xi$ and $\Omega^1$ 
by the Lie derivative.
The algebraic structures of the objects above are
\begin{flalign}
 A\in {^{U\Xi,\mathcal{L}}}\AAA~,\quad \Xi,\Omega^1\in{^{U\Xi,\mathcal{L}}_A}\MMM_A~.
\end{flalign}
We assume in the following that $\Omega^1$ and thus also $\Xi$ are finitely generated and projective modules 
(since $A$ and $\Xi$ are commutative, these modules are finitely generated and projective 
as left and right $A$-modules).

For quantizing this system we have to extend it by formal powers in a deformation parameter $\lambda$.
In order to simplify our discussion we do not explicitly display this extension 
and also do not focus on topological aspects of formal power series, see the Appendix \ref{app:basicsdefq}.
For a discussion of the latter see the next chapter.
In this setting a twist is an element $\FF\in U\Xi\otimes U\Xi$, which leads to a quantization of
the Hopf algebra and its modules, see Theorems \ref{theo:HAdef}, \ref{theo:algebradef} and \ref{theo:moduledef}.
The twist quantized Hopf algebra is denoted by $U\Xi^\FF$, the twist quantized algebra by
 $A_\star$ and the twist quantized one-forms by $\Omega^1_\star$.
The twist quantized vector fields $\Xi_\star=\Hom_A(\Omega^1,A)_\star$ are isomorphic 
to $(\Omega_\star^1)^\prime = \Hom_{A_\star}(\Omega^1_\star,A_\star)$
via the map $D_\FF$, see Theorem \ref{theo:homodef}. 
The classical metric $\mathrm{g}:\Xi\to\Omega^1$ can be quantized 
 via the following diagram
\begin{flalign}\label{eqn:defmetrc}
 \xymatrix{
(\Omega^1_\star)^\prime \ar[d]_-{D_\FF^{-1}} \ar[rr]^-{\mathrm{g}_\star}& & \Omega^1_\star\\
\Xi_\star \ar[rru]_-{D_\FF(\mathrm{g})}
}
\end{flalign}
to yield a deformed metric $\mathrm{g}_\star:(\Omega^1_\star)^\prime\to \Omega^1_\star$.
Vice versa, any deformed metric is the quantization of a classical one, see Theorem \ref{theo:homodef}.
On the level of metric tensors, (\ref{eqn:defmetrc}) implies that
$g_\star = \iota^{-1}(g)\in\Omega^1_\star\otimes_{A_\star}\Omega^1_\star$,
where $\iota$ is the isomorphism of Lemma \ref{lem:iotaiso}.

Let us make our first observation: Consider the inverse classical metric $\mathrm{g}^{-1}:\Omega^1\to\Xi$ defined 
by $\mathrm{g}\circ \mathrm{g}^{-1} = \id$ and $\mathrm{g}^{-1}\circ \mathrm{g} =\id$. 
Then its quantization $(\mathrm{g}^{-1})_\star:\Omega^1_\star \to (\Omega^1_\star)^\prime$ defined by
\begin{flalign}
\xymatrix{
 \Omega^1_\star \ar[rrd]_-{D_\FF(\mathrm{g}^{-1})}\ar[rr]^-{(\mathrm{g}^{-1})_\star} & & (\Omega^1_\star)^\prime\\
 & & \Xi_\star \ar[u]_-{D_\FF}
}
\end{flalign}
is in general not equal to the $\star$-inverse metric $\mathrm{g}_\star^{-1}:\Omega^1_\star\to(\Omega^1_\star)^\prime$ defined by
$\mathrm{g}_{\star}^{-1}\circ \mathrm{g}_\star = \id$ and $\mathrm{g}_\star\circ  \mathrm{g}_\star^{-1}=\id$.
Indeed, we find
\begin{flalign}\label{eqn:zwischensol}
\mathrm{g}_\star \circ  (\mathrm{g}^{-1})_\star=  D_\FF(\mathrm{g})\circ D_\FF(\mathrm{g}^{-1}) 
= D_\FF(\mathrm{g}\circ_\star\mathrm{g}^{-1})\neq \id~,
\end{flalign}
due to the $\star$-composition (\ref{eqn:starcompo}). In order to ensure the $U\Xi^\FF$-covariance of the noncommutative
gravity theory the correct choice for an inverse metric is the $\star$-inverse metric $\mathrm{g}_\star^{-1}$ and not
$(\mathrm{g}^{-1})_\star$.

Our method for extracting exact solutions of the noncommutative Einstein equations in Chapter \ref{chap:ncgsol}
was based on searching for special deformations of symmetric pseudo-Riemannian manifolds $(\MM,g)$,
where the deformation has no effect on the curvature and on the Einstein tensor.
We denote the Lie algebra of Killing vector fields by $\mathfrak{g}$.
In the setting above, this means that we are looking for deformations such that, in particular,
$\mathrm{g}_\star^{-1}=(\mathrm{g}^{-1})_\star$. A sufficient condition for the inequality  
in (\ref{eqn:zwischensol}) to become an equality is to assume that the inverse twist is of the form 
\begin{flalign}\label{eqn:twistsemikilling}
 \FF^{-1}-1\otimes 1 \in U\Xi\,\mathfrak{g}\otimes U\Xi + U\Xi\otimes U\Xi\,\mathfrak{g}~.
\end{flalign}
This means that for all higher orders in the inverse twist there is a Killing vector field
on the very right in the left or right leg.
This is a generalization of the condition found for the abelian twist in Chapter \ref{chap:ncgsol}, Section \ref{sec:ncgsolgen}.
Due to this property and the $\mathfrak{g}$-invariance of $\mathrm{g}$ and $\mathrm{g}^{-1}$,
the $\star$-composition in (\ref{eqn:zwischensol}) reduces to the usual composition
and we obtain
\begin{flalign}
 \mathrm{g}_\star \circ (\mathrm{g}^{-1})_\star= \id~,\quad (\mathrm{g}^{-1})_\star\circ\mathrm{g}_\star =\id~.
\end{flalign}
Note that on the level of metric tensors, $(\mathrm{g}^{-1})_\star = \mathrm{g}_\star^{-1}$ implies that
$g_\star^{-1} = (D_\FF\otimes D_\FF)\bigl(\iota^{-1}(g^{-1})\bigr)\in(\Omega_\star^1)^\prime\otimes_{A_\star}(\Omega_\star^1)^\prime$,
where $g^{-1}\in\Xi\otimes_A\Xi$ is the classical inverse metric tensor.

The classical Levi-Civita connection $\nab\in\Con_A(\Omega^1)$ is defined uniquely by demanding metric compatibility 
and the torsionfree condition. 
Regarding the metric as a tensor field $g\in\Omega^1\otimes_A\Omega^1$, the former condition reads
\begin{flalign}
 (\nab\oplus\nab)\bigl(g\bigr) = 0~,
\end{flalign}
where $\oplus$ is the sum of connections (\ref{eqn:tensorcon}) with trivial $R$-matrix $1\otimes 1$.
The torsion as given in Definition \ref{def:torsion} is a right $A$-linear map $\mathrm{T}:\Omega^1\to\Omega^2$ defined by, 
for all $\omega\in\Omega^1$,
\begin{flalign}
 \mathrm{T}(\omega) = \wedge\bigl(\nab \omega\bigr) + \dd\omega ~.
\end{flalign}

In the deformed case, we say that a connection $\nab_\star\in\Con_{A_\star}(\Omega^1_\star)$
is metric compatible, if
\begin{flalign}
 (\nab_\star\oplus_R\nab_\star)\bigl(g_\star\bigr)=0~,
\end{flalign}
where $R=\FF_{21}\,\FF^{-1}$ and $g_\star\in\Omega^1_\star\otimes_{A_\star}\Omega^1_\star$ is the metric tensor.
The torsion is the right $A_\star$-linear map $\mathrm{T}_\star:\Omega_\star^1\to\Omega_\star^2$ defined by, for all
$\omega\in\Omega^1_\star$,
\begin{flalign}
 \mathrm{T}_\star(\omega) = \wedge_\star\bigl(\nab_\star \omega \bigr) + \dd \omega~.
\end{flalign}

From Theorem \ref{theo:condef} we know that we can find for all connections $\nab_\star\in\Con_{A_\star}(\Omega^1_\star)$
a unique connection $\nab\in\Con_A(\Omega^1)$, such that $\nab_\star = \iota^{-1}\circ D_\FF(\nab)$.
Even more, Theorem \ref{theo:prodcondef} tells us that
\begin{flalign}
 \nab_\star\oplus_R \nab_\star = \iota_{123}^{-1}\circ D_\FF(\nab\oplus\nab)\circ \iota~.
\end{flalign}
The deformed metric compatibility condition is then equivalent to
\begin{flalign}
  0=\iota_{123}\circ (\nab_\star\oplus_R\nab_\star)\bigl(g_\star\bigr)
= D_\FF(\nab\oplus\nab)\bigl(\iota(g_\star)\bigr)~.
\end{flalign}
Using that (\ref{eqn:defmetrc}) implies $g_\star = \iota^{-1}(g)$,
the deformed metric compatibility reduces to the condition
\begin{flalign}\label{eqn:metriccondsimpl}
 D_\FF(\nab\oplus\nab)\bigl(g\bigr)=0~.
\end{flalign}
As shown in Proposition \ref{propo:deftorsion}, the torsion $\mathrm{T}_\star$ of $\nab_\star$ and torsion
$\mathrm{T}$ of the corresponding undeformed connection $\nab$ are related by $\mathrm{T}_\star = D_\FF(\mathrm{T})$.
Thus, the torsion of a connection $\nab_\star$ vanishes, if and only if the torsion of
the corresponding undeformed connection $\nab$ vanishes.

Let now the twist be of the form (\ref{eqn:twistsemikilling}) and consider the classical Levi-Civita
connection $\nab\in\Con_A(\Omega^1)$. Then the remaining $D_\FF$ in the deformed metric compatibility condition
(\ref{eqn:metriccondsimpl}) drops out, since $g$ and $\nab$ are invariant under $\mathfrak{g}$.
The deformed torsion condition is also fulfilled, thus $\nab_\star = \iota^{-1}\circ D_\FF(\nab)$ is
a deformed Levi-Civita connection. We now show that all curvatures of this deformed Levi-Civita connection
coincide (up to natural isomorphisms) with the undeformed ones.

As shown in Proposition \ref{propo:defcurvature}, the curvature $\mathrm{R}_\star$ of the connection $\nab_\star$
can be expressed in terms of the corresponding undeformed connection $\nab$ via
\begin{flalign}
\mathrm{R}_\star = \iota^{-1}\circ D_\FF\Bigl( (\id\otimes\wedge) \circ_\star (\nab\otimes\id)\circ_\star \nab 
+ (\id\otimes \dd)\circ_\star \nab\Bigr)~.
\end{flalign}
For a twist of the form (\ref{eqn:twistsemikilling}) the remaining $\star$-compositions (\ref{eqn:starcompo}) reduce to
the usual compositions and we find that $\mathrm{R}_\star$ is equal to $\mathrm{R}$ up to the natural isomorphisms
\begin{flalign}
 \mathrm{R}_\star = \iota^{-1}\circ D_\FF(\mathrm{R})~.
\end{flalign}
The corresponding curvature tensor 
$\mathbf{R}_\star \in \Omega^1_\star\otimes_{A_\star}\Omega^1_\star\otimes_{A_\star}\Omega^1_\star\otimes_{A_\star}
(\Omega^1_\star)^\prime$ then can be expressed in terms of the undeformed curvature tensor
$\mathbf{R}$ by 
\begin{flalign}
 \mathbf{R}_\star = \bigl(\id\otimes\id\otimes\id\otimes D_\FF\bigr)\Bigl(\iota_{1234}^{-1}\bigl(\mathbf{R}\bigr)\Bigr)~.
\end{flalign}
Using (\ref{eqn:riccidef}) we can express the deformed Ricci curvature in terms of the undeformed one, $\mathrm{Ric}$, by
\begin{flalign}
 \mathrm{Ric}_\star = \iota^{-1}(\mathrm{Ric})~.
\end{flalign}
We define the curvature scalar as the double contraction
\begin{flalign}
\mathfrak{R}_\star :=  \bigl(\mathrm{contr}_\star\circ(\wedge_\star\otimes\id)\circ(\id\otimes\mathrm{contr}_\star\otimes\id)\bigr)
(\mathrm{Ric}_\star\otimes_{A_\star} g_\star^{-1})~.
\end{flalign}
The relation (\ref{eqn:rsdef}) and (\ref{eqn:twistsemikilling}) implies
 that the  deformed and undeformed curvature scalars coincide
\begin{flalign}
 \nn \mathfrak{R}_\star&=  \bigl(\mathrm{contr}\circ(\wedge\otimes\id)\circ(\id\otimes\mathrm{contr}\otimes\id)\bigr)
\bigl(\iota_{1234}(\iota^{-1}(\mathrm{Ric}) \otimes_{A_\star} \iota^{-1}(g^{-1})\bigr)\\
&= \bigl(\mathrm{contr}\circ(\wedge\otimes\id)\circ(\id\otimes\mathrm{contr}\otimes\id)\bigr)
\bigl(\mathrm{Ric} \otimes_{A} g^{-1}\bigr) =\mathfrak{R}~.
\end{flalign}
Finally, the deformed Einstein tensor is up to the natural isomorphisms equal to the undeformed one
\begin{flalign}\label{eqn:defeinstein}
 \mathrm{Ric}_\star -\frac{1}{2} g_\star\star \mathfrak{R}_\star
= \iota^{-1}\left( \mathrm{Ric} -\frac{1}{2} g\star\mathfrak{R} \right) = 
 \iota^{-1}\left( \mathrm{Ric} -\frac{1}{2} g\cdot\mathfrak{R} \right)~,
\end{flalign}
where in the last equality we have used that $g$ and $\mathfrak{R}$ are $\mathfrak{g}$-invariant.

We can now make statements about exact solutions of the noncommutative Einstein equations,
extending the results of Chapter \ref{chap:ncgsol}, Section \ref{sec:ncgsolgen}.
We start with the simplest case of Einstein manifolds\footnote{
Remember that $(\MM,g)$ is an Einstein manifold, if $\mathrm{Ric}-\frac{1}{2}g\cdot\mathfrak{R} =\Lambda\,g$ with $\Lambda\in\bbR$.
}.
Let $(\MM,g)$ be a classical Einstein manifold with Lie algebra of Killing vector fields
$\mathfrak{g}$ and let the twist $\FF$ be of the form (\ref{eqn:twistsemikilling}).
Then the deformed manifold $(\MM,g_\star=\iota^{-1}(g),\FF)$ is deformed Einstein, since due to (\ref{eqn:defeinstein}) 
we have
\begin{flalign}
 \mathrm{Ric}_\star -\frac{1}{2} g_\star\star \mathfrak{R}_\star =
 \iota^{-1}\left( \mathrm{Ric} -\frac{1}{2} g\cdot\mathfrak{R} \right) = \iota^{-1}\left(\Lambda\,g \right) = \Lambda\,g_\star~.
\end{flalign}

Consider now a classical manifold $\MM$ with a metric field $g\in\Omega^1\otimes_A\Omega^1$
and a collection of matter fields (tensor fields) $\lbrace \Phi_i\rbrace$.
Assume that $g$ and $\lbrace\Phi_i\rbrace$ are invariant under a symmetry Lie algebra
$\mathfrak{g}$ and that the fields solve exactly the classical Einstein equations and geometric differential equations
for the matter fields. If the twist is of the form (\ref{eqn:twistsemikilling}), then the deformed Einstein tensor
reduces, up to natural isomorphisms, to the undeformed one. 
Constructing a deformed stress-energy tensor with the same twist deformation methods,
it will also reduce (up to the natural isomorphisms) to the undeformed tensor when evaluated on $\mathfrak{g}$-symmetric
configurations. This is because the $\mathfrak{g}$-invariance of all fields prevents any corrections
in the deformation parameter $\lambda$ coming from the twist (\ref{eqn:twistsemikilling}).
The same holds true for the deformed equations of motion for the matter fields.
Thus, the classical solution gives rise to an exact solution of the deformed system via the natural isomorphisms.


\chapter{\label{chap:outlookmath}Open problems and outlook}
\subsection*{Topological algebras and modules:}
In this part we have considered the situation where the underlying commutative and unital ring
$\bfK$ is generic. The tensor product over $\bfK$ was taken to be the algebraic tensor product.
However, in deformation quantization the underlying ring is $\bbK[[\lambda]]$ (where $\bbK$ is a field),
which is equipped with the $\lambda$-adic topology, see the Appendix \ref{app:basicsdefq}. 
As also explained in this appendix, the algebraic tensor product over $\bbK[[\lambda]]$
is not always appropriate in this setting, and it has to be replaced by its $\lambda$-adic completion, called
the topological tensor product. A prime example which explains why this completion is required
comes from the Moyal-Weyl twist $\FF=\exp(-\frac{i\lambda}{2}\Theta^{\mu\nu}\partial_\mu\otimes\partial_\nu)$.
Note that this twist is not an element of $H[[\lambda]]\otimes_{\bbK[[\lambda]]} H[[\lambda]]$, but only an element
of the completion of this $\bbK[[\lambda]]$-module, i.e.~the topological tensor product module $H[[\lambda]]\totimes H[[\lambda]]$.
As a consequence, algebras should be replaced by topological algebras, where the product is a $\bbK[[\lambda]]$-linear
map $\mu: A[[\lambda]]\totimes A[[\lambda]]\to A[[\lambda]]$. Note that not all elements
in $A[[\lambda]]\totimes A[[\lambda]]$ are given by finite sums of products $a\otimes b$, but the elements
which can be written in this way are only dense. 
Since in the part above we have frequently checked properties of maps by acting on the simple elements 
$a\otimes b$, we are in general missing some elements in the topological setting. 
Thus, as presented in the chapters above, our statements are only valid in a non-topological setting.

As investigated in the Appendix \ref{app:basicsdefq}, for a special choice of
the underlying $\bbK[[\lambda]]$-modules, the topologically free modules $V[[\lambda]]$ given by formal
extension of vector spaces $V$, the $\lambda$-adic completion is quite harmless.
In particular, given any $\bbK[[\lambda]]$-linear map between incomplete $\bbK[[\lambda]]$-modules,
there is a canonical construction of a $\bbK[[\lambda]]$-linear map between the completions.
Using this construction we can for example associate to any usual algebra with
product $\mu:A[[\lambda]]\otimes_{\bbK[[\lambda]]}A[[\lambda]]\to A[[\lambda]]$ a topological algebra
with product $\mu:A[[\lambda]]\totimes A[[\lambda]] \to A[[\lambda]]$. Since 
$A[[\lambda]]\otimes_{\bbK[[\lambda]]}A[[\lambda]]$ is by definition dense in the topological tensor product
and any $\bbK[[\lambda]]$-linear map between topologically free modules is continuous,
a condition satisfied on this dense subset is indeed satisfied on all of topological 
tensor product module.

Such arguments should allow us to extend all of our results to the topological setting, at least
in case all underlying modules are topologically free. 
This is not in the scope of the present work and has to be left for future research.

\subsection*{Finishing the construction of a noncommutative gravity theory:}
We have made in this part some important steps towards constructing a noncommutative gravity theory.
In particular, we obtained a better understanding of module homomorphisms and connections,
which are the main ingredients of noncommutative gravity. The extension of connections to the dual
module and to tensor products of modules was clarified, thus allowing us to include arbitrary tensor
fields into our theory.

In order to complete the construction of a noncommutative theory of gravity there are still some unsolved
issues. In particular, implementing reality conditions and clarifying the existence of a unique deformed
Levi-Civita connection are not yet completely understood.
Since we have shown above that all deformed homomorphisms, tensor fields and connections
can be obtained by quantizing undeformed homomorphisms, tensor fields and connections,
we can rewrite the noncommutative gravity theory in terms of the unquantized
quantities, paying the price that $\star$-compositions might appear. 
An interesting project for future work is to find out if this reformulation
in terms of the unquantized quantities can lead to any new insights into reality conditions and
the existence of a unique Levi-Civita connection.

\subsection*{Noncommutative gauge theory:}
The formalism developed in this part also allows us to study kinematical aspects of noncommutative gauge theories.
For this case one takes $V$ to be a representation module of the desired gauge group,
i.e.~the module of sections of an associated vector bundle of the principal bundle
underlying the classical gauge theory. The gauge field is described in terms of a
connection on this module. Due to our work above, the twist quantization of this setting 
is understood and we can in particular express every noncommutative gauge field (described by
a connection $\nab_\star$) in terms of a commutative connection $\nab$.
In the deformed curvature (\ref{eqn:defcurvaturesimpl}) there will be $\star$-compositions,
making it in general different to the undeformed one. 

An interesting project for future research is to find out if and under which conditions
twist deformations can modify the value of the topological action, given
for a four-dimensional manifold by  \cite{Connes:2000tj,Landi:2006qw}
\begin{flalign}
\mathrm{Top}(V)= \int\limits_\MM \mathrm{Tr}_V\bigl(\mathrm{R}^2\bigr)~.
\end{flalign}
Note that all operations required to extend this definition to the twist deformed setting
have been clarified above, in particular the trace $\mathrm{Tr}_V$ can be expressed in terms of the contraction map
$\mathrm{contr}$ (\ref{eqn:contrdiag}).

For studying dynamical problems in noncommutative gauge theory, like e.g.~instantons  \cite{Connes:2000tj,Landi:2006qw},
we are still missing one ingredient in the deformed setting, the Hodge operator.
For twists of the form (\ref{eqn:twistsemikilling}), it is expected that the deformed
Hodge operator is (up to the natural isomorphisms) equal to the undeformed one.
Since the deformations (\ref{eqn:twistsemikilling}) do not assume all vector fields in the twist to be Killing vector fields,
these models go beyond the ones discussed in \cite{Connes:2000tj,Landi:2006qw}. It is interesting to
study if there are qualitatively new features in the properties of noncommutative instantons in this more general
setting. A detailed investigation of gauge invariance in these models, see Chapter \ref{chap:qftproblems},
would also be important.


\part*{Conclusions}


\vspace{3mm}

\begin{center}
 {\bf Summary of the results}
\end{center}
\vspace{2mm}

In this thesis we have studied different, but strongly connected, aspects of noncommutative geometry,
gravity and quantum field theory.
A particular focus was on developing the required mathematical formalism and methods in order to
bring noncommutative gravity closer to physical applications, like for example in noncommutative cosmology
and black hole physics.

In Part \ref{part:ncg} we have investigated symmetry reduction and exact solutions in noncommutative gravity.
The usual symmetry reduction procedure, which is frequently applied to construct 
exact solutions of Einstein's equations,
could not be directly applied and had to be modified in order to be compatible
 with the deformed Hopf algebra structure of the diffeomorphisms.
We have proposed in Chapter \ref{chap:symred} a possible generalization of symmetry reduction to the noncommutative setting
by making use of almost quantum Lie subalgebras of the quantum Lie algebra of diffeomorphisms.
These algebraic structures replace the usual Lie algebra structure of Killing vector fields in
Einstein's theory of gravity, and shall be interpreted as isometries of noncommutative spacetimes.
We have studied under which conditions a classical Lie algebra of Killing vector fields
can be quantized to yield an almost quantum Lie subalgebra and we have found that there
are compatibility conditions among the classical symmetries and the twist deformation we use.
This singles out preferred choices of twists.
For the special case of abelian Drinfel'd twists, we have shown that these conditions
reduce to the very simply conditions $[X_\alpha,\mathfrak{g}] \subseteq \mathfrak{g}$, for all $\alpha$,
where $X_\alpha$ are the vector fields generating the twist and $\mathfrak{g}$ is the classical Lie algebra
we want to deform. Taking the physically interesting models of a spatially flat Friedmann-Robertson-Walker universe and
a Schwarzschild black hole, we have classified all possible abelian twist deformation of these systems satisfying
our axioms for deformed symmetry reduction. This classification yielded interesting noncommutative 
cosmological and black hole models.

After establishing a formalism for symmetry reduction in noncommutative gravity we have taken the natural next
step in Chapter \ref{chap:ncgsol} and focussed on exact solutions of the noncommutative Einstein equations.
This is in general a highly nontrivial task, since the noncommutative Einstein equations
derived by Wess et al.~\cite{Aschieri:2005yw,Aschieri:2005zs} are not only nonlinear, but also contain
arbitrarily high derivatives due to the noncommutative deformation.
We have explicitly studied the noncommutative Einstein equations for symmetric
background spacetimes, in particular for the noncommutative Friedmann-Robertson-Walker and
Schwarzschild spacetimes found in our classification before.
We have found that for special choices of the twist, all noncommutative corrections
in the Einstein equations drop out, thus leading to exact solutions where the noncommutative metric field
coincides with the undeformed metric field. It was recognized that this is in particular the case
when the twist is constructed from sufficiently many Killing vector fields, such that
all $\star$-products in the noncommutative Einstein equations reduce to the usual products due to the invariance
of the tensors. Even though the noncommutative metric field coincides with the undeformed one for these solutions,
the underlying manifold is noncommutative, leading to distinct effects in the coordinate algebra, and thus
on the localization of spacetime points, and on fields propagating on these spacetimes.

In order to understand quantitatively the new effects in these models we have turned in
Part \ref{part:qft} to noncommutative quantum field theory.
In the literature, noncommutative quantum field theory is typically studied on the Moyal-Weyl deformed
Minkowski spacetime. However, for our models a more general formalism is required, covering also curved spacetimes
and more general deformations.
This provided us with the motivation to develop and study an approach
to quantum field theory on noncommutative curved spacetimes, which was presented in Chapter \ref{chap:qftdef}.
Our strategy was to combine methods from the algebraic approach to quantum field theory on curved spacetimes
with noncommutative differential geometry. 
The starting point is a compactly deformed
normally hyperbolic operator on a compactly deformed Lorentzian manifold acting on a scalar field. 
We have provided explicit examples of such operators coming from an action principle.
We have shown that, on the level of formal deformations, this operator admits unique retarded and advanced
Green's operators, provided the deformed Lorentzian manifold is time-oriented, connected and globally hyperbolic in the classical limit. 
The solution space of the deformed wave equation was characterized abstractly
in terms of the image of the retarded-advanced Green's operator, a feature which is analogous to commutative
field theories.
We have proven that the solution space can be equipped with a symplectic structure,
and thus can be quantized in terms of $\ast$-algebras of field polynomials or $\ast$-algebras of Weyl-type.
This yielded the algebra of observables for a scalar quantum field theory on a large class of noncommutative curved spacetimes.
In Chapter \ref{chap:qftcon} we have studied mathematical aspects of the quantum field theories constructed in 
Chapter \ref{chap:qftdef}. We in particular have shown that there is an easier, but equivalent,
formulation of such theories, where reality properties are more obvious.
Also in this chapter we have derived a remarkable relation between the quantum field theories
on noncommutative curved spacetimes and their commutative counterparts. We have first shown that there are
one-to-one maps between solutions of the deformed and undeformed wave equation. These maps
gave rise to symplectic isomorphisms, such that the deformed field theory can be equivalently described
by an undeformed one. The symplectic isomorphisms lift to algebra isomorphisms between the algebras of observables
of the deformed and undeformed quantum field theory, meaning that the deformed quantum field theory can be described
mathematically in terms of the corresponding undeformed one. 
We have shown that despite this mathematical equivalence, the deformed quantum field theory is able to
lead to new physical effects, provided we include a physical interpretation.

In order to study properties of quantum field theories on noncommutative curved spacetimes in more detail
we have considered explicit applications in Chapter \ref{chap:qftapp}.
The first goal was to derive deformed wave operators for explicit examples of noncommutative Minkowski, de Sitter, Schwarzschild
and anti-de Sitter spacetimes. We then have studied the quantum field theory on toy-models
of Friedmann-Robertson-Walker spacetimes, which were deformed by a twist generated by homothetic Killing vector fields.
The special structure of these deformations allowed us to go beyond the formal power series setting.
In the convergent approach we have obtained interesting new features, in particular, we have shown that
the deformed quantum field theory is not anymore isomorphic to the undeformed one, but only to a
subalgebra which does not include strongly localized observables. This is a particular realization of the
intuitive picture that noncommutative geometry should lead to an improvement of quantum field theories
in the ultraviolet. As a further application, we have considered a more realistic cosmological model
and we have explicitly constructed the symplectic isomorphism mapping between the deformed and the undeformed theory.
We have closed the chapter on applications by discussing briefly a new model for a perturbatively interacting
quantum field theory on a nonstandard noncommutative Euclidean space. This model has completely opposite features
compared to the usual Moyal-Weyl space, and we have shown in particular that 
our model has improved quantum properties at the one-loop level and that the infamous UV/IR-mixing does
not appear. We have pointed out a remarkable relation between our model and the recently proposed Ho{\v r}ava-Lifshitz theories.

Motivated by the open problems and unsolved issues we have faced during our studies on
noncommutative gravity, we have concentrated in Part \ref{part:math} on mathematical
developments in this field. A particular point which had to be clarified was the extension of connections
to tensor fields. 
In Chapter \ref{chap:prelim} and \ref{chap:HAdef} 
we have provided a mathematical introduction to the algebraic structures which are relevant for this part.
Motivated by noncommutative field theory and gravity, we have considered a Hopf algebra
$H$, describing the symmetries of our system, acting covariantly on a noncommutative algebra $A$, describing the quantized
functions on spacetime. We have further considered modules $V$ and $W$ over $A$, transforming covariantly under $H$,
which shall be interpreted as quantized sections of a vector bundle, e.g.~quantized vector fields or one-forms.
We have focused in Chapter \ref{chap:modhom} on endomorphisms of and homomorphisms between
modules of this type. The rich algebraic structure and twist quantization of these spaces was discussed.
We have proven that there is an isomorphism between the
quantized algebra of endomorphisms and the algebra of endomorphisms of the quantized module.
This has shown, in particular, that any undeformed endomorphism can be quantized to yield a deformed
one, and that all deformed endomorphisms can be obtained via this quantization map.
We have extended the result to homomorphisms between two modules.
Motivated again by noncommutative gravity, we have studied in detail the situation where
the Hopf algebra is triangular and all algebras and bimodules are commutative up the 
action of the $R$-matrix. In this quasi-commutative setting we have proven that there is an isomorphism
between the endomorphisms respecting the left module structure and the ones respecting the right module structure of the bimodule.
The same result applies to homomorphisms between quasi-commutative bimodules.
As a consequence, we were able to establish a precise relation between the left linear and right linear
dual of a bimodule. The extension of homomorphisms to tensor products of bimodules was studied in detail, and we have obtained
an explicit expression for the quantization of a product module homomorphism.

In Chapter \ref{chap:con} we have turned to connections on modules and also bimodules.
We have developed a new theory for the extension of connections to products of bimodules,
making heavy use of our results on homomorphisms. We have shown that this extension 
is compatible with twist quantization. Even more, we have proven that all connections on the quantized
module can be obtained by quantizing connections on the undeformed module. 
In the quasi-commutative setting we have also shown that there is an isomorphism between the connections
satisfying the left Leibniz rule and the connections satisfying the right Leibniz rule.
In Chapter \ref{chap:curvature} we gave a preliminary study of the curvature and torsion of
connections. We have derived explicit expressions for the deformed curvature and torsion in terms
of the undeformed connection.
The mathematical results of Chapters \ref{chap:modhom}, \ref{chap:con} and \ref{chap:curvature}
have been applied in Chapter \ref{chap:ncgmath} to noncommutative gravity solutions.
Making use of the new mathematical techniques, we were able to rederive our results
on exact noncommutative gravity solutions in an elegant, global and basis free way, 
and also to extend them to nonabelian twists.

\newpage

\vspace{3mm}

\begin{center}
 {\bf Outlook}
\end{center}
\vspace{2mm}

There are still open issues and room for further investigations in noncommutative gravity and
quantum field theory.
On the mathematical side, noncommutative gravity
is not yet a complete theory, and there remains a lot of hard work in order to consistently 
fill up the missing parts. Of major importance is to provide an existence and uniqueness theorem for a noncommutative
Levi-Civita connection and to understand reality properties of the curvature, torsion and Ricci tensor
in the realm of noncommutative geometry. Comparing, and maybe also combining, our approach with
the recently developed methods by Beggs and Majid \cite{Beggs201195,1742-6596-254-1-012002}
may be fruitful.
A detailed investigation of gravitational fluctuations around a given noncommutative
gravity background and finding out the number of propagating degrees of freedom is also a relevant topic for future research.

On the phenomenological side, we have now reached the stage where we are able to 
go beyond the standard example of the Moyal-Weyl Minkowski spacetime. The methods developed
in this thesis allow us to study a large variety of deformations of not only Minkowski spacetimes, 
but also curved ones, like e.g.~Friedmann-Robertson-Walker universes and Schwarzschild black holes.
A first phenomenological application of our formalism appeared in \cite{Ohl:2010bh}, where we have studied
deformations of Randall-Sundrum spacetimes and their effects on scattering experiments in particle physics.
By using a non-Moyal-Weyl deformation, possible experimental signatures of noncommutative geometry 
in this model are different to the usual noncommutative standard model. In particular, in our model only the
Kaluza-Klein excitations of gravitons are subject to noncommutative geometry effects,
and deviations from the standard model are only expected at the graviton resonance.

The investigation of noncommutative geometry effects, which are not of the Moyal-Weyl type,
in cosmology is an important next step. In particular, the model investigated in Chapter \ref{chap:qftapp},
 Section \ref{sec:isotrop}, is an excellent candidate to study. Since this model is isotropic around one point,
an observer located at this point will not detect any preferred directions in the cosmological microwave background
radiation, opposed to noncommutative cosmologies based on the Moyal-Weyl deformation.
The possible experimental signatures of this model will thus be distinct, and are worth to be worked out in detail.

Another important test for noncommutative quantum field theories is the Unruh effect and Hawking radiation.
These two effects are very robust predictions of quantum field theories on commutative curved spacetimes,
which, however, still wait for their experimental detection. Any deviation from the standard predictions
would be a very interesting result, which also may offer the possibility to pin down in future experiments
 the noncommutative structure of spacetime.


\begin{appendix}

\part*{Appendix}

\chapter{\label{app:basicsdefq}Formal power series and the $\lambda$-adic topology}

In this appendix we introduce the required mathematical methods in order to precisely  
 describe formal deformations. We follow mostly \cite{Kassel:1995xr}, but at some points we add more information
 and interpretations, which we expect to help the reader to understand the basic ideas and features.
 
\bgroup
\makeatletter
\def\@tocwrite#1#2{\relax}
\makeatother

\section{The ring $\bbK[[\lambda]]$ and its topology}
Let $\bbK=\bbR$ or $\bbC$ be the field of real or complex numbers. We consider the formal power series extension
$\bbK[[\lambda]]$, whose elements are given by 
\begin{flalign}
\bbK[[\lambda]]\ni \beta =  \sum\limits_{n=0}^\infty \lambda^n \,\beta_{(n)}~,
\end{flalign}
where $\beta_{(n)}\in\bbK$ for all $n\geq 0$.
Using the field structure of $\bbK$ we can equip $\bbK[[\lambda]]$ with the structure of a commutative 
and unital ring by defining
\begin{subequations}
\begin{flalign}
 \label{eqn:Klamplus}\beta +\gamma &:= \sum\limits_{n=0}^\infty\lambda^n\,(\beta_{(n)}+\gamma_{(n)})~,\\
\label{eqn:Klamtimes}\beta\,\gamma &:= \sum\limits_{n=0}^\infty \lambda^n \sum\limits_{m+k=n} \beta_{(m)}\,\gamma_{(k)}~,
\end{flalign}
\end{subequations}
for all $\beta,\gamma\in\bbK[[\lambda]]$.
Any polynomial in $\lambda$ can be considered as an element of $\bbK[[\lambda]]$.
The polynomials in $\lambda$ with coefficients in $\bbK$ are denoted by $\bbK[\lambda]$.
In particular, $1\in\bbK$ gives the unit in $\bbK[[\lambda]]$ by this identification.
Note that $\bbK[[\lambda]]$ is no field due to the following
\begin{lem}
 An element $\beta\in\bbK[[\lambda]]$ is invertible in $\bbK[[\lambda]]$ if and only if $\beta_{(0)}\neq 0$.
\end{lem}
\begin{proof}
 $\beta\in\bbK[[\lambda]]$ is invertible if and only if there is a $\gamma\in\bbK[[\lambda]]$, such that
$\beta\,\gamma = 1$. By (\ref{eqn:Klamtimes}) this is equivalent to
$\beta_{(0)}\,\gamma_{(0)}=1$ and
\begin{flalign}
\label{eqn:Klaminv}
\beta_{(0)}\,\gamma_{(n)}+ \beta_{(1)}\gamma_{(n-1)}+ \dots + \beta_{(n)}\gamma_{(0)}=0~,
\end{flalign}
for all $n>0$. The equation $\beta_{(0)}\,\gamma_{(0)}=1$ shows that $\beta_{(0)}\neq 0 $ is a necessary condition
for the invertability of $\beta$. This condition is also sufficient, since $\gamma_{(0)}=1/\beta_{(0)}$ and 
(\ref{eqn:Klaminv}) can be used to determine recursively the elements $\gamma_{(n)}$ for $n>0$.

\end{proof}

As explained in Section \ref{sec:invlim} of this appendix, the ring $\bbK[[\lambda]]$ is isomorphic to 
the inverse limit $\lim\limits_{\overleftarrow{~n~}} \bbK[\lambda]/(\lambda^n)$,
where $(\lambda^n)$ is the ideal generated by $\lambda^n$. It thus can be equipped with a topology, called
the inverse limit topology.  In our setting this topology is also called the $\lambda$-adic topology. The $\lambda$-adic topology
is a metric topology, where the metric is defined as follows: For any nonzero $\beta\in\bbK[[\lambda]]$
we define $\omega(\beta)$ as the unique nonnegative integer, such that $\beta_{(\omega(\beta))}\neq 0$ and
$\beta_{(n)}=0$ for all $n<\omega(\beta)$. We set $\omega(0)=+\infty$. We define a map $\vert\cdot\vert:\bbK[[\lambda]] \to [0,\infty)$
by
\begin{flalign}
 \vert \beta\vert = 2^{-\omega(\beta)}~,
\end{flalign}
for $\beta\neq 0$ and $\vert 0 \vert =0$.
The map $\vert\cdot\vert$ satisfies obviously for all $\beta,\gamma\in\bbK[[\lambda]]$
\begin{flalign}
 \vert \beta \vert=0~\Longleftrightarrow~\beta=0~,\quad \vert -\beta\vert = \vert \beta\vert~,\quad \vert \beta+\gamma\vert \leq
\text{max}(\vert \beta\vert,\vert\gamma\vert)~.
\end{flalign}
\begin{cor}
\label{cor:Kadicmetric}
 Define $d(\beta,\gamma)=\vert\beta-\gamma\vert$ for any $\beta,\gamma\in\bbK[[\lambda]]$. Then $d$ is an ultrametric
on $\bbK[[\lambda]]$, i.e.~we have for all $\beta,\gamma,\delta\in\bbK[[\lambda]]$
\begin{subequations}
 \begin{flalign}
  \nn(i)&\quad d(\beta,\gamma)=0~\Longleftrightarrow~\beta=\gamma~,\\
 \nn(ii)&\quad d(\beta,\gamma)=d(\gamma,\beta)~,\\
\nn(iii)&\quad d(\beta,\delta)\leq \text{max}\bigl(d(\beta,\gamma),d(\gamma,\delta)\bigr)~.
 \end{flalign}
\end{subequations}
\end{cor}

\begin{lem}
\label{lem:densepoly}
 With respect to the $\lambda$-adic topology the polynomials $\bbK[\lambda]$ are dense in $\bbK[[\lambda]]$.
\end{lem}
\begin{proof}
 Let $\beta\in \bbK[[\lambda]]$ be arbitrary. We have to find a family $\beta_k\in \bbK[\lambda]$, $k>0$, 
such that $\beta-\beta_k\in \lambda^k \bbK[[\lambda]]$ for all $k>0$. We define
$\beta_k=\sum_{n=0}^{k-1}\lambda^n\,\beta_{(n)}$ and find
$\beta-\beta_k = \sum_{n=k}^\infty\lambda^n\,\beta_{(n)}\in \lambda^k \bbK[[\lambda]]$ for all $k>0$.

\end{proof}


\section{Topologically free $\bbK[[\lambda]]$-modules}
As a reminder, a module over a unital ring is a generalization of a vector space over a field.
\begin{defi}
 A left $R$-module $M$ over a unital ring $R$ consists of an abelian group $(M,+)$ and a map $\cdot: R\times M\to M$, such that
for all $r,s\in R$ and $m,n\in M$ 
\begin{subequations}
 \begin{flalign}
  r\cdot (m+n)&=r\cdot m+ r\cdot n~,\\
  (r+s)\cdot m&= r\cdot m + s\cdot m~,\\
  (r\,s)\cdot m &= r\cdot(s\cdot m)~,\\
  1\cdot m &= m~.
 \end{flalign}
\end{subequations}
\end{defi}
\noindent For a commutative ring we can drop the term {\it left}. We will also drop $\cdot$ in this section
and write the action of the ring $\bbK[[\lambda]]$ on its modules simply by juxtaposition.

Let $M,N$ be two $\bbK[[\lambda]]$-modules. A $\bbK[[\lambda]]$-module homomorphism 
$P\in\Hom_{\bbK[[\lambda]]}(M,N)$ is a map $P:M\to N$, such that
\begin{subequations}
 \begin{flalign}
  P(\beta\, m) &= \beta\, P(m)~,\\
  P(m+m^\prime)&= P(m)+P(m^\prime)~,
 \end{flalign}
\end{subequations}
for all $\beta\in\bbK[[\lambda]]$ and $m,m^\prime\in M$. We call this type of maps
also $\bbK[[\lambda]]$-linear maps.

Let $M$ be a $\bbK[[\lambda]]$-module. The family of $\bbK[[\lambda]]$-modules 
$(M_n=M/\lambda^nM)_{n\geq 0}$ together with the natural projections 
$(p_n:M_n\to M_{n-1})_{n>0}$ forms an inverse system of $\bbK[[\lambda]]$-modules.
We define the $\bbK[[\lambda]]$-module
\begin{flalign}
\widetilde{M} := \lim\limits_{\overleftarrow{~n~}}M_n~,
\end{flalign}
which has a natural topology, namely the inverse limit topology.
The module $\widetilde{M}$ is called the completion of $M$.
The projections $i_n:M\to M_n$ induce a unique $\bbK[[\lambda]]$-linear map
$i:M\to\widetilde{M}$, such that $\pi_n\circ i = i_n$ for all $n$ (see Proposition \ref{propo:invlimhom}).
We say that the module $M$ is separated, if $\Ker(i)=\lbrace 0\rbrace$, and that 
$M$ is complete, if $i$ is surjective.
For a separated and complete module $i$ provides an isomorphism $M\cong \widetilde{M}$
and we can induce a topology on $M$. As before, we call the topology induced by the inverse limit topology
the $\lambda$-adic topology.  
\begin{propo}
 Let $M,N$ be two separated and complete $\bbK[[\lambda]]$-modules and let
$P\in \Hom_{\bbK[[\lambda]]}(M,N)$ be arbitrary. Then $P$ is a continuous map with respect to the
$\lambda$-adic topology.
\end{propo}
\begin{proof}
 We have to show that for all $m\in M$ and all open balls $P(m)+\lambda^n N$, $n>0$,
 there is an open neighborhood $U_n$ of $m$, such that $P(U_n)\subseteq P(m)+\lambda^n N$.
 Let $U_n = m+ \lambda^n M$, then $P(U_n) = P(m + \lambda^nM) = P(m)+ \lambda^n P(M)
\subseteq P(m) + \lambda^n N$, where we used that $P$ is a $\bbK[[\lambda]]$-module homomorphism.

\end{proof}

Provided a $\bbK[[\lambda]]$-linear map between $M$ and $N$ there is a canonical construction of
a $\bbK[[\lambda]]$-linear map between the completions $\widetilde{M}$ and $\widetilde{N}$.
Consider the inverse systems $(M_n= M/\lambda^nM,p_n)_{n\geq 0}$ 
and $(N_n=N/\lambda^nN,p_n^\prime)_{n\geq 0}$. The map $P:M\to N$ induces a family of $\bbK[[\lambda]]$-linear maps by
\begin{flalign}
 P_n:M_n\to N_n~,~[m]\mapsto [P(m)]~.
\end{flalign}
These maps are well-defined since $P[\lambda^n M]\subseteq \lambda^nN$. The requirements of
Proposition \ref{propo:invlimhom2} (more precisely the analogous statement for modules)
 are satisfied and we obtain the $\bbK[[\lambda]]$-linear map 
\begin{flalign}
\lim\limits_{\overleftarrow{~n~}}P_n: \widetilde{M}\to \widetilde{N}~.
\end{flalign}

We now consider a restricted, but very important, class of separated and complete $\bbK[[\lambda]]$-modules.
Let $V$ be a vector space over $\bbK$ and let $V[[\lambda]]$ be its formal power series extension.
Any element of $V[[\lambda]]$ is of the form
\begin{flalign}
 V[[\lambda]]\ni v=\sum\limits_{n=0}^\infty\lambda^n\,v_{(n)}~,
\end{flalign}
where $v_{(n)}\in V$ for all $n\geq 0$. The $\bbK$-vector space structure of $V$ allows us to induce
a $\bbK[[\lambda]]$-module structure on $V[[\lambda]]$ by defining
\begin{subequations}
\begin{flalign}
 \label{eqn:Vlamplus}v +v^\prime &:= \sum\limits_{n=0}^\infty\lambda^n\,(v_{(n)}+v^\prime_{(n)})~,\\
\label{eqn:Vlamtimes}\beta\, v &:= \sum\limits_{n=0}^\infty \lambda^n \sum\limits_{m+k=n} \beta_{(m)}\,v_{(k)}~,
\end{flalign}
\end{subequations}
for all $\beta\in\bbK[[\lambda]]$ and $v,v^\prime\in V[[\lambda]]$. 
$\bbK[[\lambda]]$-modules of the type $V[[\lambda]]$, where $V$ is a $\bbK$-vector space, are called
topologically free.

The $\bbK[[\lambda]]$-module $V[[\lambda]]$ can be equipped with the $\lambda$-adic topology.
This topology can be induced from a metric
similarly to Corollary \ref{cor:Kadicmetric}:
For any nonzero $v\in V[[\lambda]]$ we define $\omega(v)$ as the unique nonnegative integer, 
such that $v_{(\omega(v))}\neq 0$ and $v_{(n)}=0$ for all $n<\omega(v)$. 
We set $\omega(0)=+\infty$. We define a map $\vert\cdot\vert: V[[\lambda]] \to [0,\infty)$
by
\begin{flalign}
 \vert v \vert = 2^{-\omega(v)}~,
\end{flalign}
for $v\neq 0$ and $\vert 0 \vert =0$. The metric on $V[[\lambda]]$ is then defined by 
$d(v,v^\prime) = \vert v-v^\prime\vert$, for all $v,v^\prime\in V[[\lambda]]$.

Analogously to Lemma \ref{lem:densepoly} one can show that the polynomials $V[\lambda]$ are
dense in $V[[\lambda]]$ with respect to the $\lambda$-adic topology.

There is a simple characterization of topologically free modules.
We state it without proof and refer to \cite{Kassel:1995xr} for details.
\begin{propo}
A $\bbK[[\lambda]]$-module is topologically free if and only if
it is separated, complete and torsion-free (i.e.~$\lambda m\neq 0$ for all
$m\neq 0$).
\end{propo}


\section{Homomorphisms between topologically free $\bbK[[\lambda]]$-modules}
Let $V,W,Z$ be $\bbK$-vector spaces and let $V[[\lambda]],W[[\lambda]],Z[[\lambda]]$
be the corresponding topologically free $\bbK[[\lambda]]$-modules.
It turns out that all $P\in\Hom_{\bbK[[\lambda]]}(V[[\lambda]],W[[\lambda]])$ can be equivalently
described by a family $(P_{(n)}:V\to W)_{n\geq 0}$ of vector space homomorphisms.
\begin{propo}
\label{eqn:topmodulemaps}
 Let $(P_{(n)}:V\to W)_{n\geq 0}$ be an arbitrary family of $\bbK$-linear maps, then
 \begin{flalign}
\label{eqn:vectormodulemap}
  P:V[[\lambda]]\to W[[\lambda]]\,,~v\mapsto P(v) = \sum\limits_{n=0}^\infty \lambda^n \sum\limits_{m+k=n} P_{(m)}(v_{(k)})~
 \end{flalign}
is a $\bbK[[\lambda]]$-linear map.\vspace{1mm}\\
 Let now $P\in\Hom_{\bbK[[\lambda]]}(V[[\lambda]],W[[\lambda]])$ be arbitrary, then we can define
a family of $\bbK$-linear maps $(P_{(n)}:V\to W)_{n\geq 0}$ by $P_{(n)}(v):= \bigl(P(v)\bigr)_{(n)}$ for all $v\in V$.
The original $P$ is equal to the map (\ref{eqn:vectormodulemap}) constructed by these $P_{(n)}$.
\end{propo}
\begin{proof}
 We have to check if (\ref{eqn:vectormodulemap}) is a $\bbK[[\lambda]]$-linear map.
We have for all $v,v^\prime\in V[[\lambda]]$ and $\beta\in\bbK[[\lambda]]$
\begin{flalign}
 \nn P(v+v^\prime) &= \sum\limits_{n=0}^\infty \lambda^n \sum\limits_{m+k=n} P_{(m)}(v_{(k)}+v^\prime_{(k)})\\
 \nn &=  \sum\limits_{n=0}^\infty \lambda^n \Bigl(\sum\limits_{m+k=n} P_{(m)}(v_{(k)}) + \sum\limits_{m+k=n} P_{(m)}(v^\prime_{(k)})\Bigr)\\
 &= P(v)+P(v^\prime)~,
\end{flalign}
and 
\begin{flalign}
 \nn P(\beta\,v) &= \sum\limits_{n=0}^\infty \lambda^n\,\sum\limits_{m+k=n} P_{(m)}\Bigl(\sum\limits_{i+j=k}\beta_{(i)}\,v_{(j)}\Bigr)\\
 \nn &=\sum\limits_{n=0}^\infty \lambda^n\,\sum\limits_{m+i+j=n} P_{(m)}\Bigl(\beta_{(i)}\,v_{(j)}\Bigr)\\
 \nn &=\sum\limits_{n=0}^\infty \lambda^n\,\sum\limits_{m+i+j=n} \beta_{(i)}\,P_{(m)}(v_{(j)})\\
 \nn &=\sum\limits_{n=0}^\infty \lambda^n\,\sum\limits_{i+k=n}\beta_{(i)}\,\sum\limits_{m+j=k}P_{(m)}(v_{(j)})\\
 &= \beta\,P(v)~.
\end{flalign}

Let now $P\in\Hom_{\bbK[[\lambda]]}(V[[\lambda]],W[[\lambda]])$ and define a family of $\bbK$-linear
maps $(P_{(n)}:V\to W)_{n\geq 0}$ by $P_{(n)}(v):= \bigl(P(v)\bigr)_{(n)}$, for all $v\in V$. 
Obviously, $P_{(n)}(v+v^\prime) = \bigl(P(v+v^\prime)\bigr)_{(n)} = \bigl(P(v)+P(v^\prime)\bigr)_{(n)} = P_{(n)}(v)+P_{(n)}(v^\prime)$
and $P_{(n)}(\beta\,v) = \bigl(P(\beta\,v)\bigr)_{(n)}= \bigl(\beta\, P(v)\bigr)_{(n)} = \beta\,P_{(n)}(v)$, for
all $v,v^\prime\in V$ and $\beta\in\bbK$.  Since all $\bbK[[\lambda]]$-linear maps are continuous and
$V[\lambda]$ is dense in $V[[\lambda]]$ it is sufficient to proof that $P$ is equal to (\ref{eqn:vectormodulemap})
constructed by the $P_{(n)}$ above on $V[\lambda]$. We obtain for all $v\in V[\lambda]$
\begin{flalign}
\nn  P(v) &= P\Bigl(\sum\limits_{n=0}^N \lambda^n\,v_{(n)}\Bigr) = \sum\limits_{n=0}^N\lambda^n P(v_{(n)})\\
\nn &= \sum\limits_{n=0}^N\lambda^n \sum\limits_{m=0}^\infty \lambda^m\,\bigl(P(v_{(n)})\bigr)_{(m)} 
= \sum\limits_{m=0}^\infty \sum\limits_{n=0}^N\lambda^{n+m}  \,P_{(m)}(v_{(n)})\\
 &= \sum\limits_{m=0}^\infty \lambda^m \sum\limits_{n+k=m}P_{(n)}(v_{(k)})~.
\end{flalign}

\end{proof}
For the map (\ref{eqn:vectormodulemap}) we introduce the following convenient notation
\begin{flalign}
\label{eqn:vectormodulemap2}
 P = \sum\limits_{n=0}^\infty \lambda^n\,P_{(n)}~.
\end{flalign}
Given $P,Q\in\Hom_{\bbK[[\lambda]]}(V[[\lambda]],W[[\lambda]])$ and $\beta\in\bbK[[\lambda]]$
 we can define new $\bbK[[\lambda]]$-linear maps
$\beta\,P$ and $P+Q$ by
\begin{flalign}
 (\beta\,P)(v) := \beta\,P(v)~,\quad (P+Q)(v) := P(v)+ Q(v)~,
\end{flalign}
for all $v\in V[[\lambda]]$.  This means that $\Hom_{\bbK[[\lambda]]}(V[[\lambda]],W[[\lambda]])$ is a $\bbK[[\lambda]]$-module
itself. Note that due to (\ref{eqn:vectormodulemap2}) the $\bbK[[\lambda]]$-module
$\Hom_{\bbK[[\lambda]]}(V[[\lambda]],W[[\lambda]])$ is topologically free. 
In the notation (\ref{eqn:vectormodulemap2}) we have
\begin{flalign}
 \beta\,P = \sum\limits_{n=0}^\infty\lambda^n\sum\limits_{m+k=n}\beta_{(m)}\,P_{(k)}~,\quad 
 P+Q = \sum\limits_{n=0}^\infty \lambda^n \bigl( P_{(n)}+ Q_{(n)}\bigr)~,
\end{flalign}
where on the right hand side we have used the $\bbK$-vector space structure on $\bbK$-linear maps.

Let now $P\in\Hom_{\bbK[[\lambda]]}(V[[\lambda]],W[[\lambda]]) $ and $Q\in\Hom_{\bbK[[\lambda]]}(W[[\lambda]],Z[[\lambda]])$.
We can compose $P$ and $Q$, yielding  $Q\circ P\in\Hom_{\bbK[[\lambda]]}(V[[\lambda]],Z[[\lambda]])$
defined by $(Q\circ P)(v) := Q(P(v))$, for all $v\in V[[\lambda]]$. In the notation
(\ref{eqn:vectormodulemap2}) the composed map reads
\begin{flalign}
 Q\circ P = \sum\limits_{n=0}^\infty\lambda^n \sum\limits_{m+k=n} Q_{(m)}\circ P_{(k)}~,
\end{flalign}
where on the right hand side the composition is the composition of $\bbK$-linear maps.


\section{The topological tensor product}
Let $M,N$ be $\bbK[[\lambda]]$-modules. 
The algebraic tensor product of $M$ and $N$, denoted by $M\otimes_{\bbK[[\lambda]]} N$, is as usual
defined as the quotient of the free $\bbK[[\lambda]]$-module over the Cartesian product $M\times N$
by the submodule generated by
\begin{subequations}
\begin{flalign}
 &(m+m^\prime, n) - (m,n) - (m^\prime , n)~,\\
 &(m,n +n^\prime) - (m,n) - (m, n^\prime)  ~,\\
 &(\beta\, m,\gamma\,n) - \beta\,\gamma\,(m,n)~,
\end{flalign}
\end{subequations}
for all $m,m^\prime \in M$, $n,n^\prime\in N$ and $\beta,\gamma\in\bbK[[\lambda]]$.
Given $m\in M$ and $n\in N$ we denote by $m\otimes_{\bbK[[\lambda]]}n$ the image of $(m,n)$ under the natural 
$\bbK[[\lambda]]$-bilinear map $M\times N\to M\otimes_{\bbK[[\lambda]]} N$.

For reasons which become clear later in this section, the algebraic tensor product is not always suitable when
dealing with formal power series. 
\begin{defi}
Let $M,N$ be two $\bbK[[\lambda]]$-modules. The topological tensor product $M\widetilde{\otimes} N$ 
is defined as the $\lambda$-adic completion of $M\otimes_{\bbK[[\lambda]]}N$, i.e.~
\begin{flalign}
M\widetilde{\otimes}N := \left(M\otimes_{\bbK[[\lambda]]}N\right)^{\widetilde{~~}} = \lim_{\overleftarrow{~n~}} \left(M\otimes_{\bbK[[\lambda]]}N\right)/\lambda^n\left(M\otimes_{\bbK[[\lambda]]}N\right) ~.
\end{flalign} 
\end{defi}
Given $m\in M$ and $n\in N$ we denote by $m\widetilde{\otimes} n $ the image of $(m,n)$ under the natural
maps $M\times N\to M\otimes_{\bbK[[\lambda]]}N \to M\widetilde{\otimes} N$. The span of these elements lies dense in 
$M\widetilde{\otimes}N$.
Furthermore, the topological tensor product has the following properties, for all $\bbK[[\lambda]]$-modules
$M,N,L$
\begin{subequations}
\begin{flalign}
(M\totimes N)\totimes L &\cong M\totimes (N\totimes L)~,\\
\label{eqn:commutoptens}M\totimes N &\cong N\totimes M~,\\
\bbK[[\lambda]]\totimes M &= \widetilde{M} = M\totimes \bbK[[\lambda]]~.
\end{flalign}
\end{subequations}
This is a direct consequence of the definition of the topological tensor product as an inverse limit,
the corresponding $\bbK[[\lambda]]$-module isomorphisms for the algebraic tensor product 
and Proposition \ref{propo:invlimhom2}.
For a better understanding we prove (\ref{eqn:commutoptens}):
Let $\tau: M\otimes_{\bbK[[\lambda]]} N \to N\otimes_{\bbK[[\lambda]]} M $ be the corresponding
$\bbK[[\lambda]]$-module isomorphism for the algebraic tensor product.
The $\bbK[[\lambda]]$-linear maps 
\begin{flalign}
 \tau_n: &M\otimes_{\bbK[[\lambda]]} N /\lambda^n(M\otimes_{\bbK[[\lambda]]} N) \to 
N\otimes_{\bbK[[\lambda]]} M/\lambda^n(N\otimes_{\bbK[[\lambda]]} M)~,~[v]\mapsto [\tau(v)]~,
\end{flalign}
are well-defined isomorphisms since $\tau[\lambda^n(M\otimes_{\bbK[[\lambda]]} N)] = \lambda^n(N\otimes_{\bbK[[\lambda]]} M)$.
The requirements for Proposition \ref{propo:invlimhom2} are satisfied and we obtain a unique $\bbK[[\lambda]]$-module
isomorphism $M\totimes N \to N\totimes M$.

The topological tensor product is functorial, i.e.~provided $\bbK[[\lambda]]$-linear
maps $P:M\to M^\prime$ and $Q:N\to N^\prime$ then there is a $\bbK[[\lambda]]$-linear
map
\begin{flalign}
P\totimes Q :M\totimes N\to M^\prime\totimes N^\prime
\end{flalign}
enjoying the formal properties of the algebraic tensor product. Let us be more detailed at this point
and explain the construction of $P\totimes Q$. Let $P\otimes_{\bbK[[\lambda]]} Q:
M\otimes_{\bbK[[\lambda]]} N\to M^\prime\otimes_{\bbK[[\lambda]]} N^\prime$ be the $\bbK[[\lambda]]$-linear map
defined by $P\otimes_{\bbK[[\lambda]]}Q(m\otimes_{\bbK[[\lambda]]}n) = P(m)\otimes_{\bbK[[\lambda]]}Q(n)$,
for all $m\in M$ and $n\in N$. We define a family of $\bbK[[\lambda]]$-linear maps
\begin{flalign}
\nonumber (P\otimes_{\bbK[[\lambda]]} Q)_n:&
M\otimes_{\bbK[[\lambda]]} N/\lambda^n(M\otimes_{\bbK[[\lambda]]} N) \to M^\prime\otimes_{\bbK[[\lambda]]} 
N^\prime/\lambda^n(M^\prime\otimes_{\bbK[[\lambda]]} N^\prime)~,\\
& [v] \mapsto [P\otimes_{\bbK[[\lambda]]}Q(v)]~.
\end{flalign}
These maps are well-defined since due to $\bbK[[\lambda]]$-linearity 
$(P\otimes_{\bbK[[\lambda]]}Q) [\lambda^n( M\otimes_{\bbK[[\lambda]]} N)]
\subseteq \lambda^n(M^\prime\otimes_{\bbK[[\lambda]]} N^\prime)$.
The requirements for Proposition \ref{propo:invlimhom2} are satisfied and we define 
$P\totimes Q := \lim\limits_{\overleftarrow{~n~}}(P\otimes_{\bbK[[\lambda]]}Q)_n$.
For composeable maps $P\totimes Q$ and $P^\prime\totimes Q^\prime$ we obtain
\begin{flalign}
(P^\prime\totimes Q^\prime)\circ (P\totimes Q) = (P^\prime\circ P)\totimes (Q^\prime\circ Q)~.
\end{flalign}

The topological tensor product closes on topologically free modules.
\begin{propo}
Let $V,W$ be $\bbK$-vector spaces. Then $V[[\lambda]]\totimes W[[\lambda]] = (V\otimes_\bbK W)[[\lambda]]$.
\end{propo}
\noindent For a proof of this proposition we refer to \cite{Kassel:1995xr}.


\section{Topological algebras, Hopf algebras and modules}
The definition of topological algebras, Hopf algebras and  modules 
is analogous to the algebraic case (see Chapter \ref{chap:prelim}), with the difference that all algebraic
tensor products are replaced by topological ones. For completeness, we are going to review them.

A topological algebra is a triple $(A,\mu,\mathbf{e})$, where 
$A$ is a $\bbK[[\lambda]]$-module, $\mu:A\totimes A\to A$ and
$\mathbf{e}:\bbK[[\lambda]]\to A$ are $\bbK[[\lambda]]$-linear maps, such that
\begin{subequations}
 \begin{flalign}
  \mu\circ (\mu\totimes\id_A) &= \mu\circ (\id_A\totimes \mu)~,\\
  \mu\circ (\mathbf{e}\totimes \id_A) &= \id_A =\mu\circ (\id_A\totimes \mathbf{e})~.
 \end{flalign}
\end{subequations}
A  homomorphism $\varphi:(A,\mu,\mathbf{e})\to(A^\prime,\mu^\prime,\mathbf{e}^\prime)$ of topological algebras
is a $\bbK[[\lambda]]$-linear map $\varphi:A\to A^\prime$, such that
\begin{flalign}
 \varphi\circ \mu = \mu^\prime \circ (\varphi\totimes\varphi)~,\quad \varphi\circ\mathbf{e}=\mathbf{e}^\prime~.
\end{flalign}

Let us provide some examples.
Identifying $\bbK[[\lambda]]\totimes \bbK[[\lambda]]$ with $\bbK[[\lambda]]$,
we have the topological algebra $(\bbK[[\lambda]],\id_{\bbK[[\lambda]]},\id_{\bbK[[\lambda]]})$.
Furthermore, given two topological algebras $(A,\mu,\mathbf{e})$ and $(A^\prime,\mu^\prime,\mathbf{e}^\prime)$ then
\begin{flalign}
 \bigl(A\totimes A^\prime, \mu\totimes \mu^\prime \circ(\id_A \totimes \tau\totimes \id_{A^\prime}),
\mathbf{e}\totimes\mathbf{e}^\prime\bigr)~,
\end{flalign}
where $\tau:A^\prime \totimes A \to A\totimes A^\prime$ is the canonical isomorphism,
is a topological algebra.

A  topological bialgebra is a quintuple $(A,\mu,\mathbf{e},\Delta,\epsilon)$, where $(A,\mu,\mathbf{e})$
is a topological algebra, $\Delta:A\to A\totimes A$ and $\epsilon: A\to \bbK[[\lambda]]$ are 
homomorphisms of topological algebras, such that
\begin{subequations}
\begin{flalign}
 (\id_A\totimes \Delta)\circ \Delta &= (\Delta\totimes \id_A)\circ \Delta~,\\
(\epsilon\totimes \id_A)\circ \Delta &= \id_A =(\id_A\totimes \Delta)\circ\Delta~.
\end{flalign}
\end{subequations}
A homomorphism $\varphi:(A,\mu,\mathbf{e},\Delta,\epsilon)\to
(A^\prime,\mu^\prime,\mathbf{e}^\prime,\Delta^\prime,\epsilon^\prime)$  of topological bialgebras
is a homomorphism $\varphi$ between the underlying topological algebras, such that
\begin{flalign}
 (\varphi\totimes\varphi)\circ \Delta = \Delta^\prime\circ \varphi~,\quad \epsilon=\epsilon^\prime\circ \varphi~.
\end{flalign}

A  topological Hopf algebra is a sextuple $(H,\mu,\mathbf{e},\Delta,\epsilon,S)$, where
$(H,\mu,\mathbf{e},\Delta,\epsilon)$ is a topological bialgebra and $S: H\to H$ is a $\bbK[[\lambda]]$-linear
map, such that
\begin{flalign}
 \mu \circ(\id_H\totimes S)\circ\Delta = \mu \circ(S\totimes \id_H)\circ\Delta = \mathbf{e}\circ\epsilon~.
\end{flalign}

A topological right $A$-module is a triple $\bigl(V,(A,\mu,\mathbf{e}),\cdot\bigr)$, where $V$ is a $\bbK[[\lambda]]$-module,
$(A,\mu,\mathbf{e})$ is a topological algebra and $\cdot:V\totimes A\to V$ is a $\bbK[[\lambda]]$-linear map,
such that
\begin{flalign}
\cdot \circ (\id_V\totimes\mu) = \cdot \circ (\cdot \totimes\id_A)~,\quad \cdot\circ (\id_V\totimes \mathbf{e} ) = \id_V~.
\end{flalign}
A topological left $A$-module and $A$-bimodule are defined analogously.


\section{\label{sec:invlim}Appendix: Inverse limits}
Since inverse limits might not be familiar to the reader, we briefly review the basic ideas and 
definitions following \cite{Kassel:1995xr}. An inverse system of abelian groups $(A_n,p_n)$ is a family $(A_n)_{n\geq 0}$
of abelian groups and of group homomorphisms $(p_n:A_n\to A_{n-1})_{n>0}$. Given this data
we can define the inverse limit $\lim\limits_{\overleftarrow{~n~}}A_n$ by
\begin{flalign}
\lim\limits_{\overleftarrow{~n~}}A_n := \Bigl\lbrace (x_n)_{n\geq 0}\in \prod\limits_{n\geq 0} A_n\,:\,p_n(x_n) = x_{n-1}\,\text{ for all }n>0 \Bigr\rbrace~.
\end{flalign}
The inverse limit can be equipped with an abelian group structure by defining 
$(x_n)_n +(y_n)_n = (x_n+y_n)_n$, for all $(x_n)_n,(y_n)_n\in \lim\limits_{\overleftarrow{~n~}}A_n$.
Remark:  $(x_n)_n +(y_n)_n\in\lim\limits_{\overleftarrow{~n~}}A_n$, since
 $p_n(x_n+y_n)=p_n(x_n)+p_n(y_n)=x_{n-1}+y_{n-1}$ for all $n>0$.
 
There is a natural group homomorphism $\pi_k:\lim\limits_{\overleftarrow{~n~}}A_n\to A_k\,,~(x_n)_n\mapsto x_k$, for all
$k\in\bbN^0$. If all $p_n$ are surjective, then so are the $\pi_n$. By definition  of the inverse limit we have
$p_n\circ \pi_n = \pi_{n-1}$, for all $n>0$.

Let $C$ be any abelian group and $f:C\to \lim\limits_{\overleftarrow{~n~}} A_n$ be a group homomorphism.
Then we can define a family of group homomorphisms $(f_n:=\pi_n \circ f: C\to A_n)_{n\geq 0}$, satisfying
$p_n\circ f_n=p_n\circ \pi_n\circ f = \pi_{n-1}\circ f = f_{n-1}$, for all $n>0$.
The other way around is also possible due to 
\begin{propo}\label{propo:invlimhom}
For any abelian group $C$ and any given family $(f_n:C\to A_n)_{n\geq 0}$ of group homomorphisms
such that $p_n\circ f_n =f_{n-1}$ for all $n>0$, there exists a unique group homomorphism
\begin{flalign}
f:C\to \lim\limits_{\overleftarrow{~n~}}A_n~,
\end{flalign}
such that $\pi_n\circ f = f_n$ for all $n\geq 0$.
\end{propo}
\begin{proof}
The family $(f_n)_n$ defines a unique group homomorphism $f$ from $C$
to $\prod_{n\geq0} A_n$. Due to $p_n\circ f_n = f_{n-1}$, for all $n>0$,
the image of $f$ is contained in the inverse limit. This shows the existence of $f$.

Assume that there are two maps $f,\tilde f$ satisfying the hypothesis above.
Then $\pi_n\circ (f-\tilde f) =f_n-f_n=0$ for all $n\geq 0$, and thus $\tilde f =f$. This shows the uniqueness of $f$.

\end{proof}
This proposition establishes a one-to-one correspondence between group homomorphisms into the inverse limit
and families of group homomorphisms into the inverse system satisfying the hypothesis above.

The inverse limit is functorial, i.e.~provided a homomorphism between inverse systems there is
a unique homomorphism between the corresponding inverse limits.
\begin{propo}\label{propo:invlimhom2}
Let $(A_n,p_n)$ and $(A_n^\prime,p_n^\prime)$ be two inverse systems of abelian groups and
let $(f_n:A_n\to A^\prime_n)_{n\geq 0}$ be a family of group homomorphisms such that
$p_n^\prime \circ f_n = f_{n-1}\circ p_n$ for all $n>0$. Then there exists a unique group homomorphism
\begin{flalign}
 f = \lim\limits_{\overleftarrow{~n~}}f_n : \lim\limits_{\overleftarrow{~n~}}A_n\to \lim\limits_{\overleftarrow{~n~}}A_n^\prime~,
\end{flalign}
such that $\pi_n^\prime \circ f = f_n\circ \pi_n$ for all $n\geq 0$.
\end{propo}
\begin{proof}
Consider the family of group homomorphisms $(f_n\circ \pi_n:\lim\limits_{\overleftarrow{~n~}}A_n\to A^\prime_n)_{n\geq0}$.
We have $p^\prime_n\circ f_n\circ \pi_n = f_{n-1}\circ p_n\circ \pi_n = f_{n-1}\circ \pi_{n-1}$, for all $n>0$.
By Proposition \ref{propo:invlimhom} there is a unique 
$f:\lim\limits_{\overleftarrow{~n~}}A_n\to \lim\limits_{\overleftarrow{~n~}}A_n^\prime$, such
that $\pi^\prime_n\circ f = f_n\circ \pi_n$, for all $n\geq 0$.

\end{proof}
Let $(A_n,p_n)$, $(A_n^\prime,p_n^\prime)$ and $(A^{\prime\prime}_n,p^{\prime\prime}_n)$ be three inverse systems
of abelian groups and let $(f_n:A_n\to A^\prime_n)_{n\geq 0}$ and $(g_n:A^\prime_n\to A^{\prime\prime}_n)_{n\geq 0}$
be two families of group homomorphisms such that $p_n^{\prime\prime}\circ g_n=g_{n-1}\circ p_n^\prime$
and $p_n^\prime\circ f_n = f_{n-1}\circ p_n$ for all $n>0$. Obviously, the family of group homomorphisms
$(g_n\circ f_n:A_n\to A_n^{\prime\prime})_{n\geq 0}$ satisfies $p_n^{\prime\prime}\circ g_n\circ f_n =
g_{n-1}\circ f_{n-1}\circ p_n$ for all $n>0$, and by Proposition \ref{propo:invlimhom2} there is a unique
group homomorphism $\lim\limits_{\overleftarrow{~n~}}(g_n\circ f_n) :\lim\limits_{\overleftarrow{~n~}}A_n\to 
\lim\limits_{\overleftarrow{~n~}}A_n^{\prime\prime} $ satisfying $\pi_n^{\prime\prime} \circ 
\lim\limits_{\overleftarrow{~n~}}(g_n\circ f_n) = g_n\circ f_n\circ \pi_n$, for all $n\geq 0$.
We obtain by employing the uniqueness
\begin{flalign}
 g\circ f = \bigl(\lim\limits_{\overleftarrow{~n~}}g_n\bigr)\circ \bigl(\lim\limits_{\overleftarrow{~n~}}f_n\bigr) = 
\lim\limits_{\overleftarrow{~n~}}(g_n\circ f_n)~,
\end{flalign}
since $\pi_n^{\prime\prime}\circ g\circ f = g_n\circ \pi_n^\prime \circ f = g_n\circ f_n\circ \pi_n$ for all $n\geq0$.
As a consequence, a group homomorphism $f=\lim\limits_{\overleftarrow{~n~}}f_n:\lim\limits_{\overleftarrow{~n~}}A_n\to 
\lim\limits_{\overleftarrow{~n~}}A_n^{\prime}$ is an isomorphism, if  $f_n:A_n\to A_n^\prime$
are group isomorphisms for all $n\geq 0$. The inverse is $f^{-1}=\lim\limits_{\overleftarrow{~n~}}f_n^{-1}:
\lim\limits_{\overleftarrow{~n~}}A^\prime_n\to \lim\limits_{\overleftarrow{~n~}}A_n$.

\begin{rem}
Analogously one can define an inverse system of rings, modules, algebras, etc., by replacing the abelian groups $A_n$
and group homomorphisms $p_n$ by the desired algebraic objects and morphisms, respectively.
\end{rem}

To become more familiar with these notions let us provide an
\begin{ex}
Let $\bbK=\bbR$ or $\bbC$ and $\bbK[[\lambda]]$ the ring of formal power series.
We denote the (ring) ideal generated by $\lambda^n$ by $(\lambda^n)$ and define $\bbK_n:= \bbK[[\lambda]]/(\lambda^n)$
for all $n\in\bbN^0$. Note that $K_n$ is isomorphic (as a ring) to the truncated polynomials $\bbK[\lambda]/(\lambda^n)$, where
 $(\lambda^n)$ here denotes the ideal w.r.t.~the polynomial ring $\bbK[\lambda]$.

Due to the inclusion $(\lambda^n)\subset (\lambda^{n-1})$  the identity map $\id:\bbK[[\lambda]]\to\bbK[[\lambda]]$
reduces to surjective ring homomorphisms $p_n:\bbK_n\to\bbK_{n-1}$.
Using representatives, $p_n$ is given by 
\begin{flalign}
p_n\left(\sum\limits_{m=0}^{n-1}\lambda^m\, \beta_{(m)}\right)=\sum\limits_{m=0}^{n-2}\lambda^m\,\beta_{(m)}~,
\end{flalign}
i.e.~it ``throws away'' the highest power in $\lambda$.

The inverse limit is then given by
\begin{flalign}
\lim\limits_{\overleftarrow{~n~}}K_n := \Bigl\lbrace (k_n)_{n\geq 0}\in \prod\limits_{n\geq 0} K_n\,:\,p_n(k_n) = k_{n-1}\,\text{ for all }n>0 \Bigr\rbrace~.
\end{flalign}
Using representatives we can get a heuristic understanding of this ring:
The elements of the inverse limit are collections of truncated formal power series, 
where (for $n>m$) $k_n$ and $k_m$ agree up to order $m-1$.  
Since the inverse limit includes all $n\in\bbN^0$ one expects that 
it is isomorphic (as a ring) to $\bbK[[\lambda]]$ itself. 

This can be made rigorous by employing Proposition \ref{propo:invlimhom}.
Consider the family of projections $(f_n:\bbK[[\lambda]]\to\bbK_n)_{n\geq0}$ where 
$f_n(\beta) = [\beta]$ for all $\beta\in\bbK[[\lambda]]$ (for notational reasons we do not label $[\cdot]$ by $n$).
It is easy to see that $p_n\circ f_n=f_{n-1}$ for all $n>0$. By Proposition \ref{propo:invlimhom}
(more precisely the analogous statement for rings) there is a unique $f:\bbK[[\lambda]]\to \lim\limits_{\overleftarrow{~n~}}K_n$,
such that $\pi_n\circ f = f_n$. This map is obviously injective and surjective, thus it is a ring isomorphism.
\end{ex}

The inverse limit can be naturally equipped with a topology, the so-called inverse limit topology.
This is done as follows: Equip all groups $A_n$ with the discrete topology, i.e.~the topology where every subset is open.
The inverse limit topology is then the restriction of the product topology on $\prod_{n\geq0}A_n$.
This means that a basis of open sets is given by $\pi_n^{-1}(U_n)$, where $n\in\bbN^0$ and $U_n\subset A_n$ 
is any subset. 
A neighborhood basis of $0\in \lim\limits_{\overleftarrow{~n~}} A_n$ is given by 
$\pi^{-1}_n(\lbrace 0 \rbrace)$, $n\geq 0$, since $\lbrace 0\rbrace$ is a neighborhood basis
of $0\in A_n$.

\egroup


\chapter{\label{app:twistqft}On the twist deformation of the algebra of field polynomials}
In this appendix we study the twist approach to quantum field theory, which is typically used for the Moyal-Weyl deformation
of a Minkowski quantum field theory \cite{Zahn:2006wt,Balachandran:2007vx,Aschieri:2007sq}. 
We show that the twist deformation of the algebra
of field polynomials along homothetic Killing vector fields is possible if and only if the vector fields are Killing.
Note that in the approach presented in Part \ref{part:qft} no such restriction exist.

Let $\mathcal{A}_{(V,\omega)}[[\lambda]]$ be the formal power series extension of the $\ast$-algebra of field polynomials 
of a commutative quantum field theory. The basic idea of the twist approach to noncommutative quantum field theory 
is to replace the usual algebra product by a $\star$-product
\begin{flalign}
 a\star b=(\bar f^\alpha\triangleright a)\,(\bar f_\alpha \triangleright b)~,
\end{flalign}
for all $a,b\in\mathcal{A}_{(V,\omega)}[[\lambda]]$. We restrict ourselves to twists 
$\mathcal{F}^{-1}=\bar f^\alpha\otimes\bar f_\alpha$ generated by homothetic 
Killing vector fields $\mathfrak{H}$, i.e.~for all $v\in\mathfrak{H}$ there is a $c_v\in\bbC$, such that
we have for the metric $\mathcal{L}_v(g) = c_v\,g$.

We assume the action $\triangleright$ of the twist on 
$\mathcal{A}_{(V,\omega)}$ to be the natural (geometric) action in order to interpret the deformation 
as a spacetime deformation. The geometric action of the Lie algebra $\mathfrak{H}$ is defined on the 
generators of $\mathcal{A}_{(V,\omega)}$ by 
\begin{subequations}
\begin{flalign}
 &v\triangleright 1:=0~,\\
 \label{eqn:geometricaction}&v\triangleright\Phi\bigl([\varphi]\bigr):= \Phi\bigl([\mathcal{L}_v(\varphi)]\bigr)~, 
\end{flalign}
\end{subequations}
for all $v\in\mathfrak{H}$ and $[\varphi]\in V$. The action is extended to $\mathcal{A}_{(V,\omega)}$ by 
$\bbC$-linearity and the Leibniz rule $v\triangleright (a\,b)=(v\triangleright a)\,b + a\,(v\triangleright b)$,
for all $a,b\in\mathcal{A}_{(V,\omega)}$ and $v\in\mathfrak{H}$. The extension to $\AA_{(V,\omega)}[[\lambda]]$
is straightforward.

It has to be checked if (\ref{eqn:geometricaction}) is well-defined. For this let
$[\varphi]=[\varphi^\prime]$, i.e.~$\varphi^\prime = \varphi + P(\psi)$, where $P=\square_g-\xi\,\mathfrak{R}$ is the equation
of motion operator and $\psi\in C^\infty_0(\mathcal{M},\bbR)$. We find
\begin{flalign}
 \mathcal{L}_v(\varphi^\prime)= \mathcal{L}_v(\varphi)+\mathcal{L}_{v}\bigl(P(\psi)\bigr) = \mathcal{L}_v(\varphi)
 +P\bigl(\mathcal{L}_v(\psi) - c_v \psi\bigr)~,
\end{flalign}
where we have used that the scaling of the d'Alembert operator is $[\mathcal{L}_v,\square_g]=-c_v\,\square_g$.
Thus, the action is well-defined for all $v\in\mathfrak{H}$.

Next, we have to check if the action of $\mathfrak{H}$ is consistent with the commutation relations in $\mathcal{A}_{(V,\omega)}$.
We obtain the consistency condition 
(omitting the brackets $[\,\cdot\,]$ denoting equivalence classes)
\begin{multline}
\label{eqn:consistencycom}
 0=v\triangleright\bigl( i\, \omega(\varphi,\psi)\, 1\bigr)=v\triangleright[\Phi(\varphi),\Phi(\psi)]=
 [v\triangleright\Phi(\varphi),\Phi(\psi)]+[\Phi(\varphi),v\triangleright\Phi(\psi)]\\
 =[\Phi(\mathcal{L}_v(\varphi)),\Phi(\psi)]+[\Phi(\varphi),\Phi(\mathcal{L}_v(\psi))]=i\bigl(
\omega(\mathcal{L}_v(\varphi),\psi)+\omega(\varphi,\mathcal{L}_v(\psi))\bigr)\,1~,
\end{multline}
for all $\varphi,\psi\in V$ and $v\in\mathfrak{H}$.
Using the explicit form of $\omega$ we find
\begin{flalign}
 \nn \omega(\mathcal{L}_v(\varphi),\psi)&=\int\limits_\mathcal{M}\mathcal{L}_v(\varphi)\,\Delta(\psi)\,\vol\stackrel{\text{PI}}{=}-
\int\limits_\mathcal{M}\varphi\,\mathcal{L}_v\bigl(\Delta(\psi)\,\vol\bigr) \\
\nn & \hspace{-1.4mm}\stackrel{\text{HKP}}{=} -\int\limits_\mathcal{M}\varphi\,
\left(\Delta\bigl(\mathcal{L}_v(\psi)\bigr) +c_v\,\left(\frac{N}{2}+1\right)\,
 \Delta(\psi) \right)\, \vol\\
\label{eqn:leftright} &=-\omega(\varphi,\mathcal{L}_v(\psi))-c_v\,\left(\frac{N}{2}+1\right)\,\omega(\varphi,\psi)~,
\end{flalign}
for all $\varphi,\psi\in V$ and $v\in\mathfrak{H}$. In this derivation we have used integration by parts (PI)
and the homothetic Killing properties $\mathcal{L}_v(\vol)=\frac{ c_v N}{2}\vol$ and $[\mathcal{L}_v,\Delta]=c_v\Delta$
(HKP).
Putting (\ref{eqn:leftright}) into (\ref{eqn:consistencycom}) the consistency condition reads
\begin{flalign}
0= -c_v\,\left(\frac{N}{2}+1\right)\,\omega(\varphi,\psi)~,~\forall \varphi,\psi\in V~,
\end{flalign}
which implies $c_v=0$ due to the (weak) nondegeneracy of the symplectic structure $\omega$.
Thus, we can only represent the Lie subalgebra of Killing vector fields $\mathfrak{K}\subseteq\mathfrak{H}$ 
on $\mathcal{A}_{(V,\omega)}$, provided we assume a geometric action.

Let us briefly consider a general vector field $v\in\Xi$.
The two consistency conditions (\ref{eqn:geometricaction}) and (\ref{eqn:consistencycom}) required 
for $v$ to be implementable have the following meaning: (\ref{eqn:consistencycom}) states that
$v$ has to be an infinitesimal symplectic automorphism and (\ref{eqn:geometricaction})
means that the equation of motion operator has to transform as $[\mathcal{L}_v,P]=P\circ \mathcal{O}_v$, with some
operator $\mathcal{O}_v$ mapping compactly supported functions to compactly supported functions. 
These two conditions are of course not fulfilled for the most general vector field $v\in\Xi$.
Since the twisted quantum field theory construction requires a representation of {\it all} vector fields entering the twist
on the algebra $\mathcal{A}_{(V,\omega)}$, the above argumentation shows that not all twists can be implemented. 
Our conjecture, which deserves a rigorous proof, is that a vector field is implementable if and only if it is Killing.

We now show for completeness that the twisted quantum field theory construction is possible if the twist is Killing.
The Lie algebra representation of $\mathfrak{K}$ on $\mathcal{A}_{(V,\omega)}$ extends to a Hopf algebra
representation of the universal enveloping algebra $U\mathfrak{K}$, equipped with the natural coproduct, counit and antipode.
The formal deformation quantization of the Hopf algebra $U\mathfrak{K}$ and algebra $\mathcal{A}_{(V,\omega)}$ 
by a Killing twist $\mathcal{F}\in (U\mathfrak{K}\otimes U\mathfrak{K})[[\lambda]]$ 
is then straightforward, extending the result for the Moyal-Weyl deformation of the Minkowski quantum field theory
\cite{Zahn:2006wt,Balachandran:2007vx,Aschieri:2007sq}.


\chapter{\label{eqn:lambda2green}$\mathcal{O}(\lambda^2)$ Green's operators for a noncommutative Minkowski spacetime}
In order to better understand the explicit construction of deformed Green's operators, see Theorem \ref{theo:greendef},
we focus on a simple example.
Consider the deformed Minkowski spacetime of Chapter \ref{chap:qftapp}, Section \ref{sec:exampleswaveop},
where the wave operator is given by
\begin{flalign} 
\widetilde{P}_\star = -\cosh\left(\frac{3\lambda}{2}\,i\partial_t\right)\circ \bigl(\partial_t^2 + M^2\bigr) 
+\cosh\left(\frac{5\lambda}{2}\,i\partial_t\right)\circ \bigtriangleup~.
\end{flalign}
Expanding up to order $\lambda^2$ we obtain
\begin{subequations}
\begin{flalign}
 \widetilde{P}_{(0)} &= P = -\partial_t^2 - M^2 +\bigtriangleup~,\\
 \widetilde{P}_{(1)} &= 0~,\\
 \widetilde{P}_{(2)} &= \frac{9}{8}\partial_t^2 \circ \bigl(\partial_t^2+M^2\bigr) - \frac{25}{8} \partial_t^2\circ\bigtriangleup
 = -\frac{9}{8} \partial_t^2\circ P - 2\,\partial_t^2\circ \bigtriangleup ~.
\end{flalign}
\end{subequations}
The corrections to the Green's operators are given by, see Theorem \ref{theo:greendef},
\begin{subequations}
\begin{flalign}
 \widetilde{\Delta}_{(1)\pm} &= 0~,\\
 \widetilde{\Delta}_{(2)\pm} &= -\Delta_\pm \circ \widetilde{P}_{(2)}\circ \Delta_{\pm} = \frac{9}{8}\Delta_\pm\circ\partial_t^2 + 
2 \,\Delta_\pm\circ\partial_t^2\circ \bigtriangleup\circ \Delta_\pm~,
\end{flalign}
\end{subequations}
where $\Delta_\pm$ are the Green's operators for $P$. 
The distribution kernel of the deformed Green's operators $\widetilde{\Delta}_{\star\pm}$ is thus given by
\begin{flalign}
 \widetilde{\Delta}_{\star\pm}(x,y) = \Delta_\pm(x,y) + \frac{9\lambda^2}{8} \partial_{t_y}^2\Delta_\pm(x,y)
+ 2\lambda^2 \int\limits_{\bbR^4} d^4z\,\Delta_\pm(x,z)\,\partial_{t_z}^2\bigtriangleup_z\Delta_\pm(z,y) +\mathcal{O}(\lambda^3)~,
\end{flalign}
where 
\begin{flalign}
 \Delta_\pm(x,y) = \lim_{\epsilon\to 0^+} \int\limits_{\bbR^4} \frac{d^4p}{(2\pi)^4}\,
\frac{e^{-ip(x-y)}}{ (p_0 \pm i\epsilon)^2 -\mathbf{p}^2 -M^2}~.
\end{flalign}
Note the translation invariance of all kernels and operators above.
The remaining integral can be explicitly evaluated and we obtain in momentum space
\begin{flalign}
 \nn \widetilde{\Delta}_{\star\pm}(x,0) &= \Delta_\pm(x,0) + \frac{9\lambda^2}{8} \partial_{t}^2\Delta_\pm(x,0)
\\
&\mp\lambda^2\Theta(\pm t)\int\limits_{\bbR^3} \frac{d^3p}{(2\pi)^3} e^{-i\mathbf{p} \mathbf{x}} \,
\mathbf{p}^2\left(t\,\cos(E_p t) +\frac{\sin(E_p t))}{E_p}\right) +\mathcal{O}(\lambda^3)~,
\end{flalign}
where $E_p = \sqrt{\mathbf{p}^2 +M^2}$ and $\Theta$ is the Heaviside step-function.
For a massless field $M^2=0$ we can also evaluate the last integral explicitly and find as expected that
$\widetilde{\Delta}_{\star\pm}$ has support in the forward/backward lightcone.
Since the resulting expression is quite long and not instructive we do not present it here.

Let us finish with one remark: Note that the deformed wave operator is not compactly deformed, however 
all integrals required for calculating $\widetilde{\Delta}_{\star\pm}$ to order $\lambda^2$ exist, 
even without an infrared regularization.
This shows that the assumption of compactly deformed wave operators made in Chapter
\ref{chap:qftdef} is sufficient to ensure the existence of the Green's operators, but not necessary in all cases.


\chapter{\label{app:diagram}Diagrammatic proof of the associativity of $\otimes_R$}
In this appendix we provide an alternative (diagrammatic) proof of (\ref{eqn:Rtensorass}), which
in our opinion is more intuitive than the corresponding computational proof.
Since $\tau_R$ and $\tau_R^{-1}$ are representations of the generators of the braid group
we can represent them as
\begin{flalign}
\tau_R\quad =\quad \parbox{10mm}{
\begin{tikzpicture}[ultra thick]
\draw (0,-2.5) -- (0,-3)    ;
\draw (1,-2.5) -- (1,-3)    ;
 \draw (0,-3) -- (1,-4)   ;
 \draw (1,-3) -- (0.6,-3.4);
 \draw (0.4,-3.6) -- (0,-4);
 \draw (0,-4)   -- (0,-4.5) ;
 \draw (1,-4)   -- (1,-4.5) ;
\end{tikzpicture}
} \qquad\,\quad~,\qquad
\tau_R^{-1}\quad =\quad \parbox{10mm}{
\begin{tikzpicture}[ultra thick]
 \draw (0,-0.5)    -- (0,-1) ;
 \draw (1,-0.5)    -- (1,-1) ;
 \draw (0,-1) -- (0.4,-1.4);
 \draw (0.6,-1.6)  -- (1,-2);
 \draw (1,-1) -- (0,-2)   ;
 \draw (0,-2)   -- (0,-2.5) ;
 \draw (1,-2)   -- (1,-2.5) ; 
\end{tikzpicture}
}~\,\quad.
\end{flalign}
From this notation it is obvious that $\tau_R^{-1}$ is the inverse of $\tau_R$
\begin{flalign}
\tau_R\circ\tau_R^{-1} \quad =\quad \parbox{10mm}{
\begin{tikzpicture}[ultra thick]
 \draw (0,-0.5)    -- (0,-1) ;
 \draw (1,-0.5)    -- (1,-1) ;
 \draw (0,-1) -- (0.4,-1.4);
 \draw (0.6,-1.6)  -- (1,-2);
 \draw (1,-1) -- (0,-2)   ;
 \draw (0,-2)   -- (0,-2.5) ;
 \draw (1,-2)   -- (1,-2.5) ; 
\draw (0,-2.5) -- (0,-3)    ;
\draw (1,-2.5) -- (1,-3)    ;
 \draw (0,-3) -- (1,-4)   ;
 \draw (1,-3) -- (0.6,-3.4);
 \draw (0.4,-3.6) -- (0,-4);
 \draw (0,-4)   -- (0,-4.5) ;
 \draw (1,-4)   -- (1,-4.5) ;
\end{tikzpicture} 
} \quad =\quad \parbox{10mm}{
\begin{tikzpicture}[ultra thick]
 \draw (0,-0.5)    -- (0,-4.5) ;
 \draw (1,-0.5)    -- (1,-4.5) ;
\end{tikzpicture} 
}\quad =\quad \id~,
\end{flalign}
simply by unbraiding. A similar diagram shows the relation $\tau_R^{-1}\circ \tau_R =\id$.
The braid relation $\tau_{R\,12}\circ\tau_{R\,23}\circ \tau_{R\,12} = \tau_{R\,23}\circ\tau_{R\,12}\circ\tau_{R\,23}$ on triple
tensor products  holds because of the properties of the quasitriangular $R$-matrix. They can be depicted as
\begin{flalign}
\tau_{R\,12}\circ\tau_{R\,23}\circ \tau_{R\,12} \quad= \quad \parbox{20mm}{
\begin{tikzpicture}[ultra thick]
 \draw (0,0.5) -- (0,0) ;
 \draw (1,0.5) -- (1,0) ;
 \draw (2,0.5) -- (2,0) ;
 \draw (0,0) -- (1,-1) ;
 \draw (1,0) -- (0.6,-0.4);
 \draw (0.4,-0.6) -- (0,-1);
\draw (2,0) -- (2,-1);
 \draw (0,-1) -- (0,-2) ;
 \draw (1,-1) -- (2,-2);
 \draw (2,-1) -- (1.6,-1.4);
 \draw (1.4,-1.6) -- (1,-2);
 \draw (0,-2) -- (1,-3);
 \draw (2,-2) -- (2,-3);
 \draw (1,-2) -- (0.6,-2.4);
 \draw (0.4,-2.6) -- (0,-3);
 \draw (0,-3) -- (0,-3.5);
\draw (1,-3) -- (1,-3.5);
\draw (2,-3) -- (2,-3.5);
\end{tikzpicture} 
}
 \quad \,\quad=\quad\,\quad \parbox{20mm}{
\begin{tikzpicture}[ultra thick]
 \draw (0,0.5) -- (0,0) ;
 \draw (1,0.5) -- (1,0) ;
 \draw (2,0.5) -- (2,0) ;
 \draw (1,0) -- (2,-1) ;
 \draw (2,0) -- (1.6,-0.4);
 \draw (1.4,-0.6) -- (1,-1);
\draw (0,0) -- (0,-1);
 \draw (0,-1) -- (1,-2) ;
 \draw (1,-1) -- (0.6,-1.4);
 \draw (0.4,-1.6) -- (0,-2);
 \draw (2,-1) -- (2,-2);
 \draw (1,-2) -- (2,-3);
 \draw (2,-2) -- (1.6,-2.4);
 \draw (1.4,-2.6) -- (1,-3);
 \draw (0,-2) -- (0,-3);
 \draw (0,-3) -- (0,-3.5);
\draw (1,-3) -- (1,-3.5);
\draw (2,-3) -- (2,-3.5);
\end{tikzpicture} 
}~\,\quad = \quad  \tau_{R\,23}\circ\tau_{R\,12}\circ\tau_{R\,23}~.
\end{flalign}
For the $R$-tensor product we introduce the following notation
\begin{flalign}
P\otimes_R Q\quad  = \quad (P\otimes\id)\circ \tau_R \circ (Q\otimes \id)\circ\tau_R^{-1}\quad=\quad \parbox{10mm}{\begin{tikzpicture}[ultra thick]
 \draw (0,-0.5)    -- (0,-1) ;
 \draw (1,-0.5)    -- (1,-1) ;
 \draw (0,-1) -- (0.4,-1.4);
 \draw (0.6,-1.6)  -- (1,-2);
 \draw (1,-1) -- (0,-2)   ;
 \draw (0,-2)   -- (0,-3) ;
 \draw (1,-2)   -- (1,-3) ; 
 \draw (0,-3) -- (1,-4)   ;
 \draw (1,-3) -- (0.6,-3.4);
 \draw (0.4,-3.6) -- (0,-4);
 \draw (0,-4)   -- (0,-5) ;
 \draw (1,-4)   -- (1,-5) ; 
\node at (0,-2.5) [rectangle,draw,fill] {};
\node at (0,-4.5) [rectangle,draw,fill] {};
\end{tikzpicture}}~\quad\,\quad,
\end{flalign}
where the rectangles denote the action of $P\otimes\id$ and $Q\otimes\id$.
On the triple tensor product we have the identity
\begin{flalign}
 \tau_{R\,23}\circ(P\otimes \id\otimes \id )\circ\tau_{R\,23}^{-1} \quad =\quad \parbox{25mm}{
\begin{tikzpicture}[ultra thick]
 \draw (0,-0.5)    -- (0,-1) ;
 \draw (1,-0.5)    -- (1,-1) ;
 \draw (0,-1) -- (0.4,-1.4);
 \draw (0.6,-1.6)  -- (1,-2);
 \draw (1,-1) -- (0,-2)   ;
 \draw (0,-2)   -- (0,-2.5) ;
 \draw (1,-2)   -- (1,-2.5) ; 
\draw (0,-2.5) -- (0,-3)    ;
\draw (1,-2.5) -- (1,-3)    ;
 \draw (0,-3) -- (1,-4)   ;
 \draw (1,-3) -- (0.6,-3.4);
 \draw (0.4,-3.6) -- (0,-4);
 \draw (0,-4)   -- (0,-4.5) ;
 \draw (1,-4)   -- (1,-4.5) ;
\draw (-1,-0.5) -- (-1,-4.5) ;
\node at (-1,-2.5) [rectangle,draw,fill] {};
\end{tikzpicture} 
} \quad =\quad \parbox{25mm}{
\begin{tikzpicture}[ultra thick]
 \draw (0,-0.5)    -- (0,-4.5) ;
 \draw (1,-0.5)    -- (1,-4.5) ;
\draw (-1,-0.5) -- (-1,-4.5) ;
\node at (-1,-2.5) [rectangle,draw,fill] {};
\end{tikzpicture} 
}\quad =\quad P\otimes \id\otimes\id~.
\end{flalign}
The associativity condition (\ref{eqn:Rtensorass}) reads diagrammatically
\begin{flalign}
 (P\otimes_R Q)\otimes_R T \quad =\quad \parbox{25mm}{
\begin{tikzpicture}[ultra thick]
 \draw (0,0.5) -- (0,0);
 \draw (1,0.5) -- (1,0);
 \draw (2,0.5) -- (2,0);
\draw (2,0) -- (1,-1);
\draw (1,0) -- (1.4,-0.4);
\draw (1.6,-0.6) -- (2,-1);
\draw (0,0) -- (0,-1);
\draw (0,-1) -- (0.4,-1.4);
\draw (0.6,-1.6) -- (1,-2);
\draw (1,-1) -- (0,-2);
\draw (2,-1) -- (2,-2);
\draw (0,-2) -- (0,-3);
\draw (1,-2) -- (1,-3);
\draw (2,-2) -- (2,-3);
\node at (0,-2.5) [rectangle,draw,fill] {};
\draw (0,-3) -- (1,-4) ;
\draw (1,-3) -- (0.6,-3.4);
\draw (0.4,-3.6) -- (0,-4);
\draw (2,-3) -- (2,-4);
\draw (0,-4) -- (0,-5);
\draw (1,-4) -- (2,-5);
\draw (2,-4) -- (1.6,-4.4);
\draw (1.4,-4.6) -- (1,-5);
\draw (0,-5) -- (0.4,-5.4);
\draw (0.6,-5.6) -- (1,-6);
\draw (1,-5) -- (0,-6);
\draw (2,-5)  -- (2,-6);
\draw (0,-6) -- (0,-7);
\draw (1,-6) -- (1,-7);
\draw (2,-6) -- (2,-7);
\node at (0,-6.5) [rectangle,draw,fill] {};
\draw (0,-7) -- (1,-8);
\draw (1,-7) -- (0.6,-7.4);
\draw (0.4,-7.6) -- (0,-8);
\draw (2,-7) -- (2,-8);
\draw (0,-8) -- (0,-9);
\draw (1,-8) -- (1,-9);
\draw (2,-8) -- (2,-9);
\node at (0,-8.5) [rectangle,draw,fill] {};
\end{tikzpicture}}
\quad = \quad 
\parbox{25mm}{
\begin{tikzpicture}[ultra thick]
  \draw (0,0.5) -- (0,0);
 \draw (1,0.5) -- (1,0);
 \draw (2,0.5) -- (2,0);
\draw (0,0) -- (0.4,-0.4);
\draw (0.6,-0.6) -- (1,-1);
\draw (1,0) -- (0,-1);
\draw (2,0) -- (2,-1);
\draw (0,-1) -- (0,-2);
\draw (1,-1) -- (1.4,-1.4);
\draw (1.6,-1.6)  -- (2,-2);
\draw (2,-1) -- (1,-2);
\draw (0,-2) -- (0.4,-2.4);
\draw (0.6,-2.6) -- (1,-3);
\draw (1,-2) -- (0,-3);
\draw (2,-2) -- (2,-3);
\draw (0,-3) -- (0,-4);
\draw (1,-3) -- (1,-4);
\draw (2,-3) -- (2,-4);
\node at (0,-3.5) [rectangle,draw,fill] {};
\draw (0,-4) -- (1,-5);
\draw (1,-4) -- (0.6,-4.4);
\draw (0.4,-4.6) -- (0,-5);
\draw (2,-4) -- (2,-5);
\draw (0,-5) -- (0,-6);
\draw (1,-5) -- (1,-6);
\draw (2,-5) -- (2,-6);
\node at (0,-5.5) [rectangle,draw,fill] {};
\draw (0,-6) -- (0,-7);
\draw (1,-6) -- (2,-7);
\draw (2,-6) -- (1.6,-6.4);
\draw (1.4,-6.6) -- (1,-7);
\draw (0,-7) -- (1,-8);
\draw (1,-7) -- (0.6,-7.4);
\draw (0.4,-7.6) -- (0,-8);
\draw (2,-7) -- (2,-8);
\draw (0,-8) -- (0,-9);
\draw (1,-8) -- (1,-9);
\draw (2,-8) -- (2,-9);
\node at (0,-8.5) [rectangle,draw,fill] {};
\end{tikzpicture}}
\quad = \quad P\otimes_R (Q\otimes_R T)~.
\end{flalign}
The proof follows by applying the identities above to these diagrams in order to
deform one into the other.


\chapter{\label{app:symbols}Symbol index for Part \ref{part:math}}
In the following we give a list of symbols frequently used in Part \ref{part:math}
with a short explanation and reference to the main text.
The basic notations of Chapter \ref{chap:prelim} and \ref{chap:HAdef} are not repeated.
\vspace{3mm}
\begin{flushleft}
\begin{tabular}{l|l|l}
$\FF=f^\alpha\otimes f_\alpha$ & twist of a Hopf algebra			& Definition \ref{defi:twist}\\
$\FF^{-1}=\bar f^\alpha\otimes \bar f_\alpha$ & inverse twist 					& \\
$R= R^\alpha\otimes R_\alpha$  & $R$-matrix of a Hopf algebra 			& Definition \ref{defi:rmatrix}\\
$R^{-1}= \bar R^\alpha\otimes \bar R_\alpha$  & inverse $R$-matrix			& \\
 $\End_\bfK(V)$ 	& $\bfK$-linear maps $V\to V$  				   &  Proposition \ref{propo:endomodulealgebra}\\
 $\End_A(V)$    	& right $A$-linear maps $V\to V$			   &  Proposition \ref{propo:endomodulealgebra}\\
 ${_A}\End(V)$  	& left $A$-linear maps $V\to V$  			   &  Proposition \ref{propo:opposite}\\
 $\Hom_\bfK(V,W)$ 	& $\bfK$-linear maps $V\to W$  				   &  Proposition \ref{propo:homomodule}\\
 $\Hom_A(V,W)$    	& right $A$-linear maps $V\to W$			   &  Proposition \ref{propo:homomodule}\\
 ${_A}\Hom(V,W)$  	& left $A$-linear maps $V\to W$  			   &  Chapter \ref{chap:modhom}, Section \ref{sec:homoquant}\\
 $\ra$			& Hopf algebra action					& Definitions \ref{defi:hmodalg} and \ref{defi:hmod}\\
$\xi\ra\,$ & endomorphism $v\mapsto \xi\ra v$  & Example \ref{ex:endoex}\\ 
 $\RA$, $\RA_\FF$, $\RA^\cop$	& adjoint $H$, $H^\FF$, $H^\cop$-action   &  (\ref{eqn:endomodulealgebra1}), (\ref{eqn:adjointhfaction}), (\ref{eqn:adjointhcopaction}) \\
$\circ$, $\circ_\star$	& usual and $\star$-composition & (\ref{eqn:starcompo})\\
 $D_\FF$		& quantization isomorphism 		& Theorems \ref{theo:endodef} and \ref{theo:homodef}\\
$D^\cop_\FF$		& left quantization isomorphism				& Theorems \ref{theo:leftendodef} and \ref{theo:lefthomodef}\\
$\otimes_R$ 		& $R$-tensor product of $\bfK$-linear maps		& Definition \ref{defi:Rtensor}\\
$\tau_R$, $\tau_R^{-1}$ & $R$-flip map 			& (\ref{eqn:Rflipmap})\\
$\iota$, $\iota_{123}$   & product module isomorphisms				& Lemma \ref{lem:iotaiso}, (\ref{eqn:higheriota})\\
$\bigl(\Omega^\bullet=\bigoplus_{n\geq0}\Omega^n,\wedge,\dd\bigr)$ & differential calculus & Chapter \ref{chap:con}, Section \ref{sec:conbas}\\
$\Con_A(V)$  		& right connections $V\to V\otimes_A\Omega^1$ & Definition \ref{defi:connection}\\
${_A}\Con(V)$  		& left connections $V\to \Omega^1\otimes_A V$ & Definition \ref{defi:connection}\\
$\widetilde{D}_\FF$	& quantization isomorphism (connections)      & Theorem  \ref{theo:condef}\\
$\widetilde{D}^\cop_\FF$& left quantization isomorphism (connections) & Theorem  \ref{theo:condefleft}\\
$\oplus_R$		& sum of connections			& Theorem \ref{theo:conplus}\\
$V^\prime=\Hom_A(V,A)$	& dual of the module $V$			& 
\end{tabular}
\end{flushleft}

\end{appendix}


\bibliographystyle{hep}
\bibliography{bibl-arxiv.bib}


\chapter*{Acknowledgements}
\noindent I would like to thank my friends, colleagues and collaborators for their support during my PhD thesis.
A particular thanks goes to:\vspace{2mm}
\begin{itemize}
\item Prof.~Dr.~Thorsten Ohl for the opportunity to write my thesis in his group, the many discussions and the fruitful collaborations.\vspace{2mm}
\item  Dr.~Paolo Aschieri, Dr.~Tim Koslowski and Christoph Uhlemann for many discussions and fruitful collaborations.\vspace{2mm}
\item  Prof.~Dr.~Christian Klingenberg and Prof.~Dr.~Werner Porod  for joining my ``Betreuungskommission''.\vspace{2mm}
\item Prof.~Dr.~Haye Hinrichsen and Prof.~Dr.~Peter Schupp for agreeing to be a referee for my thesis.\vspace{2mm}
\item All members of Theoretical Physics II and the Research Training Group GRK 1147 for the interesting and nice time.
In particular Julian Adamek for keeping alive my interest in cosmology.\vspace{2mm}
\item  Dr.~Paolo Aschieri, Prof.~Dr.~Dorothea Bahns, Dr.~Daniel Blaschke, Prof.~Dr.~Leonardo Castellani, Dr.~Claudio Dappiaggi,  
Dr.~Marija Dimitrijevi{\'c},
 Prof.~Dr.~Klaus Fredenhagen, Prof.~Dr.~Harald Grosse, Dr.~Gandalf Lechner, Prof.~Dr.~John Madore, Eric Morfa-Morales, 
 Prof.~Dr.~Denjoe O'Connor, Prof.~Dr. Jeong-Hyuck Park, Dr.~Christian S\"amann, Prof.~Dr.~Martin Schlichenmaier, 
 Paul Schreivogl, Prof.~Dr. Peter Schupp, PD~Dr.~Harold Steinacker, Dr.~Christoph Stephan, Prof.~Dr.~Richard Szabo, 
 Prof.~Dr. Dmitri Vassilevich, Prof.~Dr. Rainer Verch,
 Prof.~Dr.~Stefan Waldmann,  Dr.~Michael Wohlgenannt,
 Dr.~Hyun Seok Yang and Dr.~Jochen Zahn,
 not only for many discussions and helpful comments, but also for the enjoyable ``social events'' during conferences. \vspace{2mm}
 \item Quantum Geometry and Quantum Gravity Network of the European Science Foundation (ESF) for the financial support.\vspace{2mm}
 \item Perimeter Institute for Theoretical Physics for the invitation and the financial support.\vspace{2mm}
 \item Research Training Group 1147 \textit{Theoretical Astrophysics and Particle Physics} of the DFG for the financial support.\vspace{2mm} 
 \item Brigitte Wehner for helping me with all sorts of bureaucracy.\vspace{2mm}
 \item My parents, family and ``non-physics friends'' for the support and enjoyable free time.\vspace{2mm}
 \item And last but not least Asma Ch{\'e}rigu{\`e}ne for giving me her love, support and the valuable time together.\vspace{2mm}
\end{itemize}
Thanks!

\end{document}